\documentclass[aps,amsfonts,eqsecnum,superscriptaddress,nofootinbib,longbibliography,superscriptaddress,11pt]{revtex4-1}
\usepackage[letterpaper, left=1in, right=1in, top=1in, bottom=1in]{geometry}

\usepackage{cc}
\usepackage{comment}

\begin{document}
\title{Efficient Unitary \texorpdfstring{$T$}{T}-designs from Random Sums}
\date{February 14, 2024}

\author{Chi-Fang Chen*\footnote{*: these authors contributed equally to this work. 
}}
\email{achifchen@gmail.com}
\affiliation{Institute for Quantum Information and Matter, Caltech}
\author{Jordan Docter*}
\affiliation{Department of Computer Science, Stanford University}
\author{Michelle Xu*}
\affiliation{Stanford Institute for Theoretical Physics}
\affiliation{Department of Physics, Stanford University}
\author{Adam Bouland}
\affiliation{Department of Computer Science, Stanford University}
\author{Patrick Hayden}
\affiliation{Stanford Institute for Theoretical Physics}
\affiliation{Department of Physics, Stanford University}

\begin{abstract}
    Unitary $T$-designs play an important role in quantum information, with diverse applications in quantum algorithms, benchmarking, tomography, and communication. 
    Until now, the most efficient construction of unitary $T$-designs for $n$-qudit systems has been via random local quantum circuits, which have been shown to converge to approximate $T$-designs in the diamond norm using $\CO(T^{5+o(1)} n^2)$ quantum gates. In this work, we provide a new construction of $T$-designs via random matrix theory using $\tilde{O}(T^2 n^2)$ quantum gates. 
    Our construction leverages two key ideas. First, in the spirit of central limit theorems, we approximate the Gaussian Unitary Ensemble (GUE) by an i.i.d. \textit{sum} of random Hermitian matrices. Second, we show that the product of just two exponentiated GUE matrices is already approximately Haar random.
    Thus, multiplying two exponentiated sums over rather simple random matrices yields a unitary $T$-design, via Hamiltonian simulation. A central feature of our proof is a new connection between the polynomial method in quantum query complexity and the large-dimension ($N$) expansion in random matrix theory. 
In particular, we show that the polynomial method provides exponentially improved bounds on the high moments of certain random matrix ensembles, without requiring intricate Weingarten calculations. In doing so, we define and solve a new type of moment problem on the unit circle, asking whether a finite number of equally weighted points, corresponding to eigenvalues of unitary matrices, can reproduce a given set of moments.    
\end{abstract}

\maketitle

\setcounter{page}{1}

\tableofcontents

\newpage

\section{Introduction}
A unitary $T$-design is an ensemble of unitaries that reproduces the first $T$ moments of the Haar measure. 
It is the quantum analogue of $T$-wise independent functions or permutations. While Haar random unitaries are exponentially complex to implement, it is possible to implement
unitary $T$-designs in polynomial time, giving us efficient access to the low-degree properties of the Haar measure.
This is often sufficient for applications such as randomized benchmarking \cite{knill2008randomized,dankert2009exact}, communication~\cite{hayden2008decoupling,szehr2013decoupling}, phase retrieval~\cite{kimmel2017phase}, shadow tomography \cite{huang2020predicting} and cryptography~\cite{divincenzo2002quantum,ambainis2009nonmalleable}. We have now seen several constructions of unitary $T$-designs using a variety of methods, such as Clifford circuits \cite{webb2015clifford,zhu2017multiqubit}, random local circuits \cite{harrow2009random,harrow2009efficient, brandao2016local,harrow2023approximate,haferkamp2022random,hunter2019unitary}, and others \cite{dankert2009exact,cleve2015near,bannai2019explicit,o2023explicit,kaposi2023constructing,onorati2017mixing,nakata2017efficient}. 
Likewise, one may define state $T$-designs for Haar random states, which have been constructed using a variety of methods, e.g., \cite{ambainis2007quantum,brown2008quantum,nakata2014generating,mezher2018efficient,brakerski2019pseudo,cotler2017chaos}.

An active problem in this area is to construct efficient quantum algorithms for arbitrarily high moment unitary $T$-designs on $n$-qubit systems.
Most existing constructions yield approximate rather than exact $T$-designs (in various metrics; see \cite{low2010pseudo}).
In this work, we say an ensemble is an $\epsilon$-approximate $T$-design if the $T$-fold tensor product channel is $\epsilon$-close in diamond norm to the $T$-fold Haar channel (see Definition \ref{defn:U_parallel_design}) \cite{dankert2009exact}. That is, the output distributions on any quantum experiment performed on $T$-fold parallel copies of the unitaries are $\epsilon$-close to Haar in trace distance. 

Thus far, the only scalable algorithms to generate $T$-designs for large $T$ and $n$ are via random local circuits.
The seminal work of Brandao, Harrow, and Horodecki showed that random local circuits generate $\epsilon$-approximate $T$-designs if their depth scales as $\CO(T^{9}n ( nT+\log(\epsilon^{-1})) )$ \cite{brandao2016local}. 
Their analysis was subsequently improved by Haferkamp to show that random circuits form $T$ designs with depth $\CO(T^4 ( nT+ \log(\epsilon^{-1})) ) $, that is, using $\CO(nT^4 ( nT+ \log(\epsilon^{-1})) ) $ quantum gates  \cite{haferkamp2022random}, and also extended with weaker $T$ scaling to other random time evolutions, e.g. \cite{onorati2017mixing,harrow2023approximate,mittal2023local,allen2024approximate}. 
However, the only known lower bound is that the circuit size of any $T$-design must be $\Tilde{\Omega}(nT)$ by counting arguments \cite{brandao2016local}, and it remains a conjecture whether such an optimal scaling is achievable. Interest in the conjecture comes not only from the quantum computing community but also from reseachers in quantum chaos and quantum gravity, where rapid complexity growth and out-of-time-order correlators have been related to the $T$-design property~\cite{hayden2007black,sekino2008fast,lashkari2013towards,maldacena2016bound,roberts2017chaos,cotler2020spectral,brandao2021models}.
Thus far such linear scaling in $T$ has only been achieved in relaxed settings, such as asymptotically growing local dimension \cite{haferkamp2021improved}\footnote{That is, the local dimension increases with the desired number of matched moments $T$.} or by only considering low moments \cite{nakata2017efficient}.

\subsection{Our Results}
In this work, we describe a new algorithm for constructing approximate unitary $T$-designs on qudits via random matrix theory. Our new construction produces an $\epsilon$-approximate $T$-design in $\Tilde{O}(T^2n^2\log(\epsilon^{-1}))$ time\footnote{We measure runtime in the usual model, i.e. the number of two-qubit quantum gates without geometric locality constraints. }, improving the prior runtime which was quintic in $T$.
\begin{thm}[Informal] For any $T\leq 2^{\tilde{O}(n/\log n)}$, there exists an efficient quantum algorithm to generate an $\epsilon$-approximate unitary $T$-design (in diamond norm) using $\Tilde{O}\left(T^2 n^2 \log(\epsilon^{-1})\right)$ quantum gates. 
\end{thm}
Additionally, our construction only requires $\Tilde{O}(Tn^2)$ bits of classical randomness, in a spirit similar to recent results of \cite{o2023explicit}.
Our algorithm is fundamentally different than existing approaches for generating unitary $T$-designs~\cite{brandao2016local,haferkamp2022random,haferkamp2021improved}. In particular, instead of using random quantum circuits, we exponentiate Hamiltonians which are sums of independently random terms. Our proof uses tools from random matrix theory to create efficient unitary designs.

\subsection{Overview of the Construction}

As previously mentioned, $T$-designs have largely been generated by random local circuits: iteratively multiply some small set of unitaries, such as local gates, and prove that the associated random walk is ergodic and mixes rapidly through an intricate spectral gap calculation.

Rather than employ the mixing properties of \emph{products} of random matrices to achieve convergence, we aggregate randomness from \emph{sums} of random matrices. In particular, we focus on the Gaussian Unitary Ensemble (GUE,~Hermitian matrices with i.i.d.~Gaussian entries) as the stepping stone for Haar random unitaries. In a nutshell, our main construction is the following: 
\begin{enumerate}[label=(\roman*)]
    \item Consider a \textit{finite} sum over i.i.d.~Hermitian matrices $H_j$ which match the first $q$-moments of the GUE
\begin{align*}
    \vH \sim \frac{1}{\sqrt{m}}\displaystyle\sum_{j=1}^{m} \vH_j \quad \text{where}\quad \BE[\vH_j^{\otimes k}] \approx\BE[\vG^{\otimes k}]\quad \text{for each}\quad k = 1,\ldots,q.
\end{align*}
    \item Then, the product of two exponentiations of such matrices is approximately Haar:
\begin{align}
    \e^{\ri \theta \vH} \e^{\ri \theta \vH'} ~\stackrel{(1)}{\approx} ~\underbrace{\e^{\ri \theta \vG} \e^{\ri \theta \vG'}}_{=:\vW_{GUE}} ~\stackrel{(2)}{\approx}~ \vU_{Haar}\label{eq:intro_hybrid}.
\end{align}
\end{enumerate}

The first approximation $(1)$ results from a matrix central limit theorem: sums of random matrices matching low moments converge to the GUE $\vG$. Remarkably, random matrices as simple as signed random Pauli strings $\vsigma$, which match the second moments of GUE ($\BE [\vsigma\otimes \vsigma ] = \BE [\vG\otimes \vG ]$), can reproduce the very complex GUE.  This elementary approach circumvents the spectral gap calculations key to prior works \cite{brandao2016local,haferkamp2021improved,haferkamp2022random}.
More precisely, we significantly generalize a matrix Lindeberg principle of ~\cite{chen2023sparse} to control the rate of convergence on $m$ with a suitable choice of $q$ with respect to a very different distance; see Section~\ref{sec:intro_IIDsum} for more details. 

The second approximation $(2)$ is a highly nontrivial conversion from the GUE to Haar random unitaries. Although Gaussian matrices have a Haar-random basis (i.e. they're unitarily invariant under conjugation), they have spectra very different from Haar---roughly, the eigenvalues of GUE matrices are semicircle-distributed while the eigenvalues of Haar-random unitaries are evenly distributed around the unit circle.
This difference remains even after we exponentiate to get unitarity. Exponentiated Gaussians $e^{i\vG\theta}$ are quite far from Haar-random, again because the eigenvalues are not evenly distributed around the unit circle.
This difference between an $N \times N$ Haar unitary $\vU$ and an exponentiated GUE $\vG$ is exhibited for example by their trace moments:
\begin{align}
    \frac{1}{N}\tr[\e^{\ri \vG \theta p}]\not\approx \frac{1}{N}\tr[\vU^p]  \approx 0 \label{eq:haar_moment_zero}.
\end{align}
 
Indeed prior work of Cotler et al. \cite{cotler2017chaos} studied if exponentiated GUE matrices are unitary designs, and showed they become $t$-designs at very late times ($\theta=O(\sqrt{N})$), but one might not expect exponentiated GUE matrices to be unitary designs at short ($\poly(n)$) times due to these differences in spectra. 

Surprisingly, we show that the product of two Gaussian exponentials $\vW_{GUE}= \e^{\ri \theta \vG} \e^{\ri \theta \vG'}$ \emph{is close} to a Haar random unitary at particularly chosen values of the numerical \textit{constant} $\theta$ (independent of the system size).  That is, the product $\vW_{GUE}$ is an $\epsilon$-approximate unitary $T$-design for a very high value of $T$ ($T=2^{\Omega(n)}$) and a small value of $\epsilon$.
As we will describe shortly, this requires both developing a novel query complexity lower bound specific to distinguishing ensembles with symmetries, as well as new random matrix theory results for bounding moments of the ensemble $\vW_{GUE}$ using the polynomial method. We describe how this works in section \ref{intro:gaussiansaretdesigns}.
This then establishes that exponentiated Gaussians are already a $T$-design, and therefore our exponentiated sums of random Hamiltonians are close to a $T$-design. 

We recently became aware of the independent related work of Haah, Liu, and Tan achieving similar dependence on $T$ via a different construction \cite{haahpersonal}.

\section{Proof Sketch}

\subsection{Unitary \texorpdfstring{$T$}{T}-designs from Gaussians}
\label{intro:gaussiansaretdesigns}

The starting point of our proof is to show that the product of two Gaussian exponentials $ \vW_{GUE}=e^{\ri \theta \vG} \e^{\ri \theta \vG'}$ is a unitary $T$-design for a high value of $T$. In fact, we show something slightly stronger: the ensemble is query-indistinguishable from Haar.
That is, no quantum query algorithm can distinguish these matrices $\vW_{GUE}$ from Haar random matrices without taking at least $N^{\Omega(1)}$ queries to the matrix.
This implies it is a $T$-design for $T=N^{\Omega(1)}$ to a small constant error, as a unitary $T$-design is the special case of only allowing parallel queries to the matrix, whereas a general query algorithm might make serial or adaptive queries. 

Mathematically, the problem reduces to the query complexity for distinguishing highly symmetrical random unitaries since the two cases we wish to distinguish are both unitarily invariant under conjugation. Indeed, a Haar random unitary is conjugate-invariant by definition, and the invariance of $\vW_{GUE}$ follows from that of the GUE. A simple symmetrization argument implies that any quantum query algorithm distinguishing these ensembles is invariant under applying a random change of basis to the queried unitary, without loss of generality. Thus, the algorithm's acceptance probability must be a low-degree polynomial of the low trace moments of the input unitaries---i.e., the values of $ \frac{1}{N}\tr[\vU^p]$ where $\vU$ is either drawn from $\vW_{GUE}$ or Haar. 

This query bound is shown by first demonstrating that the moments of $\vW_{GUE}$ are close to Haar, and then arguing that they are indistinguishable via the polynomial method.

\subsubsection{Bounding the Moments of Our Ensemble}

The first step is to show that the trace moments of the ensemble $\vW_{GUE}$ nearly match those of the Haar measure---i.e. the trace moments $\Expect\frac{1}{N}\Tr[\vW_{GUE}^p]$ are close to those of Haar random matrices, up to very high moments $p$ (namely $p=2^{\Omega(n)}$). 
This is the most technically difficult part of our proof.
Intuitively, we motivate this with an observation from random matrix theory (Free probability~\cite{voiculescu2016free}). For each fixed integer $p \ge 1$ and independent unitary conjugate-invariant matrices $\vA_1, \vA_2$, it is known that 
\begin{align*}
    \lim_{N\rightarrow \infty} \frac{1}{N}\tr[ (\vA_1\vA_2)^p] = 0 \quad \text{if} \quad \lim_{N\rightarrow \infty} \frac{1}{N}\tr[\vA_i] = 0 \quad \text{and}\quad \norm{\vA_i} =O(1). 
\end{align*}

In other words, in the large $N$-limit, the trace moments of $\vW_{GUE}$ converge to their Haar counterpart, that is $\frac{1}{N}\tr[ (\vW_{GUE})^p] \rightarrow \frac{1}{N}\tr[ (\vU_{Haar})^p]$, so long as $\theta$ is precisely chosen such that $\frac{1}{N}\tr[\e^{\ri \vG \theta}] \approx 0$. 
However, this only implies that the trace moments of our ensemble converge to $0$ as $N\rightarrow \infty$, but fails to bound the rate of the convergence. This is not sufficient for our setting, where we care about how large they are at finite $N$. As such, we need a way of bounding the finite-dimensional corrections to the infinite-dimensional limit of the trace moments. 

The standard way of computing moments of finite-dimensional matrices with appropriate unitary structure is by a difficult Weingarten calculation.
The Weingarten calculus prescribes a way of integrating over the $N$-dimensional Haar measure, converting the problem if integration into computing certain trace moments of Haar tensor invariants~\cite{collins2006integration}. As the GUE ensemble is invariant under a Haar-random change of basis, we can write 
\[\frac{1}{N}\Tr[\vW_{GUE}^p] = \mathbb E_{\vU,\vV} \frac{1}{N}\Tr\left[\left(\vU \vD_1 \vU^\dagger \vV \vD_2 \vV^\dagger\right)^p\right]\] where $\vU,\vV$ are Haar random and $\vD_1,\vD_2$ are diagonal matrices with the spectra of $\e^{i \vG_1 \theta}, \e^{i \vG_2 \theta}$, approximately the exponentiated semicircle. 
In principle there is a Weingarten calculation one can perform to compute these moments in finite dimensions. 
We carefully choose $\theta$ such that $e^{i\theta \vG}$ is traceless in expectation, which simplifies the calculation and ensures the moments tend to $0$ as $N\rightarrow \infty$. Even with this special value of $\theta$, computing this quantity is quite complicated due to the enormous number of Weingarten terms contributing at leading order as we increase the moment $p$. Specifically, the number of terms scales with $p!$, so applying the triangle inequality yields naive upper bounds of $p!/N$. As the normalized trace moments of unitaries are always less than one, this bound is useless for high $p$. We therefore expect large cancellations between the Weingarten terms, but determining the structure of the cancellation requires permutation analysis too intricate to be calculated for large $p$.

To prove a better bound on these moments, we develop a polynomial method argument, described in section \ref{intro:largeN} and which may be of independent interest, which both drastically simplifies the calculation for high $p$, and exponentially improves the moment bounds from $\text{exp}(p)$ to $\text{poly}(p)$.
Providing a systematic approach to delivering such finite $N$ estimates is one of the main technical and conceptual contributions of this work, which is highlighted in Section~\ref{intro:largeN}.

\subsubsection{Small Moments Imply Query Indistinguishability from Haar}\label{sec:intro_small_moments_haar}

To complete the proof, we invoke the polynomial method again, this time in a way more familiar to those well-versed in quantum query complexity. The polynomial method uses the following fact: efficient few-query quantum algorithms for computing a function imply the existence of a low-degree polynomial approximating the function \cite{beals2001quantum}. Therefore, lower bounds from approximation theory yield quantum query lower bounds. 

Recall that we have shown that any $T$-query algorithm distinguishing the two cases depends only on the first $T$ moments of the measures, and we have shown that these moments are extremely close in the two cases. For a contradiction, a successful distinguishing algorithm's acceptance probability must be very sensitive to these small changes in moments -- it must ``jump'' in a very small range. 

To show this is not possible, we create a univariate family of valid random unitaries $\vU(x)$ with moments extrapolating from the values of the moments of the two cases (Haar vs $\vW_{GUE}$)---in particular, the moment vector is a polynomial $q$ of a single parameter $x$.
This means the acceptance probability of the algorithm run on this family of instances must be a low-degree univariate polynomial in $x$. We show this polynomial must be bounded in a large range, but must jump in a narrow range between Haar vs $\vW_{GUE}$, which allows us to lower bound the degree and obtain a contradiction. The construction of this univariate polynomial is intricate. 
The main challenge is to show that there are many valid instances $\vU(x)$ to the problem along the interpolation path so that the polynomial is bounded at many points in the interval. We discuss this in detail in Section \ref{intro:moment}, which is devoted to solving a new variant of the so-called moment problem on the unit circle. This ends up showing that to distinguish these matrices $\vW_{GUE}$ from Haar-random matrices would require $T=N^{\Omega(1)}$---i.e. these matrices $\vW_{GUE}$ form a $2^{\Omega(n)}$-design when the matrices $\vG,\vG'$ are from the true GUE ensemble.

We note this lemma---that small moments imply indistinguishability from the Haar---is applicable to any ensemble of matrices which is invariant under a change of basis. For example, it immediately implies that an $n$-qubit Haar random unitary is query-indistinguishable from the square of a Haar random unitary,\footnote{We note this would also apply to higher power of random unitaries as well.} without exponentially many queries in $n$:
\begin{cor}
    Any quantum algorithm distinguishing if an $n$-qubit unitary $U$ is either a) Haar random, or b) the square of a Haar-random unitary, with constant probability, requires $2^{\Omega(n)}$ queries to $U$.
\end{cor}
This follows immediately by noting the typical $k$th moment of a Haar random unitary is $\CO(k/N)$.
In fact, in our proof, we end up applying it to a different ensemble of matrices which interpolates between our construction and exponentiated Gaussians. We hope this may be of use in other applications.

\subsection{Approximating GUE by an Independent Sum}\label{sec:intro_IIDsum}

Thus far we have established that $\vW_{GUE}$ is a $T$-design for high $T$ and a tiny value of $\epsilon$. 
However, $\vW_{GUE}$ is not an ensemble we can access directly, as it involves exponentiating the true GUE ensemble on $2^n\times 2^n$ matrices---which requires exponential circuit depth by counting arguments. The remaining step is to invoke our independent sum as a ``poor man's GUE'' to substitute for the GUE ensemble.
That is,  recall our construction is to define two Hamiltonians $\vH,\vH'$ which are each sum of $m$ independent random matrices $\vH_i$, 
 \[\vH=\frac{1}{\sqrt{m}} \sum_{i} \vH_i\]
 and each $\vH_i$ (approximately) matches the first $q$ moments of GUE.
 We wish to show that Hamiltonian simulating $\vH,\vH'$ for time $\theta$ yields something close to exponentiated GUE matrices:
\begin{align}
    \vW_m'=e^{\ri\theta \vH} e^{i \theta \vH'} \approx \vW = e^{i\theta \vG} e^{i \theta \vG'}\label{eq:WW'}.
\end{align}
Therefore to complete the proof, we need to quantitatively control the approximation error depending on the free parameters: the number of terms in the sum $m$ and the quality of the individual terms $q$ to begin with. Simultaneously, we need our Hamiltonians $\vH,\vH'$ to be efficiently Hamiltonian simulatable to minimize the gate complexity, which restricts us to very small values of $q\ll T$. The natural way to obtain a Gaussian $q$-design is to separately match its spectrum and the basis: 
\[\vH_j = \vU_j \vD \vU_j^\dagger\]
where $\vD$ is a diagonal matrix matching the first $q$ moments of GUE, and 
where each $\vU_j$ is drawn from a unitary $q$-design.
Crucially, this method uses a unitary $q$-design construction as a subroutine! As the current best $q$-design construction runs in $\CO(q^5)$ time \cite{haferkamp2022random}, we must set $q\ll T$ if we are going to bootstrap into the prior construction as part of our algorithm, and still achieve an improved $T$ scaling in the overall construction. 

As a result, the individual Hamiltonian terms $\vH_j$ will only match small moments of the GUE.
Therefore, in order for $\vH$ to converge to Gaussian in $T$ moments, we will need to take large sums of terms (i.e., $m$ will be large) to make up for the fact the individual terms are ``low-quality,'' in the spirit of central limit theorems.

\subsubsection{A Lindeberg Argument for Convergence to GUE}

In order to complete our construction, we need to quantify the rate at which our ensembles converge to GUE under the summing of additional terms.
In other words, we need to prove a \textit{non-asymptotic} central limit theorem. The related recent work of Chen, Dalzell, Berta, Brandao, and Tropp \cite{chen2023sparse} considered Hamiltonians as random sums of signed Pauli operators. These random Paulis match the first $q=3$ tensored moments (especially the second moment $\BE \vsigma\otimes \vsigma = \BE \vG \otimes \vG$) of GUE. They showed that the $p$-th trace moments of these Hamiltonians are $O\left(\frac{\poly(p)}{\poly(m)}\right)$-close to each other. In other words, increasing the number of summed Pauli terms $m$ does decrease the trace moments towards those of GUE. Furthermore, starting with higher matching moments $q> 3$ implies a better convergence rate. Their results were obtained by a Lindeberg exchange argument: a hybrid argument where one exchanges the individual terms of the Hamiltonian one by one with GUE, and bounds the distance moved by expressing the matrix exponential as a Taylor series.

For our application, however, we need to generalize the prior result significantly by showing convergence to GUE in the much stronger norm (diamond norm of the $T$-fold channel of $\e^{\ri \vH\theta}$).\footnote{We emphasize that the i.i.d. sum is not invariant under unitary conjugation; our prior arguments that trace norm bounds can drive query indistinguishability (Section \ref{sec:intro_small_moments_haar}) do not apply. The Lindeberg argument essentially implies that adding up more low-quality matrices leads to a more random basis.}
Additionally, we need to handle the presence of errors in the moments.
This is because the first $q$ moments do not match exactly (as assumed in \cite{chen2023sparse}) for two reasons. First, there are errors in our implementation of the random basis choice since the starting $q$-design is approximate. Second, our diagonal matrices $\vD$ emulate the low-moment properties of the semicircle distribution. But the semicircle is merely an approximation to the true spectral distribution of finite $N$ GUE, which has tiny probabilistic fluctuations and non-asymptotic corrections to the semicircle.

Therefore, to complete our result, we performed a Lindeberg exchange argument, which accounts for these errors in the Gaussian design. 
The main challenge is to quantify how these (generally noncommutative) errors propagate through the matrix exponential. 
In particular, the errors appear in an interleaved manner, stemming from commutation relations that are difficult to analyze with current techniques. 
To handle this, we demonstrate convergence in two steps, treating the spectrum (\autoref{sec:lindeberg_spectrum}) and basis (\autoref{sec:lindeberg_basis}) separately. Altogether, our Lindeberg argument bounds the diamond norm distance between our construction and the ensembles $\vW_{GUE}$ of exponentiated GUE matrices. In particular, we show roughly that
\[ \dnorm{\vW_m'^{\otimes T} - \vW_{GUE}{^{\otimes T}}} \leq O\left( m\left(\frac{T}{O(\sqrt{m})}\right)^{O(q)}\right) + \text{subleading terms}\]
where we have simplified the expression to show asymptotics to leading order---see Theorem \ref{lem:clt_basis} for the complete statement.

Since $\vW^{\otimes T}$ is already close to the $T$-fold Haar measure, this distance is essentially our distance to Haar as well.
Therefore, setting $m=O(T^2)$ and $q=O(\log T)$ suffices to ensure that we are an approximate $T$-design with constant additive error.

\subsection{Implementation via the QSVT}\label{sec:intro_QSVT_Hi}

We also show how to efficiently simulate the Hamiltonian (see Section \ref{sec:construction})
\[\vH=\sum_{i=1}^m \frac{1}{\sqrt{m}}\vU_i \vD \vU_i^\dagger,\]
since $e^{i\vH\theta}$ is our ``poor man's substitute'' for $e^{i\vG\theta}$. Recall here the $\vU_i$ are from a $q$-design and $\vD$ is fixed Hamiltonian matching the semi-circular moments, while the Hamiltonian simulation time is $\theta=O(1)$.
We show this is possible in $\Tilde{O}(m\log(1/\epsilon)n^2\text{poly}(q))$ time using the quantum singular value transform (QSVT) method \cite{gilyen2019quantum}. As $m=O(T^2)$ and $q=O(\log T)$ suffices for a $T$-design by our Lindeberg argument, this yields the overall runtime of $\Tilde{O}(T^2n^2)$ time as desired.

Achieving this requires some work, as an off-the-shelf implementation via QSVT would naively require $\CO(m^{3/2})=O(T^3)$ time.
At a high level, the QSVT runtime is given by the product of two terms: first, how much the norm of the block-encoding of the Hamiltonian is suppressed. This factor is $\CO(m^{1/2})$ and comes from the $\ell_1$ norm of the coefficients of the sum of the Hamiltonian terms.
The second factor is the number of gates needed to perform the block-encoding of the individual Hamiltonian terms.
This is the sum of two parts: the cost of implementing the diagonal matrix $\vD$, and the cost of implementing the $\vU_i$.
The diagonal implementation turns out to be easy: $\vD$ is a diagonal (hence, row-sparse) matrix with easy to compute entries, and one can directly apply existing lemmas showing how to efficiently block encode known functions (e.g.~\cite{gilyen2019quantum} Lemma 48) to block encode $\vD$ in time $n\log(1/\epsilon)$.
The more tricky part is the block encoding of the $\vU_i$---i.e.,~applying the ``Select'' unitary, which given $i$ on the first register, applies the unitary $\vU_i$ to the second register.
Directly doing this requires $\Omega(m)$ time---as each $\vU_i$ is a uniformly random element of a $q$-design, and there are $m$ of them, this is encoding a string of $m$ independent random variables, which requires $\Omega(m)$ time by a counting argument. 
Applying QSVT in this manner yields a runtime of $\CO(m^{3/2})=O(T^3)$.

We show that we can perform a different ``Select'' operation, equally good for our end result, that instead takes $\Tilde{O}(T)=\Tilde{O}(m^{1/2})$ time.
The intuition is that we don't need the $m=O(T^2)$ different $\vU_i$ to all be truly independent---we only ever care about the $T$-th moments of the final operator $e^{i\vH\theta}$ applied, since we are constructing a unitary $T$-design.
The QSVT Hamiltonian simulation algorithm \cite{gilyen2019quantum} (following \cite{berry2015simulating}) approximates this operator $e^{i\vH\theta}$ with a low-degree polynomial\footnote{In particular, the degree is logarithmic in the desired error, which we set to $\CO(T^{-1})$.} in the Hamiltonian $\vH$.
This means that the $T$-th moments of the Hamiltonian simulation of $e^{i\vH \theta}$ only ever involves products of $\Tilde{O}(T)$ different $\vU_i$. 
Therefore, an $\Tilde{O}(T)$-wise independent implementation of the $m=O(T^2)$ different $\vU_i$ suffices.
We show it is possible to do this with $\Tilde{O}(T)$ gates by simply using a classical $T$-wise independent function to select the random gates for the different $\vU_i$. 

For this reason, $\Tilde{O}(T)$ bits of classical randomness suffice for our construction.
This nearly matches the recent construction of O'Donnell, Servedio and Parades \cite{o2023explicit} who achieved $\CO(T)$ classical bits, which is optimal.\footnote{However, note that unlike \cite{o2023explicit}, our random bit scaling with $n,\log(\varepsilon^{-1})$ is not optimal.}
Essentially, the reason we can achieve this is that our algorithm reduces the problem of producing a $T$-design to that of producing $\CO(T^2)$ copies of a $\CO(\log T)$-design, which are merely $\Tilde{O}(T)$-wise independent. 
The fact we can use such a small value of $q=O(\log T)$ is essential to this result.

\subsubsection{Boosting the Error Dependence}
Altogether our work thus far results in a $T$-design in $\CO(T^2)$ time, but the design is relatively ``low-quality.'' This is primarily because our error scaling is poor: if $q=O(1)$, then as we increase the number of Hamiltonian terms $m$, our design quality only improves at $\CO(\text{poly}(m^{-1}))$. Moreover, our entire construction is converging to $\vW_{GUE}$, which itself is not a perfect $T$-design. This implies there exists a inverse-polynomial ``noise floor,'' which we can approach but never pass via our arguments thus far. 

Therefore, our final step is to bootstrap our construction to improve the design quality. In particular, we start by only setting our parameters to achieve a constant error (in diamond norm) to the $T$-fold Haar channel. 
We then show that repeating this construction, multiple times and independently in series, exponentially decreases the diamond norm error of our design (see~\autoref{subsec:boosting}). This improves our scaling to $\Tilde{O}\left(T^2\log(\epsilon^{-1})\right)$ for creating an $\epsilon$-approximate additive $T$-design, which completes the construction. Interestingly, to achieve high precision, we never introduce the idea of the spectral gap (related to 2-norms) but go directly to the diamond norm (a 1-norm quantity).

\subsection{Application of the Polynomial Method to Random Matrix Theory}
\label{intro:largeN}

A central idea in our proofs is a new connection between two established ideas from very different fields: the polynomial method in theoretical computer science and the large-$N$ expansion in random matrix theory. The polynomial method is a standard proof technique for quantum query lower bounds \cite{beals2001quantum}.
As previously mentioned, the main idea is that an efficient quantum algorithm for computing a function $f$ implies the existence of a low-degree polynomial approximating $f$; therefore, lower bounds on the approximate degree imply lower bounds on their quantum query complexity. 
Such proofs have found numerous applications in the study of quantum query complexity, e.g. \cite{aaronson2004quantum,barnum2001quantum,nayak1999quantum,kutin2005quantum,razborov2003quantum}. 
In this work, we adapt this method to our proof that matrix ensembles invariant under unitary conjugation are indistinguishable from Haar by any $T$-query quantum algorithm if their moments are small enough.

Most importantly, our work applies the polynomial method in a completely different context, namely to bound the large-$N$ corrections to random matrix theory quantities. In particular, we show that the moments of $\vW = \e^{\ri \theta \vG} \e^{\ri \theta \vG'}$ are very close to the Haar value
\begin{align*}
    \labs{\BE \tr[\vW^p]} \le  O\bigg(\frac{\text{poly}(p)}{\text{poly}(N)}\bigg) \approx 0.
\end{align*}
We are not aware of existing non-asymptotic bounds of this type; what is well-known from free probability is the infinite dimension limit $N\rightarrow \infty$,
\begin{align*}
    \lim_{N\rightarrow \infty} \labs{\BE \tr[\vW^p]} = 0,
\end{align*}
which is not useful if we care about how these quantities scale as a function of both $N$ and $p$.

In principle, such quantities can be computed exactly by Weingarten calculus and diagrammatic expansion (the GUE matrices are diagonal in Haar random bases), but it often requires a difficult, tedious, and nonsystematic calculation.
The number of diagrams in the Weingarten calculation increases \emph{factorially} in the number of moments $p$ one considers---even at leading order in $1/N$---and nontrivial cancellations in the Weingarten series are difficult to control. Simply applying the triangle inequality to the Weingarten series yields very loose upper bounds, like $\CO(\frac{p!}{N})$ as they cannot capture these intricate cancellations. Fortunately, in many cases, we don't need the exact value of these quantities but rather just an upper bound to the magnitude of the correction: for example, does it scale like $\CO(p/N)$, or $\CO(2^p/N^2)$?

Our key observation is that 
\begin{center}
    \emph{Low-moment properties of sufficiently nice random matrices are \\ low-degree polynomial (or rational) functions in $\frac{1}{N}$}.
\end{center}
For the right ensemble, random matrix theory arguments can tell us that the value of the moments at $N=\infty$ is $0$, so the functions above vanish at the point $1/N=0$.
The polynomial method can then be used to bound the value of the moments in finite dimensions.
The idea is to use variants of the Markov brothers' inequality, which bounds the amount a low-degree polynomial (or rational function) can ``jump'' in a small interval while remaining bounded in a larger interval. In particular, these inequalities allow us to bound the value of the random matrix quantities at $1/N$, so long as we know their value is $0$ at $1/N=0$ by free probability, and so long as we can show the polynomial is bounded in a large range.\footnote{Of course, applying these inequalities to rational functions requires showing you are sufficiently far from the poles; See \autoref{sec:products_of_gaussians} for further details.} See \autoref{sec:products_of_gaussians} for further details.

This method can be used to prove much tighter upper bounds on moments than Weingarten calculus alone: it yields bounds of the form $O\left(\frac{\text{poly}(p)}{\text{poly}(N)}\right)$, an \emph{exponential improvement} over the naive bound in $p$ dependence,\footnote{Even the naive bound isn't trivial to prove. See \autoref{lem:weingarten_asymptotics}.} and which captures the exponential cancellations occurring in the Weingarten sum. 
This observation drastically simplifies many of our moment calculations and circumvents hard Weingarten calculations, instead replacing them with (more tractable) challenges in applying the polynomial method.
Indeed, in many ways this method fits very nicely with random matrix theory: the Weingarten calculus yields expressions that are \emph{rational} polynomials of $\frac{1}{N}$ (such as $\sim \frac{1}{N(N-1)}$); moments of unitary matrices are always bounded; and the value at $\frac{1}{N}=0$ is precisely given by the leading-order expansion in random matrix theory $N=\infty$ which is much easier to calculate; often, off-the-shelf asymptotic results already tell us the answer.

\subsubsection{Challenges in Applying the Polynomial Method to Random Matrix Theory}

While this high-level idea is simple, applying the idea is more involved than expected.
The basic reason is that, although the Weingarten calculus naturally expresses random matrix quantities as rational functions in $1/N$ (as the denominators are terms like $N(N-1)(N-2)\ldots$), the coefficients of the Weingarten series are not constants.
In other words, the moments of the random matrix ensemble do not have closed-form expressions as low-degree rational functions.
Rather, the coefficients are typically some complicated functions of the \emph{trace moments} of the random matrix input.
If one divides the choice of random matrix into its spectrum and its basis (Haar in this case), then the Weingarten calculus provides a low-degree rational function for the moments, but \emph{only if the diagonal component is fixed}.
Therefore, the low-degree rational function that arises from Weingarten naturally requires a fixed spectrum and a random basis; as random matrices also have random spectra, this destroys the structure.

To circumvent this challenge, our key idea is to artificially keep the spectrum of the random matrices constant. That is, suppose our task is to bound the typical value of some trace moment of a random matrix from the ensemble $R$ in dimension $N$.
Suppose we draw $\vA$ from $R$.
Our method, then, is to produce a series of diagonal matrices $\vD_i$ which match the spectrum of $\vA$---in particular, they match the low-order trace moments of $\vA$---and define some new random matrix ensemble $R'$ by saying the ensemble in dimension $i$ is given by $\vU \vD_i \vU^\dagger$ for $\vU$ Haar random.
Crucially, this new random ensemble is custom tailored so that the coefficients in the Weingarten expansion are \emph{fixed} across the change of dimensions $N$. Therefore the trace moments of this particular ensemble of matrices are indeed a low-degree rational function of $1/N$ as desired---which allows us to apply the polynomial method. The boundedness of the polynomial in dimensions $N'\neq N$ follows simply from the fact that the elements of the ensemble $R'$ are valid unitaries at each finite dimension.

The task of applying the polynomial method to random matrix quantities, then, is reduced to this question: how does one produce a series of diagonal matrices $\vD_i$ of different sizes which match the trace moments of $\vA$? This is a variant of a more general challenge called the ``moment problem.'' In classical moment problems, the task is to find a probability distribution matching a prescribed set of moments. Here, the goal is, given a diagonal matrix $\vD$ of a certain dimension, to find matrices of other dimensions that match its first few moments. Solving this moment problem for a large number of different dimensions is essential to the proof. We now discuss how to solve the problem in general.

\subsection{The Moment Problem and its Application to the Polynomial Method}
\label{intro:moment}

A recurring ingredient in our proofs is finding solutions to the moment problem described above.  In general, a moment problem is defined as follows: given a list of values $\alpha_1,\alpha_2,\ldots \alpha_p$ and a probability space, find a measure whose $i$th moment is $\alpha_i$. 
Variants of this problem arise several places in our proof.
First, as just discussed previously, it arises in our application of the polynomial method to bounding the value of random matrix theory quantities.
In this context, we have some diagonal matrix $\vD$ of dimension $N$, and we wish to find diagonal matrices of different dimensions which yield the same moments.
We need to solve the moment problem in many different dimensions for the proof to work, as these solutions yield bounds on the polynomial at different points.

Separately, the moment problem arises in our polynomial method query lower bound to show that the matrices $e^{i\vG\theta}e^{i\vG'\theta}$ are $T$-query indistinguishable from Haar-random for a high value of $T$.
In particular, after showing the trace moments of our ensemble are close to Haar, we define a linear interpolating path between our moments and Haar in some parameter $x$. The acceptance probability of our query algorithm is then a low degree univariate polynomial in $x$, and must jump in some small range if the algorithm distinguishes our ensemble from Haar.
To show distinguishability, and hence this jump, is not possible, we must show the polynomial is bounded at many points along the interpolation path.
This requires proving that unitary matrices exist that have a wide range of moment vectors.

The moment problem is well-studied in the mathematics literature---see Chapter 11 of \cite{schmudgen1991chapter11}. 
However, to the best of our knowledge, off-the-shelf statements from the current literature are too weak to use for our proofs. 
The basic issue is that in our context, we need to show there exist \emph{finite-dimensional} unitary matrices that have certain trace moments. This is a more ``discrete'' version of the moment problem than has been previously considered.\footnote{To the best of our knowledge, the closest result that exists is for ``atomic measures'' which places mass at discrete points but with real-valued probabilities. However, this is not enough, as for $N$-dimensional matrices, the probability masses must be integer multiples of $1/N$.}
It is analogous to asking, given some list of empirical moments tabulated from finitely many samples, could those empirical moments have been exactly reproduced by a different number of samples? 

To overcome this, we produce a solution to the moment problem for finite-dimensional unitary matrices. Given a list of moments, we start by finding a ``nearby'' moment vector that admits a discrete solution with a matrix. We then show that this solution can be perturbed to cover the nearby moment vector, while maintaining the discreteness of the measure. 
This follows from the Jacobian of the moment vector---i.e.,~the derivative of the moment vector under perturbing individual points---being full rank. Therefore, perturbations of the input point completely cover a small ball around the starting point, a region that is large enough for our applications. 
See Section \ref{sec:moment_problem} for further details.


\section{Preliminaries}
This section introduces the notational conventions that we use throughout the paper, as well as the definitions of our most frequently recurring
objects of study.

\subsection{Notation}

We write constants and integration element in roman font ($\e, \pi,\ri$ and $\rd t$, $\rd x$ ), scalar variables in lowercase $a,b$, vectors in bras and kets $\ket{\psi}$, matrices in bold uppercase ($\vG, \vU$). We use curly font for several objects: super-operators ($\CN$), sets $(\CS)$, and algorithms $(\CA)$. We use $\vertiii{\cdot}$ to indicate unitarily invariant norms, and omit subscripts for norms when they are clear from context.

\begin{center}
    \begin{tabular}{ c | c  }
     symbol & definition \\
     \hline

    $\norm{\ket{\psi}}$ & the Euclidean norm of a vector $\ket{\psi}$\\
	$\norm{\vO}:= \sup_{\ket{\psi},\ket{\phi}} \frac{\bra{\phi} \vO \ket{\psi}}{\norm{\ket{\psi}}\cdot \norm{\ket{\phi}}} $& the operator norm of a matrix $\vO$\\
	$\vO^* $ & \text{the entry-wise complex conjugate of a matrix $\vO$}\\
 	$\vO^\dagger$  & \text{the Hermitian conjugate of a matrix $\vO$}\\
  	$\norm{\vO}_p := (\tr \ltup{\labs{\vO}^p})^{1/p} $ & the Schatten p-norm of a matrix $\vO$\\
  $\norm{\CN}_{p-p} := \sup_{\vO} \frac{\normp{\CN[\vO]}{p}}{\normp{\vO}{p}}$ & the induced $p-p$ norm of a superoperator $\CN$ \\
     $n$ & number of qubits \\
     $N$ & matrix dimension, $N=2^n$\\
     $[N]$ & $\{0, \dots , N-1\}$ \\
     $\text{GL}(N)$ & general linear group of dimension $N$ over $\BC$ \\ 
     $\text{U}(N)$ & $\{M \in \text{GL}(N) : M M^\dagger = 1_N \}$ \\
     $\text{Herm}(N)$ & $\{M \in \text{GL}(N) : M = M^\dagger  \}$ \\
     $\poly(n)$ & the collection of
polynomially bounded functions of $n$
 \\
    $\mu$ & the Haar measure (often used over the unitary group) \\
    $\mathbb{N}_0$ & the set of natural numbers
 \end{tabular}
\end{center}

\subsection{Definitions}
\begin{defn}[Queries Models]
Consider a quantum channel $\CN$ on $n$-qubits ($N=2^n$). A quantum algorithm $\CA$ is \textbf{adaptive/parallel} $T$-query if it accesses $T$ copies of the channel $\CN$ or its adjoint $\CN^\dag$, 
\begin{itemize}
    \item \textbf{in parallel},  $\underbrace{\CN^{(\dag)}\otimes \cdots \otimes \CN^{(\dag)}}_T$ , or
    \item \textbf{adaptively}, $(\underbrace{\CN^{(\dag)} \dotsto \CN^{(\dag)}}_T)$
\end{itemize}
with arbitrary computation and measurement (possibly interleaved in the adaptive case).
\end{defn}

\begin{defn}[Unitary $T$-design]
    An ensemble of unitaries $\CW$ is a unitary $T$-query design if 
    $$ \bexpect{ \vU^{\otimes T} \cdot \vU^{\dag \otimes T} } = \bexpect{ \vW ^{\otimes T} \cdot  \vW^{\dag \otimes T}} .$$
\end{defn}
\begin{defn}[$\delta$-approximate parallel unitary $T$-design]\label{defn:U_parallel_design}
    An ensemble of unitaries $\CW$ is a $\delta$-approximate unitary $T$-parallel query design if any parallel $T$-query quantum algorithm can only distinguish
    $$\vW \leftarrow \CW \quad \text{from} \quad \vU \leftarrow \mu \quad \text{with probability at most $\delta$.}$$ 
    Formally, let $\CN_{\vU,T}$ and $\CN_{\vW,T}$ be the $T$-fold quantum channels acting on a density matrix $\rho$ as 
    \begin{align*}
        \CN_{\vU,T}:~& \rho \mapsto \Expect_{\vU \leftarrow \mu} \vU^{\otimes T} \rho \vU^{\dag \otimes T} \\
        \CN_{\vW,T}:~& \rho \mapsto \Expect_{ \vW \leftarrow \CW} \vW^{\otimes T} \rho \vW^{\dag \otimes T},
    \end{align*}
    then
    $$ \frac{1}{2} \dnorm{\CN_{\vU,T} - \CN_{\vW,T}} \leq \delta. $$ 
\end{defn}

\begin{defn}[$\delta$-approximate $s(n)$-space adaptive unitary $T$-designs]\label{defn:U_adaptive_design}
    An ensemble of unitaries $\CW$ acting on a space of dimension $2^n$ is a $\delta$-approximate $s(n)$-space adaptive unitary $T$-query design if any adaptive $T$-query quantum algorithm using at most $s(n)$ qubits
    can only distinguish 
    $$\vW \leftarrow \CW \quad \text{from} \quad \vU \leftarrow \mu \quad \text{with probability at most $\delta$.}$$ 
    Formally, for any $s(n)$-space-bounded quantum algorithm $\CA$ which is restricted to $T$ adaptive, possibly adjoint, queries,
    $$ \Pr[ \CA(~\underbrace{\vU \dotsto \vU}_T~) \to 1] - \Pr[ \CA(~\underbrace{\vW \dotsto \vW}_T~) \to 1]  \leq \delta. $$ 
\end{defn}
Allowing adjoint and adaptive queries leads to the strongest notions of a unitary design among its cousins. For example, $T$-designs are often discussed in the diamond norm $\norm{\cdot}_{\diamond}$ (see, e.g., ~\cite{brandao2016local}), but doing so only captures the case of parallel queries without adjoints.\footnote{A closer inspection of the proof in~\cite{brandao2016local} reveals that their argument should be applicable to many other norms. See~\cite[Lemma 2.2.14]{low2010pseudo} for this general phenomenon.}

\begin{defn}[Gaussian Unitary Ensemble (GUE)] 
The Gaussian Unitary Ensemble with dimension $N$ is a family of complex Hermitian random matrices specified by
\begin{align*}
\vG_{ij} &= \frac{g_{ij}+\ri g'_{ij}}{\sqrt{2N}} \quad \text{if $j>i$}\\
\vG_{ii} &= \frac{g_{ii}}{\sqrt{N}}.    
\end{align*}
where $g_{ii}, g_{ij}, g'_{ij}$ are independent standard (real) Gaussians.
\end{defn}
Note that our definition is normalized such that $\norm{\vG}\rightarrow 2$ in the large-$N$ limit. Sometimes, other conventions are used in the literature. 

\begin{defn}[GUE $T$-design]\label{defn:G_design}
    An ensemble of random matrices $\vA$ is a GUE $T$-design if it matches the first $T$ tensor moments of GUE, i.e.
    \[\mathbb{E}_{\vA}[\vA^{\otimes k}] =\mathbb{E}_{\vG\sim GUE}[\vG^{\otimes k}]\quad \text{for each}\quad k=1\ldots T \].
\end{defn}

\section{Technical Overview}
\label{sec:technicaloverview}

Our approach to unitary $T$-designs is very different from existing strategies~\cite{brandao2016local,haferkamp2021improved,haferkamp2022random}. Thus far, these designs have been generated by random walks, i.e., sequences of local random gates. Convergence is shown by proving that the random walk is ergodic and mixes rapidly. In comparison, our construction enjoys convergence thanks to a matrix-valued central limit theorem and intuition from random matrix theory that the product of two exponentiated Guassians should be approximately Haar~\eqref{eq:intro_hybrid}. 

\begin{restatable}[Efficient T-Designs]{thm}{efficienttdesigns}\label{thm:efficient_t_designs}
    Consider $\epsilon_{\scaleto{q}{5pt}}$-approximate i.i.d. parallel $q$-query unitary designs $\Tilde{\vU}_1 \dotsto \Tilde{\vU}_m \in \unitary(N)$ and $\Tilde{\vU}^\prime_1 \dotsto \Tilde{\vU}^\prime_m \in \unitary(N)$,
    and a diagonal matrix $\vD \in \gl(N)$ whose spectrum approximates the first $q$ moments of Wigner's semicircle:
    \begin{align*}
        \abs{\btr(\vD^k) - \int x^k \rho_{sc}(x)\, \diff x } &\leq  2^k \cdot \frac{2q+4}{N} , \quad\quad \forall 1\leq k \leq q, \\
        \norm{\vD}_{op} &\leq 2.
    \end{align*}
    Define the ensemble of unitary matrices,
    \begin{align*}
        \vW_2 :=  e^{i \frac{\theta}{\sqrt{m}} \sum_{j=1}^m \tilde{\vU}_{j} \vD \tilde{\vU}_{j}^\dag}\cdot e^{i \frac{\theta}{\sqrt{m}} \sum_{j=1}^m \tilde{\vU}^{\prime}_{j}\vD\tilde{\vU}^{\prime \dag}_{j}}.
    \end{align*}
    Then there exist $\theta =\CO(1)$, $m =\CO(T^2)$, $q=\CO(\log T)$, and $\epsilon_{\scaleto{q}{5pt}} = 2^{-O(nq)}$ such that $\vW$ is a $\delta$-approximate parallel $T$-query unitary design where the error is a small constant $1>\delta =\CO(1)$, so long as $T \leq 2^{\tilde{O}(n/\log n)}$.
\end{restatable}

To prove the correctness of Theorem~\ref{thm:efficient_t_designs}, we display the particular sequence of approximations between the truly Haar random unitary and our algorithmic implementation (slightly more nuanced than we previously alluded to). Let the unitaries $\vU_j^{(\prime)}$ be i.i.d. Haar random; the unitaries $\tilde{\vU}_j^{(\prime)}$ be unitary $q$-designs for $q\ll T$; the diagonal matrix $\vD \in \gl(N)$ be an approximation to the spectrum of GUE; $\vG^{\prime}$ be GUE matrices; and $\vLambda$ be a diagonal unitary $\vLambda$ with small moments.

\begin{align*}
    \vU_{Haar} &~\stackrel{(\textbf{i})}{\approx}~  \vU_{Haar} \vLambda \vU^{\dagger}_{Haar}  &&=: \vW& \\
    &~(\approx~ 
    e^{i \theta \vG } e^{i   \theta  \vG^\prime } ) & &=: \vW_{\scaleto{GUE}{4pt}}& \tag{Not used in formal proof.}\\
    &~\stackrel{(\textbf{ii})}{=}~ e^{i \frac{\theta}{\sqrt{m}} \sum_{j=1}^m \vU_{j} \vD \vU_{j}^\dag}\cdot e^{i \frac{\theta}{\sqrt{m}} \sum_{j=1}^m \vU^{\prime}_{j}\vD\vU^{\prime \dag}_{j}}& &=: \vW_1 & \\
    &~\stackrel{(\textbf{ii}^{\prime})}{\approx}~ e^{i \frac{\theta}{\sqrt{m}} \sum_{j=1}^m \tilde{\vU}_{j} \vD \tilde{\vU}_{j}^\dag}\cdot e^{i \frac{\theta}{\sqrt{m}} \sum_{j=1}^m \tilde{\vU}^{\prime}_{j}\vD\tilde{\vU}^{\prime \dag}_{j}}& &=: \vW_2& 
\end{align*}

Let us elaborate, starting from the targeted Haar random unitaries:

We take a first step away from the Haar random unitaries by tweaking the eigenvalues (\textbf{i}) and argue that the ensemble is indistinguishable from Haar with respect to a certain \emph{adaptive query distance}. Recall that a Haar random unitary is uniquely defined by being left (and right) unitary invariant. That structure alone implies that the measure factorizes into independent distributions for its eigenbasis and spectrum. Moreover, all trace moments of a Haar random unitary are nearly zero with high probability, thus we might expect any ensemble with random basis and nearly zero trace moments to be approximately Haar.

\begin{restatable}[From small moments to Haar]{lem}{smalltohaar}\label{lem:moments_implies_Haar}
    Let $\vW$ be a random unitary that is invariant under unitary conjugation:
    $$\vW \stackrel{dist}{\sim} \vU \vD \vU^{\dagger} \quad \text{for deterministic unitary $\vD$ and Haar $\vU$}.$$ Consider its first $T$ normalized trace moments $\alpha_p :=  \frac{1}{N} \tr(\vW^p)$ for $1 \leq p \leq T$, and let $\vec{\alpha_T} := (\alpha_1 \dotsto \alpha_T)$.
    There is an absolute constant $C>0$ such that if $T$ is small enough, $\frac{T^5}{N} \leq \frac{C \delta}{\sqrt{\log{(4/\delta)}}}$, then 
    \begin{align*}
         \norm{\vec{\alpha_T}}_1 <  \frac{\delta}{32\cdot T^{7/2} }\quad \text{implies}\quad \abs{\BE \CA(\vW) - \BE \CA(\vU) } < \delta.
    \end{align*}
    That is, small trace moments implies that any adaptive $T$-query quantum decision algorithm $\CA$ can distinguish $\vW$ from a Haar random unitary $\vU$ with probability at most $\delta$.
\end{restatable}
See Section~\ref{sec:small_moments_haar} for the proof. To reiterate, if a random unitary $\vW$ that is invariant under unitary conjugation has low trace moments close to zero, it also appears Haar to quantum algorithms. Morally, a unitarily-conjugate-invariant ensemble has no structure in its basis and is only characterized by the spectrum. The lemma above quantifies how effectively any adaptive query quantum algorithm can learn the distinction.\footnote{A very natural quantum algorithm to learn the special properties is running phase estimation on the maximally mixed state to obtain spectral statistics.} We develop the polynomial method specifically to control the degree $T$ expressions that appear when analyzing the distinguishability of $\vW$ and $\vU$.

While this first step relaxed the requirement on the spectrum, it failed to simplify the Haar random eigenbasis, so how did the problem of approximating a Haar random unitary get any easier? The Gaussian Unitary Ensemble (GUE) plays a key role in our construction (and the choice of $\vD$ in $(\textbf{i})$) since 
$$\vG \distas \vU \vD_{\scaleto{GUE}{3pt}} \vU^{\dag}$$ for Haar random unitary $\vU$ and a random diagonal matrix $\vD_{\scaleto{GUE}{3pt}}$ with a GUE spectral distribution. 
Though a GUE matrix $\vG$ is not unitary, the random matrix $e^{i\vG \theta}$ is unitary and also inherits the Haar eigenbasis from GUE. Remarkably, while the matrix $e^{i\vG \theta}$ is far from Haar (Section~\ref{app:GUE_properties}), multiplying merely \emph{two} Gaussian exponentials suffices for a properly chosen $\theta$:
$$ \vW_{\scaleto{GUE}{4pt}} := e^{i \theta \vG } e^{i   \theta  \vG^\prime } \approx \vU_{Haar}.$$
In particular, we show that the product of any two Haar-conjugated unitaries that are individually nearly traceless has small trace moments. This informs the choice of a particular $\theta = const$, which ensures that the exponential is nearly traceless $\tr[e^{i\vG \theta}] \approx 0$.

\begin{restatable}[Products of Haar-conjugated unitaries suppress trace moments]{lem}{UDUVDV}\label{lem:UDUVDV_expected_moments}
Consider a unitary ensemble $\vW$ of $N$-dimensional matrices such that
\begin{align*}
    \vW \stackrel{dist}{\sim}\vU \vD_1 \vU^{\dagger}\cdot \vV \vD_2 \vV^{\dagger}, 
\end{align*}
where $\vU$ and $\vV$ are independent Haar random unitary ensembles, and $\vD_i$ are both random unitary matrices such that for each instance $\vD_i$,
\begin{align*}
    \sum_{k=1}^p |\btr{\vD_i^k}| \leq \frac{1}{4}.
\end{align*}
Then, for dimension 
$N > \frac{97}{2}p^{7/2}$,
the expectation of the $p^{\text{th}}$ normalized trace moment satisfies
\begin{align*}
\left|\BE_{\vW}\btr[\vW^p]\right| = O \left(\frac{p^{15/4}}{\sqrt{N}}\right)+p \, \BE_{\vD_i}\left(\labs{\btr{\vD_1}}+\labs{\btr{\vD_2}}+\labs{\btr{\vD_1}}\labs{\btr{\vD_2}}\right).
\end{align*}
\end{restatable}
This is a \textit{noncommutative} phenomenon that manifests at large dimension $N \gg 1$ and is useless for scalars $(N=1)$. An analogous result can be seen in the large $N$ limit with some standard Weingarten calculus, but the main difficulty is obtaining good nonasymptotic estimates at very large values of $p \gg \log(N)$. See Section~\ref{sec:products_of_gaussians} for the proof. 

Now that we have reduced the problem to simulating GUE, the next step is to efficiently approximate a GUE ensemble by a nonasymptotic central limit theorem: summing over simpler random matrices that we can implement efficiently. We are interested in a natural random Hermitian matrix $\vH:= \tilde{\vU} \vD \tilde{\vU}^\dag$, whose spectrum $\vD$ and basis $\tilde{\vU}$ matches only the low moments of a random GUE matrix. Thus, constructing $\vH$ amounts to constructing $\vD$ which matches $q$ moments of GUE and conjugating $\vD$ with an off-the-shelf approximate unitary $q$-design~\cite{brandao2016local,harrow2023approximate,haferkamp2022random}.

This leads to the following proofs of convergence to GUE properties in two steps, treating the spectrum $(\textbf{ii})$ and basis $(\textbf{ii}^{\prime})$ separately. Both steps are elaborate versions of the Lindeberg exchange principle for matrices~\cite{chen2023sparse}. First, we address approximations in the spectrum. This argument differs from ~\cite{chen2023sparse} by handling the effect of errors when the low-moments are not matched exactly, which requires an additional Weingarten calculation for very small moments $q$.\footnote{Only using coarse estimates that do not directly address the larger moment $(T)$ case. } 

\begin{restatable}[Convergence of spectrum gives small moments]{lem}{cltspectrumsmallmoments}\label{lem:clt_spectrum_small_moments}
    Let $\vU_1 \dotsto \vU_m \in \unitary(N)$ and $\vU^\prime_1 \dotsto \vU^\prime_m \in \unitary(N)$ be i.i.d. Haar random unitaries, and suppose $\vD \in \gl(N)$ is a diagonal matrix with empirical spectral distribution which approximates the first $q$ moments of Wigner's semicircle $\rho_{sc}(x)\, \diff x := \frac{1}{2\pi} \sqrt{4 - x^2} \, \diff x $,
    \begin{align*}
        \abs{\tr{\vD^k} - \int x^k \rho_{sc}(x)\, \diff x} \leq 2^k \cdot \frac{2q+4}{N}, \quad\quad \forall~ 1\leq k\leq q
    \end{align*}
    Define the unitary ensemble,
    \begin{align}
        \vW_1:= e^{i \frac{\theta}{\sqrt{m}} \sum_{j=1}^m \vU_{j} \vD \vU_{j}^\dag}\cdot e^{i \frac{\theta}{\sqrt{m}} \sum_{j=1}^m \vU^{\prime}_{j}\vD\vU^{\prime \dag}_{j}}. \label{eq:vW_clt_spectrum_small_moments}
    \end{align}
    There exist $\theta =\CO(1)$, $m=O(T^2)$, $q=O(\log T)$ such that the trace moments of $\vW$ are small,
    \begin{align*}
            \abs{\Expect_{\{U_j\}\cup\{U'_j\}}\btr \vW_1^p} \leq  O\ltup{\frac{4p(\theta )^{q+1}}{(\sqrt{m})^{q-1} (q+1)!}} 
            + \tilde{O}\left( \frac{T^4}{\sqrt{N}} \right)
    \end{align*}
    for all $1\leq p\leq T$, for $N\geq \Omega(q^{4q})$.
\end{restatable}
 At a high level, the purpose of the lemma is to bound an expression involving expectations of traces of functions of Haar-random unitary matrices. That seems like a task well-suited to the polynomial method developed for proving \autoref{lem:moments_implies_Haar}. The matrix $\vD$, however, is Hermitian instead of unitary, so to apply the polynomial method, we would need to solve the real-valued analog of the unitary moment problem. Instead of embarking on that task, we grapple directly with the Weingarten calculus for the polynomials of degree $q$ in Haar random $\vU$ that appear in the proof. Because $N \geq \Omega(q^{4q})$, we can afford to make very rough approximations in extracting our asymptotic estimates, acquiring combinatorial factors that would have been fatal for \autoref{lem:moments_implies_Haar}. Nonetheless, the analysis is quite intricate, illustrating the need for a better technique in general. See Section~\ref{sec:lindeberg_spectrum} for the proof.

Now that $e^{i \frac{\theta}{\sqrt{m}} \sum_{j=1}^m \vU_{j} \vD \vU_{j}^\dag}$ gives a proxy for the GUE exponetial $\e^{\ri \vG \theta}$ in terms of its low moments, Lemma~\ref{lem:UDUVDV_expected_moments} then implies the product $\vW_1$ also has small moments. Finally, Lemma \ref{lem:moments_implies_Haar} implies that $\vW$ is approximately Haar. 

For the basis $(\textbf{ii}^{\prime})$, we swap out the Haar random bases $\vU_j, \vU_j^{\prime} $  from \ref{eq:vW_clt_spectrum_small_moments} with approximate low-moment unitary designs that can be efficiently implemented. We then prove that the resulting object is a $T$-design using a superoperator variant of the Lindeberg replacement principle~\cite{chen2023sparse}.
\begin{restatable}[Approximate designs from sums]{lem}{cltbasis}\label{lem:clt_basis}
    Consider $\epsilon_{\scaleto{q}{5pt}}$-approximate i.i.d. parallel $q$-query unitary designs $\Tilde{\vU}_1 \dotsto \Tilde{\vU}_m \in \unitary(N)$ and $\Tilde{\vU}^\prime_1 \dotsto \Tilde{\vU}^\prime_m \in \unitary(N)$,
    and a diagonal matrix $\vD \in  \gl(N)$  bounded as $ \BE \left[\lnorm{\vD}^{q+1}_{\text{op}}\right] \leq C_D.$
    Let $\vW$ be as defined in \eqref{eq:vW_clt_spectrum_small_moments}.
    Define the random unitary,  
    \begin{align*}
        \vW_2:=e^{i \frac{\theta}{\sqrt{m}} \sum_{j=1}^m \tilde{\vU}_{j} \vD \tilde{\vU}_{j}^\dag}\cdot e^{i \frac{\theta}{\sqrt{m}} \sum_{j=1}^m \tilde{\vU}^{\prime}_{j}\vD\tilde{\vU}^{\prime \dag}_{j}}. 
    \end{align*}
    Let $\CW$ and $\tilde{\CW}$ be the unitary channels which act on a density matrix $\rho$ by conjugation,
    \begin{align*}
        \CW : \rho \mapsto \vW \rho \vW^\dag \quad \text{and} \quad\tilde{\CW} : \rho \mapsto \vW_2 \rho \vW_2^\dag.
    \end{align*}
    Then,
    \begin{align*}
        \dnorm{\bexpect{ \tilde{\CW}^{\otimes k} -  \CW^{\otimes k} } } \leq 
        \left[ 
             \epsilon_q 2^{8nq} \left(\frac{\theta^2 k^2}{m} \right)^{q+1}
             + 2^{q+2}  \ltup{\frac{\theta k}{\sqrt{m}}}^{q+1}
             \right] \frac{ 2 m C_D \, e^q}{q^q}
    \end{align*}
\end{restatable}
See Section~\ref{sec:lindeberg_basis} for details. Morally, the above states that we can obtain high-quality ($k$ moments) Haar random basis from low-quality ones ($q \ll k$), and can be regarded as the key mechanism for getting higher designs \textit{without} a spectral gap calculation. The huge prefactor of $2^{8nq}$ requires the initial error $\epsilon_q$ in the basis to be exponentially small; this is due to conversions between various notions of designs. Fortunately, the $q$-designs we are importing~\cite{brandao2016local,haferkamp2022random} also enjoy exponentially decaying error and can be implemented with exponential precision at the cost of mild algorithmic overhead.

That concludes the main analysis for our construction, and the last step is to implement it. We explicitly describe a diagonal matrix $\vD \in \gl(N)$ that approximates the first $q$ moments of Wigner's semicircle:
\begin{align*}
        \abs{\btr(\vD^k) - \int x^k \rho_{sc}(x)\, \diff x } &\leq 2^k \cdot \frac{2q+4}{N} , \quad\quad \forall 1\leq k \leq q, \\
        \norm{\vD}_{op} &\leq 2.
    \end{align*}
Since the matrix $\vD$ is simply a weighted sum over projectors, the unitary ensemble $\vW_2$ can be implemented efficiently on a quantum computer for properly chosen parameters. 

Informally, by setting $m =\CO(T^2)$, $q =\CO(\log T)$, $\theta=O(1)$ and $\epsilon_q=2^{-O(nq)}$, running Hamiltonian simulation, we obtain the following:

\begin{lem}[Algorithm (Informal version of \ref{lem:algorithm_W})]\label{lem:algorithm_W_informal}
    There exists a unitary ensemble on $n$ qubits $\vW_{alg}$ that $\epsilon$ approximates the first $T$ tensor moments of $\vW_2$ in diamond distance, and which can be implemented on a quantum computer with  using 
    \[ \Tilde{O}\Big(T^2 n^2 +  \log(1/\epsilon) \Big)\] local gates, $\CO(n^2 \polylog(T) + \log^{2.5} \frac{T}{\epsilon})$  ancilla qubits, and $\Tilde{O}\left(T n^2\right)$ random bits. 
\end{lem}

The final step is to apply boosting. Although the runtime of the above algorithm contains an additive $\log(\epsilon^{-1})$ term, this $\epsilon$ is the distance to the ensemble $\vW_2$, not to Haar. The distance between the ensemble $\vW_2$ and Haar scales inverse polynomially with our runtime and number of terms $m$. Therefore, the final step is to run the algorithm many times independently in series; this exponentially suppresses the error to Haar, and yields an overall runtime of $\Tilde{O}\Big(T^2 n^2  \log(1/\epsilon) \Big)$ quantum gates. See Section \ref{subsec:boosting} for details.

\subsection{Structure of the Paper}

A recurring theme in several of the proofs is a variant of the moment problem on the unit circle. As such, we address it first in section \ref{sec:moment_problem}; these results are crucial to the rest of our paper, but details of the proofs can be skipped upon first reading. Sections \ref{sec:small_moments_haar} proves Lemma \ref{lem:moments_implies_Haar}. Section \ref{sec:products_of_gaussians} then motivates the use of GUE as a starting point for getting close to Haar random unitaries. Sections \ref{sec:lindeberg_spectrum} and \ref{sec:lindeberg_basis} are dedicated to the proofs of Lemma \ref{lem:clt_spectrum_small_moments}, and Lemma \ref{lem:clt_basis}, respectively. All hybrids are strung together in section \ref{sec:main_thm_proof}, proving Theorem \ref{thm:efficient_t_designs}. Lastly, the details of the algorithmic design are addressed in section \ref{sec:construction}. Figure \ref{fig:dependencies} depicts the dependencies of proofs and sections of the paper. Section \ref{sec:discussion} provides a brief discussion.

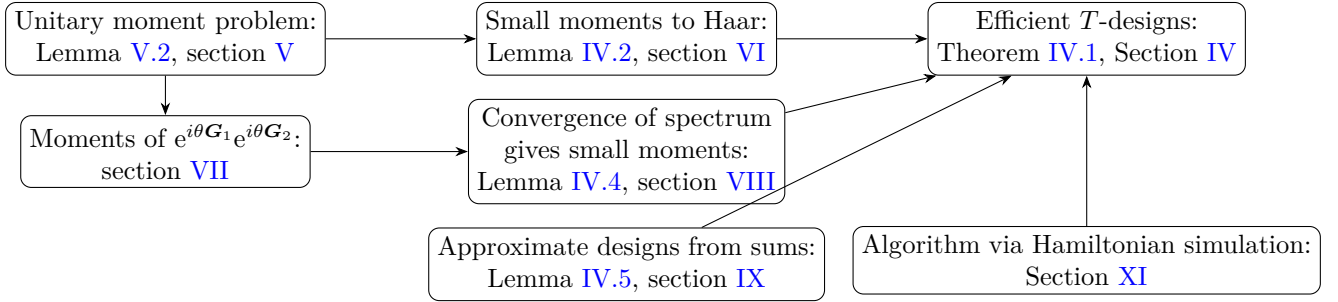
\begin{figure}
    \centering
    \begin{tikzpicture}[>=Stealth, node distance=1.5cm, every node/.style={rectangle, draw, text centered, rounded corners}]

        \node[draw, align=center] (lemma1) {Small moments to Haar:\\ Lemma \ref{lem:moments_implies_Haar}, section \ref{sec:small_moments_haar}};
        \node[draw, align=center] (lemma3) [below of= lemma1] {Convergence of spectrum \\gives small moments:\\Lemma \ref{lem:clt_spectrum_small_moments}, section \ref{sec:lindeberg_spectrum}};
        \node[draw, align=center] (lemma4) [below of= lemma3] {Approximate designs from sums:\\Lemma \ref{lem:clt_basis}, section \ref{sec:lindeberg_basis}};
        \node[draw, align=center] (moment) [left=2cm of lemma1] {Unitary moment problem:\\ Lemma \ref{lem:unitary_moment_prob}, section \ref{sec:moment_problem}};
        \node[draw, align=center] (lemma2) [below of= moment] {Moments of $\e^{i\theta \vG_1}\e^{i\theta \vG_2}$: \\ section \ref{sec:products_of_gaussians}};
        \node[draw, align=center] (theorem1) [right=2cm of lemma1] { Efficient $T$-designs: \\Theorem \ref{thm:efficient_t_designs}, Section \ref{sec:technicaloverview} };
        \node[draw, align=center] (construction) [below=2cm of theorem1] {Algorithm via Hamiltonian simulation:\\ Section \ref{sec:construction}};

        \draw[->] (moment) -- (lemma1);
        \draw[->] (moment) -- (lemma2);
        \draw[->] (lemma1) -- (theorem1);
        \draw[->] (lemma2) -- (lemma3);
        \draw[->] (lemma3) -- (theorem1);
        \draw[->] (lemma4) --  (theorem1);
        \draw[->] (construction) -- (theorem1);

    \end{tikzpicture}
    \caption{Proof dependency graph. An arrow $A \rightarrow B$ means that the proof of $B$ depends on $A$. }
    \label{fig:dependencies}
\end{figure}


\section{The Unitary Moment Problem}\label{sec:moment_problem}
Given a unitary $\vU$, consider its \emph{normalized trace moments}.
\begin{defn}[Normalized trace moments]
    The normalized $k^{th}$ trace moment of a unitary $\vU \in \BC^{N\times N}$ is defined as $\btr(\vU^k) := \frac{\tr(\vU^k)}{N}$.
\end{defn}
We define the \emph{unitary moment problem} as follows: Let $\{\alpha_k\}_{k=1}^{T}$ be a sequence of complex numbers. For a given dimension $N$, when does there exist an $N \times N$ unitary matrix $\vU \in \text{U}(N)$ such that the first $T$ normalized trace moments of $\vU$ are given by the sequence $\{\alpha_k\}_{k=1}^{T}$,
\begin{align*}
    \alpha_k = \btr(\vU^k) \quad \text{for each}\quad k \in \{0, \dots, T\}?
\end{align*}

The unitary moment problem is closely related to the \emph{trigonometric moment problem} \cite{schmudgen1991chapter11}. Let $\vec{\alpha} =\{\alpha_k\}_{k\in \BN_0}$ be a sequence of complex numbers indexed by natural numbers $\mathbb{N}_0 = \{1, 2, \cdots\}$.  When does there exist a Radon measure.\footnote{In the restricted context we are considering here, namely measures on $S^1$, a Radon measure corresponds to a positive linear functional on the space of continuous functions with compact support. Heuristically, for our purposes, it will be a probability density.} $\mu(z)$ on the complex unit circle $S^1\equiv \{z\in \BC, \labs{z} = 1\}$ such that 
\begin{align}
\alpha_k = \int_{S^1} z^{-k} \diff \mu (z) \quad \text{for each}\quad k \in \mathbb{N}_0 ? \label{eq:trig_moment}
\end{align}
The \emph{truncated trigonometric moment problem} is the corresponding problem for finite moments $\vec{\alpha}_T :=\{\alpha_k\}_{k= 0}^T$. Both of these problems have been studied thoroughly.~\cite{schmudgen1991chapter11}. It turns out that when the truncated trigonometric moment problem has a solution, it can be written as a $p$-atomic measure $\sum_{j=1}^p c_j \delta_{z_j}$ (i.e., weighted delta functions) for $p \leq 2T +1$ where $z_1 \dotsto z_p \in S^1$ are pairwise independent and $c_1 \dotsto c_p \in \BC$. For our purposes, this form of measure is not sufficient (even though quite similar) because the weights $c_j$ take on arbitrary values in $[0, 2\pi)$ while we are interested in integer multiples of some fraction $\frac{1}{N}$. That is, we want to know whether there is an empirical density for $N$ samples that has the specified moments.

One should think of the $z_j$ as the individual eigenvalues of a hypothetical unitary matrix whose finite moments coincide with a candidate $\vec{\alpha}_T$. However, since the dimension $N$ is finite, we must further constrain the moment problem: the measure must be atomic \emph{and} ``coarse'' such that there exists a finite integer $N$ where the scaled measure $N\cdot \mu$ only takes on integer values. It is of interest to determine the smallest $N$ that solves the problem, as this corresponds to the smallest matrix (w.r.t. its dimension) that gives rise to these moments. This leads us to the \emph{truncated trigonometric moment problem} with an \emph{integral, atomic} measure, which we call the \emph{unitary moment problem}: For a given $N$, when does there exist a set of values on the unit circle ${z_1, \dots, z_N} \in S^1$ such that 
$$ \alpha_k = \frac{1}{N}\sum_{j=1}^{N} z_j^{-k} \quad \text{for each}\quad k \in \{0, \dots, T\}?$$ Simply by parameter counting, one can only hope to find general solutions for this problem for $N \geq 2T+1$. In this section, we address the unitary moment problem. Surprisingly there exists a solution to the unitary moment problem provided the $\ell_1$-norm $\norm{\vec{\alpha}_T}_1$ is above bounded by a constant.  

\begin{restatable}[$\ell_1$ condition for the unitary moment problem]{lem}{unitarymomentprob}\label{lem:unitary_moment_prob}
    Let $\vec{\alpha}_T := (\alpha_1 \dotsto \alpha_T) \in \BC^T$. If $\quad N \ge 16(2T+1)T^{5/2}$, and the $\ell_1$-norm of $\vec{\alpha}_T$ is sufficiently small,
    \begin{align*}
        \norm{\vec{\alpha}_T}_1 \leq \frac{1}{4},
    \end{align*}
    then there exists a unitary $\vU \in \text{U}(N)$ such that
    \begin{align}
        \alpha_k = \btr(\vU^k) \quad \text{for each}\quad k \in \{1, \dots, T\}.\quad \label{eq:unitary_moment_prob_cond}
    \end{align}
\end{restatable}
The proof relies heavily on first solving the unitary moment problem \emph{near the origin} $\vec{\alpha}_T \approx \vec{0}$, and then bootstrapping that solution to points away from the origin.

\subsection{Existence of a Solution Near the Origin}\label{subsec:solution_near_origin}
The unitary moment problem permits a solution when the candidate moment vector is sufficiently close to the origin. However, the radius for which this solution is guaranteed to exist shrinks with the problem parameter $T$.
\begin{lem}[Local invertibility]\label{lem:discrete_unif_trunc_trig_prob}
    Let $\vec{\alpha}_T = (\alpha_1, \dots , \alpha_T) \in \BC^{T}$ be a candidate moment vector of dimension $T$ and suppose $N \geq 2T$. If its $\ell_2$-norm is sufficiently small,
    $$\norm{\vec{\alpha}_T}_2 \leq \frac{1}{8T^{3/2}},$$
    then there exist angles $\theta_1, \dots , \theta_{N}  \in \BR$ such that
    $$ \alpha_k =  \frac{1}{2T}\sum_{l=1}^{2T} e^{i  k\theta_l} \quad \text{for all} \quad 1 \leq k \leq T. $$
\end{lem}
 To simplify notation, we take $k \mapsto -k$ as compared to the convention of~\cite{schmudgen1991chapter11} as used in~\eqref{eq:trig_moment}. The solution to the discrete and uniform moment problem on the unit circle within this ball is based on the inverse mapping theorem and corresponding bounds on the neighborhoods to which the inverse mapping applies. Concretely, we will use:
\begin{lem}[\cite{lang_inverse_mappings} Theorem 1.2 (Inverse Mapping Theorem).]\label{thm:inv_mapping}
    Let $U$ be open in a Banach space $E$, and let $f\colon U \rightarrow F$ be a $C^p$ map. Let $x_0 \in U$ and assume that $f^\prime(x_0)\colon E \rightarrow F$ is a top linear isomorphism (i.e.~invertible as a continuous linear map). Then $f$ is a local $C^P$-isomorphism at $x_0$.
\end{lem}

\begin{lem}[\cite{lang_inverse_mappings} Lemma 1.3]\label{thm:inv_mapping_ball}
    Let $U$ be open in a Banach space $(E,\lnorm{\cdot}) $, and let $f \colon U \rightarrow E$ be of class $C^1$.  Assume that $f(0) = 0, f^\prime(0) = I$. Let $r > 0$ and assume that $\Bar{B}_r(0) \subset U$. Let $0 < s < 1$, and assume that  
    $$\norm{f'(z) - f'(x)} \leq s $$
    for all $x, z \in \Bar{B}_r(0)$. If $y \in E$ and $\norm{y} \leq (1 - s)r$, then there exists a unique $x \in \Bar{B}_r(0)$ such that $f(x) = y$. 
\end{lem}

Let $f\colon\BR^N \rightarrow \BC^T$ be the moment map taking a vector of $N$ angles to its first $T$ moments, 
\begin{align*}
     f\colon(\theta_1 , \dots, \theta_N) \mapsto   (\alpha_1, \dots , \alpha_T)\quad \text{where}\quad \alpha_k = \frac{1}{N}\sum_{l=1}^N e^{i \theta_l k}.
\end{align*}
In the proof, we will require that the domain and range of our map be the same complete metric space, thus we define $f_{\text{Re}}\colon \BR^N \rightarrow \BR^{2T}$ to be the same as $f$ but with the real and imaginary components of its range separated,
\begin{align*}
     f_{\text{Re}}\colon (\theta_1 , \dots, \theta_N) &\mapsto   (\Re{\alpha_1},~ \Im{\alpha_1} ,~ \dots, ~\Re{\alpha_T}, ~\Im{\alpha_T} )     \\
     \text{where} \quad \Re{\alpha_k} &:= \sum_{l=1}^N \cos{\theta_l k}\quad \text{and} \quad \Im{\alpha_k} := \sum_{l=1}^N \sin{\theta_l k}.
\end{align*}
Observe that the point (evenly distributed angles) $\vec{\theta}_{0} := (\frac{2\pi}{N},2\frac{2\pi}{N} \dots, N\frac{2\pi}{N}) \in \BR^N$ maps to the origin $\vec{\alpha}_0 := (0, \dots, 0) \in \BC^T$.  If $\nabla f$ is nonsingular at $\vec{\theta}_{0}$, then in a neighborhood $V$ of $\vec{\theta}_{0}$, $f$ is bijective onto the image of the neighborhood $f(v)$ by the inverse mapping theorem~\ref{thm:inv_mapping}. Additionally, if $\nabla f $ is $K$-Lipschitz, then it is possible to lower bound, in terms of $K$, the size of the neighborhood around the origin $\vec{\alpha}_0$ that is bijective to a neighborhood around the uniform angles $\vec{\theta}_{0}$. This is made precise in the proof.

We first start by proving a weaker statement. Note that $\vec{\alpha}_T$ is bounded in terms of $N$ instead of $T$ in this version. In our applications, $N$ is generally much larger than $T$, so this constraint on $\vec{\alpha}$ is stronger than the one in Lemma~\ref{lem:discrete_unif_trunc_trig_prob}. 
\begin{lem}[A weaker local invertibility]\label{lem:discrete_unif_trunc_trig_prob_N}
    Let $\vec{\alpha}_T = (\alpha_1, \dots , \alpha_T) \in \BC^{T}$ be a candidate moment vector of dimension $T$ and suppose $N \geq 2T$. If its $\ell_2$-norm is sufficiently small,
    $$\norm{\vec{\alpha}_T}_2 \leq \frac{1}{N^{3/2}},$$
    then there exist angles $\theta_1, \dots , \theta_{N}  \in \BR$ such that
    $$ \alpha_k =  \frac{1}{2T}\sum_{l=1}^{2T} e^{i  k\theta_l} \quad \text{for all} \quad 1 \leq k \leq T. $$
\end{lem}

We first consider the case where $N$ is even and then generalize the proof to odd $N$.
\begin{proof}[Proof of~\autoref{lem:discrete_unif_trunc_trig_prob_N}, $N$ even]
    Consider the moment map
    $f_{\text{Re},N}\colon \BR^N \rightarrow \BR^{N}$ taking a vector of $N$ angles to the real and imaginary parts of its first $N/2$ moments, 
    \begin{align*}
         f_{\text{Re}}\colon (\theta_1 , \dots, \theta_N) &\mapsto   (\Re{\alpha_1},~ \Im{\alpha_1} ,~ \dots, ~\Re{\alpha_{N/2}}, ~\Im{\alpha_{N/2}} ).     
    \end{align*}
 We will work with the function $f_{\text{Re}}\colon \BR^N \rightarrow \BR^{N}$ defined above. Note that  $\BR^N $ is open in $\BR^{N}$. The partial derivatives of $f_{\text{Re}}$ are given by 
\begin{align*}
    \frac{\partial \Re{\alpha_k} }{\partial \theta_l} = -\frac{k}{N} \sin{k\theta_l } \quad\text{and} \quad
    \frac{\partial \Im{\alpha_k} }{\partial \theta_l} = \frac{k}{N} \cos{k\theta_l }, 
\end{align*}
yielding the Jacobian matrix
\begin{align*}
    \nabla f_{\text{Re}} (\theta_1, \dots, \theta_N) = \frac{1}{N} \begin{bmatrix}
    -\cos{\theta_1} & -\cos{\theta_2} & \cdots & -\cos{\theta_N} \\
    \sin{\theta_1} & \sin{\theta_2} & \cdots & \sin{\theta_N} \\
    -2\cos{2\theta_1} & -2\cos{2\theta_2} & \cdots & -2\cos{2\theta_N} \\
    2\sin{2\theta_1} & 2\sin{2\theta_2} & \cdots & 2\sin{2\theta_N} \\
    \vdots & \vdots & \ddots & \vdots \\
    -\frac{N}{2}\cos{\frac{N}{2}\theta_1} & -\frac{N}{2}\cos{\frac{N}{2}\theta_2} & \cdots & -\frac{N}{2}\cos{\frac{N}{2}\theta_N} \\
    \frac{N}{2}\sin{\frac{N}{2}\theta_1} & \frac{N}{2}\sin{\frac{N}{2}\theta_2} & \cdots & \frac{N}{2}\sin{\frac{N}{2}\theta_N} \\
    \end{bmatrix}.
\end{align*}
We are interested in the Jabobian $\nabla f_{\text{Re}}$ evaluated at the evenly spaced angles  $\vec{\theta}_{0}=\left(\frac{2\pi}{N},2\frac{2\pi}{N} \dots, N\frac{2\pi}{N}\right)$,
\begin{align*}
\nabla f_{\text{Re}} \left(\vec{\theta}_{0}\right)
 & = \frac{1}{N} \begin{bmatrix}
    -\cos{\frac{2 \pi }{N}} & -\cos{2\frac{ 2\pi}{N}} & \cdots & -\cos{ N\frac{2\pi}{N}} \\
    \sin{\frac{2 \pi }{N}} & \sin{2\frac{2\pi}{N}} & \cdots & \sin{N \frac{2\pi}{N}} \\
    -2\cos{2\frac{2 \pi }{N}} & -2\cos{2 \cdot 2\frac{  2\pi}{N}} & \cdots & -2\cos{2\cdot N \frac{2\pi}{N}} \\
    2\sin{2\frac{2 \pi }{N}} & 2\sin{2 \cdot 2\frac{ 2\pi}{N}} & \cdots & 2\sin{2\cdot N \frac{2\pi}{N}} \\
    \vdots & \vdots & \ddots & \vdots \\
    -\frac{N}{2}\cos{\frac{N}{2}\frac{2 \pi }{N}} & -\frac{N}{2}\cos{\frac{N}{2} \cdot 2\frac{ 2\pi}{N}} & \cdots & -\frac{N}{2}\cos{\frac{N}{2}\cdot N \frac{2\pi}{N}} \\
    \frac{N}{2}\sin{\frac{N}{2}\frac{2 \pi }{N}} & \frac{N}{2}\sin{\frac{N}{2} \cdot 2 \frac{ 2\pi}{N}} & \cdots & \frac{N}{2}\sin{\frac{N}{2} \cdot N \frac{2\pi}{N}} \\
    \end{bmatrix} \\
    & = \frac{1}{N} 
    \ltup{
    \begin{bmatrix}
    1 & & \\ 
    & \ddots & \\
    & & \frac{N}{2} \\
    \end{bmatrix} \otimes \begin{bmatrix}
        -1 & 0 \\ 
        0 & 1 \\
        \end{bmatrix} }
    \begin{bmatrix}
        \cos{\frac{2 \pi }{N}} & \cos{2\frac{ 2\pi}{N}} & \cdots & \cos{ N\frac{2\pi}{N}} \\
        \sin{\frac{2 \pi }{N}} & \sin{2\frac{2\pi}{N}} & \cdots & \sin{N \frac{2\pi}{N}} \\
        \cos{2\frac{2 \pi }{N}} & -\cos{2 \cdot 2\frac{  2\pi}{N}} & \cdots & \cos{2\cdot N \frac{2\pi}{N}} \\
        \sin{2\frac{2 \pi }{N}} & \sin{2 \cdot 2\frac{ 2\pi}{N}} & \cdots & \sin{2\cdot N \frac{2\pi}{N}} \\
        \vdots & \vdots & \ddots & \vdots \\
        \cos{\frac{N}{2}\frac{2 \pi }{N}} & \cos{\frac{N}{2} \cdot 2\frac{ 2\pi}{N}} & \cdots & \cos{\frac{N}{2}\cdot N \frac{2\pi}{N}} \\
        \sin{\frac{N}{2}\frac{2 \pi }{N}} & \sin{\frac{N}{2} \cdot 2 \frac{ 2\pi}{N}} & \cdots & \sin{\frac{N}{2} \cdot N \frac{2\pi}{N}} \\
        \end{bmatrix}.
\end{align*}
Note that $\nabla f_{\text{Re}} \left(\vec{\theta}_{0}\right)$ is a square matrix. Observe that the last matrix is proportional to the discrete Fourier transform matrix, which is an orthogonal transformation (i.e., a rotation). Additionally, the matrix $\text{diag}(1, \cdots \frac{N}{2})\otimes \text{diag}(1, -1)$ is diagonal and full-rank. We conclude that $\nabla f_{\text{Re}} \left(\vec{\theta}_{0}\right)$ is non-singular. Additionally, due to sub-multiplicativity of the operator norm, the maximum and minimum singular values are bounded as $\sigma_{max} \leq \frac{N}{2}$ and $\sigma_{min} \geq 1$, respectively. By the inverse mapping theorem, the function $f$ is a local bijection from some neighborhood $U$ of $\vec{\theta}_0$ to a neighborhood $V$ of $\vec{\alpha}_0$. 

Next, we find an estimate for the size of these neighborhoods. Consider the function 
    \begin{align*}
        g \colon &\BR^{N} \rightarrow \BR^{N} \\ &\vec{\theta} \mapsto  \left[\nabla f_{\text{Re}} (\vec{\theta}_0)\right]^{-1} f_{\text{Re}}(\vec{\theta} + \vec{\theta}_0),
    \end{align*}
which transforms the function $f_{\text{Re}}$ linearly such that $g(0, \cdots, 0) = \vec{0}$, and $\nabla g (0, \cdots, 0) = I$.
\begin{align*}
    \lnorm{\nabla g (\vec{\theta}) - \nabla g (\vec{\phi}) }_{\text{op}} &\leq \lnorm{\left[\nabla f_{\text{Re}} (\vec{\theta}_0)\right]^{-1} f_{\text{Re}}(\vec{\theta} + \vec{\theta}_0) - \left[\nabla f_{\text{Re}} (\vec{\theta}_0)\right]^{-1} f_{\text{Re}}(\vec{\phi} + \vec{\theta}_0) }_{\text{op}} \\ 
    & \leq \lnorm{\left[\nabla f_{\text{Re}} (\vec{\theta}_0)\right]^{-1}}_{\text{op}} \cdot \lnorm{ \ltup{f_{\text{Re}}(\vec{\theta} + \vec{\theta}_0) -  f_{\text{Re}}(\vec{\phi} + \vec{\theta}_0) }}_{\text{op}}
\end{align*}
For $\vec{\theta}, \vec{\phi} \in \BR^{N}$, 
\begin{align*}
    \lnorm{ \nabla f_{\text{Re}}(\vec{\theta}) -  \nabla f_{\text{Re}}(\vec{\phi})}_{\text{op}}  
    &= \sqrt{\sum_{1\leq l \leq N} \sum_{1 \leq k \leq N} \abs{f_{\text{Re}}(\vec{\theta})_{k,l} -  f_{\text{Re}}(\vec{\phi})_{k,l}}^2} \\
    &= \sqrt{\sum_{1\leq l \leq N} \sum_{1 \leq k \leq \frac{N}{2}} \abs{\frac{-k}{N}  \sin{ k \theta_l} -  \frac{-k}{N}   \sin{ k \phi_l}}^2 + \abs{\frac{k}{N}   \cos{ k \theta_l} -  \frac{k}{N}   \cos{ k \phi_l}}^2 }  \\
    &\leq \sqrt{\sum_{1\leq l \leq N} \sum_{1 \leq k \leq \frac{N}{2}} \frac{k^2}{N^2} \abs{ k \theta_l -   k \phi_l}^2 + \abs{ k \theta_l - k \phi_l}^2 }  \tag{by Lemma \ref{lem:sin_lipschitz}} \\
    &\leq \sqrt{ 2\frac{\ltup{\frac{N}{2}}^5}{N^2} \sum_{1\leq l \leq N} \abs{  \theta_l -   \phi_l}^2 }   \tag{since $k\le \frac{N}{2}$}\\
    &=  \frac{N^{3/2}}{4} \norm{\vec{\theta} - \vec{\phi}}_2  .\\
\end{align*}
Since $\lnorm{\left[\nabla f (\vec{\theta}_0)\right]^{-1}}_{\text{op}} \leq \frac{1}{\sigma_{min}} = 1$, it follows that 
\begin{align*}
    \lnorm{\nabla g (\vec{\theta}) - \nabla g (\vec{\phi}) }_{\text{op}}  \leq \frac{N^{3/2}}{4}  \lnorm{\vec{\theta} - \vec{\phi}}_2 .
\end{align*}
Let $\Bar{B}_r(\vec{0}) \subset \BR^N$ be a ball of radius $r$ with respect to the 2-norm centered at the origin $\vec{0} \in \BR^N$. If $r < \frac{4}{N^{3/2}}$, then
\begin{align*}
    \lnorm{\nabla g (\vec{\theta}) - \nabla g (\vec{\phi}) }_{\text{op}} \leq \frac{N^{3/2}}{4}  r < 1 \quad \text{for each}\quad \vec{\theta}, \vec{\phi} \in \Bar{B}_r(\vec{0}).
\end{align*}
Let $s:=\frac{N^{3/2}}{4}  r $ and suppose that $r =  \frac{4}{N^{3/2}} r_{scale} $ for some $0<r_{scale}<1$.  By Lemma \ref{thm:inv_mapping_ball}, if $y \in \BR^N$ and $\norm{y} \leq (1 - s )r $, there exists a unique $\vec{\theta} \in \Bar{B}_r(\vec{0})$ such that $g(\vec{\theta}) = y$. We are interested in maximizing the size of the neighborhood around $\vec{0} \in \BR^N$ for which $g^{-1}$ is a local bijection. The radius of the ball is given by
     \begin{align*}
         (1 - s )r = \ltup{1-\frac{N^{3/2}}{4}   \frac{4}{N^{3/2}} r_{scale} } \frac{4}{N^{3/2}} r_{scale} = \frac{4}{N^{3/2}} (r_{scale}- r_{scale}^2 )
     \end{align*}
which is maximized by $r_{scale} = 1/2$. Thus, for any $y \in \BR^{N}$ which is small enough $\norm{y}_2 \leq \frac{1}{N^{3/2}} $, there exists a unique vector $\vec{\theta} \in \Bar{B}_r(\vec{0})$ such that $g(\vec{\theta}) = y$. Recall that we made the transformation
    $$ g(\vec{\theta}) = \left[\nabla f (\vec{\theta}_0)\right]^{-1} f_{\text{Re}}(\vec{\theta} + \vec{\theta}_0). $$
Consider a vector $\vec{\alpha}_{\text{Re}} \in \BR^{N}$ such that $\norm{\vec{\alpha}_{\text{Re}}}_2 \leq \frac{1}{N^{3/2}} $, and let $y = \left[\nabla f (\vec{\theta}_0)\right]^{-1} \vec{\alpha}_{\text{Re}} $. Then 
$$\norm{y}_2 =  \lnorm{\left[\nabla f (\vec{\theta}_0)\right]^{-1} \vec{\alpha}_{\text{Re}}}_2 \leq  \lnorm{\left[\nabla f (\vec{\theta}_0)\right]^{-1} }_{\text{op}} \cdot \lnorm{\vec{\alpha}_{\text{Re}}}_2 \leq \frac{1}{N^{3/2}}. $$ 
It follows that there exists a unique vector $\vec{\theta} \in \Bar{B}_r(\vec{\theta}_0)$ such that $f_{\text{Re}}(\vec{\theta}) = \vec{\alpha}_{\text{Re}}$.
Recall that originally we were interested in the existence of a solution to the unitary moment problem for complex $\vec{\alpha} \in \BC^T$. Consider the vector $\vec{\alpha}_{\text{Re}}\oplus \vec{0}  \in \BR^{N}$ derived from the real an imaginary parts of $\vec{\alpha}$, defined by
\begin{align}
    \vec{\alpha}_{\text{Re}}\oplus \vec{0} = (\Im{\alpha_1} ,~ \dots, ~\Re{\alpha_T}, ~\Im{\alpha_T}, \underbrace{0 \dotsto 0}_{N-2T} ) \label{eq:padding}
\end{align}
 Observe that $\lnorm{\vec{\alpha}_{\text{Re}}\oplus \vec{0}}_2 = \lnorm{\vec{\alpha}_{\text{Re}}}_2 = \lnorm{\vec{\alpha}}_2$, which concludes the proof for even $N$.
\end{proof}

\begin{proof}[Proof of~\autoref{lem:discrete_unif_trunc_trig_prob_N} for odd $N$]
    To handle the case of $N$ odd, we append dummy values to the candidate moment vector such that dimension of the moment vector matches that of the input. Recall the real moment function $f_{\text{Re}}\colon \BR^{2T+1} \rightarrow \BR^{2T}$ defined by 
    \begin{align*}
         f_{\text{Re}}\colon (\theta_1 , \dots, \theta_N) &\mapsto   (\Re{\alpha_1},~ \Im{\alpha_1} ,~ \dots, ~\Re{\alpha_T}, ~\Im{\alpha_T} ).     
    \end{align*}
    One can also define the moment function $f_{\text{Re},2T+1}\colon \BR^{2T+1} \rightarrow \BR^{2T+1}$, where subscript $2T+1$ indicates the additional dimension of the range, defined by 
    \begin{align*}
         f_{\text{Re},N}\colon (\theta_1 , \dots, \theta_N) &\mapsto  (\Re{\alpha_1},~ \Im{\alpha_1} ,~ \dots, ~\Re{\alpha_{\frac{N+1}{2}}})
    \end{align*}
    Again, we consider this function around the point $(\theta_1 \dotsto \theta_N) = (\frac{2\pi}{N} \dotsto N\frac{2\pi}{N})$. The rest of the proof follows from the argument for $N=2T$. This leads to a final upper bound on the $l_2$ norm of the moment vector $\vec{\alpha}_{\text{Re}, 2T+1}  \in \BR^{2T+1}$ for a solution to the unitary moment problem to exist,
    $$ \norm{\vec{\alpha}_{\text{Re}, 2T+1}}_2 \leq \frac{\sqrt{2}}{4(T+1)^{3/2}}.$$
    Recall that the original candidate moment vector $\vec{\alpha} \in \BC^T$ is $2T$ dimensional. It is possible to construct a vector $\vec{\alpha}_{\text{Re}, 2T+1} \in \BR^{2T+1}$ whose first $2T$ entries match the first $2T$ real moments of $\vec{\alpha}$ and additionally has the same $l_2$ norm as $\vec{\alpha}$  like so
    $$\vec{\alpha}_{\text{Re}, 2T+1} := (\Re{\alpha_1},~ \Im{\alpha_1} ,~ \dots, ~\Re{\alpha_T}, ~\Im{\alpha_T}, 0).$$ 
    This concludes the proof for $N$ odd.
\end{proof}
Thus far, we have shown that solutions to the unitary moment problem exist when the 
$\ell_2$ norm of $\vec{\alpha}_T$ scales inversely with the dimension $N$ of the unitary, $\norm{\vec{\alpha}_T}_2 \leq \frac{1}{N^{3/2}}$. However, one would not expect that having more angles -- that is more options and finer granularity -- would further  constrain the problem. This intuition holds and in the following, we mold the previous result into a more generous bound on the norm which depends solely on the number of moments $T$, as claimed in Lemma \ref{lem:discrete_unif_trunc_trig_prob}.
\begin{proof}[Proof of Lemma \ref{lem:discrete_unif_trunc_trig_prob}]
    Again we let $\vec{\alpha}_T = (\alpha_1, \dots , \alpha_T) \in \BC^{T}$ be a vector of dimension $T$. We are now interested in a solution to the unitary moment problem for $N \geq 2T$ and possibly $N \gg 2T$. 
    For integers $r = \left\lfloor \frac{N}{2T} \right\rfloor $ and $b= N \mod 2T$, we have that $N = r \cdot 2T + b$. Our strategy is to break up this unitary moment problem for $N$ angles into $r$ unitary moment sub-problems for the same vector $\vec{\alpha}_T$ but for angles $N_1 = \cdots = N_{r-1} = 2T$ and $N_r = 2T+b$.
    By Lemma \ref{lem:discrete_unif_trunc_trig_prob_N}, if 
    \begin{align}
        \vec{\alpha}_T \leq \frac{1}{(N_j)^{3/2}}, \label{eq:cond_lj}
    \end{align}
    then there exist $\theta^j_{1}, \cdots, \theta^j_{N_j} \in \BR$ for all $1\leq j \leq r$ such that 
    \begin{align*}
        \frac{1}{N_j}\sum_{l=1}^{N_j} e^{i  k\theta^j_{l}} &= \alpha_k  &\text{for all} &\quad 1 \leq k \leq T, \quad \text{and} \\
         \frac{1}{N_j}\sum_{l=1}^{N_j} e^{i  k\theta^j_{l}} &= 0 &\text{for all} &\quad T < k \leq T+b/2
    \end{align*}
    where we implicitly padded the moment vector as in the \eqref{eq:padding}. Consider the $N$ dimensional vector obtained by concatenating the solutions to the $r$ sub-problems as
    \begin{align*}
        (\theta_1 \dotsto \theta_{N}) := \bigoplus_{j=1}^r \left( \theta_1^j,\ldots,\theta_{N_j}^j \right).
    \end{align*}  
    We claim that $\theta_1 \dotsto \theta_{N}$ is a solution to the unitary moment problem for $\vec{\alpha}_T$.
    Indeed for every $1 \leq k \leq T$,
    \begin{align*}
        \frac{1}{N}\sum_{l=1}^{N} e^{i  k\theta_l}  = \frac{1}{N}\sum_{j=1}^{r} \sum_{l=1}^{N_j} e^{i  k\theta^j_{l}}  = \frac{1}{N}\sum_{j=1}^{r} N_j \alpha_k   = \alpha_k.
    \end{align*}
    Since $b$ is at most $2T-1$, it follows that condition \eqref{eq:cond_lj} holds (and hence the unitary moment problem has a solution) as long as
    $$ \vec{\alpha}_T \leq \frac{1}{(4T)^{3/2}}.$$
\end{proof}

\subsection{Solutions to the General Unitary Moment problem: Proof of~\autoref{lem:unitary_moment_prob}}
    In the previous section, we first showed that the unitary moment problem has a solution when the $\ell_2$ norm of the moment vector scales inversely with the dimension before strengthening the result to only require the norm to scale inversely with the number of moments. It's natural to ask whether the constraint can be made even weaker. In fact, a \textit{constant} norm bound will suffice, at the expense of switching to the $\ell_1$ norm of the moment vector.
    
    As explained earlier, the unitary moment problem can also be thought of as a version of the trigonometric moment problem with some additional constraints, namely, that the solution be an \emph{integral atomic measure:}
    \begin{defn}[$N$-integral atomic measure]
        An \emph{$N$-integral atomic measure} is an atomic measure 
        $$\mu = \sum_{j=1}^k \frac{\xi_j}{N} \delta_{\theta_j}$$
        such that $\xi_1, \dots , \xi_k \in \mathbb{Z}^+$ are \emph{integers}.
    \end{defn}
    Although requiring an integral atomic measure does constrain the problem further, we find that it does not vastly alter the conditions under which solutions exist. As such, it is worth studying the classical trigonometric moment problem. The existence of a solution has already been well-characterized~\cite{schmudgen1991chapter11} and is intimately related to properties of the moment matrix:
    \begin{defn}[Hankel (Moment) matrix]
        Let $\vec{\alpha}: = (\alpha_0 \dotsto \alpha_n) \in \BC^{n+1}$ be a vector. The $(n+1)\times (n+1)$ Hankel (or moment) matrix $H_n(\vec{\alpha})$ is given by $H_n(\vec{\alpha}) = [h_{jk}]_{j,k \in \mathbb{N}_0}$ with entries $h_{jk} = \alpha_{k-j}$ for $j, k \in \mathbb{N}_0$. Here we have set $\alpha_{-l} = \bar{\alpha_l}$ for $l \geq 1$. 
        \[H_n(\vec{\alpha}) =
        \begin{bmatrix}
        \alpha_0 & \alpha_1 & \cdots & \alpha_n \\
        \alpha_{-1} & \alpha_0 & \cdots & \alpha_{n-1} \\
        \alpha_{-2} & \alpha_{-1} & \cdots & \alpha_{n-2} \\
        \vdots & \vdots & \ddots & \vdots \\
        \alpha_{-n} & \alpha_{-(n-1)} & \cdots & \alpha_0
        \end{bmatrix}.
        \]
    \end{defn}
    In particular, when the Hankel matrix is positive semidefinite, a solution exists:
    \begin{thm}[See \cite{schmudgen1991chapter11}]\label{thm:trig_moment_prob}
        Let $n \in \mathbb{N}_0$. For a sequence $\{\alpha_k\}^T_{k=0}$, the following are equivalent:
        \begin{enumerate}[label=(\roman*)]
            \item There is a Radon measure $\mu$ on $S^1$ such that
              \begin{align}
                  \alpha_k = \int z^j \, d\mu(z) \quad \text{for } j = 0, \ldots, T \label{eq:moment_cond}
              \end{align}
            \item There exists a $p$-atomic measure $\mu$ on $S^1$, $p \leq 2T + 1$, such that condition \eqref{eq:moment_cond} holds.
            \item The Toeplitz matrix $H_n(\vec{\alpha})$ is positive semidefinite.
            \item For all $c_0, \ldots, c_T \in \mathbb{C}$, 
            \[
            \sum_{j,k=0}^{n} s_{j}c_k\overline{c_j} \geq 0.
            \] 
        \end{enumerate}
    \end{thm}
    Note that (i) indicates a solution to the \emph{trigonometric moment problem} and (ii) indicates a solution to the \emph{truncated trigonometric moment problem}. 
    
    The extra constraint that the atomic measure be integral can be accommodated by slightly more stringent criteria on the positive definiteness of the Hankel matrix. Intuitively, this translates to the trigonometric problem having more ``wiggle room'' in the moment space. With this adjustment, it is possible to get \emph{approximate} solutions with integral atomic measures, (i.e.~the unitary moment problem). Then, using the localized but exact results from section \ref{subsec:solution_near_origin}, it is possible to massage this approximation into an exact solution. We begin with the following lemma which relates the trigonometric moment problem to that with an integral, atomic measure with just an $\CO(1)$ re-scaling of the $\ell_2$ norm.  
    \begin{lem}\label{lem:int_atomic_moment}
        Let $\vec{\alpha}_T := (\alpha_1 \dotsto \alpha_T)$. If there exists a Radon measure $\mu$ on $S^1$ such that
            \begin{align}
              s \alpha_k = \int z^j \, d\mu(z) \quad \text{for } j = 0, \ldots, T 
            \end{align}
        for some integer $\BZ \ni s > 1$. Then for some $N^\prime \leq 16(2T+1)T^{5/2} $ there exists an $N^\prime$-integral atomic measure $\mu^\prime$ on $S^1$, such that 
        \begin{align}
          \alpha_k = \int z^j \, d\mu^\prime(z) \quad \text{for } j = 0, \ldots, T.
        \end{align}
        Moreover, there also exists an $N$-integral atomic measure satisfying the condition for any $N \geq N^\prime$.
    \end{lem}
    We note that the integer $s$ can be taken to be $2$. By Fourier transforming, we can interpret this to mean that there is a gap below the lowest density point of the measure, i.e.~ there is ample room before the ``measure'' exhibits negative probability. With respect to the map from a measure on the unit circle to a sequence of moments, this corresponds to the map being invertible in a small neighborhood. 
    
    To show the existence of an integral atomic measure for a sequence $\{\alpha_k\}^T_{k=0}$, we will exhibit a measure $\mu^\prime$ which has the desired form and is close in total variation distance to $\mu$. Then we will employ Lemma \ref{lem:discrete_unif_trunc_trig_prob} to solve the moment problem for the error in the corresponding moment sequences, thereby solving the problem exactly. 
\begin{proof}[Proof of Lemma \ref{lem:int_atomic_moment}]
    Our strategy will be to decompose the vector 
    $\vec{\alpha}_T $ into a vector $\vec{\alpha}^\prime_T$ which is close to $\vec{\alpha}_T $ and their difference $\vec{\beta}_T $. Then if there exists some constant $s>1$ such that it is possible to solve the moment problems with integral atomic solutions for both
    $s \vec{\alpha}^\prime_T$ and $\frac{s}{s-1} \vec{\beta}$,
    then the solutions can be combined (with weights $\frac{1}{s}$ and $\frac{s-1}{s}$) to yield a integral atomic measure for $\vec{\alpha}_T $.

    Suppose there exists a $p$-atomic measure $\mu$ on $S^1$, $p \leq 2T + 1$, such that 
    \begin{align}
        s \alpha_k = \int z^j \, d\mu(z) \quad \text{for } j = 0, \ldots, T \label{eq:moment_cond3}
    \end{align}
    for some integer $\BZ \ni s>1$ (e.g.~s=2). 
    In other words, there exist coefficients $c_1 \dotsto c_k \in [0,1]$ of atoms $z_1 \dotsto z_p \in S^1$ on the unit circle which parameterize the measure as
    $$ \mu = \sum_{j=1}^p c_j \delta_{z_j}.$$
    Let $\omega_j \in [0,2\pi)$ be the argument of $z_j \in S^1$ such that $e^{i\omega_j} = z_j$ for all $1 \leq j \leq p$.
    In terms of these parameters, the first $T$ moments can be written as
    $$ \alpha_k = \sum_j c_j e^{i \omega_j k} \quad \text{for all} \quad 0 \leq k \leq T.$$
    Consider the $N$-integral atomic measure, 
    $$ \mu^\prime := \sum_j \frac{\xi_j^\prime}{N} \delta_{z_j}$$ 
    where the integers $\xi_1^\prime \dotsto \xi_p^\prime \in \{0 \dotsto N\}$ are chosen such that the total variation distance between $\mu^\prime$ and $\mu$ is minimized. It is clear that $ \abs{\xi_j^\prime - c_j } \leq 1/N $.
    Let $\alpha_k^\prime$ indicate the moments of this measure $\mu^\prime$ scaled down by $\frac{1}{s}$,
    $$ \alpha_k^\prime = \frac{1}{s} \int z^j \, d\mu^\prime(z) \quad \text{for } j = 0, \ldots, T .$$
    Then for each $0 \leq k \leq T$,
    \begin{align*}
        \abs{\alpha_k - \alpha_k^\prime} &= \abs{ \frac{1}{s} \sum_j \frac{\xi_j^\prime}{N} e^{i \omega_j k} - \frac{1}{s} \sum_j  c_j e^{i \omega_j k}} \\ 
        &\leq \sum_j \abs{\frac{\xi_j^\prime}{N}   -c_j} \abs{e^{i \omega_j k}} \\
        &\leq \sum_j \frac{1}{N} \\
        &= \frac{p}{sN}.
    \end{align*}
    It follows that 
    $$\norm{\vec{\alpha}_T -\vec{\alpha}^\prime_T}_2 \leq \norm{\vec{\alpha}_T -\vec{\alpha}^\prime_T}_1 =\sum_{k=1}^{T} \abs{\alpha_k - \alpha_k^\prime} \leq \frac{T p}{s N}.$$
    Let $\vec{\beta}_T := \ltup{\beta_1 \dotsto \beta_T}$ be the difference in the moment vectors $\vec{\alpha}_T$ and $\vec{\alpha}^\prime_T$ defined by
    $$ \beta_k := \alpha_k - \alpha_k^\prime \quad \text{for each}\quad k = 1,\cdots,T.$$
    By Lemma \ref{lem:discrete_unif_trunc_trig_prob}, if 
    \begin{align}
    \norm{\frac{s}{s-1} \vec{\beta}_T}_2 \leq \frac{1}{8T^{3/2}} \label{eq:valid_beta}
    \end{align}
    then there exist $\phi_1 \dotsto \phi_{2T} $ such that 
    $ \frac{s}{s-1} \beta_k = \frac{1}{2T}\sum_{j=1}^{2T} e^{i \theta_j k}$ for each $1 \leq k \leq T.$ Observe that the inequality \eqref{eq:valid_beta} is satisfied for \emph{any} $N \geq \frac{s-1}{s} \cdot 8(2T+1)T^{5/2}$:
    \begin{align*}
        \vec{\beta}_T \leq \frac{s-1}{s}  \frac{Tp}{N} \leq T(2T+1) \cdot \frac{1}{8(2T+1)T^{5/2}} = \frac{1}{8T^{3/2}},
    \end{align*}
    using that $p \leq 2T+1$. Supposing that $N \geq \frac{s-1}{s} \cdot 8(2T+1)T^{5/2}$, it is possible to combine the solutions to the moment problem for $\vec{\beta}_T$ and $\vec{\alpha}_T^\prime$ to get an integral, atomic measure that solves the moment problem for the original vector $\vec{\alpha}_T$. 

    Define the angles  $\theta_1 \dotsto \theta_{p+2T} \in [0,2\pi) $ by 
    \begin{align*}
        \theta_j = \begin{cases}
            \theta_j = \omega_j & 1\leq j \leq p \\
            \theta_j = \phi_{j-p} & p < j \leq p + 2T
        \end{cases}
    \end{align*}
    and weights $\xi_1 \dotsto \xi_{2T+p} \in [0,1]$ by
    \begin{align*}
        \xi_j = \begin{cases}
            \xi_j = 2T \xi_j^\prime & 1\leq j \leq p \\
            \xi_j = N (s-1)  & p < j \leq p + 2T
        \end{cases}
    \end{align*}
     Then the $2TNs$-integral $2T+p$-atomic measure 
     $$ \mu^f = \frac{1}{2TNs} \sum_{j=1}^{p+2T} \xi_j \delta_{e^{i \theta_j}} $$
     gives rise to the moment sequence $\alpha_1 \dotsto \alpha_T$:
     \begin{align*}
         \frac{1}{2T+N} \sum_j  \sum_{j=1}^{p+2T} \xi_j e^{i \theta_j k} &= \frac{1}{2TNs}  \sum_{j=1}^{p} 2T \xi_j^\prime e^{i \omega_j k} + \frac{1}{2TNs} \sum_{j=1}^{2T}  N (s-1) e^{i \phi_j k}\\  
         &= \frac{1}{s} \sum_{j=1}^{p} \frac{\xi_j^\prime}{N} e^{i \omega_j k} + \frac{s-1}{s}  \sum_{j=1}^{2T} \frac{1}{2T}  e^{i \phi_j k}\\ 
         &= \frac{1}{s} s \vec{\alpha}^\prime_T + \frac{s-1}{s } \cdot \frac{s}{s-1 }  \vec{\beta}_T\\ 
         &= \vec{\alpha_T}\\ 
     \end{align*}
     Taking $s=2$, we have a guarantee that a solutions exists if $N \geq \frac{s-1}{s} \cdot 8(2T+1)T^{5/2} = 16(2T+1)T^{5/2}$.
\end{proof}

A natural consequence of this is Lemma \ref{lem:unitary_moment_prob} which we restate below for convenience. Note that this is merely a relaxation of Lemma \ref{lem:int_atomic_moment}, but its form is more useful throughout this work.
\unitarymomentprob*
Observe that condition \ref{eq:unitary_moment_prob_cond} is equivalent to the following statement which we focus on in the proof:
\begin{quote}\itshape
     For any $N \ge   16(2T+1)T^{5/2} $ there exists an $N$-integral measure $\mu$ on $S^1$ such that 
    \begin{align}
      \alpha_k = \int z^j \, d\mu(z) \quad \text{for } j = 0, \ldots, T. \label{eq:moment_cond2}
    \end{align}
\end{quote}
\begin{proof}
To make the connection between an integral atomic measure, and the $l_1$-norm, we make use of statements $(ii)$ and $(iii)$ from Theorem \ref{thm:trig_moment_prob} which state that 
the Toeplitz matrix $H_T(\vec{\alpha}_T)$ is positive semidefinite if and only if there exists a $p$-atomic measure $\mu$ on $S^1$, $p \leq 2T + 1$, such that 
\begin{align}
  \alpha_k = \int z^j \, d\mu(z) \quad \text{for } j = 0, \ldots, T. \label{eq:moment_cond4}
\end{align}
We are interested in an atomic measure for the \emph{scaled} moments $s \vec{\alpha}$. Particularly, we are interested in the smallest such $s$, which is 2.
As such, it suffices to determine conditions such that $H_T(2\vec{\alpha}_T)$ is positive semidefinite. Note that a hermitian matrix which is diagonally dominant with real non-negative diagonal entries is positive semidefinite. This follows from Gershgorin's circle theorem. The Hankel matrix has diagonal entries of $1$, so is diagonally dominant if for all rows $j \in \lset{1, \dotsto T+1}$,
\begin{align*}
    \sum_{i = 1}^{j} \abs{2 \bar{\alpha_j}} + \sum_{i = 1}^{T+1-j} \abs{2\alpha_j} \leq 1.
\end{align*}
A trivial upper bound on this is $4\vec{\alpha}_T \leq 1$, since any individual value can only appear at most twice (allowing for conjugates) in any particular row. Therefore, a $p$-atomic measure $\mu$ on $S^1$, $p \leq 2T + 1$, such that condition \eqref{eq:moment_cond4} holds exists if $\vec{\alpha}_T \leq \frac{1}{4}.$ Suppose that this bound holds.
By Lemma \ref{lem:int_atomic_moment}, for some $N \leq 16(2T+1)T^{5/2} $, there exists an $N$-integral atomic measure $\mu^\prime$ on $S^1$ with, such that condition \eqref{eq:moment_cond4} holds.
\end{proof}


\section{From Small Moments to the Haar Measure}\label{sec:small_moments_haar}

Consider two deterministic unitaries $\vD_1,\vD_2$ and apply an (independent) Haar random change of basis $\vU, \vU'$.
\begin{align}
    \vW_1 &= \vU \vD_1 \vU^{\dagger}\\
    \vW_2 &= \vU' \vD_2 \vU^{'\dagger}.
\end{align}
How easy is it to distinguish the two by quantum algorithms?
We can think of the above assumptions as the statement that ``the eigenbases of $\vW_1$ and $\vW_2$ are Haar-distributed.'' Intuitively, since their bases contain no information, we might expect the ensembles $\vW_1$ and $\vW_2$ to only be distinguishable by their spectra, or equivalently, their trace moments $\tr(\vD_1^p)$, $\tr(\vD_2^p)$ for $1 \leq p \leq N$. 
Then, if the trace moments $\tr(\vD_1^p)$ and $\tr(\vD_2^p)$ are close in an appropriate sense, we might further expect them to be indistinguishable. This section formalizes that intuition for the black box quantum query complexity of distinguishing $\vW_1$ from $\vW_2$. For our purpose of constructing unitary designs, one of the unitaries will be Haar random, which is indeed Haar-conjugate invariant. Additionally, the trace moments of a Haar random unitary have been characterized by Diaconis and Evans. In particular, its normalized trace moments are small, roughly $\CO(\frac{k^2}{N})$ for the $k^{th}$ trace moment, with high probability. Thus, if the random unitary $\vW$ also has small moments, $T$-query quantum algorithms will not be able to distinguish $\vW$ from $\vU$ with an appreciable probability of success.

\smalltohaar*

Roughly, if the moments $\alpha_p$ are small, the number of queries $T$ required to distinguish the distributions must be large. Lemma \ref{lem:moments_implies_Haar} is proved by first showing that the number of queries $T$ required to distinguish Haar-conjugate-invariant ensembles which have \emph{small} moments from Haar-conjugate-invariant ensembles with moments \emph{exactly zero} must be large. (From section \ref{sec:products_of_gaussians}, it was shown that the proposed ensemble of products of exponentiated Gaussians has small moments.) Second, it is known \cite{Diaconis} that Haar unitaries have small moments in expectation with small deviations. Thus, the triangle inequality implies that distinguishing a Haar random unitary from $\vW$ requires many queries. 

One of the key ingredients in proving the quantum query lower bounds is the following relationship between quantum algorithms and polynomials. By relating quantum algorithms to polynomials, we are able to use tools from approximation theory to bound the derivatives of these polynomials based on their degree. By an argument similar to that of Lemma 4.1 of~\cite{beals2001quantum}, 
\begin{lem}\label{lem:T_query_alg_to_poly}
Let $\CA$ be a quantum algorithm that makes $T$ queries to a black-box unitary $\vX \in \text{U}(N)$. Then, the final state of the algorithm can be represented by 
    $$\sum_{k\in K} p_k\ltup{\vX}|k\rangle,$$
where $p_k: \BC^{N \times N} \rightarrow \BC$ for all $k\in K$ are complex-valued multilinear polynomials in the entries $x_{ij} := \bra{i}\vX\ket{j}$ for $1 \leq i,j \leq N$, each of degree at most $T$. 
\end{lem}
Restricting to decision algorithms yields the following corollary.
\begin{cor}\label{cor:acc_prob_T}
    Let $\CA$ be a quantum decision algorithm that makes $T$ queries to a black-box unitary $\vX \in \text{U}(N)$. Then the final amplitude on the accepting state is a real-valued multilinear polynomial $p \colon \BC^{N \times N} \rightarrow \BR$ in $x_{ij} := \bra{i}\vX\ket{j}$ for $1 \leq i,j \leq N$ of degree at most $T$.
\end{cor}
Recall that we are interested in unitary oracles with additional structure; namely, that they are unitarily invariant under conjugation. This structure constrains the space of quantum algorithms that query such oracles. In particular, the polynomials representing the final acceptance amplitudes are restricted to the trace moments of the queried unitary. To formalize these notions, we first define the \emph{unitary symmetrization} of a complex-valued polynomial in $N\times N$ variables.
\begin{defn}[Unitary Symmetrization]\label{def:unitary_sym}
    Let $p\colon \BC^{N\times N} \rightarrow \BC$ be a multivariate polynomial where the argument is viewed as the entries of an $N\times N$ matrix $\vX$. The unitary symmetrization of $p$ is defined as $$ p^{usym}\ltup{\vX} = \int_{\text{U}(N)} p\ltup{ \vU \vX \vU^\dag }    \diff \vU. $$
\end{defn}
Note that if $p(\vX) = p(\vV\vX \vV^\dag)$ for any $\vV \in \text{U}(N)$, then $p = p^{usym}$. The unitary symmetrization of a polynomial $p$ has a simple form: it is a polynomial of degree $T$ in the trace moments of the arguments of $p$. 

\begin{lem}\label{lem:from_entries_to_moments}
    Let $p\colon \BC^{N\times N} \rightarrow \BC$ be a multivariate polynomial of degree $T$. There exists a multivariate polynomial $q \colon \BC^T \rightarrow \BC$  of degree at most $T$ such 
    	$$ p^{usym}\ltup{\vX} = q(\alpha_1, \cdots, \alpha_T) \quad \forall~ \vX \in \BC^{N\times N}$$ 
    where $\alpha_1, \cdots, \alpha_T$ are the first $T$ normalized trace moments of $X$.
\end{lem}

\begin{proof}
        By definition, $p(X)$ is a complex-linear combination of terms $ c_1\prod_{h=1}^k x_{i_h j_h}$ where $c_1 \in \BC$, $k<T$ and $1\leq i_h, j_h \leq N $. Under unitary symmetrization, 
	\begin{align*}	
		\prod_{h=1}^k x_{i_h j_h} \mapsto \int_{\text{U}(N)} \prod_{h=1}^k \bra{i_h}\vU \vX \vU^\dag \ket{j_h}  \diff \vU .
	\end{align*}
        By evaluating $p^{ysym}$ using the Weingarten calculus, this will be a polynomial in the first $k$ traces moments  $\btr(\vX^1) \dotsto \btr(\vX^k) $ of $\vX$, where each term takes on the form 
        \begin{align*}
            c_2 \prod_{j=1}^k \btr(\vX^j)^{a_j} \quad\quad \text{s.t.} \quad\quad \sum_{j=1}^k j a_j = k, a_j \in \BZ_{+}
        \end{align*}
        for $c_2 \in \BC$. See appendix \ref{app:weingarten} for details. It follows that $p^{usym}$ is a degree $T$, multivariate polynomial in the trace moments $\btr(X^k)$ of the argument $X$ for $1 \leq k \leq T$.
\end{proof}
The proof of Lemma~\ref{lem:moments_implies_Haar} uses the following theorem from approximation theory:
\begin{thm}[Markov Brothers' Inequality \cite{markov1889, markov1916, cheneyrivlin1966, ehlich1964} ]\label{thm:markov_brother}
    Let $p \colon \BR \to \BR$ be a polynomial of degree $d$. 
    $$ \max_{-1 \leq x \leq 1} \abs{p^\prime (x)} \leq \max_{-1 \leq x \leq 1}   d^2 \abs{p(x)}.$$
\end{thm}

With these tools, we are now ready to prove the following statement, which is the foundation of Lemma \ref{lem:moments_implies_Haar}.
\begin{lem}\label{lem:small_moments_from_zero}
    For a Haar random unitary $\vU \in \text{U}(N)$ and diagonal matrix $\vZ \in \unitary(N)$, let $\vM$ be a random unitary distributed as 
    $$ \vM  \stackrel{dist}{\sim}  \vU \vZ\vU^{\dagger} \quad \text{s.t. }\quad \btr(\vZ^k) = 0 \quad \forall 1\leq k \in \leq N. $$
    Consider another random unitary $\vW$ as diagonal unitary $\vY \in \unitary(N)$ conjugated by Haar random unitary $\vV \in \unitary(N)$,
     $$ \vW  \stackrel{dist}{\sim}  \vV \vY \vV^{\dagger}  $$
     with trace moments
    $$\vec{\alpha}_T := (\alpha_1, \dots , \alpha_T), \quad \alpha_k = \btr(\vY^k).$$ 
    No $T$ query quantum algorithm $\CA$ can distinguish $\vW$ from $\vM$ with probability greater than $16  \norm{\vec{\alpha_T}}_1 T^{7/2}$.
\end{lem}

\begin{proof}
Consider any quantum algorithm $\CA$ which makes at most $T$ queries to a black-box unitary $\vX \in \text{U}(N)$ that is invariant under unitary conjugation. Then by corollary \ref{cor:acc_prob_T}, there exists a multivariate polynomial $p\colon \BC^{N\times N} \rightarrow [0,1]$ of degree at most $T$ which represents the final amplitude of  $\CA$ on $\ket{0}$. Moreover, $\vW$ and $\vM$ are invariant under unitary conjugation. It follows that $p = p^{usym}$, and by Lemma \ref{lem:from_entries_to_moments}, there must exist some multivariate polynomial $q \colon \BC^T \rightarrow \BC$  of degree $T$ such that
	$$ p\ltup{\vX} = p^{usym}\ltup{\vX} = q\ltup{\btr(\vX^1), \cdots, \btr(\vX^T)} \quad \forall~ \vX \in \BC^{N\times N}.$$ 

Recall that the polynomial $p$ represents the \emph{amplitude} of the accepting state $\ket{0}$. The \emph{probability} of acceptance of the $T$ query algorithm $\CA$ on input $X$ is therefore 
    \begin{align*}
        \Pr[\CA(X) = 1 ] &= p(X) \overline{p(X)}\\
        &= q\ltup{\btr(\vX^1), \cdots, \btr(\vX^T)} \overline{ q\ltup{\btr(\vX^1), \cdots, \btr(\vX^T)}}.
    \end{align*}
Let $q_{acc}  \colon \BC^{T} \rightarrow \BR$ denote the acceptance probability of $\CA$ given the first $T$ moments of the input $X$,
\begin{align*}
    q_{acc}\ltup{\btr(\vX^1), \cdots, \btr(\vX^T)} := q\ltup{\btr(\vX^1), \cdots, \btr(\vX^T)} \overline{ q\ltup{\btr(\vX^1), \cdots, \btr(\vX^T)}}.
\end{align*}
Observe that $q_{acc}$ a polynomial in the first $T$ moments of $X$ and has degree at most $2T$.

Next, consider the scaled moments $(t \alpha_1, \cdots, t\alpha_T)$, and the function which maps a scaling factor $t$ to the "acceptance amplitude" of the scaled moments,
    \begin{align*}
        h\colon &\BR \rightarrow \BR \\
         & t \mapsto q_{acc}(t \alpha_1, \cdots, t\alpha_T).
    \end{align*}
Clearly, $h$ is a polynomial in $t$ of degree at most $T$. The vector $(t \alpha_1, \cdots, t \alpha_T)$, is not necessarily a moment vector and $q(t \alpha_1, \cdots, t\alpha_T)$ may not be a valid "acceptance amplitude", i.e., there may not exist a unitary $\vU \in \text{U}(N)$ such that $\btr(\vU^k) = t \alpha_k$ for all $1\leq k \leq T$. However, under certain conditions, $h(t)$ \emph{does} correspond to the acceptance amplitude of a quantum algorithm. By Lemma \ref{lem:discrete_unif_trunc_trig_prob}, if the norm of the moment vector is sufficiently small,
    $$\norm{(t \alpha_1, \cdots, t\alpha_T)}_1 \leq \frac{1}{(3T)^{3/2}},$$
then for any $N \geq 2T$ there exist $\theta_1, \dots , \theta_N  \in \BR$ such that,
    $$ t\alpha_k =  \frac{1}{N}\sum_{l=1}^N e^{i k\theta_l} \text{ for all }  1 \leq k \leq T.$$
The inequality holds up to 
$$ t_{max} := \frac{1}{\lnorm{\vec{\alpha}_T}_1 8T^{3/2}}.$$
Let $\vU_t \in \text{U}(N)$ denote the diagonal matrix  $\text{diag} \ltup{\theta^t_1, \cdots ,\theta^t_N}$ for  $t < t_{max}$ where $(\theta^t_1, \cdots ,\theta^t_N) \in \BR^N$ is the solution to the discrete, uniform, trigonometric moment problem for the moment vector $(t \alpha_1, \cdots, t\alpha_T)$. It follows that,  
$$ h(t) = q(t \alpha_1, \cdots, t\alpha_T) = p\ltup{\vU_t}.$$
The polynomial $q_{acc}$ represents the acceptance probability of a quantum algorithm $\CA$, so for unitary input, the image of $q_{acc}$ \emph{must} be contained in $[0,1]$, and so must that of $h(t)$,
$$ h(t) = q_{acc}\ltup{\btr{\vU_t}, \dots, \btr{\vU_t^T}} \in [0,1] \quad \text{for} \quad t \in [0, t_{max}]$$ 
Additionally, recall that $\text{deg}(h) \leq 2T.$ By Markov Brother's Inequality, 
    \begin{align}
        \max_{-1 \leq x \leq 1} \abs{\frac{\partial}{\partial x} h (t_{max} x)} &\leq \max_{-1 \leq x \leq 1}   d^2 \abs{h (t_{max} x)}, \nonumber \\
        \max_{-t_{max} \leq t \leq t_{max}} \abs{h^\prime (t)} &\leq \frac{(2T)^2 }{t_{max}}. \label{eq:deriv_upper_bound}
    \end{align}
Recall that the quantum algorithm $\CA$ aims to distinguish $(0, \cdots 0)$ from $\vec{\alpha_T}$. To do so with probability at least $\delta$ requires that
\begin{align}
    \abs{h(0) - h(1)} = \abs{q_{acc} (0, \cdots 0) - q_{acc} (\alpha_1, \cdots, \alpha_T)} \geq \delta . \label{eq:dist_requirement}
\end{align}
By the mean value theorem, if equation \eqref{eq:dist_requirement} holds, then there exists a $t \in [0,1]$ such that 
$$ h^\prime (t) = h(1) - h(0) \geq \delta . $$
This is a contradiction to equation \eqref{eq:deriv_upper_bound} if $t_{max} > \frac{ (2T)^2}{2\delta}$, or equivalently if  
$$ \norm{\vec{\alpha_T}}_1 <   \frac{2\delta}{8T^{3/2} \cdot (2T)^2 } < \frac{\delta}{16 T^{7/2} }. $$
\end{proof}

Now we can return to proving Lemma \ref{lem:moments_implies_Haar} which compares Haar random unitaries to random unitaries $\vW$ which are invariant under unitary conjugation and have small moments. The traces moments of Haar random unitaries have already been studied by Diaconis and Evan:
\begin{thm}[Trace moments of Haar random unitaries \cite{Diaconis}]\label{thm:haar_moments}
        Consider $a = (a_1, \ldots, a_k)$ and $b = (b_1, \ldots, b_k)$ with $a_j, b_j \in \{0, 1\}$. Let $\vU_j \in \unitary(N)$ for $1\leq j \leq k$ be independent Haar random unitaries. 
        \begin{enumerate}
        \item Then for $N \geq \max\left( \sum_{j=1}^{k} j a_j, \sum_{j=1}^{k} j b_j\right)$,
        \[
        \Expect \left[ \prod_{j=1}^{k} \left(\tr (\vU^j)\right)^{a_j} \cdot \overline{\left(\tr (\vU_j)\right)^{b_j}} \right] = \prod_{j=1}^{k}\ltup{\delta_{a_j b_j} ~j^{a_j}~a_j!} .
        \]
        \item For any $j, k$,
        \[
        \Expect \lbr{\tr(\vU^j) \overline{\tr (\vU^k)}} = \delta_{jk} \min (j, N).
        \]
    \end{enumerate}
\end{thm}
In particular, the $k^{th}$ moments of $\sqrt{j}\tr(\vU^j)$ are exactly complex Gaussian for $k \leq N$. To prove Lemma~\ref{lem:moments_implies_Haar} we are interested in the $\ell_1$ norm of the moment vector. By theorem \ref{thm:haar_moments}, for $T< N$.
\begin{align}
    \Expect_{\vU \leftarrow \mu} \sum_{k=1}^{T}\abs{\btr(\vU^k)} 
    \leq \sum_{k=1}^T \sqrt{\Expect_{\vU \leftarrow \mu} \btr(\vU^k)\overline{\btr(\vU^k)}} 
    = \sum_{k=1}^T \sqrt{\frac{k}{N^2}}
    \leq \frac{1}{N} \int_1^{T+1} \sqrt{k} \, dk
    \leq \frac{2 (T+1)^{3/2}}{3N}.
    \label{eq:haar_moment_exact}
\end{align}
The first inequality follows from the concavity of the square root function.

Additionally, an analog of Levy's lemma for the unitary group yields the following concentration bound.
\begin{lem}[Concentration of Haar unitary moments]\label{lem:conc_haar_moments}
    Let $\mu$ be the Haar measure over the unitary group. Then
    \begin{align} \label{unitary-concentration}
    \Pr_{\vU \leftarrow \mu} \lbr{ \sum_{k=1}^T \abs{\btr(\vU^k)} - \Expect_{\vU \leftarrow \mu} \sum_{k}\abs{\btr(\vU^k)} > \delta} \leq  \exp\left( - \frac{N^2 \delta^2}{12 T^2} \right).
\end{align}
    where $N$ is the dimension of $\vU$. An identical upper bound holds for the opposite one-sided deviation.
\end{lem}
\begin{proof}
    The following general concentration bound holds for $\kappa$-Lipschitz functions on $U(N)$. (See \cite{meckes2013spectral} or, for self-contained expositions, \cite{meckes2019random} and Appendix B of \cite{akers2022black}.) It is a special case of \autoref{prop:many_unitary_concentration}.
    \begin{align*}
    \Pr_{\vU} \left[ f(\vU) - \Expect f(\vU)  \geq \delta \right]
    \leq  \exp \left( - \frac{N\delta^2}{12 \kappa^2} \right).
    \end{align*}
    We need to bound the Lipschitz constant $\kappa_k$ of the function $g_k(\vU) = \btr(\vU^k)$. Using a telescoping sum, we have  that
    \begin{align*}
    \left| g_k(\vU) - g_k(\vV) \right|
    &\leq  \sum_{j=1}^{k-1} \left| \btr \left( \vU^{k-j} (\vU - \vV) \vV^{j-1} \right) \right| \\
    &\leq \frac{k}{N} \max_{\| \vX \|_{op} \leq 1} \tr \left( \vX  ( \vU - \vV ) \right) \\
    &= \frac{k}{N} \| U - V \|_1.
    \end{align*}
    Then, since $\| \vU - \vV \|_1 \leq \sqrt{N} \| \vU - \vV \|_2$, we find that $\kappa_k  \leq  k / \sqrt{N}$. It follows that the Lipschitz constant $\kappa$ of $\sum_{k=1}^T g_k$ is bounded above by $T^2 / \sqrt{N}$. Substituting into \eqref{unitary-concentration} finishes the proof.    
\end{proof}
We will now assemble these components into the main objective of this section, a proof that ensembles of isospectral matrices invariant under unitary conjugation cannot be distinguished from the Haar measure, provided their moment vectors are small. 
\begin{proof}[Proof of Lemma \ref{lem:moments_implies_Haar}]
    Let $\vec{\alpha}^{Haar}_T$ be the moment vector for the first $T$ moments of a Haar unitary. For $T \leq N$, by Lemma \ref{lem:conc_haar_moments} and equation \eqref{eq:haar_moment_exact}.
    \begin{align*}
        \Pr [\norm{\vec{\alpha}^{Haar}_T}_1 \geq  \frac{2 (T+1)^{3/2}}{3N} + \frac{\gamma \sqrt{12} T}{N}    ]
        \leq \exp( - \gamma^2 ).
    \end{align*}
    We will choose $\gamma = \sqrt{\log(4/\delta)}$ so that the probability of a large moment vector is bounded above by $\delta/4$. Otherwise, the norm of the moment vector is bounded above by
    $B_{\delta} :=  \frac{2 (T+1)^{3/2}}{3N} + \frac{\gamma \sqrt{12} T}{N}$.
    By Lemma \ref{lem:small_moments_from_zero}, if $\norm{\vec{\alpha}^{Haar}_T}_1 \leq B_{\delta} $, no $T$-query quantum algorithm $\CA$ can distinguish $\vU$ from $\vM$ with probability greater than $16 B_{\delta} T^{7/2}$. There exists a constant $C$ such that whenever
    \begin{equation}
        \frac{T^5}{N} \leq \frac{C\delta}{\sqrt{\log{(4/\delta)}}},
    \end{equation}
    this $16 B_{\delta} T^{7/2}$ will bounded above by $\delta/4$ so that the total probability of distinguishing $\vU$ from $\vM$ will be at most $\delta/2$. 

    Now let $\vec{\alpha}^{\vW}_T$ be the moment vector for the first $T$ moments of $\vW$. If
    \begin{align*}
        \norm{\vec{\alpha}^{\vW}_T}_1&\leq  \frac{\delta}{32 \cdot T^{7/2}},
    \end{align*}
    then again by Lemma \ref{lem:small_moments_from_zero}, there exists no $T$ query quantum algorithm $\CA$ that distinguishes $\vW$ from $\vM$ with probability greater than $\delta/2$. 

    Since neither $\vU$ nor $\vW$ can be distinguished from $\vM$ with probability greater than $\delta/2$, they can't be distinguished from each other by probability greater than $\delta$.     
\end{proof}


\section{Product of Two Exponentiated GUE Matrices: \texorpdfstring{\\}{} Motivation for Lemma \ref{lem:clt_spectrum_small_moments}}\label{sec:products_of_gaussians}

Our construction for a Haar random unitary design involves first creating Gaussian designs. In this section, we motivate our interest in Gaussian designs: namely, cleverly-taken products of exponentiated Gaussians have small moments, and hence by \autoref{lem:moments_implies_Haar}, are close to Haar. For the rest of this section, let $\vG$ be drawn from the GUE. The particular construct we have in mind is the product of independent Gaussian exponentials
\begin{align*}
    \vW=e^{i  \vG_1 \theta }\cdot e^{i  \vG_2 \theta } \quad \text{for precisely chosen $\theta$ such that}  \quad \btr[e^{i  \vG \theta }] \approx 0.
\end{align*}
The individual Gaussian exponentials have independent Haar-random eigenbases but pushforwards of the GUE semicircle spectra. Multiplying them manages to wash away that spectral structure, a manifestation of how incommensurate the two bases are.

This section is dedicated to studying the properties of both an individual exponentiated Gaussian and the product of two. We discuss how an exponentiated Gaussian is invariant under conjugation, and how this attribute helps control the trace moments of the product of two. While strictly speaking, our particular hybrid does not contain exact Gaussians from GUE, the results in this section play an important role in the proof of \autoref{lem:clt_spectrum_small_moments} in \autoref{sec:lindeberg_spectrum}.

\subsection{Properties of Gaussian Exponentials}\label{subsec:properties_of_exp}
For ease of analyzing this product of two independent Gaussian exponentials, we rewrite $\vW$ as
\begin{align*}
    \vW=\vU\vD_1\vU^\dagger\vV\vD_2\vV^\dagger \quad \text{such that}\quad \btr[\vD_1],\btr[\vD_2] \approx 0,
\end{align*}
where $\vU$ and $\vV$ are from Haar random unitary ensembles, and the diagonal matrices $\vD_i$ contain the eigenvalues. The two key properties of the single exponential ensemble $e^{i \vG  \theta}$ that we will need are, schematically:
\begin{enumerate}
    \item Ensemble invariance under unitary conjugation
    \item Existence of $\theta$ such that the trace moments are small as functions of $N$.
\end{enumerate}
We dive into these conditions individually.
\begin{lem}[Invariance under unitary conjugation]\label{lem:conjugation_invariance_single} The Gaussian Unitary Ensemble of matrices $\vG$, is invariant under conjugation of Haar random unitaries. That is,
\begin{align*}
   e^{i \vG  \theta} \stackrel{dist.}{\sim} \vU\vD\vU^\dagger,
\end{align*}
where $\vU$ is a Haar random unitary, and $\vD$ is an independent random diagonal matrix with entries $e^{i\lambda_i\theta}$ where $\vec{\lambda} = (\lambda_i)$ is drawn from the GUE spectral distribution (with arbitrary ordering).
\end{lem}
We emphasize that this lemma allows us to break the ensemble of $e^{\ri \vG\theta}$ into two \textit{independent} ensembles: $\vU$ and $\vD$. In the future, we will need to average just over one ensemble ($\vU$), but not the other ($\vD$), and this subtlety is important. We also want more information about the trace of instances of this ensemble, which is really information about the $\vD$ ensemble. We begin with the expected trace.

\begin{lem}[Exponentiated semicircular moments]\label{lem:exp_semicircle_moments} For a random matrix $\vG$ drawn from the Gaussian Unitary Ensemble, 
\begin{align*}
    \lim_{N\rightarrow\infty}\mathop{\BE}_{\text{GUE}}\btr[e^{i \vG  \theta}] = \frac{J_1(2\theta)}{\theta}.
\end{align*}
\end{lem}
\begin{proof}
As per \autoref{prop:catalan} in appendix \ref{app:GUE_properties}, a GUE matrix has the property that as $N\rightarrow\infty$, the odd normalized moments vanish for all $N$, and the even normalized moments $2n$ are equal to the number of non-crossing pairings of $2n$ objects, which is equivalent to the Catalan numbers $\mathrm{Cat}_n$ \cite{speicher2020lecture}: 
\begin{align*}
    \mathop{\BE}_{\text{GUE}}\btr[\vG^{2n}]=\mathrm{Cat}_n+\mathcal{O}(N^{-2})\quad\text{where} \quad\mathrm{Cat}_n:=\frac{1}{n+1}{\binom{2n}{n}}.
\end{align*}
Then the exponential trace as $N\rightarrow\infty$ is 
\begin{align*}
    \lim_{N\rightarrow\infty}\mathop{\BE}_{\text{GUE}}\btr[e^{\ri \vG  \theta}] &= \btr\left[\sum_{n=0}^\infty \frac{(i\theta)^n}{n!}\vG^n\right]\\
    &=\sum_{n=0}^{\infty} \frac{(i\theta)^{2n}}{(2n)!}\mathrm{Cat}_n\\
    &=\frac{J_1(2\theta)}{\theta}.
\end{align*}
where $J_1$ is a Bessel function of the first kind.
\end{proof}

\begin{cor}[Good $\theta$ exists]\label{cor:theta_exists} For a random matrix $\vG$ drawn from the Gaussian Unitary Ensemble, there exists an infinite number of $\theta\in\mathbb{R}$ such that as $N\rightarrow\infty$,
\begin{align*}
    \mathop{\BE}_{\text{GUE}}\btr[e^{i \vG  \theta}] = 0.
\end{align*}
There does not exist a constant greater than all the values of $\theta$.
\end{cor}
\begin{proof}
    The function $J_1$ has an infinite number of zeros spaced semi-periodically, so an infinite number of $\theta$ values exist such that $\mathop{\BE}_{\text{GUE}}\btr[e^{\ri \vG  \theta}]=0$ when $N$ goes to infinity, and there is no possible bound on $\theta$. A list of possible $\theta$ are easy to generate numerically: 1.915..., 3.507..., 5.086..., 6.661..., etc.
\end{proof}
The two lemmas above address the trace moments of exponentiated infinite-dimensional GUE---the limit in which the spectral distribution converges to the Wigner semicircle. We now address the finite $N$ corrections to the semicircle predictions.

\begin{lem}[Exponentiated finite $N$ GUE moments]\label{lem:exp_GUE_moments}
    Let $\vG$ be an $N$-dimensional random matrix drawn from the Gaussian Unitary Ensemble. Then for all $\theta>1$,
    \begin{align*}
        \left|\mathop{\BE}_{\text{GUE}}\btr[e^{i \vG  \theta}]-\frac{J_1(2\theta)}{\theta}\right| \leq  \frac{K_0\theta}{N}
    \end{align*}
    for some constant $K_0>0$.
\end{lem}
\begin{proof}
    This proof follows a format similar to the proof of Corollary B.1 in \cite{chen2023sparse}. Consider that the empirical spectral density (ESD) of a sample from the $N$-dimensional GUE is 
    \begin{align*}
        \rho(E) = \frac{1}{N}\sum_{i=1}^N\delta(E-\lambda_i)
    \end{align*}
    where $\lambda_i$ are the eigenvalues of this instance.
    As we know from \autoref{thm:semicirclelaw} in appendix \ref{app:GUE_properties}, the infinite $N$ spectral density is the semicircle 
    \begin{align*}
        \rho_{sc}(E) = \frac{1}{2\pi}\sqrt{4-E^2}, \quad E\in[-2,2].
    \end{align*}
    Then the following is true:
    \begin{align*}
        \left|\mathop{\BE}_{\text{GUE}}\btr{e^{\ri\theta\vG}} - \lim_{N\rightarrow\infty}\left[\mathop{\BE}_{\text{GUE}}\btr{e^{\ri\theta\vG}}\right]\right| &\leq \left|\mathop{\BE}_{\text{GUE}}\int_{-\infty}^\infty e^{\ri\theta E}(\rho(E)-\rho_{sc}(E))dE\right|\\
        &=\left|\mathop{\BE}_{\text{GUE}}\int_{-\infty}^\infty \ri\theta e^{\ri\theta E}\int_{-\infty}^E(\rho(E')-\rho_{sc}(E'))dE'dE\right| \tag{integration by parts}
    \end{align*}
    The last line holds since $\int_{-\infty}^\infty (\rho(E')-\rho_{sc}(E'))dE'=0$. We now want to use a combination of two propositions:
    \begin{enumerate}
        \item \autoref{prop:esd_to_sc} from appendix \ref{app:GUE_properties}, which controls the distance between the GUE ESD and the semicircle:
        \begin{align*}
        \sup_{E\in\mathbb{R}} \left|\mathop{\BE}_{\text{GUE}}\int_{-\infty}^E \rho_{GUE}(E')-\rho_{sc}(E')dE'\right|\leq \frac{K}{N}
        \end{align*}
        for some constant $K>0$.
        \item \autoref{prop:GUE_spectral_rad_conc} from appendix \ref{app:GUE_properties}, which controls the distribution of the maximum eigenvalue of the GUE: for all $t>0$,
        \begin{align*}
        \Pr\left[\norm{\vG} \geq 2+t\right]\leq C\exp(-\frac{Nt^{3/2}}{c}).
        \end{align*}
        for some constants $C,c>0$.
    \end{enumerate}
    Let $G_{\text{low}}$ be the set of matrices $\vG$ in GUE that satisfy $\norm{\vG} \geq 2+t$ and $G_{\text{high}}$ be the rest. Then, 
    \begin{align*}
        &\left|\mathop{\BE}_{\text{GUE}}\btr{e^{\ri\theta\vG}} - \lim_{N\rightarrow\infty}\left[\mathop{\BE}_{\text{GUE}}\btr{e^{\ri\theta\vG}}\right]\right|\\
        \leq&\left|\left(1-C\exp(-\frac{Nt^{3/2}}{c})\right)\mathop{\BE}_{G_\text{high}}\int_{-2-t}^{2+t} \ri\theta e^{\ri\theta E}\int_{-\infty}^E(\rho(E')-\rho_{sc}(E'))dE'dE\right|\\
        &+\left|\left(C\exp(-\frac{Nt^{3/2}}{c})\right)\mathop{\BE}_{G_\text{low}}\left[\btr{e^{\ri\theta\vG}} - \lim_{N\rightarrow\infty}\left[\mathop{\BE}_{\text{GUE}}\btr{e^{\ri\theta\vG}}\right]\right]\right| \tag{\autoref{prop:GUE_spectral_rad_conc}}\\
        \leq &\left|\left(1-C\exp(-\frac{Nt^{3/2}}{c})\right)\frac{K}{N}\int_{-2-t}^{2+t} \ri\theta e^{\ri\theta E}dE\right|+\left|2C\exp(-\frac{Nt^{3/2}}{c})\right| \tag{\autoref{prop:esd_to_sc} and $\btr{\vU}\leq 1$ for $\vU$ unitary}\\
    \end{align*}
    Note that that we needed \autoref{prop:GUE_spectral_rad_conc} to ensure that the integral in the first term would not be infinite. We continue simplifying:
    \begin{align*}
        &\left|\mathop{\BE}_{\text{GUE}}\btr{e^{\ri\theta\vG}} - \lim_{N\rightarrow\infty}\left[\mathop{\BE}_{\text{GUE}}\btr{e^{\ri\theta\vG}}\right]\right| \\
        &\quad\leq \abs{\left(\frac{K}{N}-\frac{KC}{N}\exp(-\frac{Nt^{3/2}}{c})\right)}\int_{-2-t}^{2+t} \abs{\ri\theta e^{\ri\theta E}}dE+\left|2C\exp(-\frac{Nt^{3/2}}{c})\right|\\
        &\quad\leq \abs{\left(\frac{K}{N}-\frac{KC}{N}\exp(-\frac{Nt^{3/2}}{c})\right)}(2\theta(2+t))+2C\exp(-\frac{Nt^{3/2}}{c})\tag{$\theta$, $t$, $C$ and $c$ are positive}\\ 
    \end{align*}
    Let $t^{3/2}=c.$ Then $\exp(-N)\leq 1/N^2$ and we have
    \begin{align*}
        \left|\mathop{\BE}_{\text{GUE}}\btr{e^{\ri\theta\vG}} - \lim_{N\rightarrow\infty}\left[\mathop{\BE}_{\text{GUE}}\btr{e^{\ri\theta\vG}}\right]\right| &\leq \frac{K_1\theta}{N}+\frac{K_2}{N^2}+\frac{K_3\theta}{N^3}\\
        &\leq \frac{K_0\theta}{N} \tag{Assuming $\theta >1$}
    \end{align*}
    for $K_0, K_1, K_2, K_3>0$.
    We plug in the $N\rightarrow\infty$ value of $\mathop{\BE}_{\text{GUE}}\btr{e^{\ri\theta\vG}}$ from \autoref{lem:exp_semicircle_moments} to conclude the proof.
\end{proof}

On top of the finite $N$ corrections, we also would like to understand the instance-to-instance fluctuation of the trace moments. Fortunately, this follows from an off-the-shelf Gaussian concentration inequality. A detailed proof can be found in appendix \ref{append:Proofs_concentration}.
\begin{lem}[Trace concentration for $\e^{\ri \vG\theta}$]\label{lem:theta_concentration} For a Gaussian Unitary Ensemble of $N$-dimensional matrices $\vG$, 
\begin{align*}
    \Pr[|\btr[e^{i \vG  \theta}] - \mathop{\BE}_{\text{GUE}}\btr[e^{i \vG  \theta}]| \geq t ] \leq \exp(-\frac{Nt^2}{2\theta^2}).
\end{align*}
\end{lem}

From the two lemmas above, we see that the correction from the concentration is larger than that of the finite $N$ correction. Heuristically, then, we say that the expected fluctuation of $\btr[e^{\ri \vG  \theta}]$ from its large $N$ value of zero can be taken to be $\mathcal{O}(1/\sqrt{N}).$ In other words, a more accurate way to state our condition from the beginning of this subsection is 
\begin{align*}
    \vW=\vU\vD_1\vU^\dagger\vV\vD_2\vV^\dagger \quad \text{such that}\quad \btr[\vD_i] \approx \mathcal{O}\left(\frac{1}{\sqrt{N}}\right),
\end{align*}
where $\vU$ and $\vV$ are from Haar random unitary ensembles and the $\vD_i$ are diagonal matrices from an appropriate ensemble.

Notice a consequence from \autoref{cor:theta_exists}: we can choose $\theta$ to ensure traceless $e^{i\vG\theta}$, but once we fix such a choice, we find that the other moments will not go to zero as $N\rightarrow\infty$. (One can check by taking $\theta\rightarrow p\theta$.) Hence, a single $e^{i\vG\theta}$ ensemble is still in some sense not ``close enough'' to the Haar unitary ensemble, which has all zero moments at infinite $N$. That defect motivates employing the product of two Gaussian-like exponentials for our construction.

\subsection{Expected Trace Moments of Products with Unitarily Invariant Structure}\label{subsec:bound_on_moments}
We now analyze the behavior of moments of the product of two exponentiated Gaussian ensembles:
\begin{align*}
    \BE\btr{\vW^p} = \BE\btr\left[\left(e^{i\theta\vG_1}e^{i\theta\vG_2}\right)^p\right]= \BE_{\vD_1,\vD_2}\BE_{\vU,\vV}\btr\left[\left(\vU\vD_1\vU^\dagger\vV\vD_2\vV^\dagger\right)^p\right].
\end{align*}
Given that our ultimate construction uses Gaussian designs and not precisely the GUEs, we prove a more general theorem. In this analysis, we address the randomness from the Haar random basis $(\vU,\vV)$ separately from the eigenvalues $(\vD_1,\vD_2)$:

\UDUVDV*
The proof requires several other lemmas, which we formally introduce and prove in the following subsections. We begin by providing an overarching roadmap for the proof strategy. First, we fix the randomness of the diagonal and focus on averaging over just the two Haar ensembles. Let us fix $p$ and let
\begin{align*}
    f(\vD_1,\vD_2) := \BE_{\vU,\vV\leftarrow\mu} \btr[\vW^p] \quad \text{for any}\quad \vD_1,\vD_2.
\end{align*}
It is important to consider the $\vW$ ensemble as the combination of four ensembles (two Haar ensembles $\vU$ and $\vV$, and two diagonal ensembles $\vD_1$ and $\vD_2$) because our argument only exploits invariant theory stemming from the bases. We will keep the $\vD_i$'s as \emph{fixed} diagonal matrices (which will later be sampled from the appropriate distribution).

\begin{lem}[A Markov-type inequality for rational functions]\label{lem:main_markov_inequality}
Consider an algebraic fraction of the form
\begin{align*}
    r_n(x) &= \frac{p_n(x)}{\sqrt{t_{2n}(x)}}
\end{align*}
such that 
\begin{itemize}
    \item $t_{2n}(x) = \prod^{2n}_{k=1} (1+a_k x)$ for all $ -1\leq x\leq 1,$
    \item $p_n(x)$ is an algebraic polynomial of degree at most n with complex (real) coefficients,
    \item $a_k \in \BC$ are either real or pairwise complex conjugate,
    \item $|a_k| <1$ for all $1\leq k \leq 2n$,  
    \item $|r_n(x_{i'}^*)|\leq 1$ for a set of discrete points in increasing order $-1=x^*_1<\dots < x^*_{i'}\dots <x^*_i=1$,   
    \item there exists a bound $c\in\mathbb{R}$ such that $c\geq \sup_{[-1,1]}\lambda_n(x).$
\end{itemize}
Then $r_n(x)$ is bounded as 
    \begin{align*}
        |r'_n(x)|&\leq \left(1+\frac{cI}{2\sqrt{1-x^2}-cI}\right)
    \begin{cases}
        \frac{\lambda_n(x)}{\sqrt{1-x^2}}, & x\in [x_1,x_n] \cap  X,\\
        |m'_n(x)|, & x\in ([-1,x_1]\cup[x_n,1]) \cap X
    \end{cases},
    \end{align*}
where
\begin{itemize}
    \item \{$x_k$\}, $-1 < x_1 < ... < x_n < 1$ are zeros of the cosine fraction $m_n(x)=\cos\left[\frac{1}{2}\sum^{2n}_{k=1}\arccos \left(\frac{x+a_k}{1+a_k x}\right)\right]$,
    \item $\lambda_n(x)=\frac{1}{2}\sum^{2n}_{k=1} \frac{\sqrt{1-a_k^2}}{1+a_k x}$,
    \item $I$ is the size of the largest interval between consecutively bounded points $ I = \sup_{\{x^*_i\}}{|x^*_i-x^*_{i+1}|},$
    \item $X$ is the interval defined by $X:=\left\{x\in[-1,1]:1-\frac{cI}{2\sqrt{1-x^2}}>0\right\}$.
\end{itemize}
\end{lem}
This version of a sharp Markov-type inequality for rational functions extends past the one featured in \cite{rusak1979rational} by allowing the bound on the function to only hold at a discrete set of points, provided a suitable $c$ can be found. The proof is a combination of lemmas found in \autoref{app:markov_inequality}, but for the flow of the following proofs, this can be taken as a black box result. 
\subsubsection{Checklist for the Markov-type inequality}

In order to apply this to our function of interest $\BE_{\vU,\vV}\btr[\vW^p]$, we need to
\begin{itemize}
    \item identify a low degree rational polynomial of $\frac{1}{N}$ (\autoref{subsubsec:rational_function}),
    \item obtain an a priori bound on sufficiently many points (\autoref{subsubsec:apriori_bounds}),
    \item calculate the large-$N$ limit (\autoref{subsubsec:Large_N_limit}), and
    \item adjust our function to exhibit a suitable pole structure (\autoref{subsubsec:pole_structure}).
\end{itemize}

Averaging over just the Haar ensembles allows us to invoke the Weingarten calculus. We will use the notation of the 
Weingarten primer in \autoref{app:weingarten}. The Weingarten calculus gives us both the $N\rightarrow\infty$ limit for the Haar-averaged normalized trace moments, as well as a functional form for the Haar-averaged normalized trace moments as a function of $N$. 

A natural question to ask is why the Weingarten calculus is not sufficient to find a bound on $f_{\vD_1,\vD_2}(N)$ directly. After all, unlike Markov-type inequality techniques, a proper evaluation of the Weingarten calculus would give the exact scaling of the leading order of $f(\vD_1,\vD_2)$. Such evaluation requires identifying the specific types of permutations that contribute leading order diagrams in the expansion of the Weingarten sum, which we will discuss more in depth later; since the cycles of products of permutations are related to their Cayley distance, we could analyze the transpositions needed to go between certain types of permutations and prove the leading order term must be of $\mathcal{O}(N^{-2})$. The numerator, however, is determined by the number of diagrams in the leading order, and naive approaches to such counting result in factors of $p!\sim p^{p}$. The actual scaling, which should have poly($p$) dependence, is difficult to extract from the Weingarten calculus since there are intricate cancellations that occur between terms. As such, we turn to Markov-type inequality techniques for tackling large-$N$ polynomials.

However, to make sense of the normalized trace moments as functions of $N$ that have infinite $N$ limits and prescribed pole structures, we need to consider what happens to our ensemble $\vW$ when we scale with $N$. The trace moments, by hypothesis, are random variables sampled from different distributions for different $N$. This precludes interpreting the averaged trace moments as rational functions of $1/N$. Instead, we need to artificially construct a family of ensembles $\mathcal{F}(\vW)$ such that for each instance of $(\vD_1,\vD_2)$ from  the $N_D \times N_D$ ensemble $\vW$, there is a corresponding ensemble of $N \times N$ matrices for a large range of $N$ that have the same trace moments $\btr[\vD_i^p]$ as the original instance. Then for each instance of $(\vD_1,\vD_2)$, we will have well-behaved functions of $N$, that we can then interpolate back to the original $N_D$.

\subsubsection{Rational Function of \texorpdfstring{$N$}{N} from Weingarten Calculus}\label{subsubsec:rational_function}

Weingarten calculus naturally gives rational functions. 
\begin{align}\nonumber
    f(\vD_1,\vD_2) &= \mathop{\BE}_{\vU,\vV \leftarrow \mu} \left[\btr[(\vU \vD_1 \vU^\dagger \vV \vD_2 \vV^\dagger )^p]\right] \\\nonumber
    &= \mathop{\BE}_{\vU \leftarrow \mu} \left[\btr[(\vU \vD_1 \vU^\dagger\vD_2)^p]\right] \tag{product of independent Haar ensembles is Haar}\\
    &=\frac{1}{N}\sum_{\sigma,\tau\in \mathcal{S}(p)}\delta_{\sigma}(\vec{i},\vec{i'})\delta_{\tau}(\vec{j},\vec{j'})\text{Wg}(\sigma\tau^{-1},N)\prod_{k=1}^p\vD_{1,j_k,j'_k} \vD_{2,i_k,i'_k}.\label{eq:ugly_wg}
\end{align}
The second equality comes from the fact that we can rewrite $\vV^\dagger\vU\rightarrow\vU$, and the fact that the product of independent Haar ensembles is Haar. The last equality comes from substituting the Weingarten calculus: the delta function terms and the Weingarten function are the result of averaging 
\begin{align*}
\int \vU_{i_1j_1}...\vU_{i_pj_p}\vU^\dagger_{j'_1i'_1}...\vU^\dagger_{j'_pi'_p}  \rd \vU, 
\end{align*}
which is further explained in appendix \ref{app:weingarten}, and the product of $\vD_i$ are the leftover matrices, which we write with indices so we don't lose track of how they contract. 

This expression looks complicated, but can be understood through its Weingarten wiring diagrams. In general, each term of the sum in \eqref{eq:ugly_wg} has a diagram that consists of contractions drawn onto \autoref{fig:full_wiring}. The top wire goes all the way across to the last unitary term; see the diagram in \autoref{fig:p=4wiring} for an example of the base diagram for all $p=4$ terms before contractions are made.

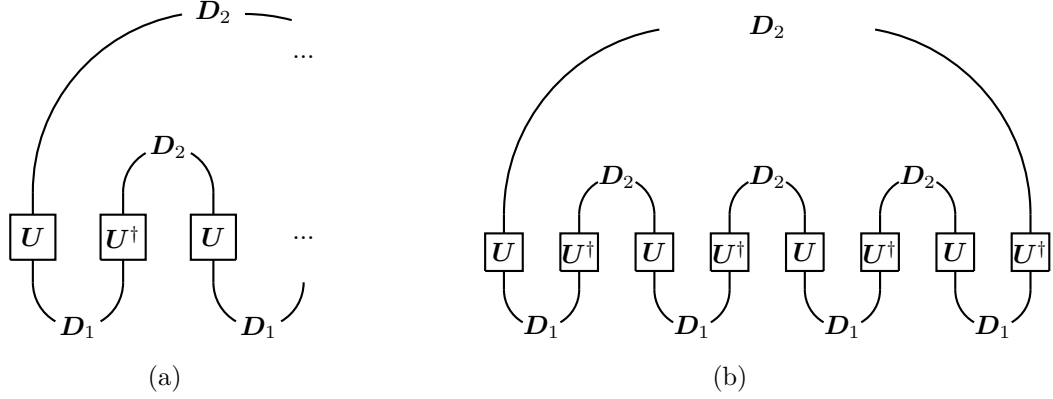
\begin{figure}
    \centering
    \begin{subfigure}{0.3\textwidth}
    \centering
    \begin{tikzpicture}[scale=0.6]
    
    \draw[thick] (-3.5,-0.5) -- (-3.5,0.5) -- (-2.5,0.5) -- (-2.5,-0.5) -- (-3.5,-0.5);
    \draw[thick] (-1.5,-0.5) -- (-1.5,0.5) -- (-0.5,0.5) -- (-0.5,-0.5) -- (-1.5,-0.5);
    \draw[thick] (1.5,-0.5) -- (1.5,0.5) -- (0.5,0.5) -- (0.5,-0.5) -- (1.5,-0.5);
    \node at (-3,0) {$\vU$};
    \node at (-1,0) {$\vU^\dagger$};
    \node at (1,0) {$\vU$};
    
    \draw[thick] (-3,-0.5) -- (-3,-1);
    \draw[thick] (-3,0.5) -- (-3,1);
    \draw[thick] (-1,-0.5) -- (-1,-1);
    \draw[thick] (-1,0.5) -- (-1,1);
    \draw[thick] (1,-0.5) -- (1,-1);
    \draw[thick] (1,0.5) -- (1,1);

    \draw[thick] (-3,-1) arc (180:240:1);
    \draw[thick] (-1,-1) arc (0:-60:1);
    \draw[thick] (-1,1) arc (180:120:1);
    \draw[thick] (1,1) arc (0:60:1);
    \draw[thick] (1,-1) arc (180:240:1);
    \draw[thick] (3,-1) arc (0:-60:1);
    \draw[thick] (-3,1) arc (180:100:4);
    \draw[thick] (1.7,4.95) arc (90:75:4);

    \node at (-2,-2) {$\vD_1$};
    \node at (0,2) {$\vD_2$};
    \node at (2,-2) {$\vD_1$};
    \node at (1,5) {$\vD_2$};

    \node at (3,0) {...};
    \node at (3,4) {...};

    \end{tikzpicture}
    \caption{}\label{fig:full_wiring}
    \end{subfigure}
    \hspace{15mm}
    \begin{subfigure}{0.4\textwidth}
    \centering
    \begin{tikzpicture}[scale=0.5]
    
    \draw[thick] (-7.5,-0.5) -- (-7.5,0.5) -- (-6.5,0.5) -- (-6.5,-0.5) -- (-7.5,-0.5);
    \draw[thick] (-5.5,-0.5) -- (-5.5,0.5) -- (-4.5,0.5) -- (-4.5,-0.5) -- (-5.5,-0.5);
    \draw[thick] (-3.5,-0.5) -- (-3.5,0.5) -- (-2.5,0.5) -- (-2.5,-0.5) -- (-3.5,-0.5);
    \draw[thick] (-1.5,-0.5) -- (-1.5,0.5) -- (-0.5,0.5) -- (-0.5,-0.5) -- (-1.5,-0.5);
    \draw[thick] (1.5,-0.5) -- (1.5,0.5) -- (0.5,0.5) -- (0.5,-0.5) -- (1.5,-0.5);
    \draw[thick] (3.5,-0.5) -- (3.5,0.5) -- (2.5,0.5) -- (2.5,-0.5) -- (3.5,-0.5);
    \draw[thick] (5.5,-0.5) -- (5.5,0.5) -- (4.5,0.5) -- (4.5,-0.5) -- (5.5,-0.5);
    \draw[thick] (7.5,-0.5) -- (7.5,0.5) -- (6.5,0.5) -- (6.5,-0.5) -- (7.5,-0.5);
    \node at (-7,0) {$\vU$};
    \node at (-5,0) {$\vU^\dagger$};    
    \node at (-3,0) {$\vU$};
    \node at (-1,0) {$\vU^\dagger$};
    \node at (1,0) {$\vU$};
    \node at (3,0) {$\vU^\dagger$};
    \node at (5,0) {$\vU$};
    \node at (7,0) {$\vU^\dagger$};
    
    \draw[thick] (-3,-0.5) -- (-3,-1);
    \draw[thick] (-3,0.5) -- (-3,1);
    \draw[thick] (3,-0.5) -- (3,-1);
    \draw[thick] (3,0.5) -- (3,1);
    \draw[thick] (-1,-0.5) -- (-1,-1);
    \draw[thick] (-1,0.5) -- (-1,1);
    \draw[thick] (1,-0.5) -- (1,-1);
    \draw[thick] (1,0.5) -- (1,1);
    \draw[thick] (-5,-0.5) -- (-5,-1);
    \draw[thick] (-5,0.5) -- (-5,1);
    \draw[thick] (5,-0.5) -- (5,-1);
    \draw[thick] (5,0.5) -- (5,1);
    \draw[thick] (-7,-0.5) -- (-7,-1);
    \draw[thick] (-7,0.5) -- (-7,1);
    \draw[thick] (7,-0.5) -- (7,-1);
    \draw[thick] (7,0.5) -- (7,1);

    \draw[thick] (-7,-1) arc (180:240:1);
    \draw[thick] (-5,-1) arc (0:-60:1);
    \draw[thick] (-5,1) arc (180:120:1);
    \draw[thick] (-3,1) arc (0:60:1);
    \draw[thick] (-3,-1) arc (180:240:1);
    \draw[thick] (-1,-1) arc (0:-60:1);
    \draw[thick] (-1,1) arc (180:120:1);
    \draw[thick] (1,1) arc (0:60:1);
    \draw[thick] (1,-1) arc (180:240:1);
    \draw[thick] (3,-1) arc (0:-60:1);
    \draw[thick] (3,1) arc (180:120:1);
    \draw[thick] (5,1) arc (0:60:1);
    \draw[thick] (5,-1) arc (180:240:1);
    \draw[thick] (7,-1) arc (0:-60:1);
    \draw[thick] (-7,1) arc (180:100:5); 
    \draw[thick] (7,1) arc (0:80:5);

    \node at (-6,-2) {$\vD_1$};
    \node at (-4,2) {$\vD_2$};
    \node at (-2,-2) {$\vD_1$};
    \node at (0,2) {$\vD_2$};
    \node at (0,6) {$\vD_2$};
    \node at (2,-2) {$\vD_1$};
    \node at (4,2) {$\vD_2$};
    \node at (6,-2) {$\vD_1$};

    \end{tikzpicture}
    \caption{}\label{fig:p=4wiring}
    \end{subfigure}
    \caption{(a) An illustration of a full wiring diagram for $\mathop{\BE}_{\vU\leftarrow \mu} \left[\btr[\vW^p]\right]$. The top wires all contract with $\vD_2$ while the bottom wires all contract with $\vD_1$. (b) A sample wiring diagram for $p=4$.}
\end{figure}

\begin{figure}
    \centering
    \begin{tikzpicture}[scale=0.5]
    
    \draw[thick] (-7.5,-0.5) -- (-7.5,0.5) -- (-6.5,0.5) -- (-6.5,-0.5) -- (-7.5,-0.5);
    \draw[thick] (-5.5,-0.5) -- (-5.5,0.5) -- (-4.5,0.5) -- (-4.5,-0.5) -- (-5.5,-0.5);
    \draw[thick] (-3.5,-0.5) -- (-3.5,0.5) -- (-2.5,0.5) -- (-2.5,-0.5) -- (-3.5,-0.5);
    \draw[thick] (-1.5,-0.5) -- (-1.5,0.5) -- (-0.5,0.5) -- (-0.5,-0.5) -- (-1.5,-0.5);
    \draw[thick] (1.5,-0.5) -- (1.5,0.5) -- (0.5,0.5) -- (0.5,-0.5) -- (1.5,-0.5);
    \draw[thick] (3.5,-0.5) -- (3.5,0.5) -- (2.5,0.5) -- (2.5,-0.5) -- (3.5,-0.5);
    \draw[thick] (5.5,-0.5) -- (5.5,0.5) -- (4.5,0.5) -- (4.5,-0.5) -- (5.5,-0.5);
    \draw[thick] (7.5,-0.5) -- (7.5,0.5) -- (6.5,0.5) -- (6.5,-0.5) -- (7.5,-0.5);
    \node at (-7,0) {$\vU$};
    \node at (-5,0) {$\vU^\dagger$};    
    \node at (-3,0) {$\vU$};
    \node at (-1,0) {$\vU^\dagger$};
    \node at (1,0) {$\vU$};
    \node at (3,0) {$\vU^\dagger$};
    \node at (5,0) {$\vU$};
    \node at (7,0) {$\vU^\dagger$};
    
    \draw[thick] (-3,-0.5) -- (-3,-1);
    \draw[thick] (-3,0.5) -- (-3,1);
    \draw[thick] (3,-0.5) -- (3,-1);
    \draw[thick] (3,0.5) -- (3,1);
    \draw[thick] (-1,-0.5) -- (-1,-1);
    \draw[thick] (-1,0.5) -- (-1,1);
    \draw[thick] (1,-0.5) -- (1,-1);
    \draw[thick] (1,0.5) -- (1,1);
    \draw[thick] (-5,-0.5) -- (-5,-1);
    \draw[thick] (-5,0.5) -- (-5,1);
    \draw[thick] (5,-0.5) -- (5,-1);
    \draw[thick] (5,0.5) -- (5,1);
    \draw[thick] (-7,-0.5) -- (-7,-1);
    \draw[thick] (-7,0.5) -- (-7,1);
    \draw[thick] (7,-0.5) -- (7,-1);
    \draw[thick] (7,0.5) -- (7,1);

    \draw[thick] (-7,-1) arc (180:240:1);
    \draw[thick] (-5,-1) arc (0:-60:1);
    \draw[thick] (-5,1) arc (180:120:1);
    \draw[thick] (-3,1) arc (0:60:1);
    \draw[thick] (-3,-1) arc (180:240:1);
    \draw[thick] (-1,-1) arc (0:-60:1);
    \draw[thick] (-1,1) arc (180:120:1);
    \draw[thick] (1,1) arc (0:60:1);
    \draw[thick] (1,-1) arc (180:240:1);
    \draw[thick] (3,-1) arc (0:-60:1);
    \draw[thick] (3,1) arc (180:120:1);
    \draw[thick] (5,1) arc (0:60:1);
    \draw[thick] (5,-1) arc (180:240:1);
    \draw[thick] (7,-1) arc (0:-60:1);
    \draw[thick] (-7,1) arc (180:100:5); 
    \draw[thick] (7,1) arc (0:80:5);

    \node at (-6,-2) {$\vD_1$};
    \node at (-4,2) {$\vD_2$};
    \node at (-2,-2) {$\vD_1$};
    \node at (0,2) {$\vD_2$};
    \node at (0,6) {$\vD_2$};
    \node at (2,-2) {$\vD_1$};
    \node at (4,2) {$\vD_2$};
    \node at (6,-2) {$\vD_1$};

    \draw[blue] (-7,-1.5) .. controls (-4,-4) .. (-1,-1.5);
    \draw[blue] (-5,-1.5) .. controls (-4,-2.5) .. (-3,-1.5);
    \draw[blue] (3,-1.5) .. controls (4,-2.5) .. (5,-1.5);
    \draw[blue] (1,-1.5) .. controls (4,-4) .. (7,-1.5);
    
    \draw[red] (-1,1.5) .. controls (2,4) .. (5,1.5);
    \draw[red] (3,1.5) .. controls (2,2.5) .. (1,1.5);
    \draw[red] (-3,1.5) .. controls (2,5) .. (7,1.5);
    \draw[red] (-7,1.5) .. controls (-6,2.5) .. (-5,1.5);

    \end{tikzpicture}
    \caption{An illustration of the $p=4$ contraction for the term corresponding to $\tau=(1)(3)(24)$, which are shown via the red lines, and $\sigma=(12)(34)$, which are shown via the blue lines.}\label{fig:p=4contractions}
\end{figure}
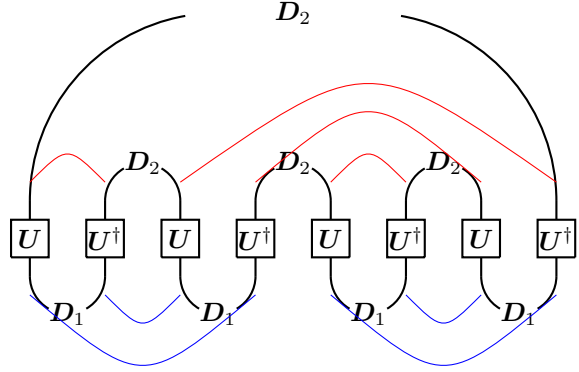

Now we take a look at contractions upon these diagrams. Each term in the sum of \eqref{eq:ugly_wg} corresponds to one of these contracted diagrams. Figure \ref{fig:p=4contractions} is an example of a contracted $p=4$ diagram. The top contraction is illustrated in red and gives us traces of products of $\vD_2$. The bottom contraction is illustrated in blue and gives us traces of products of $\vD_1$. This contraction formation tells us that any factor of $\vD_i$ will come in the form $\tr[\vD_i^q]$ for some $q\leq p$. Notice that these are factors of the regular trace, not the normalized trace. 

Each term will also have a factor of the Weingarten function dependent on both the dimension $N$ and the specific diagram. As discussed in \autoref{app:weingarten}, the Weingarten functions are rational functions, defined as sums of characters divided by Schur polynomials in $N$. Each Weingarten function has the same possible set of poles, $\{-(p-1),...,p-1\}$, and each pole $r$ can be up to order $s(r)$ where $s(s+\abs{r})\leq p$. The actual degree of any Weingarten function can be much smaller, as some poles can be cancelled by a polynomial numerator, but when added together the maximal possible denominator must account for all possible poles.

Lastly, the case we are analyzing is for the trace of a unitary, which must be bounded in magnitude by one. Hence the numerator polynomial can at most be the same degree as the denominator polynomial, if not smaller. To sum this up, we can rewrite our function $f(\vD_1,\vD_2)$ to be of the form
\begin{align}
    f(\vD_1,\vD_2) &= \frac{1}{N}\frac{\sum_{\text{diagrams}}(\text{factors of }(N\btr[\vD_i])^q)(\text{numerators of Wg})}{N^{s(0)}(N^2-1)^{s(1)}(N^2-4)^{s(2)}...(N^2-(p-1)^2)^{s(p-1)}}\label{eq:def_fD1D2}\\
    &= \frac{1}{N}\frac{\sum_{\text{diagrams}}(\text{factors of }(N\btr[\vD_i])^q)(\text{more numerators of Wg})}{N^{2p-1}(N^2-1)^{p}(N^2-4)^{p}...(N^2-(p-1)^2)^{p}}\\
    &=:f_{\vD_1,\vD_2}(N) \tag{fixing $\btr[\vD_i^q]$ and varying $N$}.
\end{align}
In the second line, we've overestimated the number of poles for convenience of later calculation---this is made up for by multiplying more factors in the numerator that would cancel these poles out if we simplified. In the third line, we defined $f_{\vD_1,\vD_2}(N)$ by fixing the trace moments and but formally allowing the dimension parameter $N$ to vary. We therefore have that $f(\vD_1,\vD_2) = f_{\vD_1,\vD_2}(N_D)$. For now, $f_{\vD_1,\vD_2}(N)$ is just a functional extension of $f_{\vD_1,\vD_2}$; we will endow more physical meaning to it in \autoref{subsubsec:apriori_bounds}. We can readily read off the rational function in $N$.

\begin{lem}[Rational function from Weingarten calculus]\label{lem:rational_function}
    Fixing the normalized traces $\btr[\vD_i^q]$, the formal expression 
    $f_{\vD_1,\vD_2}(N)$~\eqref{eq:def_fD1D2} is a rational function of $N$ of at most degree $2p^2$ in the numerator, and with $2p^2$ poles in the denominator. 
\end{lem}
\begin{proof}
    This follows from the form of the Weingarten calculus as given above, and the fact that $\abs{f_{\vD_1,\vD_2}(N_D)}=\abs{f(\vD_1,\vD_2)}\leq1,$ since $f(\vD_1,\vD_2)$ is just a normalized trace of a unitary.
\end{proof}
The form of $f_{\vD_1,\vD_2}(N)$ suggests applying our interpolation techniques and the Markov-type inequality in \autoref{lem:main_markov_inequality}, but $f_{\vD_1,\vD_2}(N)$ itself is not actually a good candidate for the process. For one, we need a function with an argument that is something like $1/N$, because the Weingarten calculus gives us the limit of $f_{\vD_1,\vD_2}(N)$ at $N\rightarrow\infty$, while the Markov inequality only gives us derivative bounds on an argument $x\in[-1,1]$. For another, the pole structure of $f_{\vD_1,\vD_2}(N)$ needs to fit the rest of the conditions in \autoref{lem:main_markov_inequality}. We narrow down the specific function we put through the Markov-type inequality and interpolation in \autoref{subsubsec:pole_structure}. For now, we investigate how much control we have over $f_{\vD_1,\vD_2}(N)$.

\subsubsection{A Priori Bounds from Solving the Moment Problem}\label{subsubsec:apriori_bounds}
So far, \autoref{lem:rational_function} is a wishful formal construction: we have really only evaluated $f$ on two particular matrices $\vD_1$ and $\vD_2$, but we are blindly extending the function for other values of $N$ while fixing the normalized traces. For other values of $N$, this may not necessarily come from moments $\BE_{\vU,\vV\leftarrow\mu} \btr[\vW^p]$ for some other matrices $\vD'_1,\vD'_2$.

However, our Markov-type inequality requires \emph{a priori bounds} at an appropriate number of other values of $N$, which the above only gives for $N_D=\dim(\vD_1)$ because we only know $f_{\vD_1,\vD_2}(N_D)$ is the trace moment of an actual unitary matrix. Thus, we further show that it is possible to \emph{create} a family of unitary matrices of \emph{different} dimension $N$ that matches the normalized traces $\btr[\vD_i^q]$ up to the $q^{\text{th}}$ moment; this is separately discussed as the \emph{moment problem} in \autoref{sec:moment_problem}. For our purposes, for any $(\vD_1, \vD_2)$, we can create a whole ensemble of unitary matrices of different dimensions $N$ with the same normalized traces $\btr[\vD_i^q]$ using \autoref{lem:unitary_moment_prob}. 
\begin{cor}[Creating a family of matrices across different dimensions.]\label{cor:moment_problem_applied_apriori_bound}
For unitaries $\vD_1$ and $\vD_2$, suppose the moments are small enough such that $\sum_{k=1}^p |\btr{\vD_i^k}|\leq \frac{1}{4}$. Then, there exists a family of matrices $(\vD'_1,\vD'_2)$ of different dimensions N such that
\begin{align*}
f_{\vD_1,\vD_2}(N) =  f(\vD'_1,\vD'_2) \quad \text{for each integer}\quad N\geq 48p^{7/2}, 
\end{align*}
which implies the a priori bound
\begin{align*}
\labs{f_{\vD_1,\vD_2}(N)} \le 1 \quad \text{for each integer}\quad N\geq 48p^{7/2}. 
\end{align*}
\end{cor}
Notice that the Markov-type inequality in \autoref{lem:main_markov_inequality} requires bounds in both negative and positive domain: $x\in[-1,1]$. Hence, what we eventually subject to the Markov-type inequality process will have an argument that is a function of $\sqrt{N}$.
 
\subsubsection{The Large-\texorpdfstring{$N$}{N} Limit}\label{subsubsec:Large_N_limit}
Now that the large $N$ limit is well-defined, we can use the Weingarten formula to calculate the large $N$ limit. The precise expansion in $1/N$ might require substantial combinatorics to count diagrams, but we merely need the leading order behavior here.
\begin{lem}[Infinite $N$ limit]\label{lem:infinite_N_limit}
\begin{align*}
    \lim_{N\rightarrow\infty} |f_{\vD_1,\vD_2}(N)| \le p\left(|\btr{\vD_1}|+|\btr{\vD_2}| + |\btr{\vD_1}||\btr{\vD_2}|\right).
\end{align*}
\end{lem}
\begin{proof}[Proof of \autoref{lem:infinite_N_limit}]
    Combining our definition of the function $f_{\vD_1,\vD_2}(N)$ with \autoref{lem:conjugation_invariance_single} and \autoref{cor:moment_problem_applied_apriori_bound}, we have that when $N\geq 48p^{7/2}$ is integer, 
    \begin{align*}
        f_{\vD_1,\vD_2}(N)=\BE_{\vU,\vV\leftarrow\mu} \btr[\left(\vU\vD_1\vU^\dagger\vV\vD_2\vV^\dagger\right)^p]
    \end{align*}
    for some unitary matrices $\vD_1$, $\vD_2$ with dimension $N$. We will shortly see that the large-$N$ limit is most transparent when the matrices are traceless; thus, we define an ensemble of matrices $\vM_i:=\vD_i-(\btr{\vD_i})\vI$, such that there is a corresponding set of $\vM_i$ for every set of $\vD_i$ at each dimension $N$. This merely incurs a mild ``error,'' and the bound between their difference can be found in \autoref{lem:lipschitz-like} in appendix \ref{append:Lipschitz_Bounds}. We take the limit to find that
    \begin{align}\label{eq:lipschitz_like}
        \left|\lim_{N\rightarrow\infty} f_{\vD_1,\vD_2}(N)-\lim_{N\rightarrow\infty}f_{\vM_1,\vM_2}(N)\right| \leq p\left(|\btr{\vD_1}|+|\btr{\vD_2}| + |\btr{\vD_1}||\btr{\vD_2}|\right).
    \end{align}
    Note the right hand side is constant since the individual $\btr{\vD_i}$ are constant for all integer $N\geq 48p^{7/2}$. We now wish to find the infinite $N$ limit of 
    \begin{align*}
        f_{\vM_1,\vM_2}=\BE_{\vU,\vV\leftarrow\mu} \btr[\left(\vU\vM_1\vU^\dagger\vV\vM_2\vV^\dagger\right)^p].
    \end{align*}
    Given its form, we can use the Weingarten calculus to understand the leading order term, which corresponds to a sum of leading order diagrams. We essentially need analyze the maximal order of $N$ that each piece of the Weingarten formula can contribute. 
    
    We begin with the factors in the top and bottom contractions that give unnormalized traces of $\vM_1$ (and similarly for $\vM_2$). Because $\vM_1$ is traceless, diagrams with factors of $\tr[\vM_1]$ will disappear. Since the normalized traces $\btr[\vD_1^q]$ are kept constant for all powers $q$, the unnormalized traces of $\tr[\vM_1^q]$ for $q>1$ will scale with $N$. The upper bound scaling from these unnormalized trace factors, then, comes from partitioning the $\vM_1$'s into as many separate traces as possible: either into pairs if $p$ is even, or pairs with one triplet if $p$ is odd. As such, for even $p$, the upper bound goes as $\tr[\vM_1^2]^{p/2}\tr[\vM_2^2]^{p/2}<N^{p/2}N^{p/2}$, and similarly for odd.

    The Weingarten functions are the other piece of the diagram; as discussed in appendix \ref{app:weingarten}, their leading order term scales as 
    \begin{align*}
        \text{Wg}(\sigma\tau^{-1},N) = \mathcal{O}\left(N^{-2p+\text{cycles}(\sigma\tau^{-1})}\right).
    \end{align*}
    The big-O notation conceals dependencies on $p$, but since we will be taking the infinite $N$ limit for fixed $p$, this dependency can be ignored. We then can tally the contributions of factors of $N$. For even $p$, we find each \emph{leading order} term in the Weingarten sum takes the form
    \begin{align*}
        f_{\vM_1,\vM_2}(N)&= \mathcal{O}\left(\frac{1}{N}N^{p/2}N^{p/2}\text{Wg}(\sigma\tau^{-1})\right) \\
        &=\mathcal{O}\left(N^{-p-1+\text{cycles}(\sigma\tau^{-1})}\right)
    \end{align*}
    For odd $p$, similar analysis shows 
    \begin{align*}
        f_{\vM_1,\vM_2}(N)=\mathcal{O}(N^{-p-2+\text{cycles}(\sigma\tau^{-1})}).
    \end{align*}
    Since the maximum number of cycles possible for any element in $\mathcal{S}(p)$ is $p$, we see that the order of $f_{\vM_1,\vM_2}(N)$ must at most be $N^{-1}$ for all $p$. Hence in the limit $N\rightarrow\infty$, $f_{\vM_1,\vM_2}(N)\rightarrow 0$. We plug this into \eqref{eq:lipschitz_like} to conclude our proof.
\end{proof}

\subsubsection{Adjusting Pole Structure} \label{subsubsec:pole_structure}
Finally, we need to find a function whose pole structure actually matches the conditions laid our for our Markov-type inequality in \autoref{lem:main_markov_inequality}. The pole structure affects a lot---to find $a_k$ that satisfy \autoref{lem:main_markov_inequality}, we need to have poles outside $[-1,1]$. Hence we need to stretch the domain of our function to move the poles outwards. Moreover, the functions $\lambda_n(x)$ and $m_n(x)$ as defined in \autoref{lem:main_markov_inequality} are completely dependent upon the $a_k$, so the poles determine whether a choice of the bound $c$ on $\lambda_n(x)$ exists. In order to combat these problems, we choose the following function to apply \autoref{lem:main_markov_inequality} to: 
\begin{align*}
    h(z):=f_{\vD_1,\vD_2}(N) \quad \text{such that } z=\sqrt{p'/N},\;\; p'=48p^{7/2}.
\end{align*} This choice of range of $z$ moves the poles outside $[-1,1]$; allows the bound of 1 from \autoref{cor:moment_problem_applied_apriori_bound}  to apply to an infinite number of $z$ including $z=\{1,-1\}$; and flips the $N\rightarrow\infty$ limit to $z\rightarrow 0$. We can prove the following properties about our new function:
\begin{lem}[Properties of $h(z)$]\label{lem:properties_of_h}
    Let $h(z)=f_{\vD_1,\vD_2}(N)$ such that $z=\sqrt{p'/N}$ and $p'=48p^{7/2}$. Then the following are true:
    \begin{enumerate}
        \item $h(z)$ is a rational function with a numerator of degree $n=4p^2.$
        \item The denominator of $h(z)$ can be written as $\sqrt{t_{2n}(z)}$ where
        \begin{align*}
    \sqrt{t_{2n}(z)}&=\sqrt{\prod_{k=1}^{8p^2} (1+a_kx)}\\
    a_k&=\left(\bigcup^{2p} b^+_k\right)\cup\left(\bigcup^{2p} b^-_k\right)\cup\left(\bigcup^{2p} i*b^+_k\right)\cup\left(\bigcup^{2p} i*b^-_k\right)\cup\{0\}^{8p}\\ 
    b^{\pm}_k&=\pm\left\{\sqrt{\frac{1}{p'}},...,\sqrt{\frac{p-1}{p'}}\right\}. 
    \end{align*}
        \item There exists a bound \begin{align*}
            \lambda_n(z):=\frac{1}{2}\sum^{2n}_{k=1} \frac{\sqrt{1-a_k^2}}{1+a_k z}<4p+4p\sqrt{2}(p-1),\quad \text{for}\quad z\in[-1,1].
        \end{align*}
    \end{enumerate}
\end{lem}
\begin{proof}[Proof of \autoref{lem:properties_of_h}]
    Recall from the structure of the function from the Weingarten calculus and \autoref{lem:rational_function} that $f_{\vD_1,\vD_2}(N)$ takes the heuristic form
\begin{align*}
    f_{\vD_1,\vD_2}(N)=\frac{\text{poly}(N,\text{order}\leq 2p^2)}{N^{2p}(N^2-1)^p(N^2-4)^p...(N^2-(p-1)^2)^p},
\end{align*}
where the numerator is a polynomial in $N$ of order $2p^2$.  This means that for corresponding $z$, $h(z)$ takes the form
\begin{align*}
    h_{\vD_1,\vD_2}(z) &= \frac{\text{poly}(1/z,\text{order}\leq 4p^2)}{(p'/z^2)^{2p}((p'/z^2)^2-1)^p((p'/z^2)^2-4)^p...((p'/z^2)^2-(p-1)^2)^p}\\ 
    &=\frac{\text{poly}(z,\text{order}\leq 4p^2)}{\left(1-\frac{1}{p'^2}z^4\right)^p\left(1-\frac{4}{p'^2}z^4\right)^p...\left(1-\frac{(p-1)^2}{p'}z^4\right)^p}.
\end{align*}
    This takes care of point 1, and rewriting the denominator gives point 2. As for the bound, we can plug $a_k$ into $\lambda_n(z)$ to find:
\begin{align*}
    \lambda_n(z)&=\frac{1}{2}\sum^{2n}_{k=1} \frac{\sqrt{1-a_k^2}}{1+a_k z}
        =4p+2p\sum^{p-1}_{k=1} \frac{\sqrt{p'(p'+k)}}{p'+kz^2} + \frac{\sqrt{p'(p'-k)}}{p'-kz^2} 
\end{align*}
We know
\begin{align*}
    \frac{\sqrt{p'(p'+k)}}{p'+kz^2} \leq \frac{\sqrt{p'(p'+k)}}{p'}= \sqrt{\frac{(p'+k)}{p'}}<\sqrt{2}, \quad k\in\{1,...,p-1\},\; z\in[-1,1].
\end{align*}
Similarly,
\begin{align*}
    \frac{\sqrt{p'(p'-k)}}{p'-kz^2} &\leq \frac{\sqrt{p'(p'-k)}}{p'-k}=\sqrt{\frac{p'}{p'-k}}< \sqrt{\frac{p'}{p'-(p-1)}}\quad k\in\{1,...,p-1\},\; z\in[-1,1]\\
    &\leq \sqrt{1+\frac{(p-1)}{p'-(p-1)}}<\sqrt{2}.
\end{align*}
Then we see a rough upper bound on $\lambda_n(z)$ is
\begin{align*}
    \lambda_n(z)<4p+2p\sum^{p-1}_{k=1} 2\sqrt{2} <4p+4p\sqrt{2}(p-1),\quad z\in[-1,1].
\end{align*}
\end{proof}

\subsection{Interpolate: Proof of \autoref{lem:UDUVDV_expected_moments}}
Armed our many lemmas about $f_{\vD_1,\vD_2}(N)$, we now assemble the proof of Lemma \ref{lem:UDUVDV_expected_moments}.
\begin{proof}[Proof of \autoref{lem:UDUVDV_expected_moments}]
As we've previously mentioned, our strategy is to use Markov-type inequality techniques to interpolate a bound on $f_{\vD_1,\vD_2}(N)$ from the infinite-$N$ limit. What does this mean specifically? Recall that the Markov-type inequality in Lemma \ref{lem:main_markov_inequality} bounds the magnitude of the derivative of an appropriately chosen algebraic rational function. Our plan, then, is to use this bound on the derivative to bound the difference between $\lim_{N\rightarrow\infty}f_{\vD_1,\vD_2}(N)$ and $f_{\vD_1,\vD_2}(N)$ for finite $N$ and ``interpolate'' a bound on $f_{\vD_1,\vD_2}(N)$ for finite $N$. 
\begin{claim}[Bound via Markov-type inequality]\label{claim:markov_inequality_W}For $N\in\left(\frac{97}{2}p^{7/2},\infty\right)$, 
\begin{align*} 
    \labs{\lim_{N\rightarrow\infty}f_{\vD_1,\vD_2}(N)-f_{\vD_1,\vD_2}(N)}\leq \sqrt{\frac{1}{N}} \frac{384\sqrt{3}p^{25/4}(\sqrt{2}p+1-\sqrt{2})}{p^{7/2}-\sqrt{2}p^2-p+\sqrt{2}p}.
\end{align*}
\end{claim}
\begin{proof}
Let us first review the conditions we need to satisfy to apply the sharp rational function Markov-type inequality on discrete points from Lemma \ref{lem:main_markov_inequality}: 
    \begin{enumerate}
        \item $r_n(x)$ must have a numerator that is an algebraic polynomial of degree $n$ in $x$ with real or complex coefficients.
        \item $r_n(x)$ must have a denominator of the form $t_{2n}(x).$ We will interpret this as a slightly stricter condition that $r_n(x)$ has no more than $n$ poles that are represented in the denominator as $(1+a_k x)$ with $a_k$ real or pairwise complex conjugate, $|a_k|<1$. 
        \item $|r_n(x)|\leq 1$ for a discrete set of points $\{x^*_i\}\in[-1,1]$, $-1=x^*_1<...<x^*_i=1$
        \item There exists a bound 
        \begin{align*}
            c\geq\sup_{x\in[-1,1]} \lambda_n(x)
        \end{align*}
        such that $1-\frac{cI}{2\sqrt{1-x^2}}> 0$ for some set of $x\in X\subseteq [-1,1]$, where $I$ is the largest interval between adjacent discrete points.
    \end{enumerate}
Recall the function $h(z)=f(N)$ with domain $z=\sqrt{p'/N}$, $p'=48p^{7/2}$ that we constructed in \autoref{subsubsec:pole_structure}. For $h(z)$, conditions 1 and 2 are satisfied via \autoref{lem:properties_of_h}. Condition 3 is satisfied by \autoref{cor:moment_problem_applied_apriori_bound}. The bound in condition 4 is found in \autoref{lem:properties_of_h}:
\begin{align*}
    c=4p+4p\sqrt{2}(p-1),
\end{align*}
but we need to ascertain the interval $X$ such that $1-\frac{cI}{2\sqrt{1-z^2}}>0$, so that any bound on the derivative we find via the Markov-inequality actually holds for some section of $z\in [-1,1].$

Since $z=\sqrt{p'/N}$, the points we have bounds on are at $x^*_i=\sqrt{p'/N}$ for integer $N\geq 48p^{7/2}.$ The points get denser near the origin, so 
\begin{align*}
    I=1-\sqrt{\frac{p'}{p'+1}}.
\end{align*}
Then we have
\begin{align*}
    1-\frac{c(z)I}{2}&=1-\frac{(4p+4p\sqrt{2}(p-1))\left(1-\sqrt{\frac{p'}{p'+1}}\right)}{2\sqrt{1-z^2}},
\end{align*}
so we would like
\begin{align*}
    \sqrt{1-z^2}>(2p+2p\sqrt{2}(p-1))\left(1-\sqrt{\frac{p'}{p'+1}}\right)
\end{align*}
We know
\begin{align*}
    1-\sqrt{\frac{p'}{p'+1}}=1-\sqrt{1-\frac{1}{p'+1}}\leq \frac{1}{p'+1}\leq \frac{1}{p'}.
\end{align*}
Hence, a slightly smaller but more manageable $X$ can also be determined via
\begin{align*}
    \sqrt{1-z^2}&>\frac{p(2+2\sqrt{2}(p-1))}{p'}
\end{align*}
or, equivalently,
\begin{align*}
    1-\left(\frac{p(2+2\sqrt{2}(p-1))}{p'}\right)^2&=1-\left(\frac{p(2+2\sqrt{2}(p-1))}{48p^{7/2}}\right)^2>z^2.
\end{align*}
We want to find an interval $X$ that will hold for all $p$. Notice that the left hand side increases with $p$, so we should really bound $X$ with $p=1$. This gives us $z^2<575/576.$ We've hence satisfied all the conditions and can apply Lemma \ref{lem:main_markov_inequality} to $h(z)$:
\begin{align*}
    |h'(z)|    \leq \left(1+\frac{cI}{2\sqrt{1-z^2}-cI}\right)
    \begin{cases}
        \frac{\lambda_n(z)}{\sqrt{1-z^2}}, & z\in [z_1,z_n] \cap  X,\\
        |m'_n(x)|, & x\in ([-1,z_1]\cup[z_n,1]) \cap X.
    \end{cases}
\end{align*}
This is the tightest bound, but since we mostly care about the order of $p$ and $N$ for practical purposes, we can resort to the slightly less tight bound found in the proof of Lemma \ref{lem:markov_rational_inequality_discrete} by plugging in $c/\sqrt{1-z^2}$ for the piecewise function. As for the range in which this is accurate in, $X$ is defined by $z^2 < 575/576$, which is smaller than the $[z_1,z_n]$ range set by the first and last roots of $m_n(z)$ for $p> 2$. Since we want this to hold for all $p$, however, we will use the lower bound on $z_n$: $z_n$ for all $p$ is lowerbounded by the $p=2$ value for $z_n$, since $m_n$ is the cosine of a sum of arccosine functions that pushes the final root to grow monotonically with $p$. This value can numerically be found to be greater than $0.995.$ Now we evaluate the function:
\begin{align*}
    |h'(z)| &< \frac{c(z)}{1-\frac{c(z)I}{2}} = \frac{4p(1+\sqrt{2}(p-1))}{\sqrt{1-z^2}-2p(1+\sqrt{2}(p-1))(1-\sqrt{\frac{p'}{p'+1}})}\\
    &\leq \frac{4p'p(1+\sqrt{2}(p-1))}{p'\sqrt{1-z^2}-2p(1+\sqrt{2}(p-1))}\\
    &\leq \frac{96p'p(1+\sqrt{2}(p-1))}{p'-48p(1+\sqrt{2}(p-1))}, \quad z^2<\frac{575}{576}
\end{align*}
We can continue to plug in $z^2<\frac{575}{576}$ for our bounds here since it simply gives us a looser bound. Finally, we take this bound and interpolate on $f_{\vD_1,\vD_2}(N)$:
\begin{align*}
    \left|\lim_{N\rightarrow\infty}f(N)-f(N)\right| &= |h(0)-h(z)|=\left|\int_{0}^z h(z') dz' \right| \leq \int_{0}^z |h(z')| dz'\\
    &< \int_{0}^z \frac{96p'p(1+\sqrt{2}(p-1))}{p'-48p(1+\sqrt{2}(p-1))} dz', \quad z^2<0.995^2\\
    &=z \frac{96p'p(1+\sqrt{2}(p-1))}{p'-48p(1+\sqrt{2}(p-1))},\quad z^2<0.990\\
    &=\sqrt{\frac{1}{N}} \frac{384\sqrt{3}p^{25/4}(\sqrt{2}p+1-\sqrt{2})}{p^{7/2}-\sqrt{2}p^2-p+\sqrt{2}p},\quad N>\frac{p'}{0.99}.
\end{align*}
\end{proof}
For our purposes, we do not need the exact constants in the bound, so we can write 
\begin{align*}
    \left|\lim_{N\rightarrow\infty}f_{\vD_1,\vD_2}(N)-f_{\vD_1,\vD_2}(N)\right|\leq \sqrt{\frac{1}{N}} \frac{384\sqrt{3}p^{25/4}(\sqrt{2}p+1-\sqrt{2})}{p^{7/2}-\sqrt{2}p^2-p+\sqrt{2}p}=\mathcal{O}\left(\frac{p^{15/4}}{\sqrt{N}}\right).
\end{align*}
Then recall that from \autoref{lem:infinite_N_limit},
\begin{align*}
    \lim_{N\rightarrow\infty} |f_{\vD_1,\vD_2}(N)| \leq p\left(|\btr{\vD_1}|+|\btr{\vD_2}| + |\btr{\vD_1}||\btr{\vD_2}|\right).
\end{align*}
Via the triangle inequality, we have that 
\begin{align*}
    |f_{\vD_1,\vD_2}(N)| \leq 
    \left|\lim_{N\rightarrow\infty} f_{\vD_1,\vD_2}(N)\right|+\left|\lim_{N\rightarrow\infty}f_{\vD_1,\vD_2}(N)-f_{\vD_1,\vD_2}(N)\right|.
\end{align*}
We then see that 
\begin{align*}
    |f_{\vD_1,\vD_2}(N)|\leq \mathcal{O}\left(\frac{p^{15/4}}{\sqrt{N}}\right)+p\left(|\btr{\vD_1}|+|\btr{\vD_2}| + |\btr{\vD_1}||\btr{\vD_2}|\right).
\end{align*}
We now thread the argument back to our original function $f(\vD_1,\vD_2)$, which doesn't depend on $N$. Recall that since our bound on $f_{\vD_1,\vD_2}(N)$ holds for any choice of $\vD_1,\vD_2$, on all $N$, that it must also apply to $f(\vD_1,\vD_2)$:
\begin{align*}
    |f(\vD_1,\vD_2)|\leq \mathcal{O}\left(\frac{p^{15/4}}{\sqrt{N}}\right)+p\left(|\btr{\vD_1}|+|\btr{\vD_2}| + |\btr{\vD_1}||\btr{\vD_2}|\right).
\end{align*}
Now recall that $f(\vD_1,\vD_2)$ is only the Haar-averaged normalized trace moments. To finish the proof, we average over the distributions of the $\vD_i$, which only affects the second term.
\end{proof}
Given the facts presented in \autoref{subsec:properties_of_exp}, we could apply this theorem to the product of two exponentiated Gaussian to show that its small moments make it indistinguishable from Haar. However, we do not analyze this explicitly in this paper -- instead, we will apply it to a product of two exponentiated Gaussian designs in the next section.


\section{A Matrix Lindeberg Principle: Convergence of Spectrum (Proof of Lemma \ref{lem:clt_spectrum_small_moments})}\label{sec:lindeberg_spectrum}

The previous section has motivated the use of the Gaussian unitary ensemble: particularly, exponentiating two GUEs gives rise to an ensemble that is already close to Haar. But how does one generate random Gaussians? An old yet simple approach is the central limit theorem, which states that under mild conditions, a centered, independent sum converges to the Gaussian distribution with the same variance 
\begin{align*}
   s= \frac{1}{\sqrt{m}} \sum_{i=1}^m x_i \stackrel{dist}{\rightarrow} g \quad \text{where}\quad \BE x_i^2 = g^2.
\end{align*}
Moreover, under stronger conditions it is also possible to control the rate of this convergence in the form of a bound on distance between the true distribution and the limiting distribution. These are generally known as Berry–Esseen type bounds. 

To generate a random GUE matrix $\vG$ it is natural to consider a sum of centered, independent random matrices $\vX_j$ and identify the conditions under which it converges:
\begin{align}
    \vS= \frac{1}{\sqrt{m}} \sum_{i=1}^m \vX_i \stackrel{?}{\rightarrow} \vG \quad \text{where}\quad \BE[\vX_i^{\otimes 2}]=\BE[\vG^{\otimes2}]. \label{eq:postulated_CLT_condition}
\end{align}
Surprisingly, the scalar approach swiftly generalizes to the matrix case when considering Gaussian designs (\autoref{defn:G_design}). For example, it is known how to control the spectral properties assuming the low tensor moments match the Gaussian case exactly~\cite{chen2023sparse}.

However, in our application, the initial quality of our low-moment designs is not perfect but rather contains small errors, unlike the exact case in equation \eqref{eq:postulated_CLT_condition}. Handling these errors is more delicate than in the scalar case. With random matrices, the error can be quantified in very different kinds of norms and propagates in an interleaved manner stemming from the noncommutative nature of the variables, which is difficult to analyze directly. Moreover, stemming from the way we implement the matrices, the error in the \textit{basis} and \textit{spectrum} are controlled in different ways (see \ref{subsec:note_on_decomp_GUE}). This leads to the following recipe of proving convergence in two steps, treating the spectrum and basis independently.
\begin{enumerate}[label=(\roman*)]
    \item \label{item:conv_spectrum} First, analyze convergence of the spectrum (with respect to trace moments). That is,
    \begin{align}
         \btr \left(\sum_j^m \vU_{j}\vD\vU_{j}^\dag\right)^k \approx  \btr \vG^k \label{eq:lindeberg_spectrum_convg}
    \end{align} 
    where the bases $\vU^{(\prime)}_{1} \dotsto \vU^{(\prime)}_{m} \in U(N)$ are Haar random (as with GUE), but the spectrum $\vD$ only approximately matches the low trace moments of GUE. 
    \item \label{item:conv_basis} Second, analyze convergence of the basis, 
    $$   \sum_j^m \Tilde{\vU}_{j}\vD\Tilde{\vU}_{j}^\dag \approx \sum_j^m \vU_{j}\vD\vU_{j}^\dag .$$
    where $ \Tilde{\vU}_{1} \dotsto \Tilde{\vU}_{m} \in \unitary(N)$ are approximate unitary designs for low moments.
    \item Then conclude a convergence of the form, 
    \begin{align*}
         \sum_j^m \Tilde{\vU}_{j}\vD\Tilde{\vU}_{j}^\dag \approx \sum_j^m \vU_{j}\vD\vU_{j}^\dag \approx \vG
    \end{align*}
\end{enumerate} 
This story is not made precise yet for the sake of simplicity, and because we are interested in slightly different objects. Item \ref{item:conv_spectrum} is the focus of this section, and item \ref{item:conv_basis} is addressed in the next section (Section~\ref{sec:lindeberg_basis}).

Specifically, in this section we analyze the unitary ensemble 
\begin{align*}
    \vW:=e^{i \frac{\theta}{\sqrt{m}} \sum_j^m \vU_{j} \vD \vU_{j}^\dag}\cdot e^{i \frac{\theta}{\sqrt{m}} \sum_j^m \vU^{\prime}_{j}\vD\vU^{\prime \dag}_{j}} \quad \text{for precisely chosen $\theta$},
\end{align*}
where again as in \eqref{eq:lindeberg_spectrum_convg}, the bases $\vU^{\prime}_{1} \dotsto \vU^{\prime}_{m}\in U(N)$ are Haar random, but the spectrum $\vD$ only approximately matches the low trace moments of GUE. One should think of this ensemble as a sequence (in $m$) of unitaries converging to the random unitary $\e^{\ri \theta \vG_1}\e^{\ri \theta \vG_2}$, which we discussed in \autoref{sec:products_of_gaussians}. We first explore the trace moments of just one of these exponentiated Gaussian designs, and then we apply \autoref{lem:UDUVDV_expected_moments} to control the moments of the product.

\subsection{A Note on Decomposing GUE}\label{subsec:note_on_decomp_GUE}

Recall that the eigenbasis of GUE matrices is distributed as Haar (see appendix \ref{app:GUE_properties}). Specifically $\vG$ decomposes as
$$\vG = \vU \vD  \vU^{\dag}$$
where $\vU$ is a Haar random unitary and $\vD$ is a random independent diagonal matrix with a GUE spectral distribution. Thus, it is possible to sample from GUE by sampling the basis and spectrum independently.  Since we are interested in employing a random Hermitian matrix $\vH$ which merely matches the first $q$ moments of a random GUE matrix, it is sufficient that $\tilde{\vU}$ and $\vD$ only match the first $2q$ and $q$ moments of the basis and spectrum of GUE, respectively, where $ \vH = \tilde{\vU} \vD \tilde{\vU}^\dag.$
If the random matrix $\vU \in \text{U}(N)$ is a unitary $q$-design and the diagonal matrix $\vD \in \text{U}(N)$ has a GUE spectral density then 
        $$  \Expect [{(\tilde{\vU} \vD \tilde{\vU}^\dag)^{\otimes q}}] = \Expect [{\vG^{\otimes q}}].$$
Thus, sampling from $\vH$ amounts to constructing $\vD$ which matches $q$ moments of GUE and conjugating $\vD$ with a unitary $q$-design. This is advantageous because we already know how to efficiently construct unitary designs which approximately match low moments \cite{brandao2016local,harrow2023approximate,haferkamp2022random}.

\subsection{Convergence of Trace Moments}\label{subsec:lindeberg_wg}
\begin{thm}\label{thm:lindeberg_wg}
    Let $\vU_1 \dotsto \vU_m \in \unitary(N)$ be Haar random unitaries, and $\vG \in \gl(N)$ be a GUE random matrix. Suppose that for a diagonal matrix $\vD \in \gl(N,\BC)$,
    \begin{align*}
        \abs{\btr(\vD^{q^{\prime}}) - \Expect_{\vG} \btr(\vG^{q^{\prime}})} = \delta_{q^{\prime}} = 2^{q'} \cdot \frac{2q+4}{N}\quad \forall~ 1\leq {q^{\prime}} \leq q.
    \end{align*}
    Then, the $p$th trace moment of the exponentiated Gaussian design satisfies
    \begin{align*}
        &\abs{\Expect_{\lset{\vU_j}_j} \btr\left(\e^{i\frac{\theta p}{\sqrt{m}} \sum_j \vU_j \vD \vU_j^\dagger}\right) - \Expect_{\vG} \btr\left(\e^{i\theta p \vG}\right)} \\
        &\le \frac{(\theta p)^{q+1}}{(\sqrt{m})^{q-1}} \frac{1}{(q+1)!}(\norm{\vD}_{op}^{q+1}+C^{q+1})  +  \ltup{\frac{2\theta p \sqrt{m} }{\sqrt{N}}} \left( 1+ \ltup{\frac{2\theta p}{\sqrt{m}}}^{q-1} \right)
    \end{align*}
    for an appropriately large $N\geq \Omega\left(q^{4q}\right)$ and some positive constant $C$.
\end{thm}

We note that our assumption $\Omega(q^{4q}) \leq N$ is what eventually limits our algorithm's applicability to $T<2^{O(n/\log n)}$ moments, as in our final algorithm $q=O(\log T)$.

\begin{proof}
    Our strategy is to use the Lindeberg replacement trick for certain forms of Taylor expansion. Recall 
    \begin{align*}
        \vG \stackrel{dist}{\sim} \frac{1}{\sqrt{m}}\sum_j^m \vG_j.
    \end{align*}
    for independent GUE matrices $\vG_j$.
    We define a sequence of random matrices that interpolates between the desired targets~\cite{chen2023sparse}:
    $$\vA_l := \sum_{j=1}^{l-1} \vU_j \vD \vU_j^\dag  + \sum_{j=l+1}^{m} \vG_j  .$$
    By a telescoping sum argument and the triangle inequality, we build the following set of hybrids,   
    \begin{align}
      &\abs{\Expect_{\lset{\vU_j}_j} \btr(\e^{i\frac{\theta p}{\sqrt{m}} \sum_j \vU_j \vD \vU_j^\dagger}) - \Expect_{\vG} \btr(\e^{i\frac{\theta p}{\sqrt{m}} \sum_j \vG_j})} \\
       &\quad\quad=  \abs{\Expect_{\lset{\vU_j}_j \cup \lset{\vG_j}_j}\btr \ltup{\sum_{l=1}^{m} \e^{i\frac{\theta p}{\sqrt{m}} (\vA_l + \vU_l \vD \vU_l^\dag)} - \e^{i\frac{\theta p}{\sqrt{m}} (\vA_l + \vG_l)}}} \nonumber\\
        &\quad\quad\leq \sum_{l=1}^{m} \abs{\Expect_{\lset{\vU_j}_j \cup \lset{\vG_j}_j}\btr \ltup{ \e^{i\frac{\theta p}{\sqrt{m}} (\vA_l + \vU_l \vD \vU_l^\dag)} - \e^{i\frac{\theta p}{\sqrt{m}} (\vA_l + \vG_l)}}}. \label{eq:tr_udu_v_G} 
    \end{align}
    A corollary of Duhamel's formula (\autoref{cor:duhamel}) yields the following, where expectations are taken over every random matrix:
    \begin{align} 
        \Expect \btr \lbr{ \e^{i\frac{\theta p}{\sqrt{m}} (\vA_l + \vU_l \vD \vU_l^\dag )}} &= \Expect \btr \lbr{e^{i \frac{\theta p}{\sqrt{m}}\vA_l }} \nonumber \\
        &+\sum_{r=1}^q (i)^l  \idotsint\limits_{\frac{2\theta p}{\sqrt{m}} >s_1 > \dots > s_r > 0} \Expect \btr \lbr{ \prod_{j=1}^{r} \diff s_j  \left(\prod_{j=1}^{r} e^{\frac{i}{2} \vA_l(s_{j-1}-s_{j})}\vU_l \frac{1}{2}\vD \vU_l^\dag \right) e^{\frac{i}{2} \vA_l s_r}} \nonumber\\
        & + (i)^{q+1} \idotsint\limits_{\frac{\theta p}{\sqrt{m}} >s_1 > \dots > s_{q+1} > 0} \Expect \btr\lbr{\prod_{j=1}^{q+1} \diff s_j  \left(\prod_{j=1}^{q+1} e^{i \vA_l(s_{j-1}-s_{j})}\vU_l \vD \vU_l^\dag \right) e^{i\left(\vA_l+\vU_l \vD \vU_l^\dag\right)s_{q+1}} }  \label{eq:lb_wg_udu} \\ 
        \Expect \btr\lbr{\e^{i\frac{\theta p}{\sqrt{m}} (\vA_l + \vG_l)}} &= \Expect \btr\lbr{e^{i \frac{\theta p}{\sqrt{m}}\vA_l}} \nonumber\\
        &+ \sum_{r=1}^q (i)^l  \idotsint\limits_{\frac{2\theta p}{\sqrt{m}} >s_1 > \dots > s_r > 0} \Expect \btr\lbr{\prod_{j=1}^{r} \diff s_j  \left(\prod_{j=1}^{r} e^{\frac{i}{2} \vA_l(s_{j-1}-s_{j})}\frac{1}{2}\vG_l \right) e^{\frac{i}{2}  \vA_l s_r}} \nonumber\\
        & + (i)^{q+1} \idotsint\limits_{\frac{\theta p}{\sqrt{m}} >s_1 > \dots > s_{q+1} > 0} \Expect \btr\lbr{\prod_{j=1}^{q+1} \diff s_j  \left(\prod_{j=1}^{q+1} e^{i \vA_l(s_{j-1}-s_{j})}\vG_l \right) e^{i\left(\vA_l+\vG_l \right)s_{q+1}}}.  \label{eq:lb_wg_g}  
    \end{align} 
    Above, we also use the linearity of expectation and trace to move the operators into the integrand. We also insert a $1/2$ into the second line of each equation for ease of analysis later. Before we make substitutions, we introduce several notations to highlight exactly where \eqref{eq:lb_wg_udu} and \eqref{eq:lb_wg_g} differ to further analyze their deviations. It is known that the eigenbasis of the GUE is distributed as Haar. Specifically, $\vG_l$ factors as $\vG_l = \vV_l \vLambda_l \vV^\dag_l$ where $\vV_l$ is a Haar random unitary and $\vLambda_l$ is distributed according to the spectrum of GUE. As such, we may rewrite 
    \begin{multline*}
        \Expect_{\substack{\vU_1 \dotsto \vU_{l-1}\\ \vG_{l} \dotsto \vG_m}} \btr \lbr{  \left(\prod_{j=1}^{r} e^{\frac{i}{2} \vA_l(s_{j-1}-s_{j})}\frac{1}{2}\vG_l \right) e^{\frac{i}{2}  \vA_l s_r}} \\ = \Expect_{\substack{\vU_1 \dotsto \vU_l \\\vG_{l+1} \dotsto \vG_m\\ \vLambda_l}}\btr \lbr{  \left(\prod_{j=1}^{r} e^{\frac{i}{2} \vA_l(s_{j-1}-s_{j})}\vU_l \frac{1}{2}\vLambda_l \vU_l^\dag \right) e^{\frac{i}{2}  \vA_l s_r}}.
    \end{multline*}
    Since we are taking expectations, the substitution $\vV_l \rightarrow \vU_l$ is trivial since $\vU_l$ is not correlated with other terms. 
    This leads us to define the matrix functions $F_r: \text{diag} (N) \rightarrow \mathbb{R}$ for $1\leq r \leq q$ in any diagonal matrix $\vQ$
    \begin{align*}
        F_r(\vQ) := \Expect_{\substack{\vU_1 \dotsto \vU_l \\ \vG_{l+1} \dotsto \vG_m}} \btr \lbr{   \left(\prod_{j=1}^{r} e^{\frac{i}{2} \vA_l(s_{j-1}-s_{j})}\vU_l \frac{1}{2}\vQ \vU_l^\dag\right) e^{\frac{i}{2}  \vA_l s_r}}.
    \end{align*}
    Note that $\vQ$ is not necessarily a random matrix. We also define the terms that contribute to the Lindeberg error from their high trace moments ($q+1$).
    \begin{align*}
        J_{UDU} &:= (i)^{q+1} \idotsint\limits_{\frac{\theta p}{\sqrt{m}} >s_1 > \dots > s_{q+1} > 0} \Expect \btr\lbr{\prod_{j=1}^{q+1} \diff s_j  \left(\prod_{j=1}^{q+1} e^{i \vA_l(s_{j-1}-s_{j})}\vU_l \vD \vU_l^\dag \right) e^{i\left(\vA_l+\vU_l \vD \vU_l^\dag\right)s_{q+1}} } \\ 
        J_{GUE} &:= (i)^{q+1} \idotsint\limits_{\frac{\theta p}{\sqrt{m}} >s_1 > \dots > s_{q+1} > 0} \Expect \btr\lbr{\prod_{j=1}^{q+1} \diff s_j  \left(\prod_{j=1}^{q+1} e^{i \vA_l(s_{j-1}-s_{j})}\vG_l \right) e^{i\left(\vA_l+\vG_l \right)s_{q+1}}}
    \end{align*}
    In this notation, equations \eqref{eq:lb_wg_udu} and \eqref{eq:lb_wg_g} simplify respectively to 
    \begin{align*} 
        \Expect \btr \lbr{ \e^{i\frac{\theta p}{\sqrt{m}} (\vA_l + \vU_l \vD \vU_l^\dag )}} &= \Expect \btr \lbr{e^{i \frac{\theta p}{\sqrt{m}}\vA_l }} + J_{UDU} +\sum_{r=1}^q (i)^l  \idotsint\limits_{\frac{2\theta p}{\sqrt{m}} >s_1 > \dots > s_r > 0} F_r(\vD) \prod_{j=1}^{r} \diff s_j      \\ 
        \Expect \btr\lbr{\e^{i\frac{\theta p}{\sqrt{m}} (\vA_l + \vG_l)})} &= \Expect \btr\lbr{e^{i \frac{\theta p}{\sqrt{m}}\vA_l}} + J_{GUE} + \sum_{r=1}^q (i)^l  \idotsint\limits_{\frac{2\theta p}{\sqrt{m}} >s_1 > \dots > s_r > 0} \Expect_{\vLambda_l}  F_r(\vLambda_l) \prod_{j=1}^{r} \diff s_j. 
    \end{align*} 
    Plugging these into \eqref{eq:tr_udu_v_G} yields,
    \begin{align*}
        &\abs{\Expect \btr \ltup{ \e^{i\frac{\theta p}{\sqrt{m}} (\vA_l + \vU_l \vD \vU_l^\dag)} - \e^{i\frac{\theta p}{\sqrt{m}} (\vA_l + \vG_l)}}} \\
        &= \abs{ J_{UDU}  - J_{GUE} + \ltup{\sum_{r=1}^q (i)^l  \idotsint\limits_{\frac{2\theta p}{\sqrt{m}} >s_1 > \dots > s_r > 0} F_r(\vD) \prod_{j=1}^{r} \diff s_j}   - \ltup{\sum_{r=1}^q (i)^l  \idotsint\limits_{\frac{2\theta p}{\sqrt{m}} >s_1 > \dots > s_r > 0} \Expect_{\vLambda_l}  F_r(\vLambda_l) \prod_{j=1}^{r} \diff s_j }} \\
    \end{align*}
    Recall that $\vD$ and $\vLambda_l$ are close in their first $q$ trace moments. As such, it is reasonable to posit that the distance between terms of $\vD$ or $\vLambda_l$ which only depend on their first $q$ trace moments are small. However, there is no reason to expect the same of the terms containing their higher trace moments: $J_{UDU}$ and $J_{GUE}$. This leads us to collect terms by the triangle inequality as follows
    \begin{align*}
        &\abs{\Expect \btr \ltup{ \e^{i\frac{\theta p}{\sqrt{m}} (\vA_l + \vU_l \vD \vU_l^\dag)} - \e^{i\frac{\theta p}{\sqrt{m}} (\vA_l + \vG_l)}}} \\
        &\leq \abs{J_{UDU}}  +\abs{J_{GUE}} +  \sum_{r=1}^q   \abs{ ~\idotsint\limits_{\frac{2\theta p}{\sqrt{m}} >s_1 > \dots > s_r > 0} \ltup{F_r(\vD) - \Expect_{\vLambda_l}  F_r(\vLambda_l) }\prod_{j=1}^{r} \diff s_j  } \\
        &\leq \abs{J_{UDU}}  +\abs{J_{GUE}} +  \sum_{r=1}^q  \;\;\idotsint\limits_{\frac{2\theta p}{\sqrt{m}} >s_1 > \dots > s_r > 0}  \abs{ F_r(\vD) - \Expect_{\vLambda_l}  F_r(\vLambda_l) }\prod_{j=1}^{r} \diff s_j  . \tag{Triangle inequality}
    \end{align*}
    In the following, we give a naive upper bound on $\abs{J_{UDU}}$ and $\abs{J_{GUE}}$ and determine the order of  $\abs{F_r(\vD) - \Expect_{\vLambda_l}  F_r(\vLambda_l) }$ given that the trace moments of $\vD$ and $\vLambda_l$ are close for $r\leq q$. 

    \subsubsection{Difference in \texorpdfstring{$F_r$}{Fr}}\label{subsubsec:F}
    We begin with the analysis of $\abs{F_r(\vD) - \Expect_{\vLambda_l}  F_r(\vLambda_l) }$. 
    \begin{align*}
        F_r(\vQ) = \Expect_{\substack{\vU_1 \dotsto \vU_l \\ \vG_{l+1} \dotsto \vG_m}} \btr \lbr{   \left(\prod_{j=1}^{r} e^{\frac{i}{2} \vA_l(s_{j-1}-s_{j})}\vU_l \frac{1}{2}\vQ \vU_l^\dag\right) e^{\frac{i}{2}  \vA_l s_r}}
    \end{align*}
    We notice that the function $F_r$ has a familiar form. It is a trace over products of interleaved Haar random unitaries as we've seen in \autoref{sec:products_of_gaussians}. This leads us to use a similar approach. Invariance to the Haar measure allows us to use Weingarten calculus, which converts the problem of integrating over the unitary group to evaluating trace moments of the ``interleaved'' matrices $\vQ$ and $e^{\frac{i}{2} \vA_l(s_{j-1}-s_{j})}$ for $1\leq j\leq r$.
    Specifically, this becomes a sum over products of particular Weingarten functions (see appendix \ref{app:weingarten}) and trace diagrams. The sum is indexed by permutations which indicate the appropriate Weingarten function and a unique product of traces of the interleaved matrices.
    \begin{align}
        F_r(\vQ) &= \frac{1}{N}\Expect_{\substack{\vU_1 \dotsto \vU_l \\ \vG_{l+1} \dotsto \vG_m}} \tr \lbr{   \left(\prod_{j=1}^{r} e^{\frac{i}{2} \vA_l(s_{j-1}-s_{j})}\vU_l \frac{1}{2}\vQ \vU_l^\dag\right) e^{\frac{i}{2}  \vA_l s_r}} \nonumber\\
        &= \frac{1}{N2^r}\Expect_{\substack{\vU_1 \dotsto \vU_{l-1} \\ \vG_{l+1} \dotsto \vG_m}} \sum_{\tau,\sigma\in S(r)} \text{Wg}(\tau\sigma^{-1},N) K(\tau,\xi_r) \prod_{k(\sigma)} \tr(\vQ^{n_k(\sigma)})^{m_k(\sigma)}, \quad  \sum_k n_k m_k = r \label{eq:lb_wg_exp_r}, 
    \end{align}
    where $\xi_{r}$ is the set of $r$ unitaries
    \begin{align*}
        \xi_{r} := \lset{e^{\frac{i}{2} \vA_l(s_{j-1}-s_{j})}}_{j=2}^q \cup \lset{e^{\frac{i}{2}\vA_ls_{q+1}}e^{\frac{i}{2} \vA_ls_1}}.
    \end{align*}    
    Let us discuss the notation in this equation. The Weingarten function is familiar by now: a weight upon each term of the sum, where each term is determined by cycle structures. Specifically, the function $K(\tau,\xi_r)$ is a product of traces over subsets of terms determined by a partition $\lambda_\tau(\xi_r)$ over the set $\xi_{r}$. The elements of $\lambda_\tau$ are subsets $B$ of $\xi_{r}$ whose disjoint union is $\xi_{r}$. That is, 
    \begin{align*}\label{eq:defn_of_ktau}
        K(\tau, \xi_r) := \prod_{B \in \lambda_{\tau}(\xi_r)} \tr\ltup{\prod_{\vV \in B} \vV}.
    \end{align*}
    Note that given the set $\xi_{r}$, only the term $K(\tau,\xi_r)$ has any dependence upon the ensembles $\vU_1 \dotsto \vU_{l-1}, \vG_{l+1} \dotsto \vG_m$. Let us now address the rest of the expression. We explicitly write out the products over traces of the diagonal matrix $\vQ$. Note that there can never be more than $r$ products since there are only $r$ instances of $\vQ$ in the function. Now we bound the difference between $F_r(\vD)$ and $\Expect F_r(\vLambda_1)$. Given that they have the same Weingarten structure,
    \begin{align*}
        &\abs{ \ltup{F_r(\vD) - \Expect_{\vLambda_l}  F_r(\vLambda_l) }} \\
        &=\frac{1}{N2^r}\abs{\Expect_{\substack{\vU_1 \dotsto \vU_{l-1} \\ \vG_{l+1} \dotsto \vG_m}} \sum_{\tau,\sigma\in S(r)} \text{Wg}(\tau\sigma^{-1},N) K(\tau,\xi_r) \left(\prod_{k(\sigma)}\tr(\vD^{n_k(\sigma)})^{m_k(\sigma)} -  \Expect_{\vLambda_l}\prod_{k(\sigma)}\tr(\vLambda_1^{n_k(\sigma)})^{m_k(\sigma)}\right)}\\
        &\leq\frac{1}{N2^r}(r!)^2  \sup_{\sigma,\tau}\left\{\abs{\text{Wg}(\tau\sigma^{-1},N) }\cdot\abs{\Expect_{\substack{\vU_1 \dotsto \vU_{l-1} \\ \vG_{l+1} \dotsto \vG_m}}K(\tau,\xi_r)} \cdot\abs{\left(\prod_{k(\sigma)}\tr(\vD^{n_k(\sigma)})^{m_k(\sigma)} -  \Expect_{\vLambda_l}\prod_{k(\sigma)}\tr(\vLambda_1^{n_k(\sigma)})^{m_k(\sigma)}\right)}\right\}.
    \end{align*}
    There are three terms we are taking a supremum over. We will need to consider all three together to get our bound, but first we want to make use of the small difference in moments (stipulated in our lemma conditions) in the third term. Recall,
    \begin{align*}
        \frac{1}{N}\abs{\tr(\vD^{q^{\prime}}) - \Expect_{\vG} \tr(\vG^{q^{\prime}})}=\abs{\btr(\vD^{q^{\prime}}) - \Expect_{\vG} \btr(\vG^{q^{\prime}})} = \delta_{q^{\prime}} =2^{q'}\cdot \frac{2q+4}{N} \quad \text{for each}\quad~ 1\leq {q^{\prime}} \leq q.
    \end{align*}
    Then, 
    \begin{multline}
    \frac{1}{2^r}\abs{\prod_{k(\sigma)}\tr(\vD^{n_k(\sigma)})^{m_k(\sigma)} -  \Expect_{\vLambda_l}\prod_{k(\sigma)}\tr(\vLambda_1^{n_k(\sigma)})^{m_k(\sigma)}} \\
    \quad\quad= \frac{1}{2^r}\abs{\left(\prod_{k(\sigma)}\left(\Expect_{\vG} \tr(\vG^{n_k(\sigma)})\pm N\delta_{n_k(\sigma)}\right)^{m_k(\sigma)} -  \Expect_{\vLambda_l}\prod_{k(\sigma)}\tr(\vLambda_1^{n_k(\sigma)})^{m_k(\sigma)}\right)} 
    \end{multline}
    We now look to upper bound
    \begin{align*}
        \abs{\Expect_{\vLambda_l}\left(\prod_{k(\sigma)}\left(\Expect_{\vG} \tr(\vG^{n_k(\sigma)})\right)^{m_k(\sigma)} -  \prod_{k(\sigma)}\tr(\vLambda_1^{n_k(\sigma)})^{m_k(\sigma)}\right)}. 
    \end{align*}
    We aim to do so by integrating the tail
    \begin{align}
        &\abs{\Expect_{\vLambda_l}\left(\prod_{k(\sigma)}\left(\Expect_{\vG} \tr(\vG^{n_k(\sigma)})\right)^{m_k(\sigma)} -  \prod_{k(\sigma)}\tr(\vLambda_1^{n_k(\sigma)})^{m_k(\sigma)}\right)} \nonumber\\
        &\quad\leq \Expect_{\vLambda_l}\abs{\prod_{k(\sigma)}\left(\Expect_{\vG} \tr(\vG^{n_k(\sigma)})\right)^{m_k(\sigma)} -  \prod_{k(\sigma)}\tr(\vLambda_1^{n_k(\sigma)})^{m_k(\sigma)}}\nonumber\\
        &\quad= \int_0^\infty \Pr\left[\abs{\prod_{k(\sigma)}\left(\Expect_{\vG} \tr(\vG^{n_k(\sigma)})\right)^{m_k(\sigma)} -  \prod_{k(\sigma)}\tr(\vLambda_1^{n_k(\sigma)})^{m_k(\sigma)}}\geq s\right]ds. \label{eq:tail_value_bound}
    \end{align}
    The last equality is the tail value formula $\BE\labs{x} = \int_{0}^{\infty} \Pr(\labs{x} \ge s) \, \rd s$.
    \subsubsection{Bounds on Deviations from GUE Trace Moments}
    How do we find the tail bounds of this product? Recall the concentration result \autoref{prop:trace_moment_conc_G} from appendix \ref{app:GUE_properties}: there exist numerical constants $\kappa>0$ and $\kappa'>0$ such that for any $\beta>0$, $t>0$, 
    \begin{align*}
        \Pr[\abs{\tr[\vG^\beta] - \mathop{\BE}_{\text{GUE}}\tr[\vG ^\beta]} \geq t ] \leq \kappa'(\beta+1)\exp(-\min\{\kappa^\beta t^2,\kappa Nt^{2/\beta}\}).
    \end{align*}
Now we do some reduction:
    \begin{align*}
        \abs{\tr[\vG^{n_k}] - \mathop{\BE}_{\text{GUE}}\tr[\vG ^{n_k}]}<t\quad 
        &\implies \abs{\tr[\vG^{n_k}]} <\abs{\mathop{\BE}_{\text{GUE}}\tr[\vG ^{n_k}]} + t\\
        &\implies \abs{\tr[\vG^{n_k}]}^{m_k} <\left(\abs{\mathop{\BE}_{\text{GUE}}\tr[\vG ^{n_k}]} + t\right)^{m_k} \tag{Both sides of inequality are positive.}
    \end{align*}
    If $A\implies B$ then $\Pr[A]\leq\Pr[B]$ or $\Pr[B']\geq \Pr[A']$, or conversely,
    \begin{align*}
    \Pr[\abs{\tr[\vG^{n_k}]}^{m_k} \geq\left(\abs{\mathop{\BE}_{\text{GUE}}\tr[\vG ^{n_k}]} + t\right)^{m_k} ] \leq \Pr[\abs{\tr[\vG^{n_k}] - \mathop{\BE}_{\text{GUE}}\tr[\vG ^{n_k}]} < t ]\\
    \leq \kappa'({n_k}+1)\exp(-\min\{\kappa^{n_k}t^2,\kappa Nt^{2/{n_k}}\}).
    \end{align*}
    Now we want to combine different $k$:
    \begin{align*}
        \prod_k \abs{\tr[\vG^{n_k}]}^{m_k} \geq \prod_k\left(\abs{\mathop{\BE}_{\text{GUE}}\tr[\vG ^{n_k}]} + t\right)^{m_k} \implies \bigcup_k \left\{\abs{\tr[\vG^{n_k}]}^{m_k} \geq\left(\abs{\mathop{\BE}_{\text{GUE}}\tr[\vG ^{n_k}]} + t\right)^{m_k} \right\}.
    \end{align*}
    We can then apply the union bound over $k$ to find
    \begin{align*}
        \Pr[\prod_k \abs{\tr[\vG^{n_k}]}^{m_k} \geq \prod_k\left(\abs{\mathop{\BE}_{\text{GUE}}\tr[\vG ^{n_k}]} + t\right)^{m_k} ]
        &\leq \sum_k \Pr[\abs{\tr[\vG^{n_k}]}^{m_k} \geq\left(\abs{\mathop{\BE}_{\text{GUE}}\tr[\vG ^{n_k}]} + t\right)^{m_k} ]\\
        &\leq \sum_k \kappa'({n_k}+1)\exp(-\min\{\kappa^{n_k}t^2,\kappa Nt^{2/{n_k}}\})\\
        &\leq 2r^2 \kappa'\exp(-\min\{\kappa^{n_k} t^2,\kappa Nt^{2/{n_k}}\}). 
    \end{align*}
    The last line holds because for all $\sigma$ permutations, $n_k\leq k\leq r$. We now expand
    \begin{align*}
        \prod_k \abs{\tr[\vG^{n_k}]}^{m_k} &\geq \prod_k\left(\abs{\mathop{\BE}_{\text{GUE}}\tr[\vG ^{n_k}]} + t\right)^{m_k} = \prod_k\abs{\mathop{\BE}_{\text{GUE}}\tr[\vG ^{n_k}]}^{m_k} + \sum^{M(\sigma)}_{\alpha=0} t^\alpha c_{\alpha}(\sigma,\vG),
    \end{align*}
    where $M(\sigma)=\sum m_k$. The function $c_{\alpha}(\sigma,\vG)$ is a complicated polynomial in the absolute value of trace moments of $\vG$. 
    There is a specific relationship between $m_k$, $n_k$, and the associated permutation $\sigma$, but we will return to this point later. Given that $t>0$ and $c_{\alpha}(\sigma,\vG)>0$ for all $\alpha$ and $\sigma$, $\sum_\alpha t^\alpha c_{\alpha}(\sigma,\vG)>0$ and $\abs{\prod_k \tr[\vG^{n_k}]^{m_k}} - \abs{\prod_k \mathop{\BE}_{\text{GUE}}\tr[\vG ^{n_k}]^{m_k}}>0$. Now recall from \autoref{prop:catalan} from appendix \ref{app:GUE_properties} that the odd trace moments of GUE are zero and even moments are positive such that the product of expected traces is positive $\prod_k \mathop{\BE}_{\text{GUE}}\tr[\vG ^{n_k}]^{m_k}\ge 0$. This implies that we can remove the absolute values
    \begin{align*}
        \abs{\abs{\prod_k \tr[\vG^{n_k}]^{m_k}} - \abs{\prod_k \mathop{\BE}_{\text{GUE}}\tr[\vG ^{n_k}]^{m_k}}} = \abs{\prod_k \tr[\vG^{n_k}]^{m_k} - \prod_k \mathop{\BE}_{\text{GUE}}\tr[\vG ^{n_k}]^{m_k}}.
    \end{align*}
    Hence we get our desired concentration bound:
    \begin{align*}
        \Pr[\abs{\prod_k \tr[\vG^{n_k}]^{m_k} - \prod_k \mathop{\BE}_{\text{GUE}}\tr[\vG ^{n_k}]^{m_k}} \geq  \sum^{M(\sigma)}_\alpha t^\alpha c_{\alpha}(\sigma,\vG)]\leq 2r^2 \kappa'\exp(-\min\{\kappa^{n_k} t^2,\kappa Nt^{2/{n_k}}\}).
    \end{align*}
    Now we return to computing the tail value bound from \eqref{eq:tail_value_bound}:
    \begin{align*}
        &\abs{\Expect_{\vLambda_l}\left(\prod_{k(\sigma)}\left(\Expect_{\vG} \tr(\vG^{n_k(\sigma)})\right)^{m_k(\sigma)} -  \prod_{k(\sigma)}\tr(\vLambda_1^{n_k(\sigma)})^{m_k(\sigma)}\right)} \\
        &\leq \int_0^\infty \Pr\left[\abs{\left(\prod_{k(\sigma)}\left(\Expect_{\vG} \tr(\vG^{n_k(\sigma)})\right)^{m_k(\sigma)} -  \prod_{k(\sigma)}\tr(\vLambda_1^{n_k(\sigma)})^{m_k(\sigma)}\right)}\geq s\right]ds \\
        &= \int_0^\infty \Pr\left[\abs{\left(\prod_{k(\sigma)}\left(\Expect_{\vG} \tr(\vG^{n_k(\sigma)})\right)^{m_k(\sigma)} -  \prod_{k(\sigma)}\tr(\vLambda_1^{n_k(\sigma)})^{m_k(\sigma)}\right)}\geq \sum^{M(\sigma)}_\alpha t^\alpha c_\alpha(\sigma,\vG)\right]\left(\sum^{M(\sigma)}_\alpha \alpha t^{\alpha-1} c_\alpha(\sigma,\vG)\right)dt\\
        &\leq \int_0^\infty 2r^2 \kappa'\exp(-\min\{\kappa^{n_k} t^2,\kappa Nt^{2/{n_k}}\})\left(\sum^{M(\sigma)}_\alpha \alpha t^{\alpha-1} c_\alpha(\sigma,\vG)\right)dt\\
        &\leq 2r^2 \kappa' \sum^{M(\sigma)}_\alpha c_\alpha(\sigma,\vG)\alpha\int_0^\infty \exp(-\kappa\min\{ t^2,Nt^{2/{n_k}}\})\left( t^{\alpha-1} \right)dt\tag{Redefine $\kappa$ as $\min\{\kappa,\kappa^{n_k}\}$}\\
        &\leq 2r^2 \kappa' \sum^{M(\sigma)}_\alpha c_\alpha(\sigma,\vG)\alpha\left(\int_0^{\infty} \exp(-\kappa t^2) t^{\alpha-1} dt +\int_{0}^\infty \exp(-\kappa Nt^{2/{n_k}}) t^{\alpha-1} dt \right),
    \end{align*}
    Now we evaluate these integrals via Gaussian integrals (e.g.,~\cite{winkelbauer2012moments}):
    \begin{align*}
        \int_0^\infty \exp(-\kappa t^2) t^{\alpha-1} dt &\le    \frac{1}{2}\sqrt{\frac{\pi}{\kappa}}   \begin{cases}
        0 & \text{if } \alpha\text{ is even,} \\
        \left(\frac{1}{\sqrt{2\kappa}}\right)^{\alpha-1} (\alpha-2)!! & \text{if }\alpha\text{ is odd.}
      \end{cases}\\
        \int_{0}^\infty \exp(-\kappa Nt^{2/{n_k}}) t^{\alpha-1} dt &= \int_{0}^\infty \exp(-\kappa Nx^{2}) x^{n_k(\alpha-1)}n_kx^{n_k-1} dx \tag{Setting $x^{n_k}=t$}\\
        &\leq \int_0^\infty \exp(-\kappa Nx^{2}) x^{(n_k\alpha-1)}n_k dx \\
        &\leq r\int_0^\infty \exp(-\kappa Nx^{2}) x^{(r\alpha-1)} dx\\
        &=      \frac{r}{2} \sqrt{\frac{\pi}{\kappa N}}\begin{cases}
        0 & \text{if }r\alpha\text{ is even,} \\
        \left(\frac{1}{\sqrt{2\kappa N}}\right)^{r\alpha-1} (r\alpha-2)!! & \text{if }r\alpha\text{ is odd,}
      \end{cases}
    \end{align*}
    where $r$, recall is the parameter of the function $F_r$.
    Then, 
    \begin{align*}
        &\abs{\Expect_{\vLambda_l}\left(\prod_{k(\sigma)}\left(\Expect_{\vG} \tr(\vG^{n_k(\sigma)})\right)^{m_k(\sigma)} -  \prod_{k(\sigma)}\tr(\vLambda_1^{n_k(\sigma)})^{m_k(\sigma)}\right)} \\
        &\quad\quad\leq r^2 \kappa' \sum^{M(\sigma)}_\alpha c_\alpha(\sigma,\vG)\alpha \sqrt{\frac{\pi}{\kappa}}\left(\left(\frac{1}{\sqrt{2\kappa}}\right)^{\alpha-1} (\alpha-2)!! +\frac{r}{\sqrt{N}}\left(\frac{1}{\sqrt{2\kappa N}}\right)^{r\alpha-1} (r\alpha-2)!!\right).
    \end{align*}
    We want to plug this back inside, along with our values for $\delta_{n_k(\sigma)}$ (introducing $2^{-r}$ factor for convenience, as the expression is $r$-homogeneous in $\vG$)
    \begin{align*}
        &\frac{1}{2^r}\abs{\left(\prod_{k(\sigma)}\tr(\vD^{n_k(\sigma)})^{m_k(\sigma)} -  \Expect_{\vLambda_l}\prod_{k(\sigma)}\tr(\vLambda_1^{n_k(\sigma)})^{m_k(\sigma)}\right)} \\
        &=\abs{\prod_{k(\sigma)}\left(2^{-n_k(\sigma)}\Expect_{\vG} \tr(\vG^{n_k(\sigma)})\pm 2^{-n_k(\sigma)}N\delta_{n_k(\sigma)}\right)^{m_k(\sigma)} -  2^{-r}\Expect_{\vLambda_l}\prod_{k(\sigma)}\tr(\vLambda_1^{n_k(\sigma)})^{m_k(\sigma)}}\\
        &=\abs{\prod_{k(\sigma)}\left(2^{-n_k(\sigma)}\Expect_{\vG} \tr(\vG^{n_k(\sigma)})\pm (2q+4)\right)^{m_k(\sigma)} -  2^{-r}\Expect_{\vLambda_l}\prod_{k(\sigma)}\tr(\vLambda_1^{n_k(\sigma)})^{m_k(\sigma)}}\\
        & \leq2^{-r}\abs{\prod_{k(\sigma)}\left(\Expect_{\vG} \tr(\vG^{n_k(\sigma)})\right)^{m_k(\sigma)} -  \Expect_{\vLambda_l}\prod_{k(\sigma)}\tr(\vLambda_1^{n_k(\sigma)})^{m_k(\sigma)}}+\abs{\sum_\alpha^{M(\sigma)}(2q+4)^{\alpha}c_{\alpha}(\sigma,\vG/2)}.
    \end{align*} 
    Notice $c_{\alpha}$ is the exact same as the one we saw before, just with a differently scaled Gaussian as input.
    \begin{align*}
        &\frac{1}{2^r}\abs{\left(\prod_{k(\sigma)}\tr(\vD^{n_k(\sigma)})^{m_k(\sigma)} -  \Expect_{\vLambda_l}\prod_{k(\sigma)}\tr(\vLambda_1^{n_k(\sigma)})^{m_k(\sigma)}\right)} \\
        & \leq  \sum^{M(\sigma)}_\alpha c_\alpha(\sigma,\vG)2^{-r}r^2 \kappa'\alpha \sqrt{\frac{\pi}{\kappa}}\left(\left(\frac{1}{ \sqrt{2\kappa}}\right)^{\alpha-1}(\alpha-2)!! +\frac{r (r\alpha-2)!!}{\sqrt{N}}\left(\frac{1}{\sqrt{2\kappa N}}\right)^{r\alpha-1}\right)+c_{\alpha}(\sigma,\vG/2)(2q+4)^{\alpha}\\
        & = \sum^{M(\sigma)}_\alpha c_\alpha(\sigma,\vG)C(\alpha,N,r)+c_{\alpha}(\sigma,\vG/2)(2q+4)^{\alpha}.
    \end{align*} 
    In the last line, we've just compressed everything for convenience.
    \subsubsection{Estimating the Power of \texorpdfstring{$N$}{N} from Weingarten Calculus}

    We now return to our original calculation:
    \begin{align*}
        &\abs{ \ltup{F_r(\vD) - \Expect_{\vLambda_l}  F_r(\vLambda_l) }}\\
        &\leq\frac{(r!)^2}{N} \sup_{\sigma,\tau}\left\{\abs{\text{Wg}(\tau\sigma^{-1},N) }\cdot \abs{\Expect_{\substack{\vU_1 \dotsto \vU_{l-1} \\ \vG_{l+1} \dotsto \vG_m}}K(\tau,\xi_r)} \cdot \abs{\sum^{M(\sigma)}_\alpha c_\alpha(\sigma,\vG)C(\alpha,N,r)+c_{\alpha}(\sigma,\vG/2)(2q+4)^{\alpha}}\right\}.
    \end{align*}
    Since we want to take the supremum over permutations, we need to perform a Weingarten analysis on the scaling of $N$ for each term inside the supremum depending on $\sigma, \tau$.
    \begin{figure}
    \centering
    \begin{tikzpicture}[scale=0.5]
    
    \draw[thick] (-7.5,-0.5) -- (-7.5,0.5) -- (-6.5,0.5) -- (-6.5,-0.5) -- (-7.5,-0.5);
    \draw[thick] (-5.5,-0.5) -- (-5.5,0.5) -- (-4.5,0.5) -- (-4.5,-0.5) -- (-5.5,-0.5);
    \draw[thick] (-3.5,-0.5) -- (-3.5,0.5) -- (-2.5,0.5) -- (-2.5,-0.5) -- (-3.5,-0.5);
    \draw[thick] (-1.5,-0.5) -- (-1.5,0.5) -- (-0.5,0.5) -- (-0.5,-0.5) -- (-1.5,-0.5);
    \draw[thick] (1.5,-0.5) -- (1.5,0.5) -- (0.5,0.5) -- (0.5,-0.5) -- (1.5,-0.5);
    \draw[thick] (3.5,-0.5) -- (3.5,0.5) -- (2.5,0.5) -- (2.5,-0.5) -- (3.5,-0.5);
    \draw[thick] (5.5,-0.5) -- (5.5,0.5) -- (4.5,0.5) -- (4.5,-0.5) -- (5.5,-0.5);
    \draw[thick] (7.5,-0.5) -- (7.5,0.5) -- (6.5,0.5) -- (6.5,-0.5) -- (7.5,-0.5);
    \node at (-7,0) {$\vU$};
    \node at (-5,0) {$\vU^\dagger$};    
    \node at (-3,0) {$\vU$};
    \node at (-1,0) {$\vU^\dagger$};
    \node at (1,0) {$\vU$};
    \node at (3,0) {$\vU^\dagger$};
    \node at (5,0) {$\vU$};
    \node at (7,0) {$\vU^\dagger$};
    
    \draw[thick] (-3,-0.5) -- (-3,-1);
    \draw[thick] (-3,0.5) -- (-3,1);
    \draw[thick] (3,-0.5) -- (3,-1);
    \draw[thick] (3,0.5) -- (3,1);
    \draw[thick] (-1,-0.5) -- (-1,-1);
    \draw[thick] (-1,0.5) -- (-1,1);
    \draw[thick] (1,-0.5) -- (1,-1);
    \draw[thick] (1,0.5) -- (1,1);
    \draw[thick] (-5,-0.5) -- (-5,-1);
    \draw[thick] (-5,0.5) -- (-5,1);
    \draw[thick] (5,-0.5) -- (5,-1);
    \draw[thick] (5,0.5) -- (5,1);
    \draw[thick] (-7,-0.5) -- (-7,-1);
    \draw[thick] (-7,0.5) -- (-7,1);
    \draw[thick] (7,-0.5) -- (7,-1);
    \draw[thick] (7,0.5) -- (7,1);

    \draw[thick] (-7,-1) arc (180:240:1);
    \draw[thick] (-5,-1) arc (0:-60:1);
    \draw[thick] (-5,1) arc (180:120:1);
    \draw[thick] (-3,1) arc (0:60:1);
    \draw[thick] (-3,-1) arc (180:240:1);
    \draw[thick] (-1,-1) arc (0:-60:1);
    \draw[thick] (-1,1) arc (180:120:1);
    \draw[thick] (1,1) arc (0:60:1);
    \draw[thick] (1,-1) arc (180:240:1);
    \draw[thick] (3,-1) arc (0:-60:1);
    \draw[thick] (3,1) arc (180:120:1);
    \draw[thick] (5,1) arc (0:60:1);
    \draw[thick] (5,-1) arc (180:240:1);
    \draw[thick] (7,-1) arc (0:-60:1);
    \draw[thick] (-7,1) arc (180:100:5); 
    \draw[thick] (7,1) arc (0:80:5);

    \node at (-6,-2) {$\vD$};
    \node at (-4,2.5) {$\e^{\ri\vA_l(s_0-s_1)}$};
    \node at (-2,-2) {$\vD$};
    \node at (0,2.5) {$\e^{\ri\vA_l(s_1-s_2)}$};
    \node at (0,6) {$\e^{\ri\vA_l(s_3-s_4)}\e^{\ri\vA_ls_4}$};
    \node at (2,-2) {$\vD$};
    \node at (4,2.5) {$\e^{\ri\vA_l(s_2-s_3)}$};
    \node at (6,-2) {$\vD$};

    \draw[blue] (-7,-1.5) .. controls (-4,-4) .. (-1,-1.5);
    \draw[blue] (-5,-1.5) .. controls (-4,-2.5) .. (-3,-1.5);
    \draw[blue] (3,-1.5) .. controls (4,-2.5) .. (5,-1.5);
    \draw[blue] (1,-1.5) .. controls (4,-4) .. (7,-1.5);
    
    \draw[red] (-1,1.5) .. controls (2,4) .. (5,1.5);
    \draw[red] (3,1.5) .. controls (2,2.5) .. (1,1.5);
    \draw[red] (-3,1.5) .. controls (2,5) .. (7,1.5);
    \draw[red] (-7,1.5) .. controls (-6,2.5) .. (-5,1.5);

    \end{tikzpicture}
    \caption{This illustrates one of the terms in the $r=4$ Weingarten expansion for $F_r(\vD)$. This term corresponds to to $\tau=(1)(3)(24)$, which are shown via the red lines, and $\sigma=(12)(34)$, which are shown via the blue lines.}\label{fig:p=4contractions_duhamels}
\end{figure}
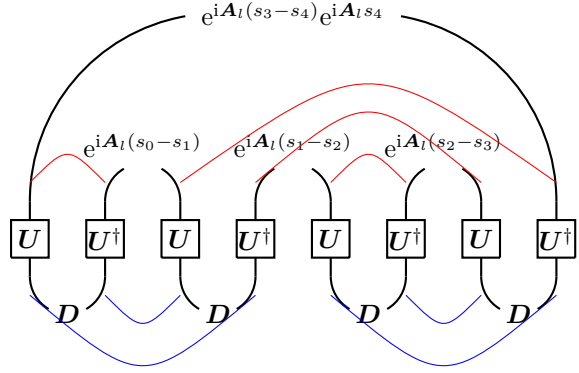
    Notice that $c_\alpha(\sigma,\vG)$ and $c_\alpha(\sigma,\vG/2)$ actually have the same $N$ scaling---they are both powers of unnormalized traces of Gaussians, and the Gaussian normalization will not affect the $N$ scaling. Hence we will account for the $N$ scaling just through $c_\alpha(\sigma,\vG)$ and analyze the $r$ scaling after. 
    
    We begin by reminding ourselves of the Weingarten wiring diagram associated with the permutations $\sigma$ and $\tau$ in this case (see \autoref{fig:p=4contractions_duhamels}).  
    Since $N\geq\Omega(q^{4q})$ by hypothesis, \autoref{lem:weingarten_asymptotics} guarantees that 
    \begin{align*}
        |\text{Wg}(\sigma\tau^{-1},N)|=N^{-2r+\text{cycles}(\pi)}\prod_i \mathrm{Cat}_{|C_i|-1} + \tilde{\mathcal{O}}\left(N^{-2r+\text{cycles}(\pi)-3/2}\right).
    \end{align*}
    The estimate, we emphasize, includes the dependence on $r \leq q$.
    As for the second factor, remember that from the structure of $K(\tau,\xi_r)$ in \autoref{eq:defn_of_ktau},
    \begin{align*}
        \abs{\Expect K(\tau,\xi_r) } = \abs{\Expect \prod_{B \in \lambda(\xi_r)} \tr\ltup{\prod_{\vU \in B} \vU}} \leq \Expect \prod_{B \in \lambda(\xi_r)} \abs{\tr\ltup{\prod_{\vU \in B} \vU}}\leq \prod_{B \in \lambda} N = N^{\lambda(\tau)} = N^{\text{cycles}(\tau\rho_1^{-1})}.
    \end{align*}
    where in the last line we remind ourselves exactly how $\lambda$ depends on $\tau$. Namely, $\lambda$ is the number of subsets $\lambda(\xi_r),$ which in turn is just a partition of $\xi_r$ based on $\tau\rho_1^{-1}$, where $\rho_1$ represents the ``original permutation'' of the base Weingarten diagram. In the convention established by all the diagrams we've drawn thus far, $\rho_1=(1r(r-1)...32)=(r(r-1)...21)$ (see \autoref{fig:p=4contractions_duhamels} for the $r=4$ example). Note that $K(\tau,\xi_r)$ does not have any $r$ dependence besides that through $\text{cycles}(\tau\rho_1^{-1}).$

    We put these two pieces together before addressing the third. In the language of permutation groups, $\text{cycles}(\sigma\tau^{-1}) = r - d_c(\sigma,\tau)$ where $d_c(\sigma,\tau)$ represents the Cayley distance between $\sigma$ and $\tau$ in $S(r)$. In what follows, let $\CO_r$ denote an asymptotic upper bound for fixed $r$. We will reintroduce the dependence on $r$ presently but, for the moment, wish to determine the leading factor of $N$.
    The Cayley distance is a proper metric, so 
    \begin{align*}
        |\text{Wg}(\sigma\tau^{-1},N)|\abs{\Expect K(\tau,\xi_r) } & = \mathcal{O}_r(N^{-2r+\text{cycles}(\sigma\tau^{-1})+\text{cycles}(\tau\rho_1^{-1})}) = \mathcal{O}_r(N^{-(d_c(\sigma,\tau)+d_c(\tau,\rho_1)}) \\
        &\leq \mathcal{O}_r(N^{-d_c(\sigma,\rho_1)})= \mathcal{O}_r(N^{-r+\text{cycles}(\sigma\rho_1^{-1})}).
    \end{align*}
    The inequality holds by the triangle inequality for the Cayley distance. This tells us that the only thing determining the $N$ scaling now is the interaction between $\sigma$ and $\rho_1.$ 
    Now we turn to the fine structure of $c_\alpha(\sigma,\vG)$. 
    Recall where $c_\alpha(\sigma,\vG)$ comes from:
    \begin{align*}
        \prod_k\left(\abs{\mathop{\BE}_{\text{GUE}}\tr[\vG ^{n_k}]} + t\right)^{m_k} = \prod_k\abs{\mathop{\BE}_{\text{GUE}}\tr[\vG ^{n_k}]}^{m_k} + \sum^{M(\sigma)}_{\alpha=1} t^\alpha c_{\alpha}(\sigma,\vG),
    \end{align*}
    The function $c_\alpha(\sigma,\vG)$ is the aggregate of $\mathop{\BE}\tr[\vG ^{n_k}]$ terms that, when appropriately summed and multiplied, give the coefficients for $t^\alpha$. Recall each $\mathop{\BE}\tr[\vG ^{n_k}]$ is at most $N2^{n_k}$---hence, there is dependence on both $N$ and $r$ to keep track of. The $N$ is because these are unnormalized trace terms; the $2^{n_k}$ because the nonzero $n$th trace moment of the GUE is the ($n/2$)th Catalan number (see \autoref{prop:catalan}). (Note that in the $c_\alpha(\sigma,\vG/2)$ case, the unnormalized trace terms are instead $N\cdot 1^{n_k}=N$). Recall also that only the even trace moments of the GUE are nonzero. This means every nonzero term in $c_\alpha(\sigma,\vG)$ must be a product of even trace moments of the GUE. 

    We now need to carefully account for the relationships between the permutations. First, let $\rho_2$ denote the ``original permutation'' of the bottom of the Weingarten diagram, which in this case is $(1)(2)...(r)$, aka the identity permutation. Then $d_c(\rho_2,\rho_1)=r-1$. Then we know by the triangle inequality in Cayley distance
    \begin{align*}
        r=d_c(\rho_2,\rho_1)+1 \leq d_c(\rho_2,\sigma) + d_c(\sigma,\rho_1) +1 = 2r+1 - (\text{cycles}(\rho_1\sigma^{-1}) + \text{cycles}(\rho_2\sigma^{-1}))
    \end{align*}
    which we can also write as
    \begin{align*}
        \text{cycles}(\rho_1\sigma^{-1}) + \text{cycles}(\rho_2\sigma^{-1}) - 1 \leq r.
    \end{align*}
    Let us just focus on the $N$ scaling for now. We now need to break into cases:
    \begin{enumerate}
        \item Let $r$ be odd. 
        Then the $N$ scaling of $c_\alpha(\sigma,\vG)$ is
        \begin{align*}
        c_\alpha(\sigma,\vG)= \mathcal{O}_r\left(N^{\text{ecycles}(\sigma\rho_2^{-1})}\right),
        \end{align*}
        where we use ecycles to denote the number of even cycles. Then we see
        \begin{align*}
        \text{ecycles}(\sigma\rho_2^{-1}) + \text{cycles}(\rho_1\sigma^{-1})\leq \text{cycles}(\sigma\rho_2^{-1})-1 + \text{cycles}(\rho_1\sigma^{-1}) \leq r.
        \end{align*}
        The first inequality is true because $r$ is odd.   
        By definition,  any contribution to $c_\alpha(\sigma,\vG)$ is always missing at least one of the original trace factors in $\prod_k\abs{\mathop{\BE}_{\text{GUE}}\tr[\vG ^{n_k}]}^{m_k}$. The contributing terms must exclude all factors with $n_k$ odd and there is always at least one. The second inequality comes from our previous derivation about how the cycles relate.
        Putting things together, this tells us that for $r$ odd, the scaling in $N$ is
    \begin{multline}
        \abs{\text{Wg}(\tau\sigma^{-1},N) }\abs{\Expect_{\substack{\vU_1 \dotsto \vU_{l-1} \\ \vG_{l+1} \dotsto \vG_m}}K(\tau,\xi_r)} \abs{ \sum^{M(\sigma)}_\alpha c_\alpha(\sigma,\vG)} \\
         =\mathcal{O}_r\left(N^{-r+\text{cycles}(\sigma\rho_1^{-1})}\right)\mathcal{O}_r\left(N^{\text{ecycles}(\sigma\rho_2^{-1})}\right)=\mathcal{O}_r\left(1\right).
    \end{multline}
        \item Now let $r$ be even. 
        If $\sigma$ only contains even cycles, then $\text{cycles}(\sigma \rho_2^{-1}) = \text{ecycles}(\sigma \rho_2^{-1})$  so the leading $N$ scaling of $c_\alpha(\sigma,\vG)$ is
        \begin{align*}
        c_\alpha(\sigma,\vG)= \mathcal{O}_r\left(N^{\text{cycles}(\sigma\rho_2^{-1})-1}\right) = \mathcal{O}_r\left(N^{\text{ecycles}(\sigma\rho_2^{-1})-1}\right).
        \end{align*}
        This is different than for the $r$ odd case: here, the number of even cycles is the same as the number of trace factors in  $\prod_k\abs{\mathop{\BE}_{\text{GUE}}\tr[\vG ^{n_k}]}^{m_k}$, but because by definition $c_\alpha$ will always be missing at least one trace factor, the $N$-degree of $c_\alpha$ can be at most one less than the ecycles. Then,
        \begin{align*}
        \text{ecycles}(\sigma\rho_2^{-1}) -1 + \text{cycles}(\rho_1\sigma^{-1})= \text{cycles}(\sigma\rho_2^{-1})-1 + \text{cycles}(\rho_1\sigma^{-1}) \leq r.
        \end{align*}
        Alternatively, if $\sigma$ induces one or more odd moment, then the leading order in $N$ is 
        \begin{align*}
        c_\alpha(\sigma,\vG)= \mathcal{O}_r\left(N^{\text{ecycles}(\sigma\rho_2^{-1})}\right)
        \end{align*}
        but then 
        \begin{align*}
        \text{ecycles}(\sigma\rho_2^{-1}) + \text{cycles}(\rho_1\sigma^{-1})\leq \text{cycles}(\sigma\rho_2^{-1})-1 + \text{cycles}(\rho_1\sigma^{-1}) \leq r.
        \end{align*}
        Again we find 
        \begin{align*}
        \abs{\text{Wg}(\tau\sigma^{-1},N) }\abs{\Expect_{\substack{\vU_1 \dotsto \vU_{l-1} \\ \vG_{l+1} \dotsto \vG_m}}K(\tau,\xi_r)} \abs{ \sum^{M(\sigma)}_\alpha c_\alpha(\sigma,\vG)}=\mathcal{O}_r\left(1\right).
    \end{align*}
    \end{enumerate}
    This gives us the scaling for $N$. Now we should think about the scaling in $r.$ 
    The function $c_\alpha$ has structure beyond just the even trace moments---there are also combinatorial factors from the binomial expansion that grow as $\alpha$ gets larger. However, these can be neglected in our analysis since $N\geq\Omega(q!^2)$: every time we go up one in $\alpha$, we pick up no more than (indeed, much less than) a $r!$ factor, while also losing a factor of $N$, so the overall bound of $\mathcal{O}(1)$ in $N$ remains accurate even when this contribution to the dependence on $r\leq q$ is taken into account.
    
    More troublesome are the factors of $2^{n_k}$ that come from the even ${n_k}^{\text{th}}$ trace moments of GUE being present in $c_\alpha.$ A rough bound can be provided by the original trace factors $\prod_k\abs{\mathop{\BE}_{\text{GUE}}\tr[\vG ^{n_k}]}^{m_k}$ that $c_\alpha(\sigma,\vG)$ is derived from: each $\mathop{\BE}_{\text{GUE}}\tr[\vG ^{n_k}]$ contributes $2^{n_k}$ and $\sum_k m_kn_k=r$, so $\prod_k\abs{\mathop{\BE}_{\text{GUE}}\tr[\vG ^{n_k}]}^{m_k}\leq 2^r$. This bound also bounds $c_\alpha(\sigma,\vG)$, which has just as many, or fewer, even trace factors as $\prod_k\abs{\mathop{\BE}_{\text{GUE}}\tr[\vG ^{n_k}]}^{m_k}$. We combine this with the previous $r$ scaling in the Weingarten function to find:
    \begin{align*}
        \abs{\text{Wg}(\tau\sigma^{-1},N) }\abs{\Expect_{\substack{\vU_1 \dotsto \vU_{l-1} \\ \vG_{l+1} \dotsto \vG_m}}K(\tau,\xi_r)} \abs{ \sum^{M(\sigma)}_\alpha c_\alpha(\sigma,\vG)}&=2^{r}\prod_i\mathrm{Cat}_{|C_i|-1}\\
        &\leq 2^{r}\prod_i 2^{|C_i|-1}\\
        &\leq 2^{r}2^{r-1} = 2^{2r-1}.
    \end{align*}
    The second line is due to a standard bound on the Catalan numbers. The last line is due to $\sum_i |C_i|\leq r$. Note this $r$ scaling is appropriate for the $c_\alpha(\sigma,\vG)$. For $c_\alpha(\sigma,\vG/2),$ there is only the Weingarten factor of $2^{r-1}.$ After all this Weingarten analysis, we finally arrive at
    \begin{align}
        &\abs{ \ltup{F_r(\vD) - \Expect_{\vLambda_l}  F_r(\vLambda_l) }} \\
        &\leq\frac{(r!)^2}{N} \sup_{\sigma,\tau}\left\{\abs{\text{Wg}(\tau\sigma^{-1},N) }\cdot \abs{\Expect_{\substack{\vU_1 \dotsto \vU_{l-1} \\ \vG_{l+1} \dotsto \vG_m}}K(\tau,\xi_r)} \cdot \abs{\sum^{M(\sigma)}_\alpha c_\alpha(\sigma,\vG)C(\alpha,N,r)+c_{\alpha}(\sigma,\vG/2)(2q+4)^{\alpha}}\right\}\\\nonumber
        &\leq\frac{(r!)^2}{N}\sup_{\sigma,\tau}\abs{\text{Wg}(\tau\sigma^{-1},N) }\abs{\Expect_{\substack{\vU_1 \dotsto \vU_{l-1} \\ \vG_{l+1} \dotsto \vG_m}}K(\tau,\xi_r)} \\
        &\quad\;\times\abs{\sum^{M(\sigma)}_\alpha c_\alpha(\sigma,\vG)2^{-r}r^2 \kappa'\alpha \sqrt{\frac{\pi}{\kappa}}\left(\frac{(\alpha-2)!!}{ \sqrt{2\kappa}^{\alpha-1}} +\frac{r (r\alpha-2)!!}{\sqrt{N}}\left(\frac{1}{\sqrt{2\kappa N}}\right)^{r\alpha-1}\right)+c_{\alpha}(\sigma,\vG/2)(2q+4)^{\alpha}}\\
        &\leq \mathcal{O}\left(\frac{1}{N}\right)(r!)^2  \sup_\alpha \abs{2^{2r-1}2^{-r}r^3 \kappa'\alpha \sqrt{\frac{\pi}{\kappa}}\left(\frac{(\alpha-2)!!}{ \sqrt{2\kappa}^{\alpha-1}} +\frac{r (r\alpha-2)!!}{\sqrt{N}}\left(\frac{1}{\sqrt{2\kappa N}}\right)^{r\alpha-1}\right)+2^{r-1}r(2q+4)^{\alpha}}\\
        &\leq \mathcal{O}\left(\frac{1}{N}\right) (r!)^2 \sup_\alpha \abs{2^{2r-1}2^{-r}r^3 \kappa'\alpha \sqrt{\frac{\pi}{\kappa}}\left((\alpha-2)!!\kappa^{\alpha-1} +\frac{r (r\alpha-2)!!}{\sqrt{N}}\left(\frac{\kappa}{\sqrt{N}}\right)^{r\alpha-1}\right)+2^{r-1}r(2q+4)^{\alpha}}\tag{Redefine $\kappa$ as $\max\{\sqrt{2\kappa},1/\sqrt{2\kappa}\}$}\\
        &\leq \mathcal{O}\left(\frac{1}{N}\right)(r!)^22^{r-1}\left(r^4 \kappa' \sqrt{\frac{\pi}{\kappa}}\left((r-2)!!\kappa^{r-1} +\frac{r (r^2-2)!!}{\sqrt{N}}\left(\frac{\kappa}{\sqrt{N}}\right)^{r^2-1}\right)+r(2q+4)^{r}\right)
    \end{align}

\subsubsection{Terms of \texorpdfstring{$|J|$}{|J|}}\label{subsubsec:J}
It remains to bound $\abs{J_{UDU}}$ and $\abs{J_{GUE}}$. Recall their definitions: 
    \begin{align*}
        J_{UDU} &:= i^{q+1} \idotsint\limits_{\frac{\theta p}{\sqrt{m}} >s_1 > \dots > s_{q+1} > 0} \Expect \btr\lbr{\prod_{j=1}^{q+1} \diff s_j  \left(\prod_{j=1}^{q+1} e^{i \vA_l(s_{j-1}-s_{j})}\vU_l \vD \vU_l^\dag \right) e^{i\left(\vA_l+\vU_l \vD \vU_l^\dag\right)s_{q+1}} } \\ 
        J_{GUE} &:= i^{q+1} \idotsint\limits_{\frac{\theta p}{\sqrt{m}} >s_1 > \dots > s_{q+1} > 0} \Expect \btr\lbr{\prod_{j=1}^{q+1} \diff s_j  \left(\prod_{j=1}^{q+1} e^{i \vA_l(s_{j-1}-s_{j})}\vG_l \right) e^{i\left(\vA_l+\vG_l \right)s_{q+1}}}
    \end{align*}
    Let us begin with $J_{UDU}.$
    \begin{align*}
        \abs{J_{UDU}} &= \abs{\;\idotsint\limits_{\frac{\theta p}{\sqrt{m}} >s_1 > \dots > s_{q+1} > 0} \Expect \btr\lbr{\prod_{j=1}^{q+1} \diff s_j  \left(\prod_{j=1}^{q+1} e^{i \vA_l(s_{j-1}-s_{j})}\vU_l \vD \vU_l^\dag \right) e^{i\left(\vA_l+\vU_l \vD \vU_l^\dag\right)s_{q+1}} }} \\
        &\leq \idotsint\limits_{\frac{\theta p}{\sqrt{m}} >s_1 > \dots > s_{q+1} > 0} \Expect \norm{\prod_{j=1}^{q+1} \diff s_j  \left(\prod_{j=1}^{q+1} e^{i \vA_l(s_{j-1}-s_{j})}\vU_l \vD \vU_l^\dag \right) e^{i\left(\vA_l+\vU_l \vD \vU_l^\dag\right)s_{q+1}} }_{op}\tag{Holder's inequality and triangle inequality}\\
        &\leq \idotsint\limits_{\frac{\theta p}{\sqrt{m}} >s_1 > \dots > s_{q+1} > 0} \Expect \left[\prod_{j=1}^{q+1} \diff s_j  \left(\prod_{j=1}^{q+1} \norm{e^{i \vA_l(s_{j-1}-s_{j})}}_{op}\norm{\vU_l \vD \vU_l^\dag}_{op}\right) \norm{e^{i\left(\vA_l+\vU_l \vD \vU_l^\dag\right)s_{q+1}} }_{op}\right]\tag{Submultiplicativity}\\
        &\leq \idotsint\limits_{\frac{\theta p}{\sqrt{m}} >s_1 > \dots > s_{q+1} > 0} \Expect \left[\prod_{j=1}^{q+1} \diff s_j  \left(\prod_{j=1}^{q+1} \norm{\vD}_{op}\right) \right]\tag{Operator norms of unitaries are 1}.
    \end{align*}
    The expectation value is taken over all of the Haar $\vU_j$ and GUE $\vG_j$, but since they all go into the unitary factors, there is no need for it by the end of this calculation. We evaluate 
    \begin{align*}
        \idotsint\limits_{\frac{\theta p}{\sqrt{m}} >s_1 > \dots > s_{q+1} > 0}  \prod_{j=1}^{q+1} \diff s_j = \frac{1}{(q+1)!} \ltup{\frac{\theta p}{\sqrt{m}}}^{q+1}.
    \end{align*}
    Then 
    \begin{align*}
        \abs{J_{UDU}} \leq \ltup{\frac{\theta p \norm{\vD}_{op}}{\sqrt{m}}}^{q+1} \frac{1}{(q+1)!} 
    \end{align*}
    As for $J_{GUE}$, a similar argument yields 
    \begin{align*}
        \abs{J_{GUE}} \leq \idotsint\limits_{\frac{\theta p}{\sqrt{m}} >s_1 > \dots > s_{q+1} > 0} \Expect \left[\prod_{j=1}^{q+1} \diff s_j  \left(\prod_{j=1}^{q+1} \norm{\vG_l}_{op}\right) \right].
    \end{align*}
    For this case, we actually have to account for the expectation value over $\vG_l$:
    \begin{align*}
        \abs{J_{GUE}} \leq \ltup{\frac{\theta p}{\sqrt{m}}}^{q+1} \frac{1}{(q+1)!} \Expect_{\vG} \norm{\vG}_{op}^{q+1}.
    \end{align*}
    We use \autoref{prop:exp_value_op_norm_products} from appendix \ref{app:GUE_properties}: there exist constants $C$ and $\alpha$ such that for any $k<\alpha N$,
    \begin{align*}
    \Expect_{\vG} \norm{\vG}_{op}^k \leq C^k.
    \end{align*}
    Since $N\geq\Omega(q!^2)$, $q+1=k$ qualifies. Then we have
    \begin{align*}
        \abs{J_{GUE}} \leq \ltup{\frac{C\theta p}{\sqrt{m}}}^{q+1} \frac{1}{(q+1)!}.
    \end{align*}
    
\subsubsection{Summation of Terms}
    We can finally put together the estimates for individual terms in
    \begin{align*}
        &\abs{\Expect \btr \ltup{ \e^{i\frac{\theta p}{\sqrt{m}} (\vA_l + \vU_l \vD \vU_l^\dag)} - \e^{i\frac{\theta p}{\sqrt{m}} (\vA_l + \vG_l)}}} \\
        &\leq \abs{J_{UDU}}  +\abs{J_{GUE}} +  \sum_{r=1}^q \quad \idotsint\limits_{\frac{2\theta p}{\sqrt{m}} >s_1 > \dots > s_r > 0}  \abs{ \ltup{F_r(\vD) - \Expect_{\vLambda_l}  F_r(\vLambda_l) }}\prod_{j=1}^{r} \diff s_j.
    \end{align*}
    We find
    \begin{align*}
        &\sum_{r=1}^q \quad \idotsint\limits_{\frac{2\theta p}{\sqrt{m}} >s_1 > \dots > s_r > 0}  \abs{ \ltup{F_r(\vD) - \Expect_{\vLambda_l}  F_r(\vLambda_l) }}\prod_{j=1}^{r} \diff s_j\\
        &\leq\sum_{r=1}^q  \quad\idotsint\limits_{\frac{2\theta p}{\sqrt{m}} >s_1 > \dots > s_r > 0}  \mathcal{O}\left(\frac{1}{N}\right)(r!)^2 2^{r-1}\\
        &\quad\quad\quad\left(r^4 \kappa' \sqrt{\frac{\pi}{\kappa}}\left((r-2)!!\kappa^{r-1} +\frac{r (r^2-2)!!}{\sqrt{N}}\left(\frac{\kappa}{\sqrt{N}}\right)^{r^2-1}\right)+r(2q+4)^{r}\right)\prod_{j=1}^{r} \diff s_j\\
        &\leq  \sum_{r=1}^q \frac{1}{r!}\ltup{\frac{2\theta p}{\sqrt{m}}}^{r}\frac{(r!)^22^{r-1}}{N} \left(r^4 \kappa' \sqrt{\frac{\pi}{\kappa}}\left((r-2)!!\kappa^{r-1} +\frac{r (r^2-2)!!}{\sqrt{N}}\left(\frac{\kappa}{\sqrt{N}}\right)^{r^2-1}\right)+r(2q+4)^{r}\right)
    \end{align*}
    Assuming a scaling $\sqrt{N}\geq\Omega(q^{2q})$ to absorb all $q$ dependencies and extra constants and using that $\sum_{r=1}^q x^r \le q (\labs{x}+\labs{x}^q)$, we can simplify the last term drastically:
    \begin{align*}
        &\sum_{r=1}^q  \idotsint\limits_{\frac{\theta p}{\sqrt{m}} >s_1 > \dots > s_r > 0}  \abs{ \ltup{F_r(\vD) - \Expect_{\vLambda_l}  F_r(\vLambda_l) }}\prod_{j=1}^{r} \diff s_j\leq  \frac{1}{\sqrt{N}}\ltup{\frac{2\theta p}{\sqrt{m}}}+ \frac{1}{\sqrt{N}}\ltup{\frac{2\theta p}{\sqrt{m}}}^{q}.
    \end{align*}
    Then, adding all the terms together, we find 
    \begin{multline}
        \abs{\Expect \btr \ltup{ \e^{i\frac{\theta p}{\sqrt{m}} (\vA_l + \vU_l \vD \vU_l^\dag)} - \e^{i\frac{\theta p}{\sqrt{m}} (\vA_l + \vG_l)}}} \\
        \leq \ltup{\frac{\theta p}{\sqrt{m}}}^{q+1} \frac{1}{(q+1)!}(\norm{\vD}_{op}^{q+1}+C^{q+1}) +   \frac{1}{\sqrt{N}}\ltup{\frac{2\theta p}{\sqrt{m}}}+ \frac{1}{\sqrt{N}}\ltup{\frac{2\theta p}{\sqrt{m}}}^{q} 
    \end{multline}
    And this leads to our final estimate,
     \begin{align*}
       &\abs{\Expect_{\lset{\vU_j}_j} \btr(\e^{i\frac{\theta p}{\sqrt{m}} \sum_j \vU_j \vD \vU_j^\dagger}) - \Expect_{\vG} \btr(\e^{i\frac{\theta p}{\sqrt{m}} \sum_j \vG})} 
        \leq \sum_{l=1}^{m} \abs{\Expect_{\lset{\vU_j}_j \cup \lset{\vG_j}_j}\btr \ltup{ \e^{i\frac{\theta p}{\sqrt{m}} (\vA_l + \vU_l \vD \vU_l^\dag)} - \e^{i\frac{\theta p}{\sqrt{m}} (\vA_l + \vG_l)}}} \\
        &= \frac{(\theta p)^{q+1}}{(\sqrt{m})^{q-1}} \frac{1}{(q+1)!}(\norm{\vD}_{op}^{q+1}+C^{q+1})  +  \frac{m}{\sqrt{N}}\ltup{\frac{2\theta p}{\sqrt{m}}} \left( 1+ \ltup{\frac{2\theta p}{\sqrt{m}}}^{q-1} \right).
    \end{align*}
    We have conveniently kept a $1/\sqrt{N}$ to counteract the $\sqrt{m}$-scaling from the second term.
\end{proof}
Note that this proof to control the expectation value of $\btr(e^{i\frac{\theta}{\sqrt{m}} \sum_{j=1}^m \vU_j \vD \vU_j^\dag})$ took quite a bit of work. In order to apply \autoref{lem:UDUVDV_expected_moments} in the next section, we'll also need the concentration statement. The statement and proof of the concentration is in \autoref{cor:moment_concentration_UDU} in appendix \ref{append:Proofs_concentration}.

\subsection{Proof of \autoref{lem:clt_spectrum_small_moments}}\label{subsec:proof_prod_gaussianlike_unitary}
In the last section and a solid part of this section, we have been gathering facts about Gaussian exponentials and Gaussian-like exponentials: the trace moments of Gaussian exponentials (\autoref{lem:exp_semicircle_moments}), the finite $N$ corrections to those trace moments (\autoref{lem:exp_GUE_moments}), the closeness of trace moments of Gaussian-like exponentials to the trace moments of Gaussian exponentials (\autoref{thm:lindeberg_wg}), and how to analyze trace moments of $\vU\vD_1\vU^\dagger\vV\vD_2\vV^\dagger$ (\autoref{lem:UDUVDV_expected_moments}).

We are finally ready to prove \autoref{lem:clt_spectrum_small_moments}, which shows that our product of two Gaussian-like exponentials  have small moments:

\cltspectrumsmallmoments*

\begin{proof}[Proof of \autoref{lem:clt_spectrum_small_moments}]
From the Taylor expansion, we know that $\e^{i\frac{\theta}{\sqrt{m}}\sum_{j=1}^m \vU_j\vD\vU_j^\dagger}$ is unitarily invariant under conjugation and can be written as
\begin{align*}
    \e^{i\frac{\theta}{\sqrt{m}}\sum_{j=1}^m \vU_j\vD\vU_j^\dagger} \distas \vU \vD_1 \vU
\end{align*}
where $\vU$ is a Haar unitary matrix, and $\vD_1$ is some random diagonal matrix determined by the structure of $\vD$ and $\vU_j$. This means that
\begin{align*}
    \e^{i\frac{\theta}{\sqrt{m}}\sum_{j=1}^m \vU_j\vD\vU_j^\dagger} \e^{i\frac{\theta}{\sqrt{m}}\sum_{j=1}^m \vU'_j\vD\vU_j^{'\dagger}} \distas \vU \vD_1 \vU \vV\vD_2\vV,
\end{align*}
with $\vU$ and $\vV$ independent Haar unitary matrices, so we would like to apply \autoref{lem:UDUVDV_expected_moments}. Now fix $T$. In order to use \autoref{lem:UDUVDV_expected_moments} for a given $p\leq T$, we need two conditions to be satisfied:
\begin{enumerate}
    \item $N\in (\frac{27648}{575}p^{7/2},\infty) \quad\forall\;\; 1\leq p \leq T$, and
    \item $\sum_{k=1}^T |\btr{\vD_i^k}| \le \frac{1}{4}$
\end{enumerate}
Remember, these two conditions are inherited from the solution to the moment problem in \autoref{sec:moment_problem}. The first is simple---we've specified $N\geq\Omega(q^{4q})$ and $q=\mathcal{O}\left(\log(T)\right)$, so $N$ is larger than polynomial in $T$. The second will require more wrangling. 

Let $\abs{\tr{\vD^k} - \int x^k \rho_{sc}(x)\, \diff x} = \delta_k < \delta = \mathcal{O}\left(\frac{2^q}{N}\right)$
for all $1\leq k\leq q$. Then from \autoref{thm:lindeberg_wg} (and a triangle inequality), we know
    \begin{align*}
        &\abs{\Expect_{\lset{\vU_j}_j} \btr\left(\e^{i\frac{\theta p}{\sqrt{m}} \sum_j \vU_j \vD \vU_j^\dagger}\right)} \\
        &\leq \abs{\Expect_{\vG} \btr\left(\e^{i\theta p \vG}\right)} + \frac{(\theta p)^{q+1}}{(\sqrt{m})^{q-1}} \frac{1}{(q+1)!}(\norm{\vD}_{op}^{q+1}+C^{q+1}) + \frac{m}{\sqrt{N}}\ltup{\frac{2\theta p}{\sqrt{m}}} \left( 1+ \ltup{\frac{2\theta p}{\sqrt{m}}}^{q-1} \right)
    \end{align*}
for some positive constant $C$, given $N\geq \Omega\left(q!^2\right)$. Further, we can bound the first term by \autoref{lem:exp_GUE_moments}
\begin{align*}
\abs{\Expect_{\vG} \btr\left(\e^{i\theta p \vG}\right)} \le \left|\frac{J_1(2p\theta)}{p\theta}\right| + \frac{K_0p\theta}{N}
    \end{align*}
    for any $p\theta >1$ and some constant $K_0>0$.
Notice that our two conditions from \autoref{lem:UDUVDV_expected_moments} are imposed on the sum of the traces of a \emph{particular instance} of our ensemble. We cannot actually guarantee (and in fact it is untrue) that every member of our ensemble obeys the condition. However, we can apply a concentration inequality to this situation. The concentration of  $\btr{\e^{i\frac{\theta}{\sqrt{m}}\sum_{j=1}^m \vU_j\vD\vU_j^\dagger}}$ is worked out in \autoref{cor:moment_concentration_UDU}: 
 \begin{align*}
    \Pr_{\substack{\vU_1 \dotsto  \vU_m \\ \leftarrow \mu} }\left[\abs{\btr(e^{i\frac{\theta}{\sqrt{m}} \sum_{j=1}^m \vU_j \vD \vU_j^\dag})^k - \underset{\substack{\vV_1 \dotsto  \vV_m \\ \leftarrow \mu}}{\Expect }\btr(e^{i\frac{\theta}{\sqrt{m}}  \sum_{j=1}^m \vV_j \vD \vV_j^\dag})^k} > s\right] \leq 2 \exp{-\frac{N^2 s^2 }{ 48 k^2 \theta^2  \lnorm{\vD}_{op}^2}}.
\end{align*}
We can then use the union bound to say
\begin{align}
    \Pr_{\substack{\vU_1 \dotsto  \vU_m \\ \leftarrow \mu} }&\left[
    \forall \, 1 \leq k \leq p \, : \, 
    \abs{\btr(e^{i\frac{k\theta}{\sqrt{m}} \sum_{j=1}^m \vU_j \vD \vU_j^\dag})^k - \underset{\substack{\vV_1 \dotsto  \vV_m \\ \leftarrow \mu}}{\Expect }\btr(e^{i\frac{k\theta}{\sqrt{m}}  \sum_{j=1}^m \vV_j \vD \vV_j^\dag})^k} \leq s\right]  \label{eq:good_moments} \\
    &\geq 1- \sum_{k=1}^p 2 \exp{-\frac{N^2 s^2 }{ 48 k^2 \theta^2  \lnorm{\vD}_{op}^2}} \\
    &\geq 1 - 2p \exp{-\frac{N^2 s^2 }{ 48 p^2 \theta^2  \lnorm{\vD}_{op}^2}}
\end{align}
This is the high probability case, and it is in this case that we will apply \autoref{lem:UDUVDV_expected_moments}. Conditioned on the event of \eqref{eq:good_moments}, we need to make sure that 
\begin{align*}
    &\frac{1}{4}\geq  \sum_{p=1}^T \bigg(\left|\frac{J_1(2p\theta)}{p\theta}\right| + \frac{K_0p\theta}{N} + \frac{(\theta p)^{q+1}}{(\sqrt{m})^{q-1}} \frac{1}{(q+1)!}(\norm{\vD}_{op}^{q+1}+C^{q+1})+ \frac{m}{\sqrt{N}}\ltup{\frac{2\theta p}{\sqrt{m}}} \left( 1+ \ltup{\frac{2\theta p}{\sqrt{m}}}^{q-1} \right)  + s\bigg)\\
        &\geq\sum_{p=1}^T \abs{ \btr\left(\e^{i\frac{\theta p}{\sqrt{m}} \sum_j \vU_j \vD \vU_j^\dagger}\right)} .
\end{align*}
Let us analyze these five terms separately and analyze the conditions necessary to shrink each term to less than $\frac{1}{4\times 5}$. Let's start with $J_1$. In appendix \ref{append:j1_props}, we analyze the behavior of $J_1$, and in particular, \autoref{lem:bound_j1} tells us for $x> 1$, $    |J_1(x)|<\sqrt{\frac{1}{x}}.$ We apply this to find (using the inequality $\sum_{p=1}^\infty \frac{1}{p^{3/2}} < 3$)
\begin{align*}
    \sum_{p=1}^T\left|\frac{J_1(2p\theta)}{p\theta}\right| < \sum_{p=1}^T\left|\frac{1}{p\theta\sqrt{2p\theta}}\right| < \frac{1}{\sqrt{2\theta^3}} \sum_{p=1}^\infty \frac{1}{p^{3/2}}  < \frac{3}{\sqrt{2\theta^3}}.
\end{align*}
Hence $\theta>2^{1/3} 30^{2/3}$ will give us that this term is less than $\frac{1}{20}$. 
For the next term, we get
\begin{align*}
    \sum_{p=1}^T \frac{K_0p\theta}{N} \leq \frac{K_0T^2\theta}{N} < \frac{1}{20},
\end{align*}
for some large enough absolute constant $\theta=\mathcal{O}(1)$ since $N\geq \Omega(q^{4q})$ and $q=\mathcal{O}(\log T)$. Consider the next term:
\begin{align*}
    \sum_{p=1}^T\ltup{\frac{(\theta p)^{q+1}}{(\sqrt{m})^{q-1}}} \frac{1}{(q+1)!}(\norm{\vD}_{op}^{q+1}+C^{q+1})  < \ltup{\frac{\theta^{q+1} T^{q+2}}{(\sqrt{m})^{q-1}}} \frac{1}{(q+1)!}(\norm{\vD}_{op}^{q+1}+C^{q+1}).
\end{align*}
For any $\theta=\mathcal{O}(1)$, this term can be less than $1/20$ for the right choice of $m=\mathcal{O}(T^2)$, $q=\mathcal{O}(\log T).$ For the fourth term, taking $N\geq \Omega(q^{4q})$, $m=\mathcal{O}(T^2)$, $q=\mathcal{O}(\log T)$ also ensures the LHS is bounded by $1/20$.  
\begin{align*}
    \sum_{p=1}^T \frac{m}{\sqrt{N}}\ltup{\frac{2\theta p}{\sqrt{m}}} \left( 1+ \ltup{\frac{2\theta p}{\sqrt{m}}}^{q-1} \right) < \frac{2\theta T^2\sqrt{m}}{\sqrt{N}} + \frac{2^q\theta^{q} T^{q+1}}{(\sqrt{m})^{q-2}}\frac{1}{\sqrt{N}}.
\end{align*}

For the last term, we can just take $s<1/20$. This means the moment problem is satisfied for our high probability case with $\theta=\mathcal{O}(1)>2^{1/3} 30^{2/3}$ and $s<\frac{1}{20}$, and appropriately chosen $m=\mathcal{O}(T^2)$, $q=\mathcal{O}(\log T).$ We can now apply \autoref{lem:UDUVDV_expected_moments} to this case:
\begin{align*}
    &\abs{\Expect_{high} \btr\left(\left(\e^{i\frac{\theta}{\sqrt{m}}\sum_{j=1}^m \vU_j\vD\vU_j^\dagger} \e^{i\frac{\theta}{\sqrt{m}}\sum_{j=1}^m \vU'_j\vD\vU_j^{\prime \dagger}}\right)^p\right)} \\
    &\leq \mathcal{O}\left(\frac{p^{15/4}}{\sqrt{N}}\right) + p \left[ \Expect_{high} \left| \btr\left(\e^{i\frac{\theta}{\sqrt{m}}\sum_{j=1}^m \vU_j\vD\vU_j^\dagger}\right) \right| \right]^2 + 2 p \Expect_{high} \left| \btr\left(\e^{i\frac{\theta}{\sqrt{m}}\sum_{j=1}^m \vU_j\vD\vU_j^\dagger} \right)\right| 
\end{align*}
Since this is the high probability situation, we have that
 \begin{align*}
    \abs{\btr(e^{i\frac{\theta}{\sqrt{m}} \sum_{j=1}^m \vU_j \vD \vU_j^\dag}) - \underset{\substack{\vV_1 \dotsto  \vV_m \\ \leftarrow \mu}}{\Expect }\btr(e^{i\frac{\theta}{\sqrt{m}}  \sum_{j=1}^m \vV_j \vD \vV_j^\dag})} < s.
\end{align*}
This means
\begin{align*}
    &\abs{\Expect_{high} \btr\left(\left(\e^{i\frac{\theta}{\sqrt{m}}\sum_{j=1}^m \vU_j\vD\vU_j^\dagger} \e^{i\frac{\theta}{\sqrt{m}}\sum_{j=1}^m \vU'_j\vD\vU_j^{'\dagger}}\right)^p\right)} \\
    &\leq \mathcal{O}\left(\frac{T^{15/4}}{\sqrt{N}}\right) +  p \left(\left|\Expect  \btr\left(\e^{i\frac{\theta}{\sqrt{m}}\sum_{j=1}^m \vU_j\vD\vU_j^\dagger}\right)\right|+s\right)^2 + 2p\left(\left|\Expect\btr\left(\e^{i\frac{\theta}{\sqrt{m}}\sum_{j=1}^m \vU_j\vD\vU_j^\dagger}\right) \right|+ s \right)\\
    &\leq \mathcal{O}\left(\frac{T^{15/4}}{\sqrt{N}}\right) + p \left[ \left|\Expect\btr\left(\e^{i\frac{\theta}{\sqrt{m}}\sum_{j=1}^m \vU_j\vD\vU_j^\dagger}\right)\right|^2+(2s+2) \left|\Expect\btr\left(\e^{i\frac{\theta}{\sqrt{m}}\sum_{j=1}^m \vU_j\vD\vU_j^\dagger}\right)\right|+ s^2 +2s \right].
\end{align*}
We now combine the high and low probability cases. 
\begin{align*}
    &\abs{\Expect \btr\left(\left(\e^{i\frac{\theta}{\sqrt{m}}\sum_{j=1}^m \vU_j\vD\vU_j^\dagger} \e^{i\frac{\theta}{\sqrt{m}}\sum_{j=1}^m \vU'_j\vD\vU_j^{'\dagger}}\right)^p\right)}\\
    &\leq \left(1-2 p\exp{-\frac{N^2 s^2 }{ 48 p^2 \theta^2  \lnorm{\vD}_{op}^2}}\right) \abs{\Expect_{high} \btr\left(\e^{ip\frac{\theta}{\sqrt{m}}\sum_{j=1}^m \vU_j\vD\vU_j^\dagger} \e^{ip\frac{\theta}{\sqrt{m}}\sum_{j=1}^m \vU'_j\vD\vU_j^{'\dagger}}\right)} \\
    &\quad\quad+ 2 p\exp{-\frac{N^2 s^2 }{ 48 p^2 \theta^2  \lnorm{\vD}_{op}^2}} \abs{\Expect_{low} \btr\left(\e^{ip\frac{\theta}{\sqrt{m}}\sum_{j=1}^m \vU_j\vD\vU_j^\dagger} \e^{ip\frac{\theta}{\sqrt{m}}\sum_{j=1}^m \vU'_j\vD\vU_j^{'\dagger}}\right)}\\
    &\leq  2 T \exp{-\frac{N^2 s^2 }{ 48 T^2 \theta^2  \lnorm{\vD}_{op}^2}}+ \left(1-2 p\exp{-\frac{N^2 s^2 }{ 48 T^2 \theta^2  \lnorm{\vD}_{op}^2}}\right) \\
    &\times \left\{ \mathcal{O}\left(\frac{T^{15/4}}{\sqrt{N}}\right) +   p \left[ \left|\Expect\btr\left(\e^{i\frac{\theta}{\sqrt{m}}\sum_{j=1}^m \vU_j\vD\vU_j^\dagger}\right)\right|^2+(2s+2) \left|\Expect\btr\left(\e^{i\frac{\theta}{\sqrt{m}}\sum_{j=1}^m \vU_j\vD\vU_j^\dagger}\right)\right|+ s^2 +2s \right] \right\}. \tag{Normalized trace of unitary at most 1.}
\end{align*}
Take $s^2=\frac{T^2\theta^2}{N}$. Given that $m=\mathcal{O}(T^2)$, this choice of $s$ will satisfy $s<\frac{1}{20}$
provided $T^2 = o(N)$. This is ensured by combining our choice $q =\CO(\log T)$ together with  $N = \Omega(q^{4q})$ from the statement of the theorem. Moreover,
\begin{align*}
    2 T\exp{-\frac{N^2 s^2 }{ 48 T^2 \theta^2  \lnorm{\vD}_{op}^2}} =  e^{-\tilde{\Omega}(N)}.
\end{align*}
Thus, there is a function $\Delta(m,T,N,\theta)$ such that
\begin{equation}
    \Delta(m,T,N,\theta) = 
        \tilde\CO\left( \max\left( \frac{T^{15/4}}{\sqrt{N}},
         e^{-\Omega(N)},
         \frac{p T}{N^{1/2}} \right)
         \right) 
    = \tilde{O}\left( \frac{T^{4}}{\sqrt{N}} \right).
\end{equation}
for which the overall scaling becomes, using \autoref{thm:lindeberg_wg},
\begin{align*}
    &\abs{\Expect \btr\left(\left(\e^{i\frac{\theta}{\sqrt{m}}\sum_{j=1}^m \vU_j\vD\vU_j^\dagger} \e^{i\frac{\theta}{\sqrt{m}}\sum_{j=1}^m \vU'_j\vD\vU_j^{'\dagger}}\right)^p\right)}\leq 4p\left|\Expect\btr\left(\e^{i\frac{\theta}{\sqrt{m}}\sum_{j=1}^m \vU_j\vD\vU_j^\dagger}\right)\right| + \Delta(m,T,N,\theta) \\
    &\leq 4p\left(\left|\frac{J_1(2\theta)}{\theta}\right| + \left[ \frac{(\theta )^{q+1}}{(\sqrt{m})^{q-1}} \frac{1}{(q+1)!}(\norm{\vD}_{op}^{q+1}+C^{q+1})  + \frac{m}{\sqrt{N}}\ltup{\frac{2\theta}{\sqrt{m}}} \left( 1+ \ltup{\frac{2\theta}{\sqrt{m}}}^{q-1} \right) \right] \right) + \Delta(m,T,N,\theta)\\
    &\leq 4p\left( \frac{((2+C)\theta )^{q+1}}{(\sqrt{m})^{q-1} (q+1)!} \right)
    +\tilde{O}\left( \frac{T^{4}}{\sqrt{N}} \right)
\end{align*}
The first line uses that the normalized trace is bounded by one and $1+2s+2 \le 4$. The last line picks an appropriate $\theta=\mathcal{O}(1)$ so that the Bessel term vanishes (\autoref{cor:theta_exists}). This gives us the bound in the lemma statement.

\end{proof}

This concludes our work in understanding the concentration and expected normalized trace moments of an ensemble of products of Gaussian-like exponentials. For our purposes, these trace moments, along with the concentration properties of the ensemble, show that this construction mimics the Haar ensemble well enough at the level of the spectral distribution; we will see this in action when we thread our lemmas together in \autoref{sec:main_thm_proof}. It is now time to consider whether convergence extends beyond the spectrum to the matrix ensemble itself.


\section{A Matrix Lindeberg Principle: Convergence of Basis
(Proof of Lemma \ref{lem:clt_basis})}\label{sec:lindeberg_basis}

In this section, we address approximation error in the bases, again using a matrix Lindeberg principle. Let $\vU^{\prime}_{1} \dotsto \vU^{\prime}_{m}\in U(N)$ be i.i.d. Haar random unitaries  and $\tilde{\vU}_{1} \dotsto \tilde{\vU}_{m}\in U(N)$ be i.i.d. approximate unitary $q$-designs. Also, let the spectrum $\vD$ be the same as in section \ref{sec:lindeberg_spectrum}, approximately matching the low trace moments of GUE. Informally, we make the comparison     
$$   \sum_j^m \Tilde{\vU}_{j}\vD\Tilde{\vU}_{j}^\dag \approx   \sum_j^m \vU_{j}\vD\vU_{j}^\dag .$$
Formally, we use the diamond norm to compare the unitary ensemble 
\begin{align*}
    \vW_2 := \e^{\ri \frac{\theta}{\sqrt{m}} \sum_j^m \Tilde{\vU}_{j} \vD \Tilde{\vU}_{j}^\dag}\cdot \e^{\ri \frac{\theta}{\sqrt{m}} \sum_j^m \Tilde{\vU}^{\prime}_{j}\vD\Tilde{\vU}^{\prime \dag}_{j}} \quad \text{for precisely chosen $\theta$}
\end{align*}
to $\vW$ as defined earlier in section \ref{sec:lindeberg_spectrum}. Since the initial low-moment parallel unitary designs are merely \textit{approximations}, once again, the error propagates in an interleaved manner stemming from noncommutativity. This interleaved structure requires the power of adaptive designs.  Since our results so far are about parallel designs, we will need to demonstrate that approximate parallel $T$-designs are also approximate adaptive $T$-designs albeit with a loss scaling exponential with $T$ and the space consumed by the adaptive procedure. 

The main objective for this section is to prove Lemma \ref{lem:clt_basis}. We first address the issue of translating between the parallel and adaptive query bounds in section \ref{sec:parallel_to_adaptive}. Then, in section \ref{sec:approx_design_from_clt} we prove the lemma.

\subsection{From Parallel to Adaptive Quantum Query Bounds}\label{sec:parallel_to_adaptive}

Our construction of unitary $T$-designs requires a pre-existing design as a subroutine, from which we bootstrap an asymptotically more efficient one. Existing designs, however, are generally studied in the parallel query model. Nonetheless, as we will show, a parallel query $T$-design is always a space-bounded adaptive $T$-design, albeit with potentially much weaker parameters. The associated cost is bearable  because the asymptotic efficiency of our construction depends only weakly on the parameters of the pre-existing unitary design subroutine.

\begin{lem} \label{lem:parallel_to_adaptive_qc}
Every $\delta$-approximate parallel unitary $T$-design acting on an $N$-dimensional system is also a $\delta 2^{(2T-1)s(n)}$-approximate adaptive $s(n)$-space unitary $T$-design, for $n = \log(N)$.
\end{lem}
Note that $s(n) \ge n$ as the space includes both the ancilla and the original qubits.
\begin{proof}
We will use the notation of definitions \ref{defn:U_parallel_design} and \ref{defn:U_adaptive_design}. Recall that our design $\vW$ 
is to be indistinguishable by any $s(n)$-space adaptive algorithm $\CA$ from a unitary $\vU$ drawn from the Haar measure. Without loss of generality, we will assume $\CA$ acts on two registers. The design acts on the first, of dimension $N=2^n$, and the second contains all the additional workspace of the algorithm. After the $j^{\text{th}}$ application of $\vW$ (or $\vU$), the algorithm applies a unitary $\vV_j \in \unitary(2^{s(n)})$ to the pair of registers. Write $\Tilde{\vW}_j = (\vW \otimes I) \vV_j$ and $\Tilde{\vU}_j = (\vU \otimes I)\vV_j$.

$\CW$ will be an $\epsilon$-approximate adaptive unitary $T$-design provided
\begin{equation} \label{eqn:adaptive_expanded}
\left| \sup_{\CO, \vV_j, \rho}  \left[  \Expect_{\vU} \tr\left(\CO  \Tilde{\vU}_T \cdots \Tilde{\vU}_2 \Tilde{\vU}_1 \rho  \Tilde{\vU}^\dagger_1 \Tilde{\vU}^\dagger_2 \cdots \Tilde{\vU}^\dagger_T \right)  - 
 \Expect_{\vW}\tr\left(\CO \Tilde{\vW}_T \cdots \Tilde{\vW}_2 \Tilde{\vW}_1  \rho \Tilde{\vW}^\dagger_1 \Tilde{\vW}^\dagger_2 \cdots \Tilde{\vW}^\dagger_T \right) \right] \right| \leq \epsilon.
\end{equation}
The supremum can be taken to be over all (not necessarily Hermitian) operators satisfying $\| \CO \|_{op} \leq 1$, families of unitary operations $\vV_j$ on $s(n)$ qubits, and density operators $\rho$. To prove the lemma, it suffices to rewrite \eqref{eqn:adaptive_expanded} in terms of tensor products of the unitaries $\Tilde{\vU}_j$ and $\Tilde{\vW}_j$. In order to do so, we will need to work in the $2T$-fold tensor power of the original Hilbert space $H$. We will have $\Tilde{\vU}_j$ act on the space $R_{1,j} \simeq H$. It will have an isomorphic ancilla copy called $R_{2,j}$. Also, let $R_1 = \otimes_{j=1}^T R_{1,j}$ and $R_2 = \otimes_{j=1}^T R_{2,j}$.  
The full Hilbert space is $R = R_1 \otimes R_2$. Then, we can write
\begin{equation} \label{eqn:adaptive_expectation}
\Expect_{\vU} \tr\left(\CO  \Tilde{\vU}_T \cdots \Tilde{\vU}_2 \Tilde{\vU}_1 \rho  \Tilde{\vU}^\dagger_1 \Tilde{\vU}^\dagger_2 \cdots \Tilde{\vU}^\dagger_T \right)
\end{equation}
as
\begin{equation} \label{eqn:adaptive_as_parallel}
    \Expect_{\vU} \tr \ltup{
    (\CO \otimes I_{R_2}) 
    \cdot \Big(\big(\bigotimes_{j=1}^T \Tilde{\vU}_j\big) \otimes I_{R_2} \Big)
    \cdot \vP_1
    \cdot \Big(\big(\bigotimes_{j=1}^{T-1} I_{R_{1,j}} \big) \otimes \rho \otimes I_{R_2} \Big)
    \cdot \Big(\big(\bigotimes_{j=1}^T \Tilde{\vU}_j^\dagger\big) \otimes I_{R_2} \Big)
    \cdot \vP_2
    }.
\end{equation}
The operators $\vP_1$ and $\vP_2$ permute the tensor factors in such a way that the index contractions match those of \eqref{eqn:adaptive_expectation}. Specifically,
\begin{align*}
\vP_1(R_{1,j}) &= R_{2,j} \quad &\text{for } j=1,\ldots,T-1 \\
\vP_1(R_{2,j}) &= R_{1,j} \quad &\text{for } j=1,\ldots,T-1 \\
\vP_1(R_{i,j}) &= R_{i,j} \quad &\text{otherwise}
\end{align*}
and
\begin{align*}
\vP_2(R_{2,j}) &= R_{1,j+1} \quad &\text{for } j=1,\ldots,T-1 \\
\vP_2(R_{1,j}) &= R_{2,j-1} \quad &\text{for } j=2,\ldots,T \\
\vP_2(R_{i,j}) &= R_{i,j} \quad &\text{otherwise}.
\end{align*}
Therefore, \eqref{eqn:adaptive_expanded} can be rewritten in a form resembling the expression in the definition of a $\delta$-approximate parallel unitary $T$-design. Since $\vP_2$ is a permutation, its operator norm is $1$ and it can be absorbed into the optimization over operators $\CO$. Each $\Tilde{\vU}_j$ has the form $(\vU\otimes I) \vV_j$. 
Likewise, $\Tilde{\vW}_j = (\vW\otimes I ) \vV_j$. The operator $((\bigotimes_{j=1}^T \vV_j) \otimes I_{R_2} ) \vP_1$ is unitary so also has operator norm $1$ and can be absorbed into the optimization over $\rho$, as can the $\bigotimes_{j=1}^T \vV_j^\dagger$ embedded in $\bigotimes_{j=1}^T \Tilde{\vU}_j^\dagger$ and $\bigotimes_{j=1}^T \Tilde{\vW}_j^\dagger$. Note, however, that
\begin{equation}
\left\| \left(\bigotimes_{j=1}^{T-1} I_{R_{1,j}} \right) \otimes \rho \otimes I_{R_2} \right\|_1
= \| \rho \|_1 \dim\left( \bigotimes_{j=1}^{T-1} R_{1,j} \right) \dim( R_2 ),
\end{equation}
which will introduce a multiplicative factor between the parallel error $\delta$ and the adaptive error $\epsilon$. 
By hypothesis, $\log \dim( R_{i,j} ) = s(n)$ so
\begin{equation}
\log \dim\left( \bigotimes_{j=1}^{T-1} R_{1,j} \right) \dim( R_2 ) = (2T-1) s(n)
\end{equation}
and the lemma follows.
\end{proof}

The lemma above is phrased in language appropriate for discussing parallel and adaptive designs in the operational quantum computing sense as defined earlier in the paper, but the proof of Lemma \ref{lem:clt_basis} requires a slight extension. The main issues are that the operators involved do not act by conjugation and are not all unitary. Let's begin by dealing with the first issue. We need to be able to deal with expressions of the form
\begin{equation} \label{eq:left_right_different}
    \left\|  \Expect_{\vU} \tilde{\vU}_L \rho \tilde{\vU}_R^\dagger
    - \Expect_{\vW} \tilde{\vW}_L \rho \tilde{\vW}_R^\dagger  \right\|_1
\end{equation}
where the operators $\tilde{\vU}_L$, $\tilde{\vU}_R$, $\tilde{\vW}_L$ and $\tilde{\vW}_R$ acting on $\rho$ include bounded numbers of applications of $\vU$ and $\vW$. In the statement of the lemma, the same operators were acting on the left and the right but they were absorbed separately into $\CO$ and $\rho$ so the proof works just as well if the operators are different.

Next, let's consider how nonunitarity will affect the bound. To begin, it is straightforward to handle nonunitary $\vV_j$ in the proof of \autoref{lem:parallel_to_adaptive_qc}. The operator norm of each of the operators $\vV_j$ and $\vV_j^\dagger$ being absorbed at each step will be bounded by $\| \vV_j \|_{op}$ instead of one, and the multiplicative factor between parallel and adaptive designs will acquire a blow-up of the product of all those operator norms. Sums of operators can be handled using the triangle inequality.

These observations prove
\begin{lem} 
Suppose that $\vW$ is a $\delta$-approximate parallel unitary $T$-design and that $\vU$ is a Haar random unitary. Let $\vX_j^U = \vU \sum \vX_{jl}$ and $\vY_j^U = \vU \sum_r \vY_{jr}$ then in turn define $\vX_U = \prod_j \vX_j^U$ and $\vY_U = \prod_j \vY_j^U$.  Suppose that $\vX_U$ and $\vY_U$ act on $s(n)$ qubits while $\vU$ on acts on only $n$ on them. Repeat the definitions for $\vW$. Then
\begin{equation} 
   \max_{\|\rho\|_1 \leq 1} \left\|  \Expect_{\vU \leftarrow \mu} \vX_U \rho \vY_U^\dagger - \Expect_{\vW \leftarrow \CW} \vX_W \rho \vY_W^\dagger
    \right\|_1 \leq \delta  2^{(2T-1)s(n)} \cdot \Gamma,
\end{equation}
where 
\begin{equation}
    \Gamma = \sum_{\vec{l}, \vec{r}} \prod_{j_1} \prod_{j_2} \left\| \vX_{j_1 l_{j_1}} \right\|_{op} \left\| \vY_{j_2 r_{j_2}} \right\|_{op}.
\end{equation}
\end{lem}
Finally, in our application, we will need to include actions from the left instead of the right of the adjoints of unitaries from our design. Had the original parallel design been capable of handling such cases, then the argument above would extend immediately to the mixed $\vU$, $\vU^\dagger$ case. Otherwise, by doubling the number of registers again, the indices can be made to contract appropriately. That finally proves
\begin{lem} \label{lem:parallel_to_adaptive}
Suppose that $\vW$ is a $\delta$-approximate parallel unitary $T$-design and that $\vU$ is a Haar random unitary. Let $\vX_j^U = \vU_j \sum_l \vX_{jl}$ and $\vY_j^U = \vU_j \sum_r \vY_{jr}$ then in turn define $\vX_U = \prod_j \vX_j^U$ and $\vY_U = \prod_j \vY_j^U$.  Suppose that $\vX_U$ and $\vY_U$ act on $s(n)$ qubits while $\vU$ on acts on only $n$ on them. Here $\vU_j \in \{ \vU, \vU^\dagger\}$. Repeat the definitions for $\vW$. Then
\begin{equation} 
   \max_{\|\rho\|_1 \leq 1} \left\|  \Expect_{\vU \leftarrow \mu} \vX_U \rho \vY_U^\dagger - \Expect_{\vW \leftarrow \CW} \vX_W \rho \vY_W^\dagger
    \right\|_1 \leq \delta  2^{4T s(n)} \cdot \Gamma,
\end{equation}
where 
\begin{equation}
    \Gamma = \sum_{\vec{l}, \vec{r}} \prod_{j_1} \prod_{j_2} \left\| \vX_{j_1 l_{j_1}} \right\|_{op} \left\| \vY_{j_2 r_{j_2}} \right\|_{op}.
\end{equation}
\end{lem}

\subsection{Approximate Designs from the Central Limit Theorem}\label{sec:approx_design_from_clt}

In this section, we will frequently need to refer to the conjugation action of a matrix. We will use the map $\CN$ to ``channelize'' an arbitrary matrix $\vA \in \text{GL}(N)$ as follows,
    \begin{align*}
        \CN \colon A & \mapsto \left( \rho \mapsto \vA \rho \vA^\dag \right).
    \end{align*}    

\begin{lem}[Approximate unitary designs from CLT]\label{lem:exp_G_design_from_CLT}
    Consider i.i.d. unitary designs $\Tilde{\vU}_j \in \unitary(N)$ such that the first $q$ tensor moments of $\Tilde{\vU}_j$ match those of Haar-distributed $\vU_j$. More precisely,
    \begin{align*}
        \dnorm{\CN(\Tilde{\vU}_j^{\otimes r}) -  \CN(\vU_j^{\otimes r})} \leq \epsilon_{\scaleto{q}{5pt}} \quad \text{for}\quad 1 \leq r \leq q.
    \end{align*}
    Suppose $\vD \in  \gl(N)$ is a random diagonal matrix bounded as $ \BE \left[\lnorm{\vD}^{q+1}_{\text{op}}\right] \leq C_D.$
    Define the random Hermitian matrices 
    \begin{align*}
        \tilde{\vH} &:= \frac{1}{\sqrt{m}}\sum_{j=1}^{m} \tilde{\vH}_j \quad\quad \text{where} \quad\quad \tilde{\vH}_j: = \Tilde{\vU}_j \vD \Tilde{\vU}_j^\dag, \\ 
        \vH &:=  \frac{1}{\sqrt{m}}\sum_{j=1}^{m} \vH_j \quad\quad \text{where} \quad\quad  \vH_j: = \vU_j \vD \vU_j^\dag.
    \end{align*}
    Then for $\theta =\CO(1)$, and $m =\Omega(\poly_t\log(N))$
    \begin{align*}
    	\dnorm{\bexpect{ \CN \ltup{e^{i\theta \Tilde{\vH}}}^{\otimes k} -  \CN \ltup{e^{i\theta \vH}}^{\otimes k} } } \leq 
     \left[ 
            \epsilon_q 2^{8nq} \left(\frac{\theta^2 k^2}{m} \right)^{q+1}
            + 2^{q+2}  \ltup{\frac{\theta k}{\sqrt{m}}}^{q+1}
            \right] \frac{ m C_D \, e^q}{q^q}
    \end{align*}
\end{lem}
To prove Lemma \ref{lem:exp_G_design_from_CLT}, we use Lindeberg's replacement principle \cite{Lindeberg1922EineNH}, as adapted to proving a central limit theorem in the matrix case~\cite{chen2023sparse}. We will require Duhamel's formula and a corollary thereof.
\begin{prop}\label{prop:duhamel}[Duhamel's Formula] For any matrix $\vA,\vB$, 
	$$\e^{(\vA+\vB)t} = \e^{\vA t} + \int_0^t  \e^{\vA(t-s)}\vB \e^{(\vA+\vB)s}\diff s.$$
\end{prop}
Applying Duhamel's formula $k$ times yields the following corollary:
\begin{cor}\label{cor:duhamel} For any matrix $\vA,\vB$ and $t\in \BR$, set $s_0 := t$. Then,
\begin{align*}
    \e^{(\vA+\vB)t} = \e^{\vA t} &+  \sum_{r=1}^q ~~ \idotsint\limits_{t>s_1 > \dots > s_l > 0} ~~\prod_{j=1}^{r} \diff s_j  \left(\prod_{j=1}^{r} \e^{\vA(s_{j-1}-s_{j})}\vB \right) \e^{\vA s_r}  \\
    & ~~~~+  \idotsint\limits_{t>s_1 > \dots > s_{q+1} > 0} \prod_{j=1}^{q+1} \diff s_j  \left(\prod_{j=1}^{l} \e^{\vA(s_{j-1}-s_{j})}\vB \right) \e^{(\vA+\vB)s_{q+1}}. 
\end{align*}
\end{cor} 
Additionally, we will make use of the following properties of the diamond norm on unitary channels. Write $\CM = \CN(\vM)$ and $\CB = \CN(\vB)$ for $\vB, \vM \in \gl(N)$. Similarly, let $\CU = \CN(\vU)$ and $\CV = \CN(\vV)$ be unitary channels derived from the unitary matrices $\vU, \vV \in \text{U}(N)$. Then~\cite{aharonov1998quantum},
\begin{enumerate}[label=(\roman*)]
    \item $\dnorm{\CM} \leq \norm{M}_{\text{op}}^2 $
    \item $ \dnorm{\CM \otimes \CB} = \dnorm{\CB \otimes \CM} $
    \item $\dnorm{\CU \otimes \CM} = \dnorm{\CM}$ 
    \item $\dnorm{\CU \circ \CM} = \dnorm{\CM}$ 
    \item $\dnorm{\CU - \CV} = \sup_{s \in \unitary(1)} 2\lnorm{s\vU - \vV}_{\text{op}}$
\end{enumerate}

\begin{proof}[Proof of Lemma \ref{lem:exp_G_design_from_CLT}] 
    Let $\vH_1, \dots , \vH_m \in \text{Herm}(N)$ be i.i.d. random matrices drawn from the GUE. Note that a sum of properly normalized GUE matrices is distributed as GUE:
    \begin{align*}
        \frac{1}{\sqrt{m}}\sum_{j=1}^{m} \vH_j  \stackrel{dist.}{\sim}  \vH.
    \end{align*} 
    In anticipation of a Lindeberg-type argument~\cite{chen2023sparse}, we define the random matrix that interpolates between the desired targets
    $$\vM_l := \frac{1}{\sqrt{m}}\ltup{\sum_{j=1}^{l} \Tilde{\vH}_j  + \sum_{j=l+1}^{m} \vH_j } \quad  \text{such that}\quad \vM_0 = \vH \quad \text{and}\quad \vM_m = \Tilde{\vH}.$$
 It follows that  
    \begin{align*}
        \dnorm{\bexpect{\CN\ltup{\ltup{e^{i\theta \Tilde{\vH}}}^{\otimes k}}-  \CN\ltup{\ltup{e^{i\theta \vH}}^{\otimes k}}}}
        &=  \dnorm{\sum_{l=0}^{m} \bexpect{ \CN\ltup{\ltup{e^{i\theta \vM_{l+1}}}^{\otimes k}} - \CN\ltup{\ltup{e^{i\theta \vM_l}}^{\otimes k}}}} 
    \end{align*}
    and by the triangle inequality,
    \begin{align}
        \dnorm{\sum_{l=0}^{m} \bexpect{\CN\ltup{ \ltup{e^{i\theta \vM_{l+1}}}^{\otimes k}} - \CN\ltup{\ltup{e^{i\theta \vM_l}}^{\otimes k}}} }
        \leq \sum_{l=0}^{m} \dnorm{ \bexpect{ \CN\ltup{\ltup{e^{i\theta \vM_{l+1}}}^{\otimes k}} - \CN\ltup{\ltup{e^{i\theta \vM_l}}^{\otimes k}}}}. \label{eq:Lindeberg}
    \end{align}
    We now introduce a new notation:
    \begin{align*}
        \Tilde{\vH}^{(k)} := \sum_{j=1}^k \vI \otimes \dots \otimes \underbrace{\Tilde{\vH}}_{j^{th}\text{ place}} \otimes \dots \otimes \vI.
    \end{align*}
    Then $ e^{\Tilde{\vH}^{(k)}} = \ltup{e^{\Tilde{\vH}}}^{\otimes k}$ so
    \begin{align}
        \dnorm{ \bexpect{ \CN\ltup{  \ltup{e^{i\theta \vM_{l+1}}}^{\otimes k}} -   \CN\ltup{ \ltup{e^{i\theta \vM_l}}^{\otimes k}}}} = \dnorm{ \bexpect{  \CN\ltup{ e^{i\theta \vM_{l+1}^{(k)} }} -   \CN\ltup{  e^{i\theta \vM_{l}^{(k)} } }}}. \label{eq:neq_notation}
    \end{align}
    Next, we use Duhamel's formula to isolate the components of the channels which contribute to the distance.
    We define the following random matrices 
    \begin{align*}
        \vA_l  &:= \sum_{j=1}^{l-1} \Tilde{\vH}_j  + \sum_{j=l+1}^{m} \vH_j \\
    \end{align*}
    such that $\sqrt{m}\vM_{l + 1} = \vA_l + \Tilde{\vH}_l $ and $\sqrt{m}\vM_{l} = \vA_l + \vH_l$, and naturally $\vM_{l + 1}^{(k)} = \vA_l^{(k)} + \Tilde{\vH}_l^{(k)} $ and $\vM_{l}^{(k)} = \vA_l^{(k)} + \vH_l^{(k)} $. This allows us to rewrite the RHS of equation \eqref{eq:neq_notation} as
    \begin{align}
        \dnorm{ \bexpect{ \CN\ltup{  e^{i\theta \vM_{l+1}^{(k)} }} -    \CN\ltup{ e^{i\theta \vM_{l}^{(k)} } }}} = \dnorm{ \bexpect{  \CN\ltup{ e^{i\theta \ltup{\vA_l^{(k)} + \Tilde{\vH}_l^{(k)}} }} -   \CN\ltup{  e^{i\theta \ltup{\vA_l^{(k)} + \vH_l^{(k)}} } }}}. \label{eq:Lindeberg_expand}
    \end{align}  
    Applying the $q$-fold Duhamel Formula where $s_0 = \theta $ yields
    \begin{align} 
        e^{i \theta (\vA_l^{(k)} + \Tilde{\vH}_l^{(k)}) } &= e^{i \theta \vA_l^{(k)}} + \sum_{r=1}^q i^l~~ \idotsint\limits_{\frac{\theta}{\sqrt{m}} >s_1 > \dots > s_r > 0} ~~\prod_{j=1}^{r} \diff s_j  \left(\prod_{j=1}^{r} e^{i \vA_l^{(k)}(s_{j-1}-s_{j})}\Tilde{\vH}_l^{(k)} \right) e^{i  \vA_l^{(k)} s_r}  \label{eq:lower_H}\\
        & \quad+ i^{q+1} \idotsint\limits_{\frac{\theta}{\sqrt{m}} >s_1 > \dots > s_{q+1} > 0} ~~\prod_{j=1}^{q+1} \diff s_j  \left(\prod_{j=1}^{q+1} e^{i \vA_l^{(k)}(s_{j-1}-s_{j})}\Tilde{\vH}_l^{(k)} \right) e^{i \left(\vA_l^{(k)}+\Tilde{\vH}_l^{(k)}\right)s_{q+1}}  \nonumber  \\ 
        e^{i \theta (\vA_l^{(k)} + \vH_l^{(k)}) } &=   e^{i  \theta \vA_l^{(k)}} + \sum_{r=1}^q i^l~~ \idotsint\limits_{\frac{\theta}{\sqrt{m}} >s_1 > \dots > s_r > 0} ~~\prod_{j=1}^{r} \diff s_j  \left(\prod_{j=1}^{r} e^{i \vA_l^{(k)}(s_{j-1}-s_{j})}\vH_l^{(k)} \right) e^{i  \vA_l^{(k)} s_r} \label{eq:lower_G} \\
        & \quad + i^{q+1} \idotsint\limits_{\frac{\theta}{\sqrt{m}} >s_1 > \dots > s_{q+1} > 0} ~~\prod_{j=1}^{q+1} \diff s_j  \left(\prod_{j=1}^{q+1} e^{i \vA_l^{(k)}(s_{j-1}-s_{j})}\vH_l^{(k)} \right) e^{i \left(\vA_l^{(k)}+\vH_l^{(k)}\right)s_{q+1}} \nonumber
    \end{align} 
    Our goal is to bound \eqref{eq:Lindeberg_expand} so it is natural to pair up lower order terms (as a functions of $\Tilde{\vH}_l$ and $\vH_l$) from lines \eqref{eq:lower_H} and \eqref{eq:lower_G}. First, we define several new variables to ease notation. We will use a variable $\vV$ to denote each term in the above set of equations. Superscripts $H$ and $G$ denote if the term depends on the random matrices $\Tilde{\vH}_l$ or $\vH_l$, respectively. Subscript $r$ denotes $r$ occurrences of the respective random matrices in the integrand. Specifically,
    \begin{align*}
        \vV^H_0 &:= e^{i  \frac{\theta}{\sqrt{m}} \vA_l^{(k)}} \\
        \vV^H_r &:= i^{r} \idotsint\limits_{\frac{\theta}{\sqrt{m}} >s_1 > \dots > s_r > 0} \prod_{j=1}^{r} \diff s_j  \left(\prod_{j=1}^{r} e^{i \vA_l^{(k)}(s_{j-1}-s_{j})}\Tilde{\vH}_l^{(k)} \right) e^{i  \vA_l^{(k)} s_r}  
        & (1\leq  r \leq q)\\
        \vV^H_{r, \infty} &:= i^{r}  \idotsint\limits_{\frac{\theta}{\sqrt{m}} >s_1 > \dots > s_l > 0} \prod_{j=1}^{r} \diff s_j  \left(\prod_{j=1}^{r} e^{i \vA_l^{(k)}(s_{j-1}-s_{j})}\Tilde{\vH}_l^{(k)} \right) e^{i \left(\vA_l^{(k)}+\Tilde{\vH}_l^{(k)}\right)s_{r}}  &  (r\geq 1).\\
    \end{align*}
    We use analogous definitions for $\vV^G_r$. Then, we may rewrite \eqref{eq:lower_H} and \eqref{eq:lower_G} as 
    \begin{align} 
        e^{i \frac{\theta}{\sqrt{m}} \ltup{\vA_l^{(k)} + \Tilde{\vH}_l^{(k)}} } = \sum_{r=0}^q \vV^H_r + \vV^H_{q+1, \infty}  \label{eq:lower_H2}   \\ 
        e^{i \frac{\theta}{\sqrt{m}} \ltup{\vA_l^{(k)} + \vH_l^{(k)}} } = \sum_{r=0}^q \vV^G_r + \vV^G_{q+1, \infty}. \label{eq:lower_G2}
    \end{align} 
    Note that this expansion is valid for any $q \geq 1$. By the definition of the diamond norm, \eqref{eq:Lindeberg_expand} becomes
    \begin{align} 
         \sup_{\rho, d} \lnorm{ \Expect \ltup{\Id_d \otimes \CN \left(e^{i \frac{\theta}{\sqrt{m}} \ltup{\vA_l^{(k)} + \Tilde{\vH}_l^{(k)}} }\right)}[\rho]- \Expect \ltup{\Id_d \otimes \CN \left(e^{i \frac{\theta}{\sqrt{m}} \ltup{\vA_l^{(k)} + \vH_l^{(k)}} }\right)}[\rho]}_1,
         \label{eq:diamond_expansion}
    \end{align} 
    where the optimization is over density operators and positive integers $d$. The supremum over $d$ is achieved for $d=2$, however, by convexity of the norm, it suffices to optimize over minimal size states purifying the density operators on the input ($d=1$) space. We will, therefore, fix $d=2$ from now on and stop writing the $\Id_2$ explicitly. As such, \eqref{eq:diamond_expansion} can be written as
    \begin{align} 
        & \lnorm{\bexpect{  \CN\ltup{ e^{i\frac{\theta}{\sqrt{m}} \ltup{\vA_l^{(k)} + \Tilde{\vH}_l^{(k)}} }}[\rho] -    \CN\ltup{ e^{i\frac{\theta}{\sqrt{m}} \ltup{\vA_l^{(k)} + \vH_l^{(k)}} } }[\rho]}}_1. \label{eq:arg_of_diamond}
    \end{align}
    The first term can be expanded as follows
    \begin{align*}
        \CN\ltup{ e^{i\frac{\theta}{\sqrt{m}} \ltup{\vA_l^{(k)} + \Tilde{\vH}_l^{(k)}} }}[\rho] 
        &= 
        \ltup{ e^{i\frac{\theta}{\sqrt{m}} \ltup{\vA_l^{(k)} + \Tilde{\vH}_l^{(k)}} }}\rho\ltup{ e^{i\frac{\theta}{\sqrt{m}} \ltup{\vA_l^{(k)} + \Tilde{\vH}_l^{(k)}} }}^\dagger    \\
        &=\left(\sum_{r=0}^{q} \vV^H_r + \vV^H_{q+1, \infty}\right)\rho\ltup{ e^{i\frac{\theta}{\sqrt{m}} \ltup{\vA_l^{(k)} + \Tilde{\vH}_l^{(k)}} }}^\dagger \tag{by \eqref{eq:lower_H2} and  \eqref{eq:lower_G2}}\\
        &=\sum_{r=0}^{q} \left( \vV^H_r \right)\rho\ltup{ e^{i\frac{\theta}{\sqrt{m}} \ltup{\vA_l^{(k)} + \Tilde{\vH}_l^{(k)}} }}^\dagger  +\ltup{\vV^H_{q+1, \infty}}\rho\ltup{ e^{i\frac{\theta}{\sqrt{m}} \ltup{\vA_l^{(k)} + \Tilde{\vH}_l^{(k)}} }}^\dagger    \\
        &= \sum_{r=0}^{q} \left( \vV^H_r \right)\rho\ltup{ \sum_{r^\prime=0}^{q-r} \vV^H_{r^\prime} + \vV^H_{q-r+1, \infty}}^\dagger + \ltup{\vV^H_{q+1, \infty}}\rho\ltup{ e^{i\frac{\theta}{\sqrt{m}} \ltup{\vA_l^{(k)} + \Tilde{\vH}_l^{(k)}} }}^\dagger \tag{by \eqref{eq:lower_H2} and  \eqref{eq:lower_G2}}\\
        &=  \sum_{r=0}^{q} \sum_{r^\prime=0}^{q-r} \left( \vV^H_r \right)\rho\ltup{  \vV^H_{r^\prime}}^\dagger + \sum_{r=0}^{q}\left( \vV^H_r \right)\rho\ltup{\vV^H_{q-r+1, \infty}}^\dagger   \\
        &\quad + \ltup{\vV^H_{q+1, \infty}}\rho\ltup{ e^{i\frac{\theta}{\sqrt{m}} \ltup{\vA_l^{(k)} + \Tilde{\vH}_l^{(k)}} }}^\dagger.  
    \end{align*}
    A similar calculation shows that the second term 
    of \eqref{eq:arg_of_diamond} can also be expanded as 
    \begin{align*}
        &\CN\ltup{ e^{i\frac{\theta}{\sqrt{m}} \ltup{\vA_l^{(k)} + \vH_l^{(k)}} }}[\rho] \\
        &= \sum_{r=0}^{q} \sum_{r^\prime=0}^{q-r} \left( \vV^G_r \right)\rho\ltup{  \vV^G_{r^\prime}}^\dagger + \sum_{r=0}^{q}\left( \vV^G_r \right)\rho\ltup{\vV^G_{q-r+1, \infty}}^\dagger + \ltup{\vV^G_{q+1, \infty}}\rho\ltup{ e^{i\frac{\theta}{\sqrt{m}} \ltup{\vA_l^{(k)} + \vH_l^{(k)}} }}^\dagger. 
    \end{align*}
    Putting these together, we have 
    \begin{align*} 
        & \bexpect{  \CN\ltup{ e^{i\frac{\theta}{\sqrt{m}} \ltup{\vA_l^{(k)} + \Tilde{\vH}_l^{(k)}} }}[\rho] -    \CN\ltup{ e^{i\frac{\theta}{\sqrt{m}} \ltup{\vA_l^{(k)} + \vH_l^{(k)}} } }[\rho]} \\
        &= \bexpect{ \sum_{r=0}^{q} \sum_{r^\prime=0}^{q-r} \left( \vV^H_r \right)\rho\ltup{  \vV^H_{r^\prime}}^\dagger + \sum_{r=0}^{q}\left( \vV^H_r \right)\rho\ltup{\vV^H_{q-r+1, \infty}}^\dagger  + \ltup{\vV^H_{q+1, \infty}}\rho\ltup{ e^{i\frac{\theta}{\sqrt{m}} \ltup{\vA_l^{(k)} + \Tilde{\vH}_l^{(k)}} }}^\dagger}  \\
        &\quad\quad - \bexpect{\sum_{r=0}^{q} \sum_{r^\prime=0}^{q-r} \left(\vV^G_r \right) \rho \ltup{  \vV^G_{r^\prime} }^\dagger + \sum_{r=0}^{q}\left(\vV^G_r \right) \rho \ltup{  \vV^G_{q-r+1, \infty}}^\dagger + \ltup{ \vV^G_{q+1, \infty}}\rho \ltup{ e^{i\frac{\theta}{\sqrt{m}} \ltup{\vA_l^{(k)} + \vH_l^{(k)}} }}^\dagger }
    \end{align*} 
    Recall that for each $1 \leq r \leq q$,
    $$\dnorm{\BE [ \CN\ltup{\tilde{\vU}_j^{\otimes r}}] - \BE[\CN\ltup{\vU_j^{\otimes r}]} } \leq \epsilon_{\scaleto{q}{5pt}}. $$
    Observe that the matrices $\vV_r^H$ and $\vV_r^G$ defined earlier were constructed by the iterated application of Duhamel's formula, which generates an operator that has the form of $r$ adaptive queries. If $r+ r^\prime \leq q$ for any $0 \leq r,r^\prime \leq q$, then by Lemma \ref{lem:parallel_to_adaptive},
    \begin{align*}
        &\sup_{\rho} \lnorm{  \bexpect{\left(\Id_2 \otimes\vV^H_r\right)\rho\left(\Id_2 \otimes \vV^H_{r^\prime}\right)^\dagger }- \bexpect{  \left(\Id_2 \otimes\vV^G_r\right)\rho\left(\Id_2 \otimes \vV^G_{r^\prime}\right)^\dagger} }_1 \leq \epsilon_{\scaleto{q}{5pt}} 2^{4q \cdot 2n} \cdot \Gamma, \text{ and} \\
        &\sum_{r=0}^{q} \sum_{r^\prime=0}^{q-r} \sup_{\rho} \lnorm{ \bexpect{ \left(\Id_2 \otimes\vV^H_r\right)\rho\left(\Id_2 \otimes \vV^H_{r^\prime}\right)^\dagger} - \bexpect{  \left(\Id_2 \otimes\vV^G_r\right)\rho\left(\Id_2 \otimes \vV^G_{r^\prime}\right)^\dagger}}_1 \leq q! \epsilon_{\scaleto{q}{5pt}} 2^{8nq} \cdot \Gamma,
    \end{align*}
    where
    \begin{align}
    \Gamma &=  \left(~~\idotsint\limits_{\frac{\theta}{\sqrt{m}} >s_1 > \dots > s_l > 0} ~~\prod_{j=1}^{q+1} \diff s_j  
    \right)^2 k^2 
    \Expect{\left[  \left\| D \right\|_{op}^{q+1} \right]} \nonumber 
   \leq \left( 
    \frac{(\frac{\theta}{\sqrt{m}} )^{q+1}} {(q+1)!} \right)^2 
    k^2 \bexpect{\lnorm{\vD}_{\text{op}}^{q+1}} \nonumber \\
    & =: h(k,\theta,m,q)^2 \bexpect{\lnorm{\vD}_{\text{op}}^{q+1}} \label{eq:define_h} \\
    &\leq h(k,\theta,m,q)^2 C_D, \nonumber
    \end{align}
    using the upper bound $\bexpect{\lnorm{\vD^{\otimes q+1}}_{op}} \leq C_D$ in the last line.
    Note that we have used the space bound $s(n)=2n$ when applying Lemma \ref{lem:parallel_to_adaptive}. By the triangle inequality (for the diamond norm) \eqref{eq:arg_of_diamond} is bounded above as
    \begin{align*}
        &\dnorm{ \bexpect{ \CN\left( e^{i \frac{\theta}{\sqrt{m}} (\vA_l^{(k)} + \Tilde{\vH}_l^{(k)}) }\right) -  \CN\left(  e^{i \frac{\theta}{\sqrt{m}} (\vA_l^{(k)} + \vH_l^{(k)}) } \right)}} \label{eq:layer2} \\
        &\quad = q! \epsilon_{\scaleto{q}{5pt}} 2^{8nq} +  \sum_{r=0}^{q}\lnorm{\bexpect{ \left( \vV^H_r \right)\rho\ltup{\vV^H_{q-r+1, \infty}}^\dagger}}_1 +\sum_{r=0}^{q} \lnorm{\bexpect{\left(\vV^G_r \right) \rho \ltup{  \vV^G_{q-r+1, \infty}}^\dagger }}_1  \\
        & \quad+  \lnorm{\bexpect{ \ltup{\vV^H_{q+1, \infty}}\rho\ltup{ e^{i\frac{\theta}{\sqrt{m}} \ltup{\vA_l^{(k)} + \Tilde{\vH}_l^{(k)}} }}^\dagger}} +\lnorm{\bexpect{\ltup{ \vV^G_{q+1, \infty}}\rho \ltup{ e^{i\frac{\theta}{\sqrt{m}} \ltup{\vA_l^{(k)} + \vH_l^{(k)}} }}^\dagger }}_1.  \\
    \end{align*} 
    The remaining terms can only be bounded naively, since we don't have guarantees on the distance between higher moments of our starting unitary design.
    \begin{align} 
        &\lnorm{\bexpect{\left( \vV^H_{q+1, \infty}\right) \rho \ltup{ e^{\frac{i\theta}{\sqrt{m}} \ltup{\vA_l^{(k)} + \Tilde{\vH}_l^{(k)}} }}^\dagger}}_1 \leq \lnorm{\bexpect{\left( \vV^H_{q+1, \infty}\right) \ltup{ e^{\frac{i\theta}{\sqrt{m}} \ltup{\vA_l^{(k)} + \Tilde{\vH}_l^{(k)}} }}^\dagger}}_{\text{op}} \tag{Holder's inequality}\nonumber\\
        &= \Bigg\Vert \Expect \Bigg[ \left(i^{q+1} \idotsint\limits_{\frac{\theta}{\sqrt{m}} >s_1 > \dots > s_l > 0} ~~\prod_{j=1}^{q+1} \diff s_j  \left(\prod_{j=1}^{q+1} e^{i \vA_l^{(k)}(s_{j-1}-s_{j})}\Tilde{\vH}_l^{(k)} \right)e^{i \left(\vA_l^{(k)}+\Tilde{\vH}_l^{(k)}\right)s_{q+1}} \right) \cdot \ltup{ e^{\frac{i\theta}{\sqrt{m}} \ltup{\vA_l^{(k)} + \Tilde{\vH}_l^{(k)}} }}^\dagger \Bigg] \Bigg\Vert_{\text{op}} \nonumber \\
        &\leq \Expect \Bigg[ ~~\idotsint\limits_{\frac{\theta}{\sqrt{m}} >s_1 > \dots > s_l > 0} ~~\prod_{j=1}^{q+1} \diff s_j  \left(\prod_{j=1}^{q+1} \lnorm{e^{i\vA_l^{(k)}(s_{j-1}-s_{j})}}_{\text{op}}\cdot\lnorm{\Tilde{\vH}_l^{(k)}}_{\text{op}}\right)\cdot\lnorm{e^{i \left(\vA_l^{(k)}+\Tilde{\vH}_l^{(k)}\right)s_{q+1}}}_{\text{op}}  \Bigg] \label{eq:convexity}  \\
        &\leq \bexpect{ \left(~~\idotsint\limits_{\frac{\theta}{\sqrt{m}} >s_1 > \dots > s_l > 0} ~~\prod_{j=1}^{q+1} \diff s_j  \prod_{j=1}^{q+1}\lnorm{(\Tilde{\vU}_l \vD \Tilde{\vU}_l^\dag)^{(k)}}_{\text{op}}\right) } \label{eq:Hk}\\
        &\leq  h(k,\theta,m,q) 
        \bexpect{\lnorm{\vD}_{\text{op}}^{q+1}} 
        \leq  h(k,\theta,m,q) C_D. \nonumber 
    \end{align} 
    Line \eqref{eq:convexity} follows from the convexity of the operator norm and the triangle inequality. The last line uses the inequality $ \lnorm{\Tilde{\vH}_l^{(k)}}_{\text{op}}  \leq  k \lnorm{\Tilde{\vH}_l}_{\text{op}}$ and the definition of $h(k,\theta,m,q)$ in \eqref{eq:define_h}. Similar arguments show that for $1\leq r \leq q$,
    \begin{align*}
        &\dnorm{\bexpect{ \left( \vV^H_r \right)\rho\left( \vV^H_{q+1, \infty}\right)^\dagger - \left( \vV^G_r \right)  \rho\left( \vV^G_{q-r+1, \infty}\right)^\dagger }} \\
        &\quad\quad= \lnorm{\bexpect{\left( \vV^H_{q-r+1, \infty}\right)\left( \vV^H_r \right)^\dagger}}_{\text{op}} + \lnorm{\bexpect{\left( \vV^G_{q-r+1, \infty}\right)\left( \vV^G_r \right) ^\dagger} }_{\text{op}}\\
        &\quad\quad\leq \frac{ 2 C_D (\frac{\theta}{\sqrt{m}} k)^{q+1} }{(\sqrt{m})^{q+1}r!(q+1-r)!}
    \end{align*}
    Putting these bounds together,
    \begin{align*} 
        &\dnorm{ \bexpect{ \CN\left( e^{i \frac{\theta}{\sqrt{m}} (\vA_l^{(k)} + \Tilde{\vH}_l^{(k)}) }\right) -  \CN\left(  e^{i \frac{\theta}{\sqrt{m}} (\vA_l^{(k)} + \vH_l^{(k)}) } \right)}} \\
        &\quad \leq  q! \epsilon_{\scaleto{q}{5pt}} 2^{(8nq} \cdot 
        \Gamma + \sum_{r=0}^{q}\lnorm{\bexpect{  \left( \vV^H_r \right)\rho\ltup{\vV^H_{q-r+1, \infty}}^\dagger }}_1+\lnorm{\left(\vV^G_r \right) \rho \ltup{  \vV^G_{q-r+1, \infty}}^\dagger}_1   \\
        & \quad\quad\quad+  \lnorm{\bexpect{ \ltup{\vV^H_{q+1, \infty}}\rho\ltup{ e^{i\frac{\theta}{\sqrt{m}} \ltup{\vA_l^{(k)} + \Tilde{\vH}_l^{(k)}}} }}^\dagger}_1+\lnorm{\ltup{ \vV^G_{q+1, \infty}}\rho \ltup{ e^{i\frac{\theta}{\sqrt{m}} \ltup{\vA_l^{(k)} + \vH_l^{(k)}} }}^\dagger }_1  \\
        &\quad\leq q! \epsilon_{\scaleto{q}{5pt}} 2^{8nq} \cdot h(k,\theta,m,q)^2 C_D + 
        h(k,\theta,m,q) C_D + \sum_{r=0}^{q}\frac{ 2 C_D (\theta k)^{q+1} }{(\sqrt{m})^{q+1}r!(q+1-r)!} \\
        &\quad\leq q! \epsilon_{\scaleto{q}{5pt}} 2^{8nq} \cdot h(k,\theta,m,q)^2 C_D +   
        2 h(k,\theta,m,q) C_D +
        +  \frac{2 C_D (\theta k)^{q+1} }{(\sqrt{m})^{q+1}}\cdot\frac{2^{q+1}}{ (q+1)!} \tag{Binomial Thm}\\
        &\quad\leq  q! \epsilon_{\scaleto{q}{5pt}} 2^{8nq} \cdot h(k,\theta,m,q)^2 C_D 
        + 2 \left( 2^{q+1} + 1 \right) h(k,\theta,m,q) C_D
    \end{align*} 
    Finally, we may return to equation \eqref{eq:Lindeberg} to get an upper bound on the distance between the two unitary channels,
    \begin{align*}
        &\dnorm{\bexpect{ \CN \left( (e^{i\theta \vM})^{\otimes k}\right) -  \CN \left( (e^{i\theta \Tilde{\vM}})^{\otimes k}\right)}} \\
        &\quad\leq  \sum_{l=0}^{m} \dnorm{ \bexpect{ \CN \left( (e^{i\theta \vM_{l+1}})^{\otimes k}\right) - \CN \left( (e^{i\theta \vM_l})^{\otimes k}\right)}} \\
        &\quad\leq  m\ltup{q! \epsilon_{\scaleto{q}{5pt}} 2^{8nq} \cdot \Gamma +  \frac{2 C_D (\theta k)^{q+1} }{(\sqrt{m})^{q+1}}\cdot \frac{2^{q+1}+1}{ (q+1)!}}\\
          &\quad\leq \epsilon_q 2^{8nq} \left(\frac{\theta^2 k^2}{m} \right)^{q+1}  \frac{ m C_D \, e^q}{(q+1)^{q+1}} + \ltup{\frac{\theta k}{\sqrt{m}}}^{q+1} \frac{2m C_D (2^{q+1}+1)e^q}{(q+1)^{q+1}} \\
          &\quad \leq \left[ 
            \epsilon_q 2^{8nq} \left(\frac{\theta^2 k^2}{m} \right)^{q+1}
            + 2^{q+2}  \ltup{\frac{\theta k}{\sqrt{m}}}^{q+1}
            \right] \frac{ m C_D \, e^{q}}{q^q} ,
    \end{align*}
    using the inequality $q^q e^{-q} \leq (q+1)!$
    which concludes the proof. 
\end{proof}

\subsection{From One to Two Exponentiated Gaussians}
Lemma \ref{lem:clt_basis} addresses the distance to a single exponentiated Hermitian matrix. However, we are actually interested in approximating the product of \emph{two} exponentiated GUE $e^{i\theta \vH}e^{i\theta \vH^\prime}$. Lemma \ref{lem:clt_basis} follows as a corollary of Lemma \ref{lem:exp_G_design_from_CLT} and the following telescoping argument for quantum channels.

\begin{lem}[Telescoping channel differences]\label{lem:composition_channels}
    For two quantum channels $\CX$ and $\CM$, $\dnorm{\CX^k - \CM^k} \leq k \dnorm{\CX - \CM}.$
\end{lem}
\begin{proof}[Proof of Lemma~\ref{lem:composition_channels}]
    By a telescoping sum, we have that
    \begin{align*}
        \dnorm{\CX^k - \CM^k} 
        &= \dnorm{\CX^k (\CX  - \CM) + (\CX^{k-1}  - \CM^{k-1}) \CM}\\
        &\leq \dnorm{\CX^k (\CX  - \CM)} + \dnorm{(\CX^{k-1}  - \CM^{k-1}) \CM}\\
        &\leq \dnorm{\CX  - \CM} + \dnorm{\CX^{k-1}  - \CM^{k-1}}\tag{Channels contract diamond norm}\\
        &\leq k\dnorm{\CX  - \CM}
    \end{align*}
    as claimed.
\end{proof}


\section{Proof of Theorem \ref{thm:efficient_t_designs}}\label{sec:main_thm_proof}

In this section we prove theorem \ref{thm:efficient_t_designs} (restated), which integrates Lemmas \ref{lem:moments_implies_Haar}, \ref{lem:clt_spectrum_small_moments}, and \ref{lem:clt_basis}. 
\efficienttdesigns*
\begin{proof}
    Let $\vU_1 \dotsto \vU_m \in \unitary(N)$ and $\vU^\prime_1 \dotsto \vU^\prime_m \in \unitary(N)$ be i.i.d. Haar random unitaries, so that 

    $$ \vW :=  e^{i \frac{\theta}{\sqrt{m}} \sum_{j=1}^m \vU_{j} \vD \vU_{j}^\dag}\cdot e^{i \frac{\theta}{\sqrt{m}} \sum_{j=1}^m \vU^{\prime}_{j}\vD\vU^{\prime \dag}_{j}} $$
    Let $\CW$ and $\tilde{\CW}$ be the unitary channels which act on a density matrix $\rho$ by conjugation,
        \begin{align*}
            \CW : \rho \mapsto \vW \rho \vW^\dag \quad \text{and} \quad\tilde{\CW} : \rho \mapsto \vW_2 \rho \vW_2^\dag.
        \end{align*}
    We prove the theorem in the following steps,
    \begin{align*}
        \vW_2 \overset{\text{(i)}}{\approx} \vW \overset{\text{(ii)}}{\approx}\vU_{Haar} \vLambda \vU^{\dagger}_{Haar}\overset{\text{(iii)}}{\approx} \vU_{Haar} 
    \end{align*}
    where the approximations are in adaptive or parallel $T$-query distance.
    \begin{enumerate}[label=(\roman*)]
        \item By Lemma \ref{lem:clt_basis}, the tensor moments of the channels $\vW_2$ and $\vW$ are close in diamond distance,
            \begin{align*}
                \dnorm{\bexpect{ \tilde{\CW}^{\otimes T} -  \CW^{\otimes T} } } \leq 
                \left[ 
                     \epsilon_q 2^{8nq} \left(\frac{\theta^2 T^2}{m} \right)^{q+1}
                     + 2^{q+2}  \ltup{\frac{\theta T}{\sqrt{m}}}^{q+1}
                     \right] \frac{ 4 m  \, e^q}{q^q}.
            \end{align*}
            Let $1>\delta =\CO(1)$ be a small constant. It is possible to achieve, 
            \begin{align}
                \dnorm{\bexpect{ \tilde{\CW}^{\otimes T} -  \CW^{\otimes T} } } \leq \delta/2 \label{eq:main_thm_cond1}
            \end{align}
            for some choice of parameters $q=O(\log T)$ and $m=O(T^2)$ (recall that $\theta =\CO(1))$ such that 
            \begin{align*}
                \ltup{\frac{T}{\sqrt{m}}}^{q+1}  m   \ll 1. 
            \end{align*}
            The choice of mildly growing $q$ ensures that the additional factor of $m$ can be countered 
            \item
            Define the moments $\alpha_k := \btr[\vW^k]$, and the moment vector $\vec{\alpha_T} := (\alpha_1 \dotsto \alpha_T)$. We show that the moment vector $\norm{\vec{\alpha_T}}$ is small enough to be indistinguishable from Haar moments. 
            By Lemma \ref{lem:clt_spectrum_small_moments}, 
            the normalized trace moments $\alpha_k$ are small in expectation, 
            \begin{align*}
                 \abs{\Expect \alpha_k} &\leq  
                 O\ltup{\frac{4k(C\theta )^{q+1}}{(\sqrt{m})^{q-1} (q+1)!}}
                + \tilde{O}\left( \frac{T^4}{\sqrt{N}} \right)
                 \quad\quad \forall~ 1\leq k\leq T.
            \end{align*}
            By \autoref{cor:moment_concentration_UDUVDV}, the trace moments $\Expect \alpha_k$ concentrate as,
            \begin{align}
                \Pr \lbr{\abs{\alpha_k - \BE \alpha_k } > \delta_2 } \leq  2 \exp(-\frac{N^2 \delta_2^2 }{ 4 \cdot 96 \, k^2 \theta^2  }) \label{eq:main_thm_exp_alpha}
            \end{align}
            where we use $\lnorm{\vD}_{op}^2 < 4$. It follows that the magnitudes of $\alpha_k$ are small,   
            \begin{align*}
                \Expect \abs{\alpha_k}  &=  \Expect \lbr{\abs{\alpha_k} \Big| \abs{\alpha_k -\BE \alpha_k  } > \delta_2} \cdot \Pr[\abs{\alpha_k - \BE \alpha_k } > \delta_2]   + \Expect \lbr{\abs{\alpha_k} \Big| \abs{\alpha_k - \BE \alpha_k } \leq \delta_2} \cdot  \Pr[\abs{\alpha_k - \BE \alpha_k } \leq \delta_2 ] \\
                &\leq \ltup{1-2 \exp(-\frac{N^2 \delta_2^2 }{ 384 \, k^2 \theta^2 } )}(\abs{\BE \alpha_k }+\delta_2) + 2 \exp(-\frac{N^2 \delta_2^2 }{ 384 \, k^2 \theta^2  }) \tag{note $\sup \abs{\alpha_k} \leq 1$} \\
                &\leq \ltup{1-2 \exp(-\frac{N^2 \delta_2^2 }{ 384 \, k^2 \theta^2 } )} \ltup{O\ltup{\frac{4k(C\theta )^{q+1}}{(\sqrt{m})^{q-1} (q+1)!}}
                + \tilde{O}\left( \frac{T^4}{\sqrt{N}} \right) +\delta_2} + 2 \exp(-\frac{N^2 \delta_2^2 }{ 384 \, k^2 \theta^2  }). \tag{using \eqref{eq:main_thm_exp_alpha}} \\
                &=: A_{k,\delta_2}.
            \end{align*}
            where subscripts indicate the dependencies on $k$ and $\delta_2$. Define the moment vector for $T$ moments of $\vW$ as $\vec{\alpha}_T := (\alpha_1 \dotsto \alpha_T) \in \BC^T$. By the same reasoning as \autoref{cor:moment_concentration_UDUVDV}, the $\ell_1$ norm of $\vec{\alpha}_T$ is bounded as,
            \begin{align*}
            \Pr\lbr{\abs{\norm{\vec{\alpha}_T}_1 - \Expect \norm{\vec{\alpha}_T}_1} > \delta_3} 
            \leq 2 \exp(-\frac{N^2 \delta_3^2 }{ 384 \, T^3 \theta^2 } ).
            \end{align*}
            It follows that  $\norm{\vec{\alpha}_T}_1$ is bounded above by,
            \begin{align*}
            \Pr \lbr{\norm{\vec{\alpha}_T}_1 \ge \sum_k A_{k,\delta_2} + \sqrt{\gamma^2\frac{ 384 \, T^3 \theta^2 }{N^2 } + \ln 2 } } \leq  \exp(-\gamma^2).
            \end{align*}
            \item Now that we have control on the moment vector $\vec{\alpha_T}$ we may employ Lemma \ref{lem:moments_implies_Haar} (small moments imply Haar). We choose $\gamma = \sqrt{\frac{4}{\delta}}$ so that the probability of a large moment vector is bounded above by $\frac{\delta}{4}$. Otherwise, the norm of the moment vector is bounded above by 
            $$ B_{T,\delta_2} :=  \sum_{k=1}^T A_{k,\delta_2} + \sqrt{\log \frac{4}{\delta} \cdot \frac{ 384 \, T^3 \theta^2 }{N^2 } + \ln 2} \geq \norm{\vec{\alpha}_T}_1 . $$
            In a regime of $T \leq 2^{\tilde{O}(n/ \log n)}$ with properly chosen constants, and for some choice of parameters $q=O(\log T)$ $m=O(T^2)$, and $\delta_2>0$, it is possible to achieve $A_k =\CO(\frac{1}{T^{4}})$ and thus, 
            \begin{align}
               B_{T,\delta_2} \cdot 32T^{7/2} \leq \delta/4. \label{eq:main_thm_cond2}
            \end{align}
            Moreover, conditions \ref{eq:main_thm_cond1} and \ref{eq:main_thm_cond2} impose complementary constraints such that there exist $q=O(\log T)$ and $m=O(T^2)$ which satisfy both conditions.
            Lastly, recall that $\vW$ has a Haar random basis. By Lemma \ref{lem:moments_implies_Haar}, no $T$-query quantum decision algorithm $\CA$ can distinguish $\vW$ from a Haar random unitary $\vU$ with probability more than $B_{T,\delta_2} \cdot 32  T^{7/2} \leq \delta/4$. As such, the total probability of distinguishing $\vW$ from a Haar random unitary is at most $\delta/2$.
    \end{enumerate}
    It follows from (i) and (ii/iii) that we can set $1>\delta =\CO(1)$ such that  no $T$-query quantum decision algorithm $\CA$ can distinguish $\vW_2$ from Haar with probability more than $\delta$.
\end{proof}

\section{Implementation}\label{sec:construction}
Theorem \ref{thm:efficient_t_designs} provides a \emph{description} of unitary $T$-designs from i.i.d. unitary $q$-designs $\vU^{\prime}_j$ and diagonal matrices $\vD$ matching the low trace moments of GUE. In this section, we provide an explicit \emph{implementation} of $\vW_2$ via standard quantum algorithm primitives (Section~\ref{sec:standardQA}) with controllable error. Our algorithm is implemented in the Quantum Singular Value Transform (QSVT) framework \cite{gilyen2019quantum}. 
Please refer to Appendix \ref{sec:standardQA} for basic definitions and standard results within in this framework.

Our algorithm gives an efficient block encoding of the desired unitary  
\[\vW_2 = e^{i  \frac{\theta}{\sqrt{m}}  \sum_{j =1}^m \Tilde{\vU}_j \vD \Tilde{\vU}_j^{\dag}  } e^{i  \frac{\theta}{\sqrt{m}}  \sum_{j=1}^m \Tilde{\vU}_j^\prime \vD \Tilde{\vU}_j^{\prime \dag}}. \]

\begin{restatable}[Algorithm]{lem}{algorithm}\label{lem:algorithm_W}
    Suppose that  $\vD \in \gl(N)$ is a diagonal matrix that approximates the first $q$ moments of Wigner's semicircle:  
    \begin{align*}
        \abs{\btr(\vD^k) - \int x^k \rho_{sc}(x)\, \diff x } &\leq 2^k \cdot  \frac{2q+4}{N}, \quad\quad \forall 1\leq k \leq q, \\
        \norm{\vD}_{op} &\leq 2.
    \end{align*}
    Additionally, suppose the sequences $\Tilde{\vU}_1, \dotsto, \Tilde{\vU}_m \in \unitary(N)$ and $\Tilde{\vU}_1^{\prime}, \dotsto, \Tilde{\vU}_m^{\prime} \in \unitary(N)$ are each sequences of $\Tilde{O}(T)$-wise independent $\epsilon_{\scaleto{q}{5pt}}$-approximate unitary $q$-designs (where the sequences are independent or eachother). Let 
    \begin{align*}
        \vW_{alg} := e^{i  \frac{\theta}{\sqrt{m}}  \sum_{j =1}^m \Tilde{\vU}_j \vD \Tilde{\vU}_j^{\dag}  } e^{i  \frac{\theta}{\sqrt{m}}  \sum_{j =1}^m \Tilde{\vU}_j^\prime \vD \Tilde{\vU}_j^{\prime \dag}  }
    \end{align*}
    For $N> 2(q+2)$, $m =\CO(T^2 )$, $\theta=O(1)$ and $\epsilon_q=2^{-O(nq)}$ it is possible to implement a $(1, a + 2, \epsilon$-block-encoding of the unitary ensemble $\vW_{alg}$ with 
    \begin{itemize}
        \item $\Tilde{O}\ltup{n^2  T^2 \theta^2  + \log^{2.5} \frac{2\sqrt{m}}{\epsilon}  + \frac{\ln(1/\epsilon)}{\ln \ltup{\e + \ln(1/\epsilon)/(n^2 T^2 \theta^2  + \log^{2.5} \frac{T}{\epsilon}) } }} $ local gates, 
        \item $\CO(n^2 \polylog(T) + \log^{2.5} \frac{T}{\epsilon})$  ancilla qubits, and 
        \item $\tilde{O}(n^2 T )$ random bits,
    \end{itemize} 
    where $m =\CO(T^2)$ and $q =\CO(\log T)$.
\end{restatable}
In other words, this gives an efficient approximate unitary implementation of our $T$-design. Let $\vH_j^{(\prime)} := \Tilde{\vU}_j^{(\prime)} \vD \Tilde{\vU}_j^{(\prime), \dag}$. The quantum algorithm for $\vW_{alg}$ breaks down into the following steps.
\begin{enumerate}
    \item Construct block encodings of $\vH_j$ by (Lemma \ref{lem:construct_diag}).
    \item Construct block encodings of the linear combinations $\vH := \frac{1}{\sqrt{m}} \sum_{j=1}^m \vH_j$ and $\vH^\prime := \frac{1}{\sqrt{m}} \sum_{j=1}^m \vH^\prime_j$(by Lemma \ref{lem:linear_comb}).
    \item Hamiltonian simulation of $\vH$ followed by  $\vH^\prime$ (by Lemma \ref{lem:h_sim}).
\end{enumerate}
Similarly to LCU, the cost of step 2 depends on the $\ell_1$ norm of the weight vector $y$ and the select operator $\vS$. The state preparation cost is $\CO(\sqrt{m})$, but the cost of the select operator is naively $\CO(m)$. 
Careful inspection of the amount of randomness actually required in the proof of Theorem \ref{thm:efficient_t_designs} reveals that it is possible to construct a select operation over these specific $m$ unitary operators in $\CO(\sqrt{m})$ time where $T^2 = m$. The key insight into this is that the terms $\vH_j$ need only be \emph{$\Tilde{O}(T)$-wise independent} for Lemma \ref{lem:clt_basis} to hold, and thus it is possible to first plug the select index through a $\Tilde{O}(T)$-wise independent hash function to cut down costs. This idea will become clearer after formalizing the building blocks $\vU^{\prime}_j$ and $\vD$ of $\vW_{alg}$. Thus far, we haven't given any explicit description or construction for these operators, but in order to construct the select operator, it will be necessary to do so. 

We already know how to efficiently construct unitary designs that match low moments \cite{brandao2016local,harrow2023approximate,haferkamp2022random}. 
In the following, we show how to construct diagonal matrices $\vD$ by delving into the spectrum of GUE. Then, we return to constructing the select operator $\vS$ efficiently, which leads to a construction of $\vW_{alg}$.

\subsection{Matching Semicircular Moments} 
This section constructs the diagonal operators $\vD$. The main observation is that the limiting spectral distribution of the GUE is Wigner's semicircle (see \autoref{thm:semicirclelaw}), and for finite $N$, its empirical spectral distribution concentrates around this well-defined limit. This motivates us to construct $\vD$ with fixed spectrum which approximates Wigner's semicircle. The divergent behavior of the derivatives at the edges of the semicircle might raise concerns about the approximation quality. However, it is only necessary that the first $q$  moments of $\vD$ approximate that of GUE. We will find that this relaxation avoids the sharp edges of the semicircle, and yields a very simple description of $\vD$. In the construction of $\vW_{alg}$, we will need to access $\vD$ as a block encoding, so we also address how to specifically block encode $\vD$.

We aim to match only the first $1\leq k \leq q$ moments of the semicircle\footnote{We are solving a ``real''-valued moment problem instead of the unitary moment problem.} 
\begin{align}
    \int x^k \rho_{sc}(x)\, \diff x, \quad\quad \text{where } \rho_{sc}(x) := \frac{1}{2\pi} \sqrt{4 - x^2}  \label{eq:semicircle}
\end{align}
with the trace moments $\tr[\vD^k]$. However, in practice, we only need an efficiently implementable $\vD$ such that
\begin{align*}
    \labs{\tr(\vD^k) - \int x^k \frac{1}{2\pi} \sqrt{4 - x^2} \, \diff x} \leq \epsilon\quad \text{for each}\quad 1\leq k \leq q.
\end{align*}
This boils down to constructing an $N$-atomic integral measure 
\begin{align*}
    \mu = \frac{1}{N} \sum_{i=1}^N \mu_i \delta_{\xi_i}
\end{align*}
for values $\xi_i \in [-2,2]$ and integral weights $\mu_i \in \BN_0$. 
To do so, we first construct an $N$-atomic measure that matches the semicircle moments perfectly and then round each coefficient to yield an $N$-atomic integral measure. Starting from expression \eqref{eq:semicircle}, we make the change of variables $x= 2 \cos \theta$:
\begin{align}
\int 2^k\cos^k\theta \frac{1}{2\pi} \sqrt{4 - |2 \cos \theta |^2} \, (-2\sin \theta)\diff \theta  = - \frac{2^{k+2} }{2\pi} \int \cos^k\theta    \sin^2 \theta \, \diff \theta.\end{align}
By analyzing the Fourier domain, it is evident that sampling above twice the Nyquist frequency $\frac{2\pi}{k+2}$ recreates the signal. In other words, for any $L>2(k+2)$,
$$ \frac{1}{2\pi} \int \cos^k\theta    \sin^2 \theta \, \diff \theta =  \frac{1}{L} \sum_{l=1}^L \cos^k\frac{2\pi l}{L}    \sin^2 \frac{2\pi l}{L}. $$
Returning to our original variables, we have 
$$\int x^k \frac{1}{2\pi} \sqrt{4 - x^2} \, \diff x = -\frac{1}{L} \sum_{x \in \CS_L}  x^k (4 - x^2) $$
where $\CS_L := \lset{ \cos \frac{2\pi 1}{L}, \cos \frac{2\pi 2}{L} \dotsto \cos \frac{2\pi L}{L} }$. The corresponding $2k+4$-atomic measure is
\begin{align}
    \frac{1}{L} \sum_{x \in \CS_L} (4 - x^2) \delta_{x}. \label{eq:semicirc_mu}
\end{align}
Thus, as long as $L> 2(q+2)$, measure \eqref{eq:semicirc_mu} exactly matches the first $q$ moments of the  Wigner semicircle distribution. 

To get a diagonal matrix with an ESD that closely approximates $\rho_{sc}$, we round the weight $\frac{(4 - x^2)}{L}$ to its nearest integer multiple $\mu_x$ of $\frac{1}{N}$ for each $x\in \CS_L$, (barring some $\CO(1/N)$ adjustments such that the weights sum to one). Let $\vD$ be the diagonal matrix with entries $x \in \CS_L$ each with multiplicity $\mu_x$ for $x \in \CS_L$. 
Then the spectrum of $\vD$  approximates Wigner's semicircle.
\begin{lem}[Empirical semi-circle]\label{lem:approx_to_semi}
    If $L> 2(q+2)$, the matrix $\vD$ has an ESD which approximates the first $1\leq k \leq q$ moments of Wigner's semicircle as,
    \begin{align*}
        \abs{\btr{\vD^k}  - \int x^k \rho_{sc}(x) \, \diff x} \leq 2^k  \cdot \frac{2q+4}{N}.
    \end{align*}
\end{lem}
\begin{proof}
     \begin{align*}
         \abs{\btr{\vD^k}  - \int x^k \rho_{sc}(x) \, \diff x} 
         & = \abs{\sum_{x \in \CS_L} x^k \frac{\mu_x}{N}  - \int x^k \frac{1}{2\pi} \sqrt{4 - x^2} \, \diff x} \\
         &= \abs{\sum_{x \in \CS_L} x^k \frac{\mu_x}{N}  - \sum_{x \in \CS_L} x^k \frac{4 - x^2}{L}  } \\
         &\leq \sum_{x \in \CS_L} \abs{x^k} \abs{\frac{\mu_x}{N}  -  \frac{4 - x^2}{L}  } \\
         &\leq 2^k \cdot  \frac{2q+4}{N}.  
     \end{align*}
\end{proof}

Note that a lookup table for the diagonal entries of $\vD$ can be queried in $\CO(\log q)$ time with a binary search.
By standard quantum algorithms for sparse matrices (see Lemma \ref{lem:construct_diag}), it is possible to construct block encodings of $\vD$.
\begin{lem}[Block encodings of rounded semicircular spectrum]\label{lem:block_encoding_semicircle}
    For any $\epsilon_D>0$, It is possible to implement a 
    $(1,n+3,\epsilon_{\scaleto{D}{4pt}})$-block encoding of $\vD$ with $\CO(\log q + n+ \log^{2.5} \frac{1}{\epsilon_{\scaleto{D}{4pt}}})$ one and two qubit gates, and $\CO(n+\log^{2.5} \frac{1}{\epsilon_{\scaleto{D}{4pt}}})$ ancilla qubits.
\end{lem}
 
\subsection{The Select Operator}\label{sec:the_select_op}

Constructing linear combinations of block encodings requires a \textit{select operator} for the terms in the sum. Explicitly, for block encodings $\vB_{j}$ of $\vH_j$ for $j \in 1,\dots , m$, a select operator is a unitary
\begin{align*}
    \vS &= \sum_{j=0}^{m-1} \ketbra{j}{j}  \otimes \vB_{j} + \left( \Id - \sum_{j=0}^{m-1} \ketbra{j}{j} \otimes \Id \right).
\end{align*}
Recall that $\vH_j$ decomposes as $\vU_j \vD \vU_j^\dagger$ for unitaries $\vU_j$ and $\vD$, which approximately match the first $q$ moments of the basis and spectrum of GUE, respectively. Thus for i.i.d.  $q$-designs $\vU_j$, the Hamiltonian can be rewritten as $\vH = \sum_j \vU_j \vD \vU_j^\dag$, and selecting a term $\vH_j$ amounts to sampling a $q$-design $\vU_j$. The general cost for creating the select operator over $m$ terms scales linearly with $m$, and each independent unitary $q$-design can be constructed from random local gates chosen from a discrete universal gate set as shown by Haferkamp: 
\begin{lem}[{\cite[Corollary 1]{haferkamp2022random}}]\label{lem:haferkamp}
For $n \geq \lceil 2\log(4q) + 1.5 \sqrt{\log(4q)} \rceil$, local random quantum circuits drawn from discrete invertible gate sets with algebraic
entries are $\epsilon_{\scaleto{q}{5pt}}$-approximate unitary $q$-designs with $\CO(n q^{4+O(1)}( nq+ \log\frac{1}{\epsilon_{\scaleto{q}{5pt}}}))$ local gates.
\end{lem}
Let $\CG$ be such a gate set such that $|\CG| =\CO(1)$, and $l_q =\CO(n q^{4+O(1)}( nq+ \log\frac{1}{\epsilon_{\scaleto{q}{5pt}}}))$ be the number of steps achieving an $\epsilon_{\scaleto{q}{5pt}}$-approximate unitary $q$-design. 
Each step is represented by an element of $\CG$, and every instance of a random local circuit is represented by $\vec{r} \in \CL := \CG\times \cdots \times \CG = \CG^{l_q}$. The above line of thought would lead to a Hamiltonian simulation cost scaling as the $\ell_1$ norm of the LCU weight multiplied by the cost of $\vS$, which is $\sqrt{m} \times m \sim T^3$ in terms of the $T$ dependence. 

\subsection{Better gate complexity by replacing i.i.d. sums with correlated sums}
However, it is possible to do better by allowing correlation between the terms $\vH_j$ (and hence $\vU_j$). Careful inspection of the proof of Lemma \ref{lem:clt_basis} reveals that only $\Tilde{O}(T)$-wise independence is required for the sequence $\vH_1,\cdots ,\vH_j,\cdots, \vH_m$; indeed, we are merely after a $T$-design, not a $m$-design. More precisely, the $t$-fold exponential can be well-approximated by a Taylor series with a degree $\CO(T\log(1/\varepsilon))$ noncommutative polynomial of the unitaries (each monomial only contains $T\log(1/\varepsilon)$ many distinct $\vH_1,\cdots ,\vH_j,\cdots, \vH_m$.) In fact, this Taylor expansion argument is how one may design quantum simulation algorithms~\cite{berry2015simulating}. In our final algorithm, we only need to use $\varepsilon=O(1/T)$ because, in the end, we only target a constant diamond norm distance for the $T$-fold channel from Haar before boosting. Hence, the polynomial degree of the simulation algorithm is only $\tilde{O}(\log(T))$ for our implementation.
This implies the first $T$ moments of our simulation algorithm are at most $\tilde{O}(T)$ degree polynomials in the individual Hamiltonian terms.
As we only care about our algorithm's $T$-th moments to establish it as a $T$-design, this implies $\tilde{O}(T)$-wise independent $U_i$'s suffice for our application.

To construct a select operator over $m$ approximate GUE terms that are $\Tilde{O}(T)$-wise independent, we first apply a $\Tilde{O}(T)$-wise independent hash function on the select index, and then use the hash as the randomness parameter for approximate unitary $q$-designs. 

To give a careful account of the improved complexity, we recall existing properties and algorithms for $k$-wise independent hash functions. It was shown by Joffe that the following simple function family is $k$-wise independent:
\begin{lem}[$k$-wise independent hash families \cite{Joffe}]\label{lem:twise_hash}
    Let $\BF$ be a finite field. Define the family of functions $\CH = \{h_{\vec{a}} : \BF \to \BF\}$ where $\vec{a} = (a_0, a_1, \ldots, a_{k-1}) \in \BF^{k}$ and $h_{\vec{a}}(x) = a_0 + a_1x + a_2x^2 + \ldots + a_{k-1}x^{k-1}$. The family $\CH$ is k-wise independent. 
\end{lem}
As simple corollary is that $k$-wise independent hash families from $n^\prime$ to $n$ bits can be constructed efficiently:
\begin{cor}\label{cor:twise_hash}
For every $n, n^\prime, k \in \mathbb{N}$, there is a family of $k$-wise independent functions $\CH = \{h : \bit^{n^\prime} \rightarrow \bit^n \}$ such that choosing a random function from $\CH$ takes $k \cdot \max\{{n^\prime},n\}$ random bits, and evaluating a function from $\CH$ takes $\Tilde{O}(k , \max\{{n^\prime},n\} )$ time.
\end{cor}
\begin{proof}
    Using repeated squaring, $x^j$ can be calculated in $\Tilde{O}(\max\{{n^\prime},n\}) \log k)$ time. Summing over $k$, $n$-bit numbers requires $\Tilde{O}(k \max\{{n^\prime},n\} )$ gates. 
\end{proof}

Based on the above classical results, we hash the $m$ select values to ``random seeds'' which then parameterize $m$ many $\Tilde{O}(T)$-wise independent unitary $q$-designs. Let the family of functions  $\CH := \{h_{\vec{a}} : \lbr{m} \to\CL \}$ be an $\Tilde{O}(T)$-wise independent hash family as in \autoref{cor:twise_hash}. Define $\vF_{\vec{a}}$ as the unitary that acts on a computational basis state as $h_{\vec{a}},$
\begin{align*}
    \vF_{\vec{a}} \ltup{\ket{\vec{x}} \otimes \ket{\vec{y}}} \rightarrow \ket{\vec{x}} \otimes \ket{\vec{y} + h_{\vec{a}}(\vec{x}) }.
\end{align*}
Naturally, the inverse of $\vF_{\vec{a}}$ is given by  
\begin{align*}
    \vF^{\dag}_{\vec{a}} \ltup{\ket{\vec{x}} \otimes \ket{\vec{y}}  } \rightarrow \ket{\vec{x}} \otimes \ket{\vec{y} - h_{\vec{a}}(\vec{x}) }. 
\end{align*}
From \autoref{cor:twise_hash} it follows that the time cost of implementing $\vF_{\vec{a}}$ is 
$$\Tilde{O}\ltup{T\max\{\log m, \log |\CL|\} } = \Tilde{O}\ltup{T \max\lset{\log m,  n q^{4+O(1)}( nq+ \log\frac{1}{\epsilon_{\scaleto{q}{5pt}}})}}.$$
Assuming $m=O(T^2)$ and $q=O(\log T)$, the time cost is $ \Tilde{O}\ltup{T n (n+ \log\frac{1}{\epsilon_{\scaleto{q}{5pt}}}) }$ and sampling requires $T n (n+ \log\frac{1}{\epsilon_{\scaleto{q}{5pt}}})$ random bits.

Let $\vD_{\scaleto{BL}{4pt}}$ be a $(1,n+3,\epsilon_{\scaleto{D}{4pt}})$-block encoding of the sparse diagonal matrix $\vD$ whose spectrum which approximates the semicircle as in \autoref{lem:block_encoding_semicircle}. For a vector $\vec{r} \in \CL$, we denote by $\vU_{\vec{r}}$ the (parameterized) local circuit of $l_q$ steps over the gate set $\CG$ where each step is determined by an entry of the vector $\vec{r}$ in the obvious way. We then define $\vC_{\vL}$ as the controlled circuit that implements the parameterized local circuit $\vL_{\vec{r}}$ controlled on the vector $\vec{r}$.\footnote{We note the existence of such a universal circuit simulator is a standard result in quantum complexity theory; for example it is used in the proof that $\textsf{BQP}^\textsf{BQP} = \textsf{BQP}$.} 
That is, $\vC_{\vL}$ acts on the control register $\ket{\vec{r}}$ and arbitrary state $\ket{\psi}$ as follows,
\begin{align*}
    \vC_{\vL} \ltup{ \ket{\vec{r}} \otimes \ket{\psi} }\rightarrow \ket{\vec{r}} \otimes \vL_{\vec{r}} \ket{\psi}.
\end{align*}
Similarly, we define the inverse
\begin{align*}
    \vC_{\vL}^\dag \ltup{\ket{\vec{r}} \otimes \ket{\psi}} \rightarrow \ket{\vec{r}} \otimes \vL_{\vec{r}}^\dag \ket{\psi}.
\end{align*}
where $\vL_{\vec{r}}^\dag$ is simply implemented by applying the local gates of $\vL_{\vec{r}}$ in reverse. 

The select operator $\vS$ is constructed as follows.
Sample a function $h_{\vec{a}}$ from the hash family $\CH$. Then, on the select register $s$, state register $q$, and ancilla register $anc$, perform the following:
\begin{align*}
    \ket{j}_s \otimes \ket{\vec{\psi}}_q \otimes \ket{\vec{0}}_{anc} &\longrightarrow \ket{j}_s \otimes \ket{\vec{\psi}}_q \otimes \ket{h_{\vec{a}}(\vec{j})}_{} & \text{apply $\vF_{\vec{a}}$ from $s$ to $anc$} \\
    &\longrightarrow \ket{j}_s \otimes \vU^\dag_{\vec{r}}\ket{\vec{\psi}}_q \otimes \ket{\vec{r}}_{anc} & \text{where } \vec{r} := h_{\vec{a}}(j), \text{ apply $\vC^\dag_{\vL}$ from $anc$ to $q$} \\
    &\longrightarrow \ket{j}_s \otimes \vD\vU_{\vec{r}}\ket{\vec{\psi}}_q \otimes \ket{\vec{r}}_{anc} &\text{apply $\vD_{\scaleto{BL}{4pt}}$ to $q$ and $\CO(1)$ additional ancillas} \\
    &\longrightarrow \ket{j}_s \otimes \vU_{\vec{r}}\vD\vU^\dag_{\vec{r}}\ket{\vec{\psi}}_q \otimes \ket{\vec{r}}_{anc} & \text{apply $\vC_{\vL}$ from $anc$ to $q$} \\
    &\longrightarrow \ket{j}_s \otimes \vU_{\vec{r}} \vD\vU^\dag_{\vec{r}}\ket{\vec{\psi}}_q \otimes \ket{\vec{0}}_{anc}  & \text{apply $\vF^{\dag}_{\vL}$ from $s$ to $anc$}.
\end{align*}
Thus, $\vS$ is a $\CO(\log m + n + l_q)$ qubit unitary that can be implemented with $\Tilde{O}\ltup{T n (n+ \log\frac{1}{\epsilon_{\scaleto{q}{5pt}}}) + \log^{2.5} \frac{1}{\epsilon_{\scaleto{D}{4pt}}} }$ local gates (improving from the naive $T^2$ scaling in Section~\ref{sec:the_select_op}) and $\CO(l_q+n)$ ancilla qubits. 

\subsection{Proof of Lemma \ref{lem:algorithm_W}}
\begin{proof}
    
By Lemma \ref{lem:block_encoding_semicircle}, it is possible to construct a $(1,n+3,\epsilon_{\scaleto{D}{4pt}})$-block encoding $\vD_{\scaleto{BL}{4pt}}$  of $\vD$ with $\CO(\log q + n+ \log^{2.5} \frac{1}{\epsilon_{\scaleto{D}{4pt}}})$ one and two qubit gates, and $\CO(n+\log^{2.5} \frac{1}{\epsilon_{\scaleto{D}{4pt}}})$ ancilla qubits. It follows that $\vB_{\vec{r}}:=  (\vU_{\vec{r}} \otimes \Id) \vD_{\scaleto{BL}{4pt}} (\vU_{\vec{r}}^\dag \otimes \Id) $ is a $(1,n+3,\epsilon_{\scaleto{D}{4pt}})$-block encoding of $\vH_{\vec{r}}:=\vU_{\vec{r}} \vD \vU_{\vec{r}}^\dag$. 

To implement a block encoding of $\vH_{alg} := \frac{1}{\sqrt{m}} \sum_{\vec{r}} \vH_{\vec{r}}$, we require a select operator $\vS$, as described in the previous section, and as well as a state preparation pair $(\vP_L, \vP_R)$ for $\vec{y} := (\frac{1}{\sqrt{m}} \dotsto \frac{1}{\sqrt{m}})$. It suffices to let $\vP_L$ and $\vP_R$ be the unitaries which take $\ket{0}$ to the uniform superposition over basis states. Since $\norm{\vec{y}}_1 \leq  \sqrt{m} $, the unitaries $(\vP_L, \vP_R)$ are a $(\sqrt{m}, \log m, 0)$-state preparation pair. By Lemma \ref{lem:linear_comb} it is possible to implement a $(\sqrt{m}, n + \log m + 3, \sqrt{m}\epsilon_{\scaleto{D}{4pt}})$-block-encoding $\vB_{alg}$ of $\vH_{alg} := \frac{1}{\sqrt{m}} \sum_{\vec{r}} \vH_{\vec{r}}$, with a single use of $\vS$, $\vP_R$, and $\vP_L^\dagger$. This requires $\Tilde{O}\ltup{T n (n+ \log\frac{1}{\epsilon_{\scaleto{q}{5pt}}}) + \log^{2.5} \frac{1}{\epsilon_{\scaleto{D}{5pt}}}}$
local gates and $\CO(l_q+\log^{2.5} \frac{1}{\epsilon_{\scaleto{D}{4pt}}})$ ancilla qubits.

By Lemma \ref{lem:h_sim}, it is possible to implement an $\epsilon_{\scaleto{H}{3pt}}$-precise Hamiltonian simulation unitary that is an $(1, n + 2, \epsilon_{\scaleto{H}{3pt}} + \sqrt{m}\epsilon_{\scaleto{D}{4pt}})$-block-encoding of $\e^{\ri \theta \vH_{alg}}$,  with $3 \gamma \ltup{ \e\sqrt{m}\theta /2, \frac{\epsilon_{\scaleto{H}{3pt}}}{6}}$ uses of $\vB_{alg}$ or its inverse, three uses of controlled-$\vB_{alg}$ or its inverse, and with $O\left(n \gamma \ltup{ \e\sqrt{m} \theta/2, \frac{\epsilon_{\scaleto{H}{3pt}}}{6}}\right)$ two-qubit gates and using $\CO(1)$ ancilla qubits. This amounts to $O\ltup{\gamma\ltup{\sqrt{m} \theta^2 \Tilde{O}\ltup{T  n (n+ \log\frac{1}{\epsilon_{\scaleto{q}{5pt}}}) + \log^{2.5} \frac{1}{\epsilon_{\scaleto{D}{4pt}}}} }, \frac{\epsilon_{\scaleto{H}{3pt}}}{6} }$ local gates and 
$\CO(n \polylog(T)(n+\log\frac{1}{\epsilon_{\scaleto{q}{5pt}}})+\log^{2.5} \frac{1}{\epsilon_{\scaleto{D}{4pt}}})$ ancilla qubits. 

By Lemma \ref{lem:hsim_error_func}, the gate count is roughly 
$$ \Tilde{O}\ltup{T^2 \theta^2  n (n+ \log\frac{1}{\epsilon_{\scaleto{q}{5pt}}}) + \log^{2.5} \frac{1}{\epsilon_{\scaleto{D}{4pt}}}  + \frac{\ln(1/\epsilon_{\scaleto{H}{3pt}})}{\ln\ltup{\e + \ln(1/\epsilon_{\scaleto{H}{3pt}})/(T^2 n^2 \theta^2 \log \frac{1}{\epsilon_{\scaleto{q}{5pt}}}+ \log^{2.5} \frac{1}{\epsilon_{\scaleto{D}{4pt}}}) } } }.$$ 
Additionally, by Lemma \ref{lem:twise_hash}, this requires $\tilde{O}(T n (n+ \log\frac{1}{\epsilon_{\scaleto{q}{5pt}}}))$ random bits. We note that Theorem \ref{thm:efficient_t_designs} requires that the ``building block'' $\epsilon_{\scaleto{q}{4pt}}$-approximate unitary $q$-designs have extremely high precision $\epsilon_{\scaleto{q}{4pt}}$. 
Thus, setting 
\begin{align*}
 \epsilon_{\scaleto{q}{4pt}} =\CO(2^{-qn})\quad \text{and}\quad \epsilon/2= \epsilon_{\scaleto{H}{3pt}} = \sqrt{m}\epsilon_{\scaleto{D}{4pt}} ,
\end{align*} 
it is possible to implement a $(1, n + 2, \epsilon)$-block-encoding of $\e^{\ri \theta \vH_{alg}}$ with the following resource costs,  
\begin{itemize}
        \item $\Tilde{O}\ltup{n^2  T^2 \theta^2   + \log^{2.5} \frac{T}{\epsilon}  + \frac{\ln(1/\epsilon)}{\ln\ltup{\e + \ln(1/\epsilon)/(n^2 T^2 \theta^2  + \log^{2.5} \frac{T}{\epsilon}) } }} $ local gates, 
        \item $\CO(n^2 \polylog(T) + \log^{2.5} \frac{T}{\epsilon})$  ancilla qubits, and 
        \item $\tilde{O}(T n^2) $ random bits.
    \end{itemize}

Finally, recall that we compose two such unitary ensembles to implement $\vW_{alg}$. This incurs twice the cost of resources required in the single case. 

\end{proof}

\subsection{Error Suppression via Composition}\label{subsec:boosting}

Previously, in our construction, we leveraged \emph{sums} to get convergence results, but the error only decreases inversely polynomially with the algorithmic costs. Not too surprisingly, taking iterative products can exponentially reduce the error bounds at the cost of repeating the unitary \cite{o2023explicit}.
\begin{lem}
    Let $\tau$ be a representation of the unitary group, $\mu$ be the Haar measure over the unitary group, and $\vD$ be some other distribution over the unitary group. We define $\CN$ to be the channel induced by $\tau$ over $\mu$ and similarly $\Tilde{\CN}$ to be that of $\tau$ over $\vD$. That is
    \begin{align*}
        \CN(\rho) := \mexpect\limits_{U \leftarrow \mu} \tau(U)~ \rho~ \tau(U)^\dagger \\
        \Tilde{\CN}(\rho) := \mexpect\limits_{U \leftarrow D} \tau(U)~ \rho~ \tau(U)^\dagger
    \end{align*}
    Suppose that 
    $\dnorm{\Tilde{\CN} - \CN} \leq \epsilon.$
    Then repeating the $\Tilde{\CN}$ channel $k$ times yields an $\epsilon^k$-approximate unitary channel, that is, $\dnorm{\Tilde{\CN}^k - \CN} \leq \epsilon^k.$
\end{lem}
\begin{proof}
    Due to left and right invariance of the Haar measure over the unitary group, 
    $\dnorm{(\Tilde{\CN} - \CN)^k} = \dnorm{\Tilde{\CN}^k - \CN}.$ By sub-multiplicativity of the diamond norm, $\dnorm{(\CN - \Tilde{\CN})^k} \leq \epsilon^k$, proving the claim.
\end{proof}

\section{Discussion} \label{sec:discussion}

At the end of what has been a rather arduous path for the authors and perhaps the readers as well, we have succeeded in establishing our main result, the specification and analysis of a quantum algorithm for generating $\epsilon$-approximate unitary $T$-designs on $n$ qubits using $\tilde{O} \left(T^2 n^2 \log(\epsilon^{-1})\right)$ gates and $\tilde{O}(Tn^2)$ bits of classical randomness.

While that end goal was a powerful motivator, technical difficulties involved in combining the elements of the proof unfortunately obscure the beauty of the underlying ideas. Indeed, our research collaboration was studying different possible applications of the ideas and settled on unitary $T$-designs as a first illustration of their power and utility before realizing that deploying them in that context would introduce an array of distracting complications. And so, with the proof complete, it is time to highlight those essential ideas free of the need to connect them into a larger logical structure.

Briefly, they are:
\begin{itemize}
    \item The product of two exponentiated GUE matrices approximates Haar. 
    
    A single exponentiated GUE matrix $e^{i\theta \vG}$ is far from Haar-distributed. One would therefore expect to need to compose many independent such matrices in order to start to approximate the Haar measure but, almost miraculously, that isn't necessary. A product of two suffices, but will approximate the Haar measure only for very specific choices of $\theta$.
    
    \item The polynomial method can be used to bound finite-$N$ random matrix quantities. 

    In quantum information theory, the polynomial method is a standard tool for bounding query complexity, but the method originates in approximation theory, and there is no reason to restrict its application to query bounds. In particular, we have developed a strategy for applying it quite generally to the task of finding rigorous bounds on the deviation of finite-$N$ random matrix quantities from their infinite-$N$ limits.

    Constructing the necessary rational function from the random matrix quantity is a nontrivial step. Our solution was to reduce the task to a natural but, at least to our knowledge, previously unstudied moment problem. Historically, the moment problem was motivated by statistics: given a sequence of moments evaluated empirically, the task was to find a probability density consistent with those moments. The task in our case is more specific: to recover the moments using \textit{other} empirical probability densities consisting of a prescribed number of equally weighted atomic measures. That is, the question was whether moments tabulated in one experiment could be exactly reproduced in another experiment with a different number of samples. For probability densities on $S^1$, the so-called trigonometric or unitary case, we found that they can, under quite weak assumptions.
    \item Matrix central limit theorems ensure convergence to the GUE from sums of simple matrices.

    This is not a new idea, of course. A weak version of the statement follows immediately from applying the central limit theorem entrywise to the matrix. In quantum computing, however, it is more straightforward to multiply matrices than to add them. Perhaps for that reason, convergence results for sums of random matrices had not previously been applied to the construction of pseudorandom unitary transformations.
\end{itemize}

Combining these elements with a kitchen sink of other techniques, including Lindeberg replacement, $k$-wise independent hash functions, the quantum singular value transformation, measure concentration, and the Weingarten calculus eventually yielded our sought-after unitary $T$-designs. We hope and expect, however, that the individual components will be applied elsewhere, whether to prove query bounds in quantum complexity, to prove new rigorous results in finite-$N$ random matrix theory, or to address other problems in quantum computing.

\section*{Acknowledgments}

We thank Jorge Garza-Vargas, Jeongwan Haah, Isaac Kim, Yunchao Liu, Mark Meckes, Vaughan McDonald, Tselil Schramm, Douglas Stanford, Stanislaw Szarek, Joel Tropp, Cynthia Yan, Shunyu Yao, Li-Yang Tan and Ryan O'Donnell for stimulating conversations.
CF.C. thanks Ching-Ting Tsai for offering his sofa for sleep during his consecutive visits to Stanford.
A.B. and J.D. were supported in part by the AFOSR under grant FA9550-21-1-0392 and by the U.S. DOE Office of Science under Award Number DE-SC0020266.
A.B. was supported in part by the DOE QuantISED grant DE-SC0020360.
J.D. was supported in part by a Shoucheng Zhang graduate fellowship.
P.H. acknowledges support from AFOSR (award FA9550-19-1-0369), DOE (Q-NEXT), CIFAR and the Simons Foundation.
M.X. acknowledges support from a NSF Graduate Research Fellowship.


\appendix

\section{Standard Quantum Algorithm Notions}\label{sec:standardQA}
In this section, we recall some standard quantum algorithmic notions useful for our implementation.
\begin{defn}[Block-encoding]
    Suppose that $A$ is an $s$-qubit operator, $\alpha, \epsilon \in \mathbb{R}^+$, and $a \in \mathbb{N}$. Then we say that the $(s + a)$-qubit unitary $U$ is an $(\alpha, a, \epsilon)$-block-encoding of $A$ if
    \[
    \lnorm{A - \alpha \left(\bra{0}^{\otimes a} \otimes \Id \right)U \left(|0\rangle^{\otimes a}\otimes  \Id \right)}_{op} \leq \epsilon.
    \]
\end{defn}
\begin{defn}[State Preparation Pair]
    Let $y \in \mathbb{C}^m$ and $\|y\|_1 \leq \beta$. The pair of unitaries $(\vP_L, \vP_R)$ is called a $(\beta, b, \epsilon)$-state preparation pair if  $\vP_L |0\rangle^{\otimes b} = \sum_{j=0}^{2^b-1} c_j |j\rangle$ and $\vP_R |0\rangle^{\otimes b} = \sum_{j=1}^{2^b-1} d_j |j\rangle$ such that $\sum_{j=0}^{m-1} | \beta(c_j^* d_j) - y_j | \leq \epsilon$, and for all $j \in \{m, \ldots, 2^b - 1\}$ we have $c_j^* d_j = 0$.
\end{defn}
\begin{lem}{Lemma 59 \cite{gilyen2019quantum}} \label{lem:hsim_error_func}
    Let $\gamma(t,\epsilon) \geq t $ be the solution to the equation 
    $$ \epsilon = \ltup{\frac{\gamma}{t}}^\gamma ~:~ \gamma\in (t, \infty).$$
    For $t \in \mathbb{R}^+$ and $\epsilon \in (0, 1)$,
    \[\gamma(t, \epsilon) = \Theta\left( t + \frac{\ln(1/\epsilon)}{\ln(e + \ln(1/\epsilon)/t)} \right).\]
    Moreover, for all $q \in \mathbb{R}^+$,
    \[\gamma(t, \epsilon) < e^{q\left( t + \frac{\ln(1/\epsilon)}{q} \right)}.\]
    
\end{lem}

Our implementation relies on the three following results on constructing and manipulating block encodings restated from \cite{gilyen2019quantum}:
\begin{lem}[Lemma 48 \cite{gilyen2019quantum}]\label{lem:construct_diag}
    Let $\vA \in \BC^{2^n \times 2^n}$ be a diagonal matrix. Suppose we have entry-wise access to $\vA$ via an oracle $O_A$,
    $$ O_A : \ket{i}\ket{j} \ket{0}^{\otimes b} \rightarrow \ket{i}\ket{j} \ket{a_{ij}} \quad 1 \leq i,j \leq 2^n $$
    where $a_{ij}$ is a $b$-bit binary description of the entry of $A$ in the $i^{th}$ row and $j^{th}$ column. Then it is possible to implement a $(1,n+3,\epsilon) $-block encoding of $\vA$ with two queries to $O_A$, $\CO(n+ \log^{2.5} \frac{1}{\epsilon}) $ one and two qubit gates, and $\CO(b + \log^{2.5} \frac{1}{\epsilon})$ ancilla qubits.
\end{lem}
\begin{lem}[Lemma 52 \cite{gilyen2019quantum}]\label{lem:linear_comb}
    Let $\vA = \sum_{j=1}^{m} y_j \vA_j$ be an $s$-qubit operator and $\epsilon \in \mathbb{R}^+$. Suppose that $(\vP_L, \vP_R)$ is a $(\beta, b, \epsilon_1)$-state-preparation-pair for $y$, $\vS = \sum_{j=0}^{m-1} \ketbra{j}{j} \otimes U_j + \left((I - \sum_{j=0}^{m-1} \ketbra{j}{j}) \otimes I_a \otimes I_s\right)$ is an $s + a + b$ qubit unitary such that for all $j \in 0, \ldots, m$ we have that $\vU_j$ is an $(\alpha, a, \epsilon_2)$-block-encoding of $\vA_j$. Then we can implement a $(\alpha\beta, a + b, \alpha\epsilon_1 + \alpha\beta\epsilon_2)$-block-encoding of $A$, with a single use of $\vS$, $\vP_R$, and $\vP_L^\dagger$.
\end{lem}
\begin{lem}[Theorem 58 \cite{gilyen2019quantum}]\label{lem:h_sim}
    Let $t \in \mathbb{R} \setminus \{0\}$, $\epsilon \in (0, 1)$, and let $\vU$ be an $(\alpha, a, 0)$-block-encoding of the Hamiltonian $\vH$. Then we can implement an $\epsilon$-precise Hamiltonian simulation unitary $V$ which is an $(1, a + 2, \epsilon)$-block-encoding of $e^{it\vH}$, with $3r \ltup{ \e\alpha|t|/2, \frac{\epsilon}{6}}$ uses of $\vU$ or its inverse, three uses of controlled-$U$ or its inverse, and with $O\left(ar \ltup{ \e\alpha|t|/2, \frac{\epsilon}{6} }\right)$ two-qubit gates and using $\CO(1)$ ancilla qubits. 
\end{lem}
\section{Properties of GUE}\label{app:GUE_properties}
Some properties of the GUE play important roles throughout. We collect them below, citing where the proof can be found for easy reference.

The GUE is defined to have a unitarily invariant eigenbasis---that is, the eigenbasis is distributed as Haar. As a class of Wigner matrices, its empirical spectral distribution converges to Wigner's semicircle law:
\begin{thm}[Semicircular law]\label{thm:semicirclelaw}
    Let $M_n$ be the top left $n \times n$ minors of an infinite Wigner matrix $(\xi_{ij})_{i,j \geq 1}$. Then the empirical spectral distributions (ESDs) $\rho \sqrt{\frac{1}{n}} M_n$ converge almost surely (and hence also in probability and in expectation) to the Wigner semicircular distribution
    \begin{equation}
        \label{eq:semicircular}
        \rho_{\text{sc}}(x) \, \diff x := \frac{1}{2\pi} \sqrt{4 - x^2} \, \diff x.
    \end{equation}
\end{thm}
\begin{prop}[Semicircular moments are Catalan]\label{prop:catalan} (See Theorem 1.6 in \cite{speicher2020lecture})
    The moments of the semicircle distribution $\rho_{sc}$ are given by
    \begin{align*}
        \int_{-2}^2 x^n \rho_{sc}(x)\, \diff x = 
        \begin{cases}
        0 & n \text{ odd} \\
        \mathrm{Cat}_n & n=2k \text{ even}
        \end{cases}
    \end{align*}
    where $\mathrm{Cat}_n$ are the Catalan numbers, as defined by $\mathrm{Cat}_n=\frac{1}{n+1}{\binom{2n}{n}}$.
\end{prop}
The Catalan numbers are also equivalent to the number of non-crossing pairings of $2n$ objects.
\begin{prop}[Closeness of ESD and semicircle distribution]\label{prop:esd_to_sc} (See Theorem 1.1 in \cite{gotze2005rate} for the first proof. See Theorem 1.6 in \cite{gotze2018local} for a generalization and alternate proof.)
    For all $N>1$, the empirical spectral distribution of a random $N$-dimensional matrix $\vG$ from the Gaussian Unitary Ensemble
    \begin{align*}
        \rho_N(E):=\frac{1}{N}\sum_{i=1}^N\delta(E-\lambda_i(\vG))
    \end{align*}
    satisfies 
    \begin{align*}
        \sup_{E\in\mathbb{R}} \left|\mathop{\BE}_{\text{GUE}}\int_{-\infty}^E \rho_{GUE}(E')\, \diff E'-\rho_{sc}(E')\, \diff E\right|\leq \frac{K}{N}
    \end{align*}
    for some constant $K=O(1).$
\end{prop}

In \cite{gotze2018local}, their statement of the proposition applies to more general Wigner matrix ensembles satisfying particular conditions, which include the GUE. This proposition is important in estimating finite $N$ corrections to our results.

\begin{prop}[GUE concentration of spectral radius]\label{prop:GUE_spectral_rad_conc} (See Equation 8 in \cite{aubrun2005sharp})
    Let $\lambda_{max}$ denote the maximum eigenvalue of a random matrix $\vG$ from the Gaussian Unitary Ensemble of dimension $N$. Then for all $N>1$ and $t>0$, there exists some constants $C>0$, $c>0$ such that
    \begin{align*}
        \Pr\left[\norm{\vG} \geq 2+t\right]\leq C\exp(-\frac{Nt^{3/2}}{c}).
    \end{align*}
\end{prop}
This concentration over the eigenvalue spectrum helps limit integrals of eigenvalues to around $[-2,2]$. 

\begin{prop}[Expected value of operator norm products]\label{prop:exp_value_op_norm_products}(See Proposition 2.11 in \cite{collins2019operator})
    Let $\vG$ be an $N$-dimensional matrix from the Gaussian Unitary Ensemble. There exist constants $C$, and $\alpha$ such that for any $k<\alpha N$,
\begin{align*}
    \Expect_{\vG} \norm{\vG}_{op}^k \leq C^k.
\end{align*}
\end{prop}

\begin{prop}[Trace concentration for GUE]\label{prop:trace_moment_conc_G}(Special case of Corollary 7 in \cite{meckes2012concentration}.)
    Let $\vG$ be an $N$-dimensional matrix from the Gaussian Unitary Ensemble.  There exists absolute constants $\kappa>0$ and $\kappa'>0$ such that for each $p\geq 1$, $t>0$,
\begin{align*}
    \Pr[\abs{\tr[\vG^p] - \mathop{\BE}_{\text{GUE}}\tr[\vG ^p]} \geq t ] \leq \kappa'(p+1)\exp(-\min\{\kappa^pt^2,\kappa Nt^{2/p}\}).
\end{align*}
\end{prop}
Note that the traces in the equation above are \emph{not} normalized.
\section{Concentration Inequality Proofs}\label{append:Proofs_concentration}
In this appendix, we collect proofs of concentration inequalities using off-the-shelf arguments. First is Lemma \ref{lem:theta_concentration}, which we restate here for convenience:
\begin{lem}[Trace concentration for $\e^{\ri \vG\theta}$]For a Gaussian Unitary Ensemble of $N$-dimensional matrices $\vG$, 
\begin{align}
    \Pr[\left|\btr[e^{i \vG  \theta}] - \mathop{\BE}_{\text{GUE}}\btr[e^{i \vG  \theta}]\right| \geq t ] &\leq \exp(-\frac{Nt^2}{2\theta^2})  \quad \text{and} \label{eq:gconc1} 
\end{align}
\end{lem}
\begin{proof}
We want to use the Gaussian concentration inequality \cite{boucheron2013concentration}
\begin{align}\label{eq:subgaussian_concentration}
    \Pr[|f(\vec{g}) - \mathop{\BE} f(\vec{g})|\geq t] \leq e^{-t^2/(2L^2)} 
\end{align}
for $L$-Lipschitz $f$ and real Gaussian vectors $\vec{g}$ normalized so that each coefficient is $N(0,1)$.
We can treat the matrix $\vG$ as a vector by writing
\begin{align*}
    \vG=\sum_i g_i \vA_i \quad \text{where}\quad g_i\stackrel{i.i.d.}{\sim}\mathcal{N}(0,1)
\end{align*}
where the $\vA_i$ are matrices that both normalize each Gaussian and place it into an entry of the matrix. There is one $g_i$ for each real entry of $\vG$; this means for an $N$-dimensional $\vG$, there are $N$ vector entries from the diagonal that are normalized with $1/\sqrt{N}$, $N(N-1)/2$ vector entries from the real part of (independent) off diagonals that are normalized with $1/\sqrt{2N}$, and $N(N-1)/2$ vector entries for the imaginary part of (independent) off diagonals that are normalized with $\ri/\sqrt{2N}$. This means $\vec{g}$ is $N^2$-dimensional. 

We then want to find the Lipschitz constant of the function
\begin{align*}
    f(\vec{g})=\btr[e^{i  \theta\sum g_i \vA_i}].
\end{align*}
We will do this in pieces. By Lemma \ref{lem:lipschitz_mat_exp},  we see the matrix function $\vX \rightarrow e^{i\vX\theta}$ is $\theta$-Lipschitz with respect to the operator norm. In order to arrive at the vector $\ell_2$-norm on $\vec{g}$, we can use a known technique (see, e.g.,~\cite{chen2023sparse}):
\begin{align*}
    \norm{\vX-\vY} = \norm{\sum_i (x_i-y_i)\vA_i} 
    &= \sup_{\ket{w},\ket{v}}\bra{w} \sum_i (x_i-y_i)\vA_i \ket{v} \\ 
    &\leq \left(\sum_i\sup_{\ket{w},\ket{v}}|\bra{w}\vA_i\ket{v}|^2\right)^{1/2}\norm{\vx-\vy}_{2}.
\end{align*}
The inequality follows from Cauchy-Schwarz. Recall the form of $\vA_i$ from above; then, for the GUE,
\begin{align*}
    \left(\sum_i\sup_{\ket{w},\ket{v}}|\bra{w}\vA_i\ket{v}|^2\right)^{1/2} = \frac{1}{\sqrt{N}}.
\end{align*}
Putting this together with the Lipschitz constant for $e^{i\vG\theta}$, we find
\begin{align*}
    \norm{e^{i\vX\theta}-e^{\ri\vY\theta}}\leq\frac{\theta}{\sqrt{N}}\norm{\vx-\vy}_{2}. 
\end{align*}
The last part is the (normalized) trace, which is a 1-Lipschitz function in $\vX$:
\begin{align*}
    |\btr[\vX]-\btr[\vY]| \leq \norm{\vX-\vY}.
\end{align*}
Now recall that compositions of functions have multiplicative Lipschitz constants. We conclude
\begin{align*}
    \left|\btr[e^{i  \theta\sum x_i \vA_i}]-\btr[e^{i  \theta\sum y_i \vA_i}]\right|\leq \frac{\theta}{\sqrt{N}}\norm{\vx-\vy}_{2}.
\end{align*}
We substitute $L=\theta/\sqrt{N}$ into equation \eqref{eq:subgaussian_concentration} to finish our proof of \eqref{eq:gconc1}. 
\end{proof}
Equally important, we require a concentration statement for our product of exponentiated Gaussians construction to show that our construction has moments small enough to approximate Haar. 
\begin{lem}[Trace concentration for $e^{i\vG_1\theta}e^{i\vG_2\theta}$] For two Gaussian Unitary Ensembles of $N$-dimensional matrices $\vG_1$, $\vG_2$,
\begin{align*}
    \Pr[|\btr[\left(e^{i\vG_1\theta}e^{i\vG_2\theta}\right)^p]-\BE \btr[\left(e^{i\vG_1\theta}e^{i\vG_2\theta}\right)^p]| \geq t]\leq \exp(-\frac{Nt^2}{4p^2\theta^2}).
\end{align*}
\label{lem:W_concentration}
\end{lem}
\begin{proof}
    This proof is very similar to the proof of Lemma \ref{lem:theta_concentration}. We again would like to use the Gaussian concentration inequality \cite{boucheron2013concentration}
\begin{align}\label{eq:subgaussian_concentration2}
    \Pr[|f(\vec{g}) - \mathop{\BE} f(\vec{g})|\geq t] \leq e^{-t^2/(2L^2)} \quad \text{for $L$-Lipschitz $f$ and normalized Gaussian vector $\vec{g}$.}
\end{align}
This time the vector $\vec{g}$ will have to contain all the elements of both $\vG_1$ and $\vG_2$, which again take the form
\begin{align*}
    \vG_j=\sum_i g_{i} \vA_{i,j} \quad \text{where}\quad g_{i}\stackrel{i.i.d.}{\sim}\mathcal{N}(0,1),\;b=1,2.
\end{align*}
Recall that $\vec{g}_1$ and $\vec{g}_2$ are $N^2$-dimensional vectors, so $\vec{g} = (\vec{g}_1,\vec{g}_2)$ is defined as the $2N^2$-dimensional vector such that the first $N^2$ entries are $\vec{g}_1$ and the next $N^2$ entries are $\vec{g}_2$. 
From the proof of Lemma \ref{lem:theta_concentration}, we recall
\begin{align}\label{eq:eixtheta_concentration}
    \norm{\e^{\ri\vX\theta}-\e^{\ri\vY\theta}}\leq \frac{\theta}{\sqrt{N}}\norm{\vx-\vy}_{2}.
\end{align}
Also, for generic operators $\vX_1$, $\vX_2$, $\vY_1$ and $\vY_2$,
\begin{align}\label{eq:bd_on_product}
    \norm{\vX_1\vX_2 - \vY_1\vY_2} &= \norm{\vX_1\vX_2-\vY_1\vX_2+\vY_1\vX_2-\vY_1\vY_2}\\
    &\leq \norm{(\vX_1-\vY_1)\vX_2} +\norm{\vY_1(\vX_2-\vY_2)}\tag{by triangle inequality}\\
    &\leq \norm{\vX_1-\vY_1}\norm{\vX_2} +\norm{\vY_1}\norm{\vX_2-\vY_2} \tag{by submultiplicativity of operator norm},
\end{align}
Applied to our particular product, we find
\begin{align*}
    \norm{\e^{\ri\vX_1\theta}\e^{\ri\vX_2\theta} - \e^{\ri\vY_1\theta}\e^{\ri\vY_2}} &\leq \norm{\e^{\ri\vX_1\theta}-\e^{\ri\vY_1\theta}}\norm{\e^{\ri\vX_2\theta}} + \norm{\e^{\ri\vY_1\theta}}\norm{\e^{\ri\vX_2\theta}-\e^{\ri\vY_2\theta}}\tag{by eq \ref{eq:bd_on_product}}\\
    &= \norm{\e^{\ri\vX_1\theta}-\e^{\ri\vY_1\theta}} + \norm{\e^{\ri\vX_2\theta}-\e^{\ri\vY_2\theta}} \\
    &\leq \frac{\theta}{\sqrt{N}}\left(\norm{\vx_1-\vy_1}_{2}+\norm{\vx_2-\vy_2}_{2}\right)\tag{by eq \ref{eq:eixtheta_concentration}}\\
    &=\theta\sqrt{\frac{2}{N}}\norm{\vx-\vy}_{2}\tag{by the AM-GM inequality    $a+b
    \leq \sqrt{2}\sqrt{a^2+b^2}$.},
\end{align*}
where $\vx$ is the vector that is $\vx_1$ for the first $N^2$ entries and $\vx_2$ for the next $N^2$ entries,and similarly for $\vy$.
Then recall from the proof of Lemma \ref{lem:theta_concentration} that the normalized trace is 1-Lipschitz on the operator norm, and also that if $\vX$ and $\vY$ are unitary,
\begin{align*}
    \norm{\vX^p-\vY^p} &= \norm{\vX^p-\vX^{p-1}\vY+\vX^{p-1}\vY-...+\vX\vY^{p-1}-\vY^p}\\
    &\leq\norm{\vX^{p-1}}\norm{\vX-\vY}+\norm{\vX^{p-2}}\norm{\vX-\vY}\norm{\vY}+\cdots+\norm{\vX-\vY}\norm{\vY^{p-1}}\\
    &=p\norm{\vX-\vY}\tag{by unitarity $\norm{\vX}=\norm{\vY}=1$}.
\end{align*}
Since compositions of functions have multiplicative Lipschitz constants, we conclude that
\begin{align*}
    |\btr{\left(\e^{\ri\vX_1\theta}\e^{\ri\vX_2\theta}\right)^p}-\btr{\left(\e^{\ri\vY_1\theta}\e^{\ri\vY_2\theta}\right)^p}|\leq p\theta\sqrt{\frac{2}{N}}\norm{\vx-\vy}_{2}.
\end{align*}
We substitute $L=p\theta\sqrt{2/N}$ into equation~\eqref{eq:subgaussian_concentration2} to finish our proof.
\end{proof}

\begin{prop} \label{prop:many_unitary_concentration} (Corollary 6 of \cite{meckes2013spectral}.) Let $\vU_1 \in \unitary(N_1), \vU_2 \in \unitary(N_2), \ldots, \vU_m \in \unitary(N_m)$ be independent Haar unitary random matrices. Equip $\unitary(N_1) \times \unitary(N_2) \times \cdots \times \unitary(N_m)$ with the $\ell_2$-sum metric
\begin{equation}
    \| (\vX_1,\ldots,\vX_m) \|_{2,2} = \sqrt{\sum_j \| \vX_j \|_2^2}.
\end{equation}
Let $N = \min(N_1,\ldots,N_m)$.
Then
\begin{equation}
    \Pr \left( \left| f(\vU_1,\ldots,\vU_m) - \Expect f(\vU_1, \ldots, \vU_m) \right| \geq \delta \right)
    \leq e^{-\frac{N\delta^2}{12L^2}},
\end{equation}
where $L$ is the Lipschitz constant of $f$.
\end{prop}
The following table collects Lipschitz constants with respect to the $\ell_2$ and $\ell_2$-sum metrics. All matrices are $N \times N$. The symbol $\vM$ denotes an arbitrary matrix, $\vH$ a Hermitian matrix, and $\vU$ a unitary matrix.
\begin{center}
    \begin{tabular}{|c|c|}
        \hline
        Function & Lipschitz constant \\
        \hline 
      $\gl(N) \ni \vM \mapsto \btr{\vM}$   & $N^{-1/2}$  \\
       $\herm(N) \ni  \vH \mapsto e^{i\vH \theta}$ & $\theta$ \\
       $\unitary(N) \ni \vU \mapsto \vU\vD\vU^\dagger$ & $2 \| \vD \|_{op}$ \\
        $\unitary(N) \ni  \vU \mapsto \vU^p$ & $p$ \\
       $\gl(N)^{m} \ni (\vM_1,\ldots,\vM_m) \mapsto \sum_{j=1}^m \vM_j$ & $m^{1/2}$ \\
       ${\unitary(N)^2} \ni (\vU_1,\vU_2) \mapsto \vU_1 \vU_2$ & $\sqrt{2}$ \\
       \hline
    \end{tabular}
\end{center}
Most of them are proved in appendix \ref{append:Lipschitz_Bounds}. For the sum, consider a deviation $(\Delta_1,\ldots,\Delta_m)$.
\begin{align*}
    \left\| \sum_{j=1}^m (\vM_j + \Delta_j ) - \sum_{j=1}^m \vM_j \right\|_2
    \leq \sum_{j=1}^m \| \Delta_j \|_2 
    \leq \sqrt{m} \sqrt{\sum_{j=1}^m  \| \Delta_j \|_2^2}
    = \| (\Delta_1,\ldots,\Delta_m) \|_{2,2}
\end{align*}
where the second inequality follows from Cauchy-Schwarz.

\begin{cor}[Moment concentration for $e^{i\theta \sum_j \vU_j \vD \vU_j^\dag}$]\label{cor:moment_concentration_UDU}
    Let $\mu$ be the Haar measure over the unitary group $\unitary(N)$ and suppose $\vD \in \gl(N)$ is diagonal. 
    \begin{align*}
        \Pr \left[\abs{\btr \left(e^{i\frac{\theta}{\sqrt{m}} \sum_{j=1}^m \vU_j \vD \vU_j^\dag}\right)^p - 
        {\Expect }\btr\left(e^{i\frac{\theta}{\sqrt{m}}  \sum_{j=1}^m \vU_j \vD \vU_j^\dag}\right)^p} \geq \delta \right] \leq 2 \exp{-\frac{N^2 \delta^2 }{ 48 \, p^2 \theta^2  \lnorm{\vD}_{op}^2}}.
    \end{align*}
\end{cor}
\begin{proof}
    The Lipschitz constant of the function in question can be evaluated by composing entries from the table. The result is that the Lipschitz constant $L$ is bounded above by $2  p \theta \| D \|_{op} / \sqrt{N}$. It suffices to plug that value into \autoref{prop:many_unitary_concentration}, acquiring a factor of 2 in front of the exponential from bounding deviations from both sides.
\end{proof}

\begin{cor}[Moment concentration for $\vW$]
\label{cor:moment_concentration_UDUVDV}
    Suppose that $\vU_1,\ldots,\vU_m$ and $\vU_1',\ldots,\vU_m'$ are independent Haar-random $\unitary(N)$ matrices and that $\vD \in \gl(N)$ is diagonal. 
    \begin{multline*}
        \Pr\left[\abs{\btr\left(e^{i\frac{\theta}{\sqrt{m}} \sum_{j=1}^m \vU_j \vD \vU_j^\dag} e^{i \frac{\theta}{\sqrt{m}} \sum_{j=1}^m \vU_j' \vD \vU_j'^\dag}\right)^p  - 
        {\Expect }\btr\left(e^{i\frac{\theta}{\sqrt{m}}  \sum_{j=1}^m \vU_j \vD \vU_j^\dag} e^{i\frac{\theta}{\sqrt{m}} \sum_{j=1}^m \vU_j' \vD \vU_j'^\dag}\right)^p} \geq \delta \right] \\ \leq 2 \exp{-\frac{N^2 \delta^2 }{ 96 \, p^2 \theta^2  \lnorm{\vD}_{op}^2}}.
    \end{multline*}
\end{cor}
\begin{proof}
    The proof is essentially the same as that of \autoref{cor:moment_concentration_UDU} except that the Lipschitz constant acquires an extra $\sqrt{2}$ factor from multiplying the two exponentials. 
\end{proof}

\section{Lipschitz Bounds}\label{append:Lipschitz_Bounds}
\begin{lem}\label{lem:exp_lipschitz}
    The function $\BR \ni x \mapsto e^{i x} \in S^1$ is 1-Lipschitz.
\end{lem}
\begin{proof}
    Without loss of generality, suppose that $a \geq b \in \BR$. Then,
    \begin{align*}
        \abs{e^{ia} - e^{i b}} = \labs{\int_b^a -ie^{it}\diff t} \leq  \int_b^a \labs{-ie^{it}} \diff t 
        = a-b.
    \end{align*}
\end{proof}

\begin{lem}\label{lem:sin_lipschitz}
    The functions $\BR \ni x \mapsto \cos{x} \in S^1$ and $\BR \ni x \mapsto \sin{x} \in S^1$ are 1-Lipschitz.
\end{lem}
\begin{proof}
    Without loss of generality, suppose that $a \geq b \in \BR$. Then,
    \begin{align*}
        \abs{\cos{a} - \cos{b} }= \labs{\int_b^a \sin{t} ~\diff t} \leq  \int_b^a \labs{\sin{t}}  \diff t \leq \int_b^a \diff t = a-b.
    \end{align*}
    For $\sin{x}$, one can just swap the $\cos$ and $\sin$ to complete the proof.
\end{proof}
\begin{lem}[Lipschitz-like bound]\label{lem:lipschitz-like}
    Let  $\vX_1$ and $\vX_2$ be unitary matrices, and $\vM_i=\vX_i-(\btr{\vX_i})\vI$. Then for
    \begin{align*}
    f(\vD_1,\vD_2) = \BE_{\vU,\vV\leftarrow\mu} \btr\L[\left(\vU\vD_1\vU^\dagger\vV\vD_2\vV^\dagger\right)^p\R],
    \end{align*}
    we have that
    \begin{align*}
        \labs{f(\vX_1,\vX_2)-f(\vM_1,\vM_2)}\leq p\left(|\btr{\vX_1}|+|\btr{\vX_2}|+|\btr{\vX_1}||\btr{\vX_2}|\right).
    \end{align*}
\end{lem}
\begin{proof}[Proof of~\autoref{lem:lipschitz-like}]
    Recall that from the proof of Lemma \ref{lem:W_concentration} in appendix \ref{append:Proofs_concentration}, we've shown that
    \begin{align*}
        |\btr{\vX^p}-\btr{\vY^p}|\leq p\norm{\vX-\vY}.
    \end{align*}
    Moreover, we see
    \begin{align*}
        &\norm{\vD_1\vU^\dagger\vV\vD_2\vV^\dagger\vU - \vM_1\vU^\dagger\vV\vM_2\vV^\dagger\vU} \\
        &\quad\leq \norm{\vD_1\vU^\dagger\vV\vD_2 - \vM_1\vU^\dagger\vV\vM_2}\\
        &\quad =\norm{\vD_1\vU^\dagger\vV\vD_2 -\vD_1\vU^\dagger\vV\vM_2 + \vD_1\vU^\dagger\vV\vM_2- \vM_1\vU^\dagger\vV\vM_2}\\
        &\quad \leq \norm{\vD_1\vU^\dagger\vV}\norm{\vD_2 -\vM_2}+\norm{\vD_1-\vM_1}\norm{\vU^\dagger\vV\vM_2}\\
        &\quad =\norm{\vD_2 -\vM_2}+\norm{\vD_1-\vM_1}\norm{\vM_2}\\
        &\quad =\norm{\vI\btr{\vD_2}}+\norm{\vI\btr{\vD_1}}\norm{\vD_2-\vI\btr{\vD_2}}\\
        &\quad \leq|\btr{\vD_2}|+|\btr{\vD_1}|+|\btr{\vD_2}||\btr{\vD_1}|
    \end{align*}
    Together, this gives us
    \begin{align*}
        \labs{\btr[\left(\vU\vD_1\vU^\dagger\vV\vD_2\vV^\dagger\right)^p]-\btr[\left(\vU\vM_1\vU^\dagger\vV\vM_2\vV^\dagger\right)^p]}\leq p (|\btr{\vD_2}|+|\btr{\vD_1}|+|\btr{\vD_2}||\btr{\vD_1}|)
    \end{align*}
    Then, taking the expectation value, we find
    \begin{align*}
        \labs{f(\vD_1,\vD_2)-f(\vM_1,\vM_2)}&\leq\BE_{\vU,\vV\leftarrow\mu}\labs{\btr[\left(\vU\vD_1\vU^\dagger\vV\vD_2\vV^\dagger\right)^p]-\btr[\left(\vU\vM_1\vU^\dagger\vV\vM_2\vV^\dagger\right)^p]}\\
        &\leq p \L(|\btr{\vD_2}|+|\btr{\vD_1}|+|\btr{\vD_2}||\btr{\vD_1}|\R)
    \end{align*}
\end{proof}

\begin{lem}\label{lem:lipschitz_mat_exp}
    The matrix function $\text{Herm}(N) \ni \vX \mapsto e^{i\vX\theta} \in \unitary(N)$ is $\theta$-Lipschitz with respect to any unitarily invariant norm $\vertiii{\cdot}$.
\end{lem}
\begin{proof}
    
From Duhamel's Formula (\autoref{prop:duhamel}) we have:
\begin{align*}
    e^{i\vX\theta} &= e^{i\vY\theta} + \int^{\theta}_{0} e^{i\vY(t-s)}\ri(\vX-\vY)e^{i\vX s} \rd s. 
\end{align*}
and
\begin{align*}
    \vertiii{e^{i\vX\theta}-e^{i\vY\theta}} &= \vertiii{\int^{\theta}_{0} e^{i\vY(t-s)}\ri(\vX-\vY)e^{i\vX s}}\rd s\\
    &\leq \int^{\theta}_{0} \vertiii{e^{i\vY(t-s)}\ri(\vX-\vY)e^{i\vX s}}\rd s\tag{by triangle inequality}\\
    &\leq \int^{\theta}_{0} \vertiii{\ri(\vX-\vY)} \rd s= \theta \vertiii{\vX-\vY} \tag{by unitary invariance}.
\end{align*}
\end{proof}

\begin{lem}\label{lem:lipschitz_UDU}
    The function $\unitary(N) \ni U \mapsto UDU^\dag \in \BC^{N\times N}$ is $2  \ltup{\norm{ \vD}_{op}}$-Lipschitz with respect to the Frobenius norm, $\lnorm{\cdot}_2$.   
\end{lem}
\begin{proof}
     Consider a deviation $\Delta \in \BC^{N\times N}$ from any $\vU \in \unitary(N)$ such that $\vU + \Delta \in \unitary(N)$.
    \begin{align*}
        \lnorm{(\vU+ \Delta) \vD (\vU+ \Delta)^\dag -\vU \vD \vU^\dag }_2 &= \norm{\vU \vD \Delta^\dag + \Delta \vD ( \vU+\Delta)^\dag }_2 \\
        &\leq 2 \norm{\Delta}_2 \norm{ \vD}_{op}, 
    \end{align*}
    using triangle inequality, unitary invariance to drop $\vU$ and $(\vU +\Delta)$, and sub-multiplicativity.
\end{proof}

\begin{lem}\label{lem:lipschitz_trace}
The function $\vX \mapsto \btr(\vX)$ on $N \times N$ matrices has Lipschitz constant $1/\sqrt{N}$ with respect to $\|\cdot\|_2$.
\end{lem}
\begin{proof}
    The normalized trace is $N^{-1/2}$-Lipschitz with respect to the Frobenius norm because
    \begin{align*}
    \frac{1}{N} \left| \tr X - \tr Y \right|
    &\leq  \frac{1}{N} \max_{\|M\|_\infty \leq 1} \tr\left[ M(X-Y )\right] \\
    &=  \frac{1}{N} \| X - Y \|_1 \\
    &\leq \frac{1}{\sqrt{N}} \| X - Y \|_2 \leq \| X - Y \|_{op}.
    \end{align*}
    The last inequality follows from Cauchy-Schwarz.    
\end{proof}
\begin{lem}\label{lem:lipschitz_tr_mat_exp}
    The matrix function $\herm(N) \ni H \mapsto \btr(\e^{i\vH\theta}) \in \BC$ is $\theta/\sqrt{N}$-Lipschitz with respect to the Frobenius norm, $\lnorm{\cdot}_2$, and $\theta$-Lipschitz with respect to the operator norm.
\end{lem}
\begin{proof}

    By multiplying the Lipschitz constants of lemmas \ref{lem:lipschitz_mat_exp} and \ref{lem:lipschitz_trace}, we get a Lipschitz constant of $\theta/\sqrt{N}$ for the composition of the two functions.
\end{proof}

\section{Weingarten Calculus}\label{app:weingarten}
This section serves as a short introduction to the techniques of Weingarten calculus and wiring diagrams used in \autoref{sec:products_of_gaussians}.

When we want to analyze the expectation values of quantities involving Haar random unitaries, we often use the \emph{Weingarten calculus}. Specifically, \cite{collins2006integration} showed that an integral of $q$ copies of $N$-dimensional Haar random $\vU$ and $\vU^\dagger$, i.e.
\begin{align*}
\int \vU_{i_1j_1}...\vU_{i_qj_q}\vU^\dagger_{j'_1i'_1}...\vU^\dagger_{j'_qi'_q}  \rd U, 
\end{align*}
is evaluated via
\begin{align}\label{eq:wgsum}
\sum_{\sigma,\ta\vU \in \mathcal{S}(q)} \delta_\sigma(\vec{i},\vec{i'})\delta_\tau(\vec{j},\vec{j'})\text{Wg}(\sigma\tau^{-1},N).
\end{align}
where
\begin{align*}
\delta_\sigma(\vec{i},\vec{i'})=\prod_{s=1}^q \delta_{i_si'_{\sigma(s)}}=\delta_{i_1 i'_{\sigma(1)}}...\delta_{i_q i_{\sigma(q)}}
\end{align*}
Note that the indices of the unitary matrices are written out explicitly here for clarity. In practice, these indices will be contracted with other operators or between themselves. Note also that integrals over $\vU^{\otimes k}\otimes\vU^{\dagger \otimes k'}$ for $k\neq k'$ are uniformly zero.

Let's understand eq. \ref{eq:wgsum} piece by piece.

\subsection{Permutations and Wiring diagrams}
The sum in eq. \ref{eq:wgsum} is over all permutations of $q$. The $\sigma$ permutes $i$ indices via the first $q$ $\delta$-functions in the sum, while the $\tau$ permutes $j$ indices via the last $q$ $\delta$-functions. Intuitively, these permutations appear in the Weingarten calculus because integrating over the Haar measure creates strict symmetrization. The action of these permutation terms is often illustrated using \emph{wiring diagrams}.

Let's consider the case of 
\begin{align*}
\int \vU_{i_1j_1}\vU_{i_2j_2}\vU^\dagger_{j'_1i'_1}\vU^\dagger_{j'_2i'_2}  \rd U.
\end{align*}
We would begin illustrating the terms of such an integral with a wiring diagram like so:
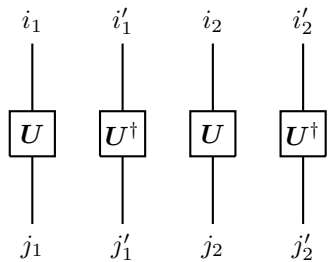
\begin{figure}[H]
    \centering
    \begin{tikzpicture}[scale=0.6]
    
    \draw[thick] (-3.5,-0.5) -- (-3.5,0.5) -- (-2.5,0.5) -- (-2.5,-0.5) -- (-3.5,-0.5);
    \draw[thick] (-1.5,-0.5) -- (-1.5,0.5) -- (-0.5,0.5) -- (-0.5,-0.5) -- (-1.5,-0.5);
    \draw[thick] (1.5,-0.5) -- (1.5,0.5) -- (0.5,0.5) -- (0.5,-0.5) -- (1.5,-0.5);
    \draw[thick] (3.5,-0.5) -- (3.5,0.5) -- (2.5,0.5) -- (2.5,-0.5) -- (3.5,-0.5);
    \node at (-3,0) {$\vU$};
    \node at (-1,0) {$\vU^\dagger$};
    \node at (1,0) {$\vU$};
    \node at (3,0) {$\vU^\dagger$};
    
    \draw[thick] (-3,-0.5) -- (-3,-2);
    \draw[thick] (-3,0.5) -- (-3,2);
    \draw[thick] (3,-0.5) -- (3,-2);
    \draw[thick] (3,0.5) -- (3,2);
    \draw[thick] (-1,-0.5) -- (-1,-2);
    \draw[thick] (-1,0.5) -- (-1,2);
    \draw[thick] (1,-0.5) -- (1,-2);
    \draw[thick] (1,0.5) -- (1,2);

    \node at (-3,-2.5) {$j_1$};
    \node at (-1,-2.5) {$j_1'$};
    \node at (1,-2.5) {$j_2$};
    \node at (3,-2.5) {$j_2'$};
    \node at (-3,2.5) {$i_1$};
    \node at (-1,2.5) {$i_1'$};
    \node at (1,2.5) {$i_2$};
    \node at (3,2.5) {$i_2'$};

    \end{tikzpicture}
\end{figure}
All the top wires are labeled by the $i$ and $i'$ indices, alternating, while the bottom wires are labeled by the $j$ and $j'$ indices. Also, we keep the $n$th index and the $n$th prime index next to each other. The $\sigma$ $\delta$-function terms, then, are various contractions of the top indices with each other. The lines can only go from $i$ to $i'$, so for our $q=2$ case, the possible top contractions are 
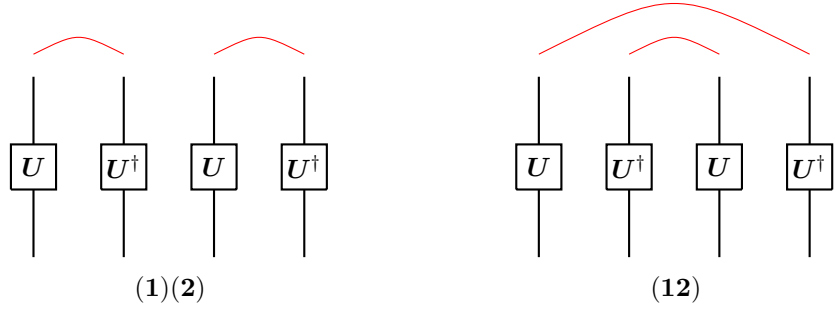
\begin{figure}[H]
    \centering
    \begin{subfigure}{0.4\textwidth}
    \centering
    \begin{tikzpicture}[scale=0.6]
    \draw[thick] (-3.5,-0.5) -- (-3.5,0.5) -- (-2.5,0.5) -- (-2.5,-0.5) -- (-3.5,-0.5);
    \draw[thick] (-1.5,-0.5) -- (-1.5,0.5) -- (-0.5,0.5) -- (-0.5,-0.5) -- (-1.5,-0.5);
    \draw[thick] (1.5,-0.5) -- (1.5,0.5) -- (0.5,0.5) -- (0.5,-0.5) -- (1.5,-0.5);
    \draw[thick] (3.5,-0.5) -- (3.5,0.5) -- (2.5,0.5) -- (2.5,-0.5) -- (3.5,-0.5);
    \node at (-3,0) {$\vU$};
    \node at (-1,0) {$\vU^\dagger$};
    \node at (1,0) {$\vU$};
    \node at (3,0) {$\vU^\dagger$};
    
    \draw[thick] (-3,-0.5) -- (-3,-2);
    \draw[thick] (-3,0.5) -- (-3,2);
    \draw[thick] (3,-0.5) -- (3,-2);
    \draw[thick] (3,0.5) -- (3,2);
    \draw[thick] (-1,-0.5) -- (-1,-2);
    \draw[thick] (-1,0.5) -- (-1,2);
    \draw[thick] (1,-0.5) -- (1,-2);
    \draw[thick] (1,0.5) -- (1,2);

    \draw[red] (-3,2.5) .. controls (-2,3) .. (-1,2.5);
    \draw[red] (3,2.5) .. controls (2,3) .. (1,2.5);

    \end{tikzpicture}
    \caption*{$\mathbf{(1)(2)}$}
    \end{subfigure}
    \begin{subfigure}{0.4\textwidth}
    \centering
    \begin{tikzpicture}[scale=0.6]
    \draw[thick] (-3.5,-0.5) -- (-3.5,0.5) -- (-2.5,0.5) -- (-2.5,-0.5) -- (-3.5,-0.5);
    \draw[thick] (-1.5,-0.5) -- (-1.5,0.5) -- (-0.5,0.5) -- (-0.5,-0.5) -- (-1.5,-0.5);
    \draw[thick] (1.5,-0.5) -- (1.5,0.5) -- (0.5,0.5) -- (0.5,-0.5) -- (1.5,-0.5);
    \draw[thick] (3.5,-0.5) -- (3.5,0.5) -- (2.5,0.5) -- (2.5,-0.5) -- (3.5,-0.5);
    \node at (-3,0) {$\vU$};
    \node at (-1,0) {$\vU^\dagger$};
    \node at (1,0) {$\vU$};
    \node at (3,0) {$\vU^\dagger$};
    
    \draw[thick] (-3,-0.5) -- (-3,-2);
    \draw[thick] (-3,0.5) -- (-3,2);
    \draw[thick] (3,-0.5) -- (3,-2);
    \draw[thick] (3,0.5) -- (3,2);
    \draw[thick] (-1,-0.5) -- (-1,-2);
    \draw[thick] (-1,0.5) -- (-1,2);
    \draw[thick] (1,-0.5) -- (1,-2);
    \draw[thick] (1,0.5) -- (1,2);

    \draw[red] (-3,2.5) .. controls (0,4) .. (3,2.5);
    \draw[red] (-1,2.5) .. controls (0,3) .. (1,2.5);

    \end{tikzpicture}
    \caption*{$\mathbf{(12)}$}
    \end{subfigure}
\end{figure}
Note that each wiring diagram is labeled with the appropriate cycle notation for the corresponding permutation. Similarly, the $\tau$ $\delta$-function are various contractions of the bottom indices, and the $q=2$ possibilities are
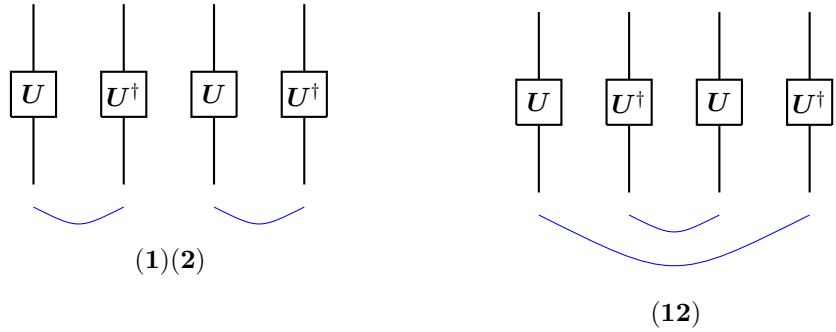
\begin{figure}[H]
    \centering
    \begin{subfigure}{0.4\textwidth}
    \centering
    \begin{tikzpicture}[scale=0.6]
    \draw[thick] (-3.5,-0.5) -- (-3.5,0.5) -- (-2.5,0.5) -- (-2.5,-0.5) -- (-3.5,-0.5);
    \draw[thick] (-1.5,-0.5) -- (-1.5,0.5) -- (-0.5,0.5) -- (-0.5,-0.5) -- (-1.5,-0.5);
    \draw[thick] (1.5,-0.5) -- (1.5,0.5) -- (0.5,0.5) -- (0.5,-0.5) -- (1.5,-0.5);
    \draw[thick] (3.5,-0.5) -- (3.5,0.5) -- (2.5,0.5) -- (2.5,-0.5) -- (3.5,-0.5);
    \node at (-3,0) {$\vU$};
    \node at (-1,0) {$\vU^\dagger$};
    \node at (1,0) {$\vU$};
    \node at (3,0) {$\vU^\dagger$};
    
    \draw[thick] (-3,-0.5) -- (-3,-2);
    \draw[thick] (-3,0.5) -- (-3,2);
    \draw[thick] (3,-0.5) -- (3,-2);
    \draw[thick] (3,0.5) -- (3,2);
    \draw[thick] (-1,-0.5) -- (-1,-2);
    \draw[thick] (-1,0.5) -- (-1,2);
    \draw[thick] (1,-0.5) -- (1,-2);
    \draw[thick] (1,0.5) -- (1,2);

    \draw[blue] (-3,-2.5) .. controls (-2,-3) .. (-1,-2.5);
    \draw[blue] (3,-2.5) .. controls (2,-3) .. (1,-2.5);
    \end{tikzpicture}
    \caption*{$\mathbf{(1)(2)}$}
    \end{subfigure}
    \begin{subfigure}[t]{0.4\textwidth}
    \centering
    \begin{tikzpicture}[scale=0.6]
    \draw[thick] (-3.5,-0.5) -- (-3.5,0.5) -- (-2.5,0.5) -- (-2.5,-0.5) -- (-3.5,-0.5);
    \draw[thick] (-1.5,-0.5) -- (-1.5,0.5) -- (-0.5,0.5) -- (-0.5,-0.5) -- (-1.5,-0.5);
    \draw[thick] (1.5,-0.5) -- (1.5,0.5) -- (0.5,0.5) -- (0.5,-0.5) -- (1.5,-0.5);
    \draw[thick] (3.5,-0.5) -- (3.5,0.5) -- (2.5,0.5) -- (2.5,-0.5) -- (3.5,-0.5);
    \node at (-3,0) {$\vU$};
    \node at (-1,0) {$\vU^\dagger$};
    \node at (1,0) {$\vU$};
    \node at (3,0) {$\vU^\dagger$};
    
    \draw[thick] (-3,-0.5) -- (-3,-2);
    \draw[thick] (-3,0.5) -- (-3,2);
    \draw[thick] (3,-0.5) -- (3,-2);
    \draw[thick] (3,0.5) -- (3,2);
    \draw[thick] (-1,-0.5) -- (-1,-2);
    \draw[thick] (-1,0.5) -- (-1,2);
    \draw[thick] (1,-0.5) -- (1,-2);
    \draw[thick] (1,0.5) -- (1,2);

    \draw[blue] (-3,-2.5) .. controls (0,-4) .. (3,-2.5);
    \draw[blue] (-1,-2.5) .. controls (0,-3) .. (1,-2.5);

    \end{tikzpicture}
    \caption*{$\mathbf{(12)}$}
    \end{subfigure}
    \caption*{}
\end{figure}
For each term in the Weingarten formula, we choose a top contraction and a bottom contraction, which gives four total possible terms for $q=2$. In general, though, there may be other operators inserted between the unitary operations. For example, if you have a traceless operator $O_{i_1i'_1}$ in your integral, all terms with the top contraction of $(1)(2)$ will disappear. In practice, it helps to analyze which contractions will survive before further evaluation.

\subsection{Weingarten Functions}
Each term of permutation sum is then weighted by the appropriate \emph{Weingarten function} on $\pi=\sigma\tau^{-1}$. The Weingarten function is a map from $\mathcal{S}(q)$ and to rational functions on $N$. It can be evaluated via 
\begin{align} \label{eq:wg_from_partitions}
\text{Wg}(\pi, N) = \frac{1}{q!^2}\sum_{\lambda \vdash q} \frac{\chi^\lambda(1)^2\chi^\lambda(\pi)}{s_{\lambda,N}(1)}.
\end{align}
The sum is over all integer partitions $\lambda$ of $q$.\footnote{For $q\geq N$, we also require the length of $\lambda$ to be less than $N$, but we won't run into that case for our purposes.} Recall that such partitions can be used to label the conjugacy classes and hence, irreducible representations of $\mathcal{S}(q).$ Here $\chi^\lambda(\pi)$ represents the character of the element $\pi\in\mathcal{S}(q)$ in the irrep labeled by $\lambda$. In particular, $\chi^\lambda(1)$ gives the dimension of the irrep $\lambda$. These characters can be calculated via the Murnaghan–Nakayama rule using border strip Young tableaux. In the denominator is the Schur polynomial for $\lambda$, evaluated on the partition $1^q$. It can be shown to equal 
\begin{align*}
s_\lambda(1,...,1)=\prod_{1\leq i < j \leq N} \frac{\lambda_i-\lambda_j+j-i}{j-i}.
\end{align*}
In this formula, the partition $\lambda$ is formulated as a vector that begins with the partition elements of $\lambda$ in decreasing order and then is filled out with zeros to be length $N$. Each $\lambda_i$ is an entry of this extended vector. The zeros of the Schur polynomial $s_{\lambda,N}(1)$ must be a subset of the integers $\{-(q-1),...,q-1\}$ since $\abs{\lambda_i-\lambda_N-1}\leq \abs{\lambda_i-1}\leq q-1.$ Each Schur polynomial is at most degree $q$.
Since these Schur polynomials are in the denominator, the possible poles on the Weingarten function must be $\{-(q-1),...,q-1\}$, and the Weingarten function is a rational function of degree at most $-q$. Because each Schur polynomial has different poles, each pole $i\in[-q+1,q-1]$ in the Weingarten function can be order $n$ for $n(n+\abs{i})\leq q$ \cite{procesi2021note}.

Notice that the denominator of the entire Weingarten function is naively independent of the permutation $\pi$ the Weingarten function is applied to; it depends only on $N$ and $q.$
However, upon simplification, the order of each pole may be dependent on $\pi$, since different permutations can generator numerators that cancel with some poles. The placement and multiplicity of these poles are important for our use of the Markov-type inequality.

Note that the Weingarten functions are class functions and hence are defined by the cycle shapes of their argument $\pi$.  In general, \cite{collins2006integration} shows that the $N\rightarrow\infty$ asymptotics are given by
\begin{align} \label{eq:weingarten_asymptotics_N}
\text{Wg}(\pi)= N^{-2q+\text{cycles}(\pi)}\prod_i (-1)^{|C_i|-1}\mathrm{Cat}_{|C_i|-1} + 
\sum_{d \leq -2q+\text{cycles}-2} w_{\pi,d} N^d.
\end{align}
where $|C_i|$ is the length of the $i$th cycle in $\pi$ and $\mathrm{Cat}_n$ are the Catalan numbers, as defined by $\mathrm{Cat}_n=\frac{1}{n+1}{\binom{2n}{n}}$. 
For the $q=2$ example, the Weingarten functions are
\begin{align*}
\text{Wg}(1^2)=\frac{1}{N^2-1} && \text{Wg}(2)=\frac{-1}{N(N^2-1)}.
\end{align*}
Here, we notate the inputs based on their cycle shapes. For $q=2$, $1^2$ uniquely refers to $\pi=(1)(2)$, and $2$ uniquely refers to $\pi=(12)$. Diagrammatically, then, we have
\begin{equation*}
\int \vU_{i_1j_1}\vU_{i_2j_2}\vU^\dagger_{j'_1i'_1}\vU^\dagger_{j'_2i'_2}  dU=\text{Wg}(1^2)\left(
\vcenter{\hbox{
\begin{tikzpicture}[scale=0.2]
    \draw[thick] (-3.5,-0.5) -- (-3.5,0.5) -- (-2.5,0.5) -- (-2.5,-0.5) -- (-3.5,-0.5);
    \draw[thick] (-1.5,-0.5) -- (-1.5,0.5) -- (-0.5,0.5) -- (-0.5,-0.5) -- (-1.5,-0.5);
    \draw[thick] (1.5,-0.5) -- (1.5,0.5) -- (0.5,0.5) -- (0.5,-0.5) -- (1.5,-0.5);
    \draw[thick] (3.5,-0.5) -- (3.5,0.5) -- (2.5,0.5) -- (2.5,-0.5) -- (3.5,-0.5);
    
    \draw[thick] (-3,-0.5) -- (-3,-2);
    \draw[thick] (-3,0.5) -- (-3,2);
    \draw[thick] (3,-0.5) -- (3,-2);
    \draw[thick] (3,0.5) -- (3,2);
    \draw[thick] (-1,-0.5) -- (-1,-2);
    \draw[thick] (-1,0.5) -- (-1,2);
    \draw[thick] (1,-0.5) -- (1,-2);
    \draw[thick] (1,0.5) -- (1,2);

    \draw[blue] (-3,-2.5) .. controls (0,-4) .. (3,-2.5);
    \draw[blue] (-1,-2.5) .. controls (0,-3) .. (1,-2.5);
    \draw[red] (-3,2.5) .. controls (0,4) .. (3,2.5);
    \draw[red] (-1,2.5) .. controls (0,3) .. (1,2.5);

\end{tikzpicture}}}
+
\vcenter{\hbox{
\begin{tikzpicture}[scale=0.2]
    \draw[thick] (-3.5,-0.5) -- (-3.5,0.5) -- (-2.5,0.5) -- (-2.5,-0.5) -- (-3.5,-0.5);
    \draw[thick] (-1.5,-0.5) -- (-1.5,0.5) -- (-0.5,0.5) -- (-0.5,-0.5) -- (-1.5,-0.5);
    \draw[thick] (1.5,-0.5) -- (1.5,0.5) -- (0.5,0.5) -- (0.5,-0.5) -- (1.5,-0.5);
    \draw[thick] (3.5,-0.5) -- (3.5,0.5) -- (2.5,0.5) -- (2.5,-0.5) -- (3.5,-0.5);
    
    \draw[thick] (-3,-0.5) -- (-3,-2);
    \draw[thick] (-3,0.5) -- (-3,2);
    \draw[thick] (3,-0.5) -- (3,-2);
    \draw[thick] (3,0.5) -- (3,2);
    \draw[thick] (-1,-0.5) -- (-1,-2);
    \draw[thick] (-1,0.5) -- (-1,2);
    \draw[thick] (1,-0.5) -- (1,-2);
    \draw[thick] (1,0.5) -- (1,2);

    \draw[blue] (-3,-2.5) .. controls (-2,-3) .. (-1,-2.5);
    \draw[blue] (3,-2.5) .. controls (2,-3) .. (1,-2.5);
    \draw[red] (-3,2.5) .. controls (-2,3) .. (-1,2.5);
    \draw[red] (3,2.5) .. controls (2,3) .. (1,2.5);
\end{tikzpicture}}}\right)
+\text{Wg}(2)\left(
\vcenter{\hbox{
\begin{tikzpicture}[scale=0.2]
    \draw[thick] (-3.5,-0.5) -- (-3.5,0.5) -- (-2.5,0.5) -- (-2.5,-0.5) -- (-3.5,-0.5);
    \draw[thick] (-1.5,-0.5) -- (-1.5,0.5) -- (-0.5,0.5) -- (-0.5,-0.5) -- (-1.5,-0.5);
    \draw[thick] (1.5,-0.5) -- (1.5,0.5) -- (0.5,0.5) -- (0.5,-0.5) -- (1.5,-0.5);
    \draw[thick] (3.5,-0.5) -- (3.5,0.5) -- (2.5,0.5) -- (2.5,-0.5) -- (3.5,-0.5);
    
    \draw[thick] (-3,-0.5) -- (-3,-2);
    \draw[thick] (-3,0.5) -- (-3,2);
    \draw[thick] (3,-0.5) -- (3,-2);
    \draw[thick] (3,0.5) -- (3,2);
    \draw[thick] (-1,-0.5) -- (-1,-2);
    \draw[thick] (-1,0.5) -- (-1,2);
    \draw[thick] (1,-0.5) -- (1,-2);
    \draw[thick] (1,0.5) -- (1,2);

    \draw[blue] (-3,-2.5) .. controls (0,-4) .. (3,-2.5);
    \draw[blue] (-1,-2.5) .. controls (0,-3) .. (1,-2.5);
    \draw[red] (-3,2.5) .. controls (-2,3) .. (-1,2.5);
    \draw[red] (3,2.5) .. controls (2,3) .. (1,2.5);
\end{tikzpicture}}}
+
\vcenter{\hbox{
\begin{tikzpicture}[scale=0.2]
    \draw[thick] (-3.5,-0.5) -- (-3.5,0.5) -- (-2.5,0.5) -- (-2.5,-0.5) -- (-3.5,-0.5);
    \draw[thick] (-1.5,-0.5) -- (-1.5,0.5) -- (-0.5,0.5) -- (-0.5,-0.5) -- (-1.5,-0.5);
    \draw[thick] (1.5,-0.5) -- (1.5,0.5) -- (0.5,0.5) -- (0.5,-0.5) -- (1.5,-0.5);
    \draw[thick] (3.5,-0.5) -- (3.5,0.5) -- (2.5,0.5) -- (2.5,-0.5) -- (3.5,-0.5);
    
    \draw[thick] (-3,-0.5) -- (-3,-2);
    \draw[thick] (-3,0.5) -- (-3,2);
    \draw[thick] (3,-0.5) -- (3,-2);
    \draw[thick] (3,0.5) -- (3,2);
    \draw[thick] (-1,-0.5) -- (-1,-2);
    \draw[thick] (-1,0.5) -- (-1,2);
    \draw[thick] (1,-0.5) -- (1,-2);
    \draw[thick] (1,0.5) -- (1,2);

    \draw[blue] (-3,-2.5) .. controls (-2,-3) .. (-1,-2.5);
    \draw[blue] (3,-2.5) .. controls (2,-3) .. (1,-2.5);
    \draw[red] (-3,2.5) .. controls (0,4) .. (3,2.5);
    \draw[red] (-1,2.5) .. controls (0,3) .. (1,2.5);
\end{tikzpicture}
}}
\right)
\end{equation*}
\begin{equation}\label{eq:q2wiring}
\;\;\;\;\;\;=\frac{1}{N^2-1}\left(
\vcenter{\hbox{
\begin{tikzpicture}[scale=0.2]
    \draw[thick] (-3.5,-0.5) -- (-3.5,0.5) -- (-2.5,0.5) -- (-2.5,-0.5) -- (-3.5,-0.5);
    \draw[thick] (-1.5,-0.5) -- (-1.5,0.5) -- (-0.5,0.5) -- (-0.5,-0.5) -- (-1.5,-0.5);
    \draw[thick] (1.5,-0.5) -- (1.5,0.5) -- (0.5,0.5) -- (0.5,-0.5) -- (1.5,-0.5);
    \draw[thick] (3.5,-0.5) -- (3.5,0.5) -- (2.5,0.5) -- (2.5,-0.5) -- (3.5,-0.5);
    
    \draw[thick] (-3,-0.5) -- (-3,-2);
    \draw[thick] (-3,0.5) -- (-3,2);
    \draw[thick] (3,-0.5) -- (3,-2);
    \draw[thick] (3,0.5) -- (3,2);
    \draw[thick] (-1,-0.5) -- (-1,-2);
    \draw[thick] (-1,0.5) -- (-1,2);
    \draw[thick] (1,-0.5) -- (1,-2);
    \draw[thick] (1,0.5) -- (1,2);

    \draw[blue] (-3,-2.5) .. controls (0,-4) .. (3,-2.5);
    \draw[blue] (-1,-2.5) .. controls (0,-3) .. (1,-2.5);
    \draw[red] (-3,2.5) .. controls (0,4) .. (3,2.5);
    \draw[red] (-1,2.5) .. controls (0,3) .. (1,2.5);

\end{tikzpicture}}}
+
\vcenter{\hbox{
\begin{tikzpicture}[scale=0.2]
    \draw[thick] (-3.5,-0.5) -- (-3.5,0.5) -- (-2.5,0.5) -- (-2.5,-0.5) -- (-3.5,-0.5);
    \draw[thick] (-1.5,-0.5) -- (-1.5,0.5) -- (-0.5,0.5) -- (-0.5,-0.5) -- (-1.5,-0.5);
    \draw[thick] (1.5,-0.5) -- (1.5,0.5) -- (0.5,0.5) -- (0.5,-0.5) -- (1.5,-0.5);
    \draw[thick] (3.5,-0.5) -- (3.5,0.5) -- (2.5,0.5) -- (2.5,-0.5) -- (3.5,-0.5);
    
    \draw[thick] (-3,-0.5) -- (-3,-2);
    \draw[thick] (-3,0.5) -- (-3,2);
    \draw[thick] (3,-0.5) -- (3,-2);
    \draw[thick] (3,0.5) -- (3,2);
    \draw[thick] (-1,-0.5) -- (-1,-2);
    \draw[thick] (-1,0.5) -- (-1,2);
    \draw[thick] (1,-0.5) -- (1,-2);
    \draw[thick] (1,0.5) -- (1,2);

    \draw[blue] (-3,-2.5) .. controls (-2,-3) .. (-1,-2.5);
    \draw[blue] (3,-2.5) .. controls (2,-3) .. (1,-2.5);
    \draw[red] (-3,2.5) .. controls (-2,3) .. (-1,2.5);
    \draw[red] (3,2.5) .. controls (2,3) .. (1,2.5);
\end{tikzpicture}}}
-\frac{1}{N}
\vcenter{\hbox{
\begin{tikzpicture}[scale=0.2]
    \draw[thick] (-3.5,-0.5) -- (-3.5,0.5) -- (-2.5,0.5) -- (-2.5,-0.5) -- (-3.5,-0.5);
    \draw[thick] (-1.5,-0.5) -- (-1.5,0.5) -- (-0.5,0.5) -- (-0.5,-0.5) -- (-1.5,-0.5);
    \draw[thick] (1.5,-0.5) -- (1.5,0.5) -- (0.5,0.5) -- (0.5,-0.5) -- (1.5,-0.5);
    \draw[thick] (3.5,-0.5) -- (3.5,0.5) -- (2.5,0.5) -- (2.5,-0.5) -- (3.5,-0.5);
    
    \draw[thick] (-3,-0.5) -- (-3,-2);
    \draw[thick] (-3,0.5) -- (-3,2);
    \draw[thick] (3,-0.5) -- (3,-2);
    \draw[thick] (3,0.5) -- (3,2);
    \draw[thick] (-1,-0.5) -- (-1,-2);
    \draw[thick] (-1,0.5) -- (-1,2);
    \draw[thick] (1,-0.5) -- (1,-2);
    \draw[thick] (1,0.5) -- (1,2);

    \draw[blue] (-3,-2.5) .. controls (0,-4) .. (3,-2.5);
    \draw[blue] (-1,-2.5) .. controls (0,-3) .. (1,-2.5);
    \draw[red] (-3,2.5) .. controls (-2,3) .. (-1,2.5);
    \draw[red] (3,2.5) .. controls (2,3) .. (1,2.5);
\end{tikzpicture}}}
-\frac{1}{N}
\vcenter{\hbox{
\begin{tikzpicture}[scale=0.2]
    \draw[thick] (-3.5,-0.5) -- (-3.5,0.5) -- (-2.5,0.5) -- (-2.5,-0.5) -- (-3.5,-0.5);
    \draw[thick] (-1.5,-0.5) -- (-1.5,0.5) -- (-0.5,0.5) -- (-0.5,-0.5) -- (-1.5,-0.5);
    \draw[thick] (1.5,-0.5) -- (1.5,0.5) -- (0.5,0.5) -- (0.5,-0.5) -- (1.5,-0.5);
    \draw[thick] (3.5,-0.5) -- (3.5,0.5) -- (2.5,0.5) -- (2.5,-0.5) -- (3.5,-0.5);
    
    \draw[thick] (-3,-0.5) -- (-3,-2);
    \draw[thick] (-3,0.5) -- (-3,2);
    \draw[thick] (3,-0.5) -- (3,-2);
    \draw[thick] (3,0.5) -- (3,2);
    \draw[thick] (-1,-0.5) -- (-1,-2);
    \draw[thick] (-1,0.5) -- (-1,2);
    \draw[thick] (1,-0.5) -- (1,-2);
    \draw[thick] (1,0.5) -- (1,2);

    \draw[blue] (-3,-2.5) .. controls (-2,-3) .. (-1,-2.5);
    \draw[blue] (3,-2.5) .. controls (2,-3) .. (1,-2.5);
    \draw[red] (-3,2.5) .. controls (0,4) .. (3,2.5);
    \draw[red] (-1,2.5) .. controls (0,3) .. (1,2.5);
\end{tikzpicture}
}}
\right)
\end{equation}
Notice that for each wiring diagram, we can count the number of cycles in the permutation $\pi=\sigma\tau^{-1}$ just by tracing through the lines. This is important since the Weingarten function weighting gives priority to wiring diagrams with the maximum possible number of cycles. As such, the largest possible contributing terms are always wiring diagrams that have the same top and bottom contraction. That being said, the insertion of say, traceless operators, could make such matching impossible, so sometimes the saddle point is at a lower order.

\subsection{Coarse Asymptotics of the Weingarten Function}

If $q$ is sufficiently slowly growing, then \eqref{eq:weingarten_asymptotics_N} can be used to prove a rigorous bound on the size of the corrections to the leading order term, a fact that we use in \autoref{sec:lindeberg_spectrum}. Note that the correction is slightly larger asymptotically than the size of the leading order correction in \eqref{eq:weingarten_asymptotics_N}.
\begin{lem}\label{lem:weingarten_asymptotics}
If $N \geq \Omega( q^{4q} )$, then
\begin{equation}
\text{Wg}(\pi)= N^{-2q+\text{cycles}(\pi)}\prod_i (-1)^{|C_i|-1}\mathrm{Cat}_{|C_i|-1} + \tilde{\mathcal{O}}\left(N^{-2q+\text{cycles}(\pi)-3/2}\right).
\end{equation}
\end{lem}
\begin{proof}
    Write
    \begin{align}
    \text{Wg}(\pi,N) &= \sum_d w_{\pi,d} N^d \quad \text{and} \quad
    \frac{\chi^\lambda(1)^2\chi^\lambda(\pi)}{s_{\lambda,N}(1)}
        = \sum_d \omega_{\pi,d,\lambda} N^d \label{eq:wg_summands}
    \end{align}
    as a Laurent series in $N$ around $N=\infty$, treating $q$ as fixed. We know from \eqref{eq:wg_from_partitions} that
    \begin{equation}
        w_{\pi,d} =
         \frac{1}{q!^2} \sum_d \sum_{\lambda \vdash q} \omega_{\pi,d,\lambda} N^d
   \end{equation}
   and wish to bound
   \begin{align*}
       \left| \text{Wg}(\pi, N) - N^{d_{\max}} \prod_i (-1)^{|C_i|-1}\mathrm{Cat}_{|C_i|-1} \right| 
       &=   \left| \frac{1}{q!^2}\sum_{d \leq d_{\max}-2} \sum_{\lambda \vdash q}  \omega_{\pi,d,\lambda} N^d \right| \\
       &\leq  \frac{1}{q!^2}\sum_{d \leq d_{\max}-2} \sum_{\lambda \vdash q}  \left| \omega_{\pi,d,\lambda} \right| N^d \\
       &\leq  \frac{e^{\CO(\sqrt{q})}}{(q!)^2}\sum_{d \leq d_{\max}-2} N^d \max_{\lambda \vdash  q} \, \left| \omega_{\pi,d,\lambda} \right| ,
   \end{align*}
   where $d_{\max} = -2q+\text{cycles}(\pi)$. The last line follows from the fact that the number of partitions of $q$ is bounded above by $e^{\CO(\sqrt{q})}$~\cite{van2001course}.
    Completing the analysis will require getting a better understanding \eqref{eq:wg_summands}.
    Since the representation labeled by $\lambda$ can be taken to be unitary, $|\chi^\lambda(\pi)| \leq \chi^\lambda(1)$,
    which is the dimension of the irreducible representation labeled by $\lambda$. It is at most $\sqrt{q!}$ because all irreducible representations occur with their dimension as multiplicity in the regular representation, whose dimension is $q!$. 
    The $N$ dependence appears via the factor
    \begin{align*}
        \frac{1}{s_\lambda(1,...,1)}
            =\prod_{1\leq i < j \leq N} 
            \frac{j-i}{\lambda_i-\lambda_j+j-i}.
    \end{align*}
    When $i$ and $j$ are both greater than the number of parts $|\lambda|$ of $\lambda$, the corresponding factors are equal to 1 and can be ignored. Now consider the part of the product when $i$ and $j$ are both less than or equal to $|\lambda|$. This contributes an $N$-independent factor that is always less than or equal to one because the numerator is less than or equal to the denominator for a given $i$ and $j$. The $N$ dependence therefore all comes from the remaining factors, when $i \leq |\lambda|$ and $j > |\lambda|$. Note that in that case there is a polynomial $D_i(N,\lambda)$ with integer coefficients such that
    \begin{align}
        \prod_{|\lambda| < j \leq N} 
            \frac{j-i}{\lambda_i-\lambda_j+j-i}
       &= \frac{(|\lambda| +1 -i)(|\lambda|+2-i) \cdots (N-i)}{(\lambda_i + |\lambda|+1-i)(\lambda_i + |\lambda|+2-i)\cdots(\lambda_i +N-i)} \\
        &= \frac{ (|\lambda| +1 -i)(|\lambda|+2-i) \cdots (\lambda_i+|\lambda|-i)}{(N-i+1)(N-i+2)\cdots(N-i+\lambda_i)}  \\
        &=: \frac{ (|\lambda| +1 -i)(|\lambda|+2-i) \cdots (\lambda_i+|\lambda|-i)}{D_i(N,\lambda)} \label{eq:nice_ratio} \\
        &\leq (2q)^{\lambda_i -1} \frac{1}{D_i(N,\lambda)}.
    \end{align}
    The inequality holds because each of the $\lambda_i-1$ factors in the numerator of \eqref{eq:nice_ratio} can be at most $2q$. 
    Next, we apply \autoref{lem:expand_simple_poles} to expand $D_i(\lambda,N)^{-1} = \sum_{d_i=-1}^{-\infty} c_{i,d_i} N^{d_i}$ where
    \begin{equation}
        c_{i,d_i} = \begin{dcases} \sum_{k=1}^{\lambda_i} 
            \frac{(i-k)^{-d_i-1}}{\prod_{1 \leq j \leq \lambda_i; j \neq k} (k-j)} & \text{if $i > \lambda_i$} \\
            \sum_{k=1, k\neq i}^{\lambda_i} 
            \frac{(i-k)^{-d_i-2}}{\prod_{1 \leq j \leq \lambda_i; j \not\in \{i,k\}}(k-j)} & \text{if $i \leq \lambda_i$}. \label{eq:error_coefs}
            \end{dcases}
    \end{equation}
    The reason that the case $i \leq \lambda_i$ has to be treated separately is that it will have a simple pole at the origin that we factor out before applying \autoref{lem:expand_simple_poles}. 
    Restoring the product over $i$, we find that
    \begin{align*}
        \frac{1}{\abs{s_\lambda(1,...,1)}}
        \leq \prod_{1 \leq i \leq |\lambda|} (2q)^{\lambda_i-1} \frac{1}{\abs{D_i(\lambda,N)}}.
    \end{align*}
    The coefficient of $N^d$ in the expansion of $s_\lambda(1,\ldots,1)^{-1}$ will come from summing over contributions from all sets of nonpositive integers $d_1 + d_2 + \cdots + d_{|\lambda|} = d$, of which there are fewer than $q! \cdot e^{\CO(\sqrt{|d|})}$ since we are counting \textit{ordered} partitions into at most $q$ parts.
    Each contribution is bounded above by
    \begin{equation}
        \prod_{i=1}^{|\lambda|} (2q)^{\lambda_i-1} |c_{i,d_i} |
        \leq  \prod_{i=1}^{|\lambda|} (2q)^{\lambda_i-1} \lambda_i q^{-d_i-1}
        \leq (2q)^q 3^q q^{|d|}
    \end{equation}
    since the absolute values of the denominators in \eqref{eq:error_coefs} are always at least 1 and the product of parts of a partition of $q$ is bounded above by $3^q$~\cite{van2001course}.
    Therefore, any nonleading coefficient of $N^d$ in the expansion of the Weingarten function will have a coefficient bounded above by
    \begin{align*}
        \frac{e^{\CO(\sqrt{q})}}{(q!)^2} \left|\chi^\lambda(1)\right|^2 \left|\chi^\lambda(\pi)\right| q^q e^{\CO(\sqrt{d})} (2q)^q 3^q q^{|d|}
        &\leq e^{\CO(\sqrt{q})} (q!)^{-1/2} q^q e^{\CO(\sqrt{d})} (2q)^q 3^q q^{|d|} \\
        &\leq q^{3q/2}  q^{|d|}  e^{\CO(q+\sqrt{d})}.
    \end{align*}
    Now we will determine the magnitude $\Delta$ of the total correction to the leading asymptotics of $\text{Wg}$  by summing over the correction terms of $N$-degree $d \leq d_{\max} - 2$. There is a constant $C$ such that
    \begin{align*}
        \Delta 
        &\leq \sum_{d \leq d_{\max}-2} q^{3q/2} e^{\CO(q+\sqrt{|d|})} q^{-d} N^d \tag{Note that $d<0$ so $q^{|d|} = q^{-d}$.} \\
        &\leq q^{3q/2} e^{\CO(q)} \left( \frac{N}{Cq} \right)^{d_{\max}-2} \sum_{d  \leq 0} \left( \frac{N}{Cq} \right)^d \\
        &\leq q^{3q/2} e^{\CO(q)} \left( \frac{N}{q} \right)^{d_{\max}-2} \frac{1}{1 - Cq/N} \\
        &\leq q^{3q/2} e^{\CO(q)} \left( \frac{N}{q} \right)^{d_{\max}-2},
    \end{align*}
    using that $Cq/N < 1$ for sufficiently large $N$. But by hypothesis, $N \geq \Omega(q^{4q})$, so $q^{3q/2} = \CO(N^{3/8})$, $e^{\CO(q)} = \CO(N^{1/8})$ 
    and $q = \CO( \log N )$. Therefore,
    \begin{equation}
        \Delta 
        \leq \CO\left( \left( \frac{N}{\log N} \right)^{d_{\max} - 2  + 1/2} \right)
        \leq \tilde{\CO}\left( N^{d_{\max}- 3/2 } \right),
    \end{equation}
    which completes the proof of the lemma when combined with \eqref{eq:weingarten_asymptotics_N}.
\end{proof}

\begin{lem}\label{lem:expand_simple_poles}
Let 
\begin{equation}
    f(z) = \frac{1}{\prod_{j=1}^m (z - a_j)}
\end{equation}
for distinct $a_j \neq 0$. Then in a neighborhood of $z=\infty$, there is a convergent expansion  $f(z) = \sum_{d=-1}^{-\infty} c_d z^{d}$ with
\begin{equation}
    c_d = \sum_{k=1}^m \frac{a_k^{-d-1}}{\prod_{j\neq k} (a_k-a_j)}.
\end{equation}
\end{lem}
 
\begin{proof}
    It suffices to evaluate the residues
    \begin{equation}
        \text{Res}(f,a_k) 
        = \lim_{z\rightarrow a_j} (z-a_j) f(z)
        = \frac{1}{\prod_{j\neq k} (a_k - a_j)}.
    \end{equation}
    Then the partial fraction expansion of $f(z)$ is
    \begin{equation}
        \sum_{k=1}^m \frac{\text{Res}(f,a_k)}{z-a_k}.
    \end{equation}
    Expanding around $z=\infty$ gives
    \begin{equation}
        \frac{1}{z-a_k} = \frac{1}{z} \sum_{d=0}^\infty 
            \left( \frac{a_k}{z} \right)^d
    \end{equation}
    so
    \begin{equation}
        f(z) = \sum_{k=1}^m \text{Res}(f,a_k) \frac{1}{z} \sum_{d=0}^\infty 
            \left( \frac{a_k}{z} \right)^d,
    \end{equation}
    which is what we set out to prove.
\end{proof}

\section{Markov-type Inequality for Rational Functions}\label{app:markov_inequality}
The proof of Theorem \ref{lem:UDUVDV_expected_moments} relies heavily on the machinery of a Markov-type inequality. The original Markov inequality applies only to bounded polynomial functions on an interval, but we require a version applicable to rational functions with bounds on only a discrete number of points, Lemma \ref{lem:main_markov_inequality}. In this appendix we prove the Lemma. Specifically, we adapt the following version of a sharp Markov-type inequality for rational functions, which we state without proof, to deal with functions that are bounded on a discrete number of points.
\begin{lem}[Sharp Markov-type inequality for rational functions]
\label{lem:markov_rational_inequality}
(See \cite{akturk2016sharp,rusak1979rational})
Consider an algebraic fraction of the form
\begin{align*}
    r_n(x) &= \frac{p_n(x)}{\sqrt{t_{2n}(x)}}
\end{align*}
such that 
\begin{itemize}
    \item $t_{2n}(x) = \prod^{2n}_{k=1} (1+a_k x)$ for all $ -1\leq x\leq 1,$
    \item $p_n(x)$ is an algebraic polynomial of degree at most n with complex (real) coefficients,
    \item $a_k$ are either real or pairwise complex conjugate,
    \item $|a_k| <1$ for all $1\leq k \leq 2n$, and 
    \item $|r_n(x)|\leq 1$ for all $ -1\leq x\leq1.$ 
\end{itemize}
Then $r_n(x)$ is bounded as 
\begin{align*}
    |r'_n(x)|\leq 
    \begin{cases}
        \frac{\lambda_n(x)}{\sqrt{1-x^2}}, & x_1 \leq x \leq x_n,\\
        |m'_n(x)|, & -1\leq x\leq x_1, x_n \leq x \leq 1,
    \end{cases}
\end{align*}
Here \{$x_k$\}, $-1 < x_1 < ... < x_n < 1$ are zeros of the cosine fraction 
\begin{align*}
m_n(x)=\cos{\frac{1}{2}\sum^{2n}_{k=1}\arccos \left(\frac{x+a_k}{1+a_k x}\right)},
\end{align*}
and $\lambda_n(x)$ is defined as
\begin{align*}
    \lambda_n(x)=\frac{1}{2}\sum^{2n}_{k=1} \frac{\sqrt{1-a_k^2}}{1+a_k x}
\end{align*}
\end{lem}
The following is our adaptation; in order to compensate for no longer having a bound that holds on the entire interval $[-1,1]$, we need some function $c(x)$ that bounds $\frac{\lambda_n(x)}{\sqrt{1-x^2}}$ and $|m'_n(x)|$ to some degree. This means the bound is no longer sharp, unless $c(x)$ is sharp. This method is inspired by Theorem 3.3 in \cite{nisan1994degree}.
\begin{cor}[Markov-type inequality for rational functions with discrete bounds]\label{lem:markov_rational_inequality_discrete}
    If an algebraic fraction $r_n(x)$ satisfies all the conditions of Lemma \ref{lem:markov_rational_inequality}, except that it only satisfies the bound
    \begin{align*}
        |r_{n}(x)|\leq 1
    \end{align*}
    for a set of discrete points in increasing order $-1=x^*_1<...<x^*_i=1$, and also there exists a function $c(x)$ such that 
    \begin{align*}
        c(x)&\geq    \begin{cases}
        \frac{\lambda_n(x)}{\sqrt{1-x^2}}, & x_1 \leq x \leq x_n,\\
        |m'_n(x)|, & -1\leq x\leq x_1, x_n \leq x \leq 1.
    \end{cases}
    \end{align*}
    where $\{x_i\}$, $\lambda_n(x)$, and $m_n(x)$ for $r_n(x)$ are defined as they are in Lemma \ref{lem:markov_rational_inequality}, then
    \begin{align*}
        |r'_n(x)|&\leq \left(1+\frac{c(x)I}{2-c(x)I}\right)
    \begin{cases}
        \frac{\lambda_n(x)}{\sqrt{1-x^2}}, & x\in [x_1,x_n] \cap  X,\\
        |m'_n(x)|, & x\in ([-1,x_1]\cup[x_n,1]) \cap X,
    \end{cases}\\
    X&:=\left\{x\in[-1,1]:1-\frac{c(x)I}{2}>0\right\}, \quad I = \sup_{\{x^*_i\}}{|x^*_i-x^*_{i+1}|}.
    \end{align*}
\end{cor}
\begin{proof}
    This proof is inspired by the proof of Theorem 3.3 in \cite{nisan1994degree}. Let $I\in\mathbb{R}$ be the largest distance between two adjacent discrete points where $r_n(x)$ obeys the bound, i.e.
    \begin{align*}
        I = \sup_{\{x^*_i\}}{|x^*_i-x^*_{i+1}|}.
    \end{align*}
    and suppose we know some $c$ as defined in the lemma statement. Let us take $c_0$ such that
    \begin{align*}
        c_0&=\sup_{[-1,1]} |r_n'(x)|.
    \end{align*}
    We know $c_0$ exists: $r_n(x)$ is a rational functions with no poles in $[-1,1]$, so $r_n'(x)$ is a continuous function on a compact domain, and hence its image must be compact as well. From these definitions, we know that
    \begin{align*}
        |r_n(x)|\leq 1+\frac{c_0I}{2}, \quad -1\leq x\leq 1.
    \end{align*}
    We now define the algebraic fraction $s_n(x)$ such that
    \begin{align*}
        s_n(x)= \frac{r_n(x)}{1+\frac{c_0I}{2}}.
    \end{align*}
    Then $s_n(x)$ follows all the conditions of Lemma \ref{lem:markov_rational_inequality}, so 
    \begin{align*}
        |s_n'(x)|\leq 
        \begin{cases}
        \frac{\lambda_n(x)}{\sqrt{1-x^2}}, & x_1 \leq x \leq x_n,\\
        |m'_n(x)|, & -1\leq x\leq x_1, x_n \leq x \leq 1.
    \end{cases}
    \end{align*}
    Notice that $\lambda_n(x)$, $m_n(x)$, $\{x_i\}$ for $s_n(x)$ are constructed with the same $a_k$ as those present in $r_n(x)$, so they are the same functions for both $s_n(x)$ and $r_n(x)$. Then we know 
    \begin{align*}
        |s_n'(x)|\leq c(x), \quad -1\leq x\leq 1.
    \end{align*} From this, given that $1-\frac{c(x)I}{2}> 0$ for $x\in X\subseteq [-1,1]$,
    \begin{align*}
    |r_n'(x)|&=|s_n'(x)(1+\frac{c_0I}{2})|\leq c(x) (1+\frac{c_0I}{2})\\
    \implies c_0&\leq c(x) (1+\frac{c_0I}{2})\\
    \implies c_0&\leq\frac{c(x)}{1-\frac{c(x)I}{2}},\quad x\in X.
    \end{align*}
    Now we could just conclude here for a bound on $|r'_n(x)|$, but the appending the functional forms from Lemma \ref{lem:markov_rational_inequality} is a tighter bound. As such,
    \begin{align*}
    |r_n'(x)|=|s_n'(x)(1+\frac{c_0I}{2})|&\leq \left(1+\frac{c_0I}{2}\right)
        \begin{cases}
        \frac{\lambda_n(x)}{\sqrt{1-x^2}}, & x_1 \leq x \leq x_n,\\
        |m'_n(x)|, & -1\leq x\leq x_1, x_n \leq x \leq 1.
    \end{cases}\\
    &\leq \left(1+\frac{c(x)I}{2-c(x)I}\right)
    \begin{cases}
        \frac{\lambda_n(x)}{\sqrt{1-x^2}}, & x\in [x_1,x_n] \cap  X,\\
        |m'_n(x)|, & x\in ([-1,x_1]\cup[x_n,1]) \cap X.
    \end{cases}
    \end{align*}
\end{proof}
Of course, the hard part of applying this discrete Markov-type inequality is finding $c(x)$. In some cases, one might a priori have a functional bound in mind, but here we develop a method to find a functional bound by looking at how the pole structure interacts in $\lambda_n(x).$
\begin{lem}[A $c(x)$ for sharp Markov-type inequality for rational functions with discrete bounds]\label{lem:c_x_markov_inequality}
    Suppose some algebraic fraction $r_n(x)$ that satisfies all the conditions of Lemma \ref{lem:markov_rational_inequality} except for the bound on $|r_n(x)|$, and define $\{x_i\}$, $\lambda_n(x)$, and $m_n(x)$ for $r_n(x)$ as in Lemma \ref{lem:markov_rational_inequality}. Let $c\in\mathbb{R}$ such that
    \begin{align*}
        c\geq \sup_{[-1,1]}\lambda_n(x).
    \end{align*}
    Then 
    \begin{align*}
        \frac{c}{\sqrt{1-x^2}}\geq    
        \begin{cases}
        \frac{\lambda_n(x)}{\sqrt{1-x^2}}, & x_1 \leq x \leq x_n,\\
        |m'_n(x)|, & -1\leq x\leq x_1, x_n \leq x \leq 1.
        \end{cases}
    \end{align*}
\end{lem}
\begin{proof}
    The first part of the piecewise function follows trivially since $[x_1,x_n]\in[-1,1]$. As for the second piece,
    \begin{align*}
    |m_n'(z)|&=\left|-\sin{\frac{1}{2}\sum^{2n}_{k=1}\arccos \left(\frac{x+a_k}{1+a_k x}\right)}\times \frac{1}{2}\sum^{2n}_{k=1}\frac{\sqrt{\frac{(a_k^2-1)(x^2-1)}{(a_kx+1)^2}}}{x^2-1}\right|\\
    &\leq\left|\frac{1}{2}\sum^{2n}_{k=1}\frac{\sqrt{\frac{(1-a_k^2)(1-x^2)}{(a_kx+1)^2}}}{1-x^2}\right|,\quad x\in[-1,1]\\
    &=\left|\frac{1}{2}\sum^{2n}_{k=1}\sqrt{\frac{(1-a_k^2)}{(a_kx+1)^2(1-x^2)}}\right|,\quad x\in[-1,1]\\
    &=\frac{1}{2\sqrt{1-x^2}}\left|\sum^{2n}_{k=1}\frac{\sqrt{(1-a_k^2)}}{a_kx+1}\right|,\quad x\in[-1,1]\\
    &=\frac{\lambda_n(x)}{\sqrt{1-x^2}},\quad x\in[-1,1].
    \end{align*}
\end{proof}
The combination of this lemma, corollary, and lemma, are what give us Lemma \ref{lem:main_markov_inequality}.

\section{Properties of \texorpdfstring{$J_1(x)$}{J1(x)}}\label{append:j1_props}
The first Bessel function of the first kind, $J_1(x)$, appears in the infinite $N$ limit of the expectation value of the trace moments of a single exponentiated Gaussian, $\BE\btr{(\e^{\ri\theta\vG})^p}$. To find the specific $\theta$ that satisfies the conditions of the moment problem, we need an upper bound on the value of $J_1(x).$ This appendix addresses such details.

From the standard source on Bessel functions properties \cite{watson1922treatise}, we use the following results:
\begin{lem}[Properties of $J_{\pm\nu}(x)$](See pgs 206-210 in \cite{watson1922treatise}) \label{lem:j_nu}
    Let $J_{\pm\nu}(x)$ be the $\nu$th Bessel function of the first kind. Then we can expand $J_{\pm\nu}(x)$ in the form
    \begin{align*}
        J_{\pm\nu}(x) = \sqrt{\frac{2}{\pi x}}\left[ \cos{(x\mp \frac{1}{2}\nu \pi-\frac{1}{4}\pi)}P(x,\nu) - \sin{(x\mp \frac{1}{2}\nu \pi-\frac{1}{4}\pi)}Q(x,\nu)\right],
    \end{align*}
    where $P(x,\nu)$, $Q(x,\nu)$ satisfy the following conditions:
    \begin{itemize}
        \item Both are analytic in $\nu$ and $x$.
        \item For integer $p\geq 0$ such that $2p>\nu-1/2$, the remainder after the $p$th term in the expansion of $P(x,\nu)$ does not exceed the $(p+1)$th term in absolute value.
        \item Similarly, for integer $q\geq 0$ such that $2q> \nu-3/2$, the remainder after $q$ terms in the expansion of $Q(x,\nu)$ does not exceed the $(q+1)$th term in absolute value
        \item For either $P(x,\nu)$ or $Q(x,\nu)$, the remainder has the same sign as the $p+1$th or $q+1$th terms, respectively.
    \end{itemize}
\end{lem}
Essentially, an upper bound can be estimated on $J_{\pm\nu}(x)$ by analyzing a finite number of terms of $P(x,\nu)$ and $Q(x,\nu)$.
\begin{lem}[Bound on $J_1(x)$] \label{lem:bound_j1}For $x\geq 1$,
\begin{align*}
    |J_1(x)|\leq \sqrt{\frac{22753}{8192\pi x}}<\sqrt{\frac{1}{x}}
\end{align*}
\end{lem}
\begin{proof}
    Suppose that for some constant $C$,
    \begin{align*}
        C \geq \left(P(x,1)^2+Q(x,1)^2\right).
    \end{align*}
    Then by analyzing the form of $J_1(x)$ in \autoref{lem:j_nu} and by Cauchy's inequality we have 
    \begin{align*}
        \sqrt{\frac{2C}{x\pi}} \geq \sqrt{\frac{2}{x\pi}\left(\cos^2(x-\frac{3}{4}\pi)+\sin^2(x-\frac{3}{4}\pi)\right)\left(P(x,1)^2+Q(x,1)^2\right)}\geq |J_1(x)|.
    \end{align*}
    Hence we just need to find such a constant $C$ using properties of $P(x,1)$ and $Q(x,1).$ For $\nu=1$, we see $p=1$ is the first integer that satisfies $2p>\nu-1/2$, and $q=0$ is the first integer that satisfies $2q\geq \nu-3/2$. Using either the integral formulas in \cite{watson1922treatise} or Mathematica, we can find that
    \begin{align*}
        P(x,1)&=1+\frac{15}{128x^2}-\frac{4725}{32768x^4}+\mathcal{O}\left(\frac{1}{x^6}\right),\\
        Q(x,1)&=\frac{3}{8x}-\frac{105}{1024x^3}+\mathcal{O}\left(\frac{1}{x^5}\right).
    \end{align*}
    Then based on \autoref{lem:j_nu}, we see the remainders are bounded such that
    \begin{align*}
        1+\frac{15}{128x^2}-\frac{4725}{32768x^4}< &P(x,1)<1+\frac{15}{128x^2},\\
        \frac{3}{8x}-\frac{105}{1024x^3}< &Q(x,1)<\frac{3}{8x}.
    \end{align*}
    For $x\geq 1$, we see that
    \begin{align*}
        \left(P(x,1)^2+Q(x,1)^2\right)&\leq \left(\left(\frac{3}{8x}\right)^2+\left(1+\frac{15}{128x^2}\right)^2\right)\\
        &\leq  \left(\left(\frac{3}{8}\right)^2+\left(1+\frac{15}{128}\right)^2\right) = 22753/16384 = 1.3887...
    \end{align*}
    Notice $22753/16384\times 2/\pi<1$. This completes the proof.
\end{proof}

\bibliographystyle{amsalpha}
\bibliography{ref}

\newcommand{\etalchar}[1]{$^{#1}$}
\providecommand{\bysame}{\leavevmode\hbox to3em{\hrulefill}\thinspace}
\providecommand{\MR}{\relax\ifhmode\unskip\space\fi MR }
\providecommand{\MRhref}[2]{%
  \href{http://www.ams.org/mathscinet-getitem?mr=#1}{#2}
}
\providecommand{\href}[2]{#2}
\begin{thebibliography}{BCHJ{\etalchar{+}}21}

\bibitem[ABG{\etalchar{+}}24]{allen2024approximate}
James Allen, Daniel Belkin, Soumik Ghoush, Christopher Kang, Sophia Lin, James
  Sud, Frederic Chong, Bill Fefferman, and Bryan Clark, \emph{Approximate
  t-design in general architectures}, 2024.

\bibitem[ABW09]{ambainis2009nonmalleable}
Andris Ambainis, Jan Bouda, and Andreas Winter, \emph{Nonmalleable encryption
  of quantum information}, Journal of Mathematical Physics \textbf{50} (2009),
  no.~4, 042106.

\bibitem[AE07]{ambainis2007quantum}
Andris Ambainis and Joseph Emerson, \emph{Quantum t-designs: t-wise
  independence in the quantum world}, Twenty-Second Annual IEEE Conference on
  Computational Complexity (CCC'07), IEEE, 2007, pp.~129--140.

\bibitem[AEH{\etalchar{+}}22]{akers2022black}
Chris Akers, Netta Engelhardt, Daniel Harlow, Geoff Penington, and Shreya
  Vardhan, \emph{The black hole interior from non-isometric codes and
  complexity}, 2022.

\bibitem[AKN98]{aharonov1998quantum}
Dorit Aharonov, Alexei Kitaev, and Noam Nisan, \emph{Quantum circuits with
  mixed states}, Proceedings of the thirtieth annual ACM symposium on Theory of
  computing, 1998, pp.~20--30.

\bibitem[AL16]{akturk2016sharp}
Mehmet~Ali Akturk and Alexey Lukashov, \emph{Sharp markov-type inequalities for
  rational functions on several intervals}, Journal of Mathematical Analysis
  and Applications \textbf{436} (2016), no.~2, 1017--1022.

\bibitem[AS04]{aaronson2004quantum}
Scott Aaronson and Yaoyun Shi, \emph{Quantum lower bounds for the collision and
  the element distinctness problems}, Journal of the ACM (JACM) \textbf{51}
  (2004), no.~4, 595--605.

\bibitem[Aub05]{aubrun2005sharp}
Guillaume Aubrun, \emph{A sharp small deviation inequality for the largest
  eigenvalue of a random matrix}, S{\'e}minaire de Probabilit{\'e}s XXXVIII
  (2005), 320--337.

\bibitem[BBC{\etalchar{+}}01]{beals2001quantum}
Robert Beals, Harry Buhrman, Richard Cleve, Michele Mosca, and Ronald De~Wolf,
  \emph{Quantum lower bounds by polynomials}, Journal of the ACM (JACM)
  \textbf{48} (2001), no.~4, 778--797.

\bibitem[BCC{\etalchar{+}}15]{berry2015simulating}
Dominic~W Berry, Andrew~M Childs, Richard Cleve, Robin Kothari, and Rolando~D
  Somma, \emph{Simulating hamiltonian dynamics with a truncated taylor series},
  Physical review letters \textbf{114} (2015), no.~9, 090502.

\bibitem[BCHJ{\etalchar{+}}21]{brandao2021models}
Fernando~GSL Brand{\~a}o, Wissam Chemissany, Nicholas Hunter-Jones, Richard
  Kueng, and John Preskill, \emph{Models of quantum complexity growth}, PRX
  Quantum \textbf{2} (2021), no.~3, 030316.

\bibitem[BHH16]{brandao2016local}
Fernando~GSL Brandao, Aram~W Harrow, and Micha{\l} Horodecki, \emph{Local
  random quantum circuits are approximate polynomial-designs}, Communications
  in Mathematical Physics \textbf{346} (2016), 397--434.

\bibitem[BLM13]{boucheron2013concentration}
S.~Boucheron, G.~Lugosi, and P.~Massart, \emph{Concentration inequalities: A
  nonasymptotic theory of independence}, OUP Oxford, 2013.

\bibitem[BNZZ19]{bannai2019explicit}
Eiichi Bannai, Mikio Nakahara, Da~Zhao, and Yan Zhu, \emph{On the explicit
  constructions of certain unitary t-designs}, Journal of Physics A:
  Mathematical and Theoretical \textbf{52} (2019), no.~49, 495301.

\bibitem[BS19]{brakerski2019pseudo}
Zvika Brakerski and Omri Shmueli, \emph{(pseudo) random quantum states with
  binary phase}, Theory of Cryptography Conference, Springer, 2019,
  pp.~229--250.

\bibitem[BSS01]{barnum2001quantum}
Howard Barnum, Michael Saks, and Mario Szegedy, \emph{Quantum decision trees
  and semidefinite programming.}, Tech. report, Los Alamos National Lab.(LANL),
  Los Alamos, NM (United States), 2001.

\bibitem[BWV08]{brown2008quantum}
Winton~G Brown, Yaakov~S Weinstein, and Lorenza Viola, \emph{Quantum
  pseudorandomness from cluster-state quantum computation}, Physical Review A
  \textbf{77} (2008), no.~4, 040303.

\bibitem[CDB{\etalchar{+}}23]{chen2023sparse}
Chi-Fang~(Anthony) Chen, Alexander~M Dalzell, Mario Berta, Fernando~GSL
  Brand{\~a}o, and Joel~A Tropp, \emph{Sparse random hamiltonians are quantumly
  easy}, 2023.

\bibitem[CGP19]{collins2019operator}
Beno{\^\i}t Collins, Alice Guionnet, and F{\'e}lix Parraud, \emph{On the
  operator norm of non-commutative polynomials in deterministic matrices and
  iid gue matrices}, 2019.

\bibitem[CHJ20]{cotler2020spectral}
Jordan Cotler and Nicholas Hunter-Jones, \emph{Spectral decoupling in many-body
  quantum chaos}, Journal of High Energy Physics \textbf{2020} (2020), no.~12,
  205.

\bibitem[CHJLY17]{cotler2017chaos}
Jordan Cotler, Nicholas Hunter-Jones, Junyu Liu, and Beni Yoshida, \emph{Chaos,
  complexity, and random matrices}, Journal of High Energy Physics
  \textbf{2017} (2017), no.~11, 1--60.

\bibitem[CLLW16]{cleve2015near}
Richard Cleve, Debbie~W. Leung, Li~Liu, and Chunhao Wang, \emph{Near-linear
  constructions of exact unitary 2-designs}, Quantum Inf. Comput. \textbf{16}
  (2016), no.~9, 721--756.

\bibitem[C{\'S}06]{collins2006integration}
Beno{\^\i}t Collins and Piotr {\'S}niady, \emph{Integration with respect to the
  haar measure on unitary, orthogonal and symplectic group}, Communications in
  Mathematical Physics \textbf{264} (2006), no.~3, 773--795.

\bibitem[DCEL09]{dankert2009exact}
Christoph Dankert, Richard Cleve, Joseph Emerson, and Etera Livine, \emph{Exact
  and approximate unitary 2-designs and their application to fidelity
  estimation}, Physical Review A \textbf{80} (2009), no.~1, 012304.

\bibitem[DE01]{Diaconis}
Persi Diaconis and Steven Evans, \emph{Linear functionals of eigenvalues of
  random matrices}, Transactions of the American Mathematical Society
  \textbf{353} (2001), 2615--2633.

\bibitem[DLT02]{divincenzo2002quantum}
David~P DiVincenzo, Debbie~W Leung, and Barbara~M Terhal, \emph{Quantum data
  hiding}, IEEE Transactions on Information Theory \textbf{48} (2002), no.~3,
  580--598.

\bibitem[EZ64]{ehlich1964}
H.~Ehlich and K.~Zeller, \emph{Schwankung von polynomen zwischen
  gitterpunkten}, Mathematische Zeitschrift \textbf{86} (1964), 41--44.

\bibitem[GNTT18]{gotze2018local}
Friedrich G{\"o}tze, Alexey Naumov, Alexander Tikhomirov, and Dmitry Timushev,
  \emph{{On the local semicircular law for Wigner ensembles}}, Bernoulli
  \textbf{24} (2018), no.~3, 2358 -- 2400.

\bibitem[GSLW19]{gilyen2019quantum}
Andr{\'a}s Gily{\'e}n, Yuan Su, Guang~Hao Low, and Nathan Wiebe, \emph{Quantum
  singular value transformation and beyond: exponential improvements for
  quantum matrix arithmetics}, Proceedings of the 51st Annual ACM SIGACT
  Symposium on Theory of Computing, 2019, pp.~193--204.

\bibitem[GT05]{gotze2005rate}
Friedrich G{\"o}tze and Alexander Tikhomirov, \emph{The rate of convergence for
  spectra of gue and lue matrix ensembles}, Central European Journal of
  Mathematics \textbf{3} (2005), no.~4, 666--704.

\bibitem[Haf22]{haferkamp2022random}
Jonas Haferkamp, \emph{Random quantum circuits are approximate unitary $ t
  $-designs in depth $ o (nt^{5+ o (1)})$}, Quantum \textbf{6} (2022), 795.

\bibitem[HHJ21]{haferkamp2021improved}
Jonas Haferkamp and Nicholas Hunter-Jones, \emph{Improved spectral gaps for
  random quantum circuits: large local dimensions and all-to-all interactions},
  Physical Review A \textbf{104} (2021), no.~2, 022417.

\bibitem[HHWY08]{hayden2008decoupling}
Patrick Hayden, Micha{\l} Horodecki, Andreas Winter, and Jon Yard, \emph{A
  decoupling approach to the quantum capacity}, Open Systems \& Information
  Dynamics \textbf{15} (2008), no.~01, 7--19.

\bibitem[HJ19]{hunter2019unitary}
Nicholas Hunter-Jones, \emph{Unitary designs from statistical mechanics in
  random quantum circuits}, 2019.

\bibitem[HKP20]{huang2020predicting}
Hsin-Yuan Huang, Richard Kueng, and John Preskill, \emph{Predicting many
  properties of a quantum system from very few measurements}, Nature Physics
  \textbf{16} (2020), no.~10, 1050--1057.

\bibitem[HL09a]{harrow2009efficient}
Aram~W Harrow and Richard~A Low, \emph{Efficient quantum tensor product
  expanders and k-designs}, International Workshop on Approximation Algorithms
  for Combinatorial Optimization, Springer, 2009, pp.~548--561.

\bibitem[HL09b]{harrow2009random}
\bysame, \emph{Random quantum circuits are approximate 2-designs},
  Communications in Mathematical Physics \textbf{291} (2009), 257--302.

\bibitem[HLT24]{haahpersonal}
Jeongwan Haah, Yunchao Liu, and Xinyu Tan, \emph{Personal communication},
  January 2024.

\bibitem[HM23]{harrow2023approximate}
Aram~W Harrow and Saeed Mehraban, \emph{Approximate unitary t-designs by short
  random quantum circuits using nearest-neighbor and long-range gates},
  Communications in Mathematical Physics (2023), 1--96.

\bibitem[HP07]{hayden2007black}
Patrick Hayden and John Preskill, \emph{Black holes as mirrors: quantum
  information in random subsystems}, Journal of high energy physics
  \textbf{2007} (2007), no.~09, 120.

\bibitem[Jof74]{Joffe}
A.~Joffe, \emph{{On a Set of Almost Deterministic $k$-Independent Random
  Variables}}, The Annals of Probability \textbf{2} (1974), no.~1, 161 -- 162.

\bibitem[KKS{\etalchar{+}}23]{kaposi2023constructing}
{\'A}goston Kaposi, Zolt{\'a}n Kolarovszki, Adrian Solymos, Tam{\'a}s Kozsik,
  and Zolt{\'a}n Zimbor{\'a}s, \emph{Constructing generalized unitary group
  designs}, International Conference on Computational Science, Springer, 2023,
  pp.~233--245.

\bibitem[KL17]{kimmel2017phase}
Shelby Kimmel and Yi-Kai Liu, \emph{Phase retrieval using unitary 2-designs},
  2017 International Conference on Sampling Theory and Applications (SampTA),
  2017, pp.~345--349.

\bibitem[KLR{\etalchar{+}}08]{knill2008randomized}
E.~Knill, D.~Leibfried, R.~Reichle, J.~Britton, R.~B. Blakestad, J.~D. Jost,
  C.~Langer, R.~Ozeri, S.~Seidelin, and D.~J. Wineland, \emph{Randomized
  benchmarking of quantum gates}, Phys. Rev. A \textbf{77} (2008), 012307.

\bibitem[Kut05]{kutin2005quantum}
Samuel Kutin, \emph{Quantum lower bound for the collision problem with small
  range}, Theory of Computing \textbf{1} (2005), no.~1, 29--36.

\bibitem[Lan93]{lang_inverse_mappings}
Serge Lang, \emph{Inverse mappings and differential equations}, pp.~360--384,
  Springer, 1993.

\bibitem[Lin22]{Lindeberg1922EineNH}
J.~W. Lindeberg, \emph{Eine neue herleitung des exponentialgesetzes in der
  wahrscheinlichkeitsrechnung}, Mathematische Zeitschrift \textbf{15} (1922),
  211--225.

\bibitem[Low10]{low2010pseudo}
Richard~A Low, \emph{Pseudo-randomness and learning in quantum computation},
  Ph.D. thesis, University of Bristol, UK, 2010, arXiv:1006.5227.

\bibitem[LSH{\etalchar{+}}13]{lashkari2013towards}
Nima Lashkari, Douglas Stanford, Matthew Hastings, Tobias Osborne, and Patrick
  Hayden, \emph{Towards the fast scrambling conjecture}, Journal of High Energy
  Physics \textbf{2013} (2013), no.~4, 1--33.

\bibitem[Mar89]{markov1889}
A.A. Markov, \emph{On a problem of d.i. mendeleev (russian)}, Zapishi Imp.
  Akad. Nauk \textbf{I12} (1889), 1--24.

\bibitem[Mar16]{markov1916}
V.A. Markov, \emph{Uber polynome die in einem gegebenen intervalle m{\"o}lichst
  wenig von null abweichen}, Math. Annalen \textbf{77} (1916), 213--258.

\bibitem[Mec19]{meckes2019random}
Elizabeth~S Meckes, \emph{The random matrix theory of the classical compact
  groups}, vol. 218, Cambridge University Press, 2019.

\bibitem[MGDM18]{mezher2018efficient}
Rawad Mezher, Joe Ghalbouni, Joseph Dgheim, and Damian Markham, \emph{Efficient
  quantum pseudorandomness with simple graph states}, Physical Review A
  \textbf{97} (2018), no.~2, 022333.

\bibitem[MHJ23]{mittal2023local}
Shivan Mittal and Nicholas Hunter-Jones, \emph{Local random quantum circuits
  form approximate designs on arbitrary architectures}, 2023.

\bibitem[MM13]{meckes2013spectral}
Elizabeth Meckes and Mark Meckes, \emph{{Spectral measures of powers of random
  matrices}}, Electronic Communications in Probability \textbf{18} (2013),
  no.~none, 1 -- 13.

\bibitem[MS12]{meckes2012concentration}
Mark Meckes and Stanis{\l}aw Szarek, \emph{Concentration for noncommutative
  polynomials in random matrices}, Proceedings of the American Mathematical
  Society \textbf{140} (2012), no.~5, 1803--1813.

\bibitem[MSS16]{maldacena2016bound}
Juan Maldacena, Stephen~H Shenker, and Douglas Stanford, \emph{A bound on
  chaos}, Journal of High Energy Physics \textbf{2016} (2016), no.~8, 1--17.

\bibitem[NHKW17]{nakata2017efficient}
Yoshifumi Nakata, Christoph Hirche, Masato Koashi, and Andreas Winter,
  \emph{Efficient quantum pseudorandomness with nearly time-independent
  hamiltonian dynamics}, Physical Review X \textbf{7} (2017), no.~2, 021006.

\bibitem[NKM14]{nakata2014generating}
Yoshifumi Nakata, Masato Koashi, and Mio Murao, \emph{Generating a state
  t-design by diagonal quantum circuits}, New Journal of Physics \textbf{16}
  (2014), no.~5, 053043.

\bibitem[NS94]{nisan1994degree}
Noam Nisan and Mario Szegedy, \emph{On the degree of boolean functions as real
  polynomials}, Computational complexity \textbf{4} (1994), 301--313.

\bibitem[NW99]{nayak1999quantum}
Ashwin Nayak and Felix Wu, \emph{The quantum query complexity of approximating
  the median and related statistics}, Proceedings of the thirty-first annual
  ACM symposium on Theory of computing, 1999, pp.~384--393.

\bibitem[OBK{\etalchar{+}}17]{onorati2017mixing}
Emilio Onorati, Oliver Buerschaper, Martin Kliesch, Winton Brown, Albert~H
  Werner, and Jens Eisert, \emph{Mixing properties of stochastic quantum
  hamiltonians}, Communications in Mathematical Physics \textbf{355} (2017),
  905--947.

\bibitem[OSP23]{o2023explicit}
Ryan O’Donnell, Rocco~A Servedio, and Pedro Paredes, \emph{Explicit
  orthogonal and unitary designs}, 2023 IEEE 64th Annual Symposium on
  Foundations of Computer Science (FOCS), IEEE, 2023, pp.~1240--1260.

\bibitem[Pro21]{procesi2021note}
Claudio Procesi, \emph{A note on the formanek weingarten function}, Note di
  Matematica \textbf{41} (2021), no.~1, 69--110.

\bibitem[Raz03]{razborov2003quantum}
Alexander~A Razborov, \emph{Quantum communication complexity of symmetric
  predicates}, Izvestiya: Mathematics \textbf{67} (2003), no.~1, 145.

\bibitem[RC66]{cheneyrivlin1966}
T.~J. Rivlin and E.~W. Cheney, \emph{A comparison of uniform approximations on
  an interval and a finite subset thereof}, SIAM Journal of Numerical Analysis
  \textbf{3} (1966), no.~2, 311--320.

\bibitem[Rus79]{rusak1979rational}
VN~Rusak, \emph{Rational functions as an apparatus of approximation}, Ph.D.
  thesis, Belorussian University, Minsk, 1979.

\bibitem[RY17]{roberts2017chaos}
Daniel~A. Roberts and Beni Yoshida, \emph{Chaos and complexity by design},
  Journal of High Energy Physics \textbf{2017} (2017), no.~4, 121.

\bibitem[Sch91]{schmudgen1991chapter11}
Konrad Schmüdgen, \emph{The moment problem on the unit circle}, ch.~11,
  pp.~257--279, Springer-Verlag, New York, NY, 1991.

\bibitem[SDTR13]{szehr2013decoupling}
Oleg Szehr, Frédéric Dupuis, Marco Tomamichel, and Renato Renner,
  \emph{Decoupling with unitary approximate two-designs}, New Journal of
  Physics \textbf{15} (2013), no.~5, 053022.

\bibitem[Spe20]{speicher2020lecture}
Roland Speicher, \emph{Lecture notes on" random matrices"}, 2020.

\bibitem[SS08]{sekino2008fast}
Yasuhiro Sekino and Leonard Susskind, \emph{Fast scramblers}, Journal of High
  Energy Physics \textbf{2008} (2008), no.~10, 065.

\bibitem[VLW01]{van2001course}
Jacobus~Hendricus Van~Lint and Richard~Michael Wilson, \emph{A course in
  combinatorics}, Cambridge university press, 2001.

\bibitem[VSW16]{voiculescu2016free}
D.V. Voiculescu, N.~Stammeier, and M.~Weber, \emph{Free probability and
  operator algebras}, M{\"u}nster Lectures in Mathematics, European
  Mathematical Society Publishing House, 2016.

\bibitem[Wat22]{watson1922treatise}
George~Neville Watson, \emph{A treatise on the theory of bessel functions},
  vol.~2, The University Press, 1922.

\bibitem[Web15]{webb2015clifford}
Zak Webb, \emph{The clifford group forms a unitary 3-design}, 2015.

\bibitem[Win12]{winkelbauer2012moments}
Andreas Winkelbauer, \emph{Moments and absolute moments of the normal
  distribution}, 2012.

\bibitem[Zhu17]{zhu2017multiqubit}
Huangjun Zhu, \emph{Multiqubit clifford groups are unitary 3-designs}, Physical
  Review A \textbf{96} (2017), no.~6, 062336.

\end{thebibliography}

\end{document}